%% file: thesis.tex
\documentclass{puthesis}
\usepackage{arabtex}
\usepackage[utf8]{inputenc}

\author{Houssam Yassin}
\adviser{Stephen M. Griffies}
\title{The Geostrophic Turbulence of Boundary Buoyancy Anomalies}
\abstract{Quasigeostrophic flows are induced by spatial variations in interior potential vorticity and boundary buoyancy. 
In the first part of this dissertation, we develop the geostrophic turbulence theory of boundary buoyancy anomalies in a quasigeostrophic fluid with vanishing potential vorticity. 
We find that the vertical stratification controls both the interaction range of boundary buoyancy anomalies and the dispersion of boundary-trapped Rossby waves. 
Buoyancy anomalies generate longer range velocity fields and more dispersive Rossby waves over decreasing stratification [$\mathrm{d}N(z)/\mathrm{d}{z} \leq 0$, where $N(z)$ is the buoyancy frequency] than over increasing stratification [$\mathrm{d}N(z)/\mathrm{d}{z} \geq 0$].  
Consequently, the surface kinetic energy spectrum is steeper over decreasing (mixed-layer like) stratification than in the classical uniformly stratified model. 
We therefore suggest that this steepening of the spectrum over mixed-layer like stratification accounts for the $k^{-2}$ spectrum found in the wintertime upper ocean. 
This suggestion is consistent with numerical and observational evidence indicating that surface geostrophic velocities over wintertime extratropical currents are largely induced by surface buoyancy anomalies. 

We also find that, under certain conditions, the nonlinear interplay of boundary-trapped Rossby waves with the turbulence spontaneously reorganizes the flow into homogenized zones of surface buoyancy separated by surface buoyancy discontinuities, with sharp eastward jets centered at the discontinuities and weaker westward flows in between. Jet dynamics then depend on the vertical stratification. Over decreasing stratification, we obtain straight jets perturbed by dispersive eastward propagating waves. Over increasing stratification, we obtain meandering jets whose shape evolves in time due to westward propagating weakly dispersive waves.

In the second part of this dissertation, we investigate normal modes in the presence of boundary-confined restoring forces, with the ultimate aim of obtaining an energy-conserving modal truncation of the quasigeostrophic equations. Such a modal truncation would generalize classical $N$-layer models to account for non-isentropic boundaries. Although we obtain orthogonal sets of vertical modes that diagonalize the energy and potential enstrophy in the presence of non-isentropic boundaries, we find that the loss of a crucial symmetry in the vertical coupling between the modes prevents modal truncations from conserving energy. Consequently, energy conserving modal truncations are not possible in the presence of non-isentropic boundaries.}

\acknowledgements{
Thank you to my advisor, Stephen Griffies, for guidance, encouragement, and many insightful conversations. Thank you to my committee members Robert Hallberg, Isaac Held, and Sonya Legg for their feedback and support. 
I also express my gratitude to several scientists outside of Princeton; namely, to Guillaume Lapeyre, Shafer Smith, and William Young for suggestions and discussions that have improved this dissertation. 

This dissertation would not have been possible without funding from the Cooperative Institute for Modeling the Earth System, under award NA18OAR4320123 from the National Oceanic and Atmospheric Administration, U.S. Department of Commerce.

Finally, my deepest gratitude goes to my family, especially my Mom and Dad. None of this would have been possible without you.
}

\dedication{ 
\setarab
 \fullvocalize
 \arabtrue
 \center{
 \begin{RLtext}
 b--is--mi al-ll_ahi al--rra.h--m_a--ni al-rra.h--I--mi
 \end{RLtext}}
 \center{
 \begin{RLtext}
 li-'umm-I wa 'ab--I
 \end{RLtext}}
 }
 
\pubs{This dissertation consists of four investigations (chapters 2,3,4 and 5), each of which is either published, or is in review for publication, in a scientific journal. Chapters 2 and 3 are in review, and will be published as Yassin \& Griffies (2022b) and Yassin (2022) respectively; preprints are of these two articles are available on \emph{arXiv}. Chapter 4 has been published as Yassin (2021) in the \emph{Journal of Mathematical Physics}. Chapter 5 has been published as Yassin \& Griffies (2022a) in the \emph{Journal of Physical Oceanography}. 
}

\usepackage{amsmath,amsfonts,amssymb,amsthm,bm}

\usepackage{titlesec}
\usepackage{graphicx}
\usepackage{bmpsize}

\usepackage{mathptmx}
\usepackage{newtxtext}
\usepackage{newtxmath}
\usepackage{nicefrac}
\usepackage{mathrsfs}

\usepackage{blindtext}
\usepackage{appendix}
\usepackage{url}
\usepackage[round]{natbib}

\usepackage[all]{nowidow}

\newtheorem{theorem}{Theorem}[section]
\newtheorem{proposition}[theorem]{Proposition}
\newtheorem{lemma}[theorem]{Lemma}

\newtheorem{definition}[theorem]{Definition}

\newenvironment{abstractchapter}
 {\normalsize
  \begin{center}
  \bfseries \abstractname\vspace{-.5em}\vspace{0pt}
  \end{center}
  \list{}{%
    \setlength{\leftmargin}{7.5mm}
    \setlength{\rightmargin}{\leftmargin}%
  }%
  \item\relax}
 {\endlist}

\makeatletter
\def\@makechapterhead#1{%
  \vspace*{50\p@}%
  {\parindent \z@ \raggedleft \normalfont
    \ifnum \c@secnumdepth >\m@ne
        \Huge \bfseries \textsc{\@chapapp}\space \thechapter
        \par\nobreak
        \vskip 10\p@
        \rule[.5ex]{\textwidth}{.4pt}%
        \vskip 10\p@
    \fi
    \interlinepenalty\@M
    \huge \normalfont \textsc{#1}\par\nobreak
    \vskip 40\p@
  }}
\makeatother

\include{PersonalDef}

\begin{document}

\nocite{yassin_surface_2022,yassin_normal_nodate,yassin_qg,yassin_jets_2022}

\setcounter{page}{11}

\include{1-Introduction/Introduction}

\include{2-SQG/SQG}

\include{3-Jets/Jets}

\include{4-Modes/Modes}

\include{5-QG_Modes/QG_Modes}

\include{6-Conclusion/Conclusion}

\bibliographystyle{abbrvnat}
\bibliography{references}

\end{document}

%% file: PersonalDef.tex
\newcommand{\abs}[1]{\left \lvert #1 \right\rvert} 
\newcommand{\inner}[2]{\left< #1, #2 \right>} 

\renewcommand{\vec}[1]{\boldsymbol{#1}  } 

\renewcommand{\d}[2]{\frac{\mathrm{d} #1}{\mathrm{d} #2}} 
\newcommand{\dd}[2]{\frac{\mathrm{d}^2 #1}{\mathrm{d} #2^2}} 
\newcommand{\pd}[2]{\frac{\partial #1}{\partial #2}} 
\newcommand{\grad}[1]{\vec{\nabla} #1} 
\newcommand{\lap}{\nabla^2} 
\newcommand{\laptwo}{\nabla^4} 
\newcommand{\unit}[1]{\hat{\vec{#1}}}

\newcommand{\ov}[1]{{#1}^\ast }

\newcommand{\wh}[1]{\widehat{#1} }

\newcommand{\e}{\mathrm{e}}
\newcommand{\irm}{\mathrm{i}}
\newcommand{\di}[1]{\mathrm{d}#1}

\newcommand{\Esc}{\mathscr{E}}

\newcommand{\Ssc}{\mathscr{S}}
\newcommand{\Ksc}{\mathscr{K}}
\newcommand{\Psc}{\mathscr{P}}

\newcommand{\J}[2]{ \mathrm{J}\left(#1,#2\right) }

\newcommand{\psik}{\hat \psi_{\vec k}}

\newcommand{\thetak}{\hat \theta_{\vec k}}

\newcommand{\keps}{k_\varepsilon}
\newcommand{\krh}{k_\mathrm{Rh}}
\newcommand{\kr}{k_r}

\newcommand{\Ac}{\mathcal{A}}
\newcommand{\Bc}{\mathcal{B}}
\newcommand{\Cc}{\mathcal{C}}

\newcommand{\Lc}{\mathcal{L}}

\newcommand{\Dc}{{\mathcal{D}}}

\newcommand{\Qf}{\left(Q,R_1,R_2\right)}

\newcommand{\C}{\mathbb{C}}
\newcommand{\Pb}{\mathscr{P}}

\newcommand{\intz}{\int_{z_1}^{z_2}}

%% file: 1-Introduction/Introduction.tex
\chapter{Introduction}\label{Ch-intro}

    \section{Vertical structure and geostrophic turbulence}
    
    \subsection{Quasigeostrophy and baroclinic instability}
    
	Quasigeostrophy is a regime of fluid motion that emerges in the limit of rapid rotation and strong stratification, with a dynamical state uniquely determined by the potential vorticity, $q$, in the fluid interior along with the buoyancy anomalies, $b$, at the fluid's lower and upper boundaries \citep[][chapter 5]{Vallis2017}. The geostrophic flow is then recovered by inverting a diagnostic relation between the potential vorticity and the geostrophic streamfunction, $\psi$, with boundary conditions determined by the boundary buoyancy anomalies. In the ocean, the quasigeostrophic equations describe motion at horizontal scales ranging from a few kilometres to a few hundred kilometres and timescales longer than one day \citep{charney_oceanic_1981,lapeyre_surface_2017}. In contrast, quasigeostrophic flows in the atmosphere have horizontal scales of thousands of kilometres \citep[][chapter 5]{Vallis2017}.
		
	The spatial distribution of potential vorticity and boundary buoyancy is set up by external mechanical and thermodynamical forcing. 
	Certain commonly occurring spatial configurations of potential vorticity and boundary buoyancy are unstable to baroclinic instability
	\textemdash{} an instability spawning eddies that transport potential vorticity and boundary buoyancy anomalies so as to destroy these unstable spatial configurations \citep[][chapter 9]{Vallis2017}. A dynamical balance is ultimately attained as external forcing maintains unstable configurations of potential vorticity and boundary buoyancy against the destructive tendencies of baroclinic eddies. 
	
	\begin{figure}
	\centerline{\includegraphics[width=\textwidth]{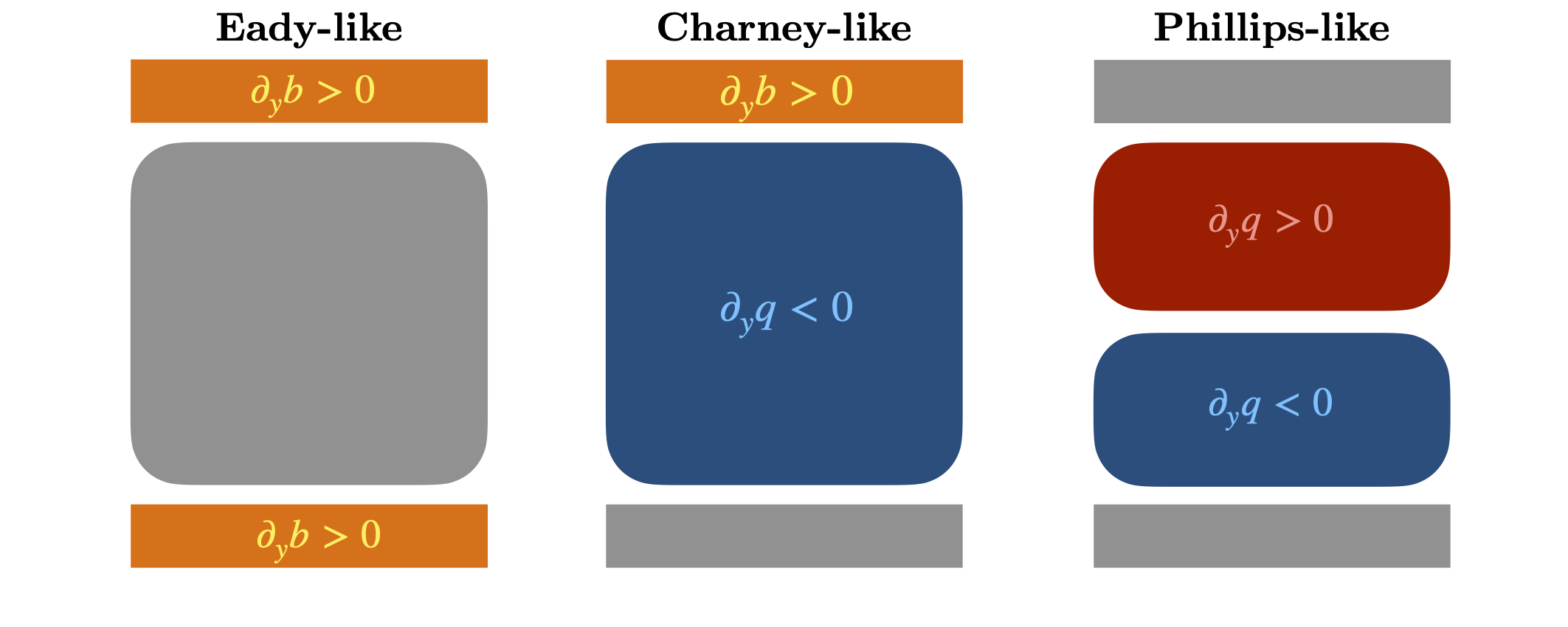}}
	\caption{The three classes of baroclinic instability.}
	\label{F-BC}
  	\end{figure}
	
	There are three classes of baroclinic instability as classified by which linear stability conditions are violated (figure \ref{F-BC}).
	The first is an Eady-like instability \citep{Eady1949}, which occurs if the boundary buoyancy gradients have the same sign at the lower and upper boundaries. 
	 The second is a Charney-like instability \citep{Charney1947}, which occurs if the horizontal potential vorticity gradient has the opposite sign to the buoyancy gradient at the fluid's upper boundary (or the same sign as the buoyancy gradient at the lower boundary). 
	 Finally, there is a Phillips-like instability \citep{phillips_energy_1954}, which occurs if the horizontal potential vorticity gradient switches sign in the fluid interior. 
	 Baroclinic instability in the atmosphere is either an Eady-type instability, with the buoyancy gradient at the tropopause having the same sign as the buoyancy gradient at the Earth's surface \citep{lapeyre_surface_2017}, or a Charney-type instability, with the horizontal potential vorticity gradient (dominated by the planetary $\beta$-effect) having the same sign as the buoyancy gradient at the Earth's surface \citep[][chapter 9]{Vallis2017} .  
	 
	 In the ocean, studies mapping out potential vorticity and boundary buoyancy have found that the energetically dominant instability in the Southern Ocean is an Eady-like instability, which is intensified at the lower and upper boundaries \citep{feng_four_2021}. In contrast, the energetically dominant instability in the Gulf Stream and Kuroshio is a Phillips-like instability, which reaches deep into the water column \citep{smith_geography_2007,tulloch_scales_2011,feng_four_2021}. 	
	 The Charney-like instability, which is intensified at the upper boundary, is mainly found in the subtropical oceans. Moreover, unlike the Eady-like and Phillips-like instabilities, which have characteristic time scales on the order of weeks, the Charney-like instability grows more slowly, with characteristic time scales on the order of months \citep{tulloch_scales_2011,feng_four_2021}.
	
	 However, over much of the World Ocean, the fastest growth rates (on the order of days) are due to a smaller scale mixed-layer baroclinic instability \citep{smith_geography_2007,boccaletti_mixed_2007}. Mixed-layer baroclinic instability is a Charney-like instability occurring because the potential vorticity gradient at the mixed-layer base has the opposite sign to the surface buoyancy gradient \citep{Callies2016}. Unlike the energetically dominant baroclinic instability at larger scales, mixed-layer instability is seasonal, with a seasonality following that of mixed-layer depth \citep{mensa_seasonality_2013, sasaki_impact_2014,callies_seasonality_2015}. Mixed-layer instability is most active in high eddy kinetic energy regions with strong buoyancy gradients and deep mixed-layers \textemdash{} in particular, major extratropical currents such as the Gulf stream and Kuroshio \citep{sasaki_regionality_2017,khatri_role_2021}.

	 \subsection{Truncated models with isentropic boundaries}
	 Traditionally, boundary buoyancy anomalies were neglected in physical oceanography, and the Phillips two-layer model, as well as more general $N$-layer models, were the main paradigm for ocean geostrophic turbulence. These layered models may be equivalently thought of as vertical modal truncations of the quasigeostrophic equations with isentropic boundaries\footnote{In quasigeostrophy, isentropic boundaries are those with uniform buoyancy and no topographic gradients.}. This equivalency can be seen in the following manner. First, the baroclinic modes are obtained by solving a Sturm-Liouville problem for the vertical structure of Rossby waves in a quiescent ocean with isentropic boundaries; there are infinitely many modes $\{\phi_n(z)\}_{n=0}^\infty$ satisfying $\phi_n'(z)=0$ at the lower and upper boundaries, where $\phi_n'(z)$ denotes the vertical derivative of $\phi_n(z)$, and these modes form an orthonormal set,
	 \begin{equation}
	     \int_{-H}^0 \phi_m \, \phi_n \, \mathrm{d}z = \delta_{mn},
	 \end{equation}
	 where $\delta_{mn}$ is the Kronecker delta \citep[][chapter 6]{Vallis2017}.
	 Given a quasigeostrophic streamfunction, $\psi$, with vanishing boundary buoyancy anomalies ($\partial_z \psi = 0$ at $z=-H,0$), we expand such a streamfunction as
	 \begin{equation}\label{eq:expansion_baroclinic}
	 	\psi(\vec x ,z,t) = \sum_{\vec k} \sum_{n=0}^\infty \psi_{\vec k n}(t)\, \phi_n(z) \, \mathrm{e}^{\mathrm{i} \vec k \cdot \vec x},
	 \end{equation}
	 where we have assumed a doubly periodic domain in the horizontal.
	 In the above expression, $\vec x =(x,y)$ is the horizontal position vector, $\vec k=(k_x,k_y)$ is the horizontal wavevector, and $t$ is the time coordinate. 
	 Substituting such an expansion into the time-evolution equation for potential vorticity,
	 \begin{equation}
	 	\pd{q}{t} + \beta \, \pd{\psi}{x} +  \J{\psi}{q} = 0,
	 \end{equation}
	 where $\J{\psi}{q} = \partial_x \psi \, \partial_y q -  \partial_y \psi \, \partial_x q$ is the Jacobian operator and $\beta$ is the latitudinal vorticity gradient, then yields a time-evolution equation for the modal amplitudes \citep{flierl_models_1978},
	 \begin{equation}\label{eq:time_modal}
		\d{q_{\vec k n}}{t} + \mathrm{i} \, \beta \, k_x \, \psi_{\vec k n}+  \sum_{\vec a , \vec b} \sum_{ l m} A_{\vec a \vec b \vec k} \, \varepsilon_{l m n} \, \psi_{\vec a l} \, q_{\vec b m} = 0.
	\end{equation}
	In this equation,  the modal potential vorticity amplitude, $q_{\vec k n}$, is related to the modal streamfunction amplitude, $\psi_{\vec k n}$, through
	\begin{equation}
		q_{\vec k n} = -(k^2 + \lambda_n) \psi_{\vec k n},
	\end{equation}
	where $k=\abs{\vec k}$ is the horizontal wavenumber and $\lambda_n$ is the eigenvalue corresponding the eigenfunction $\phi_n$.
	
	If the amplitudes, $\psi_{\vec k n}$, are small, then the time-evolution equation \eqref{eq:time_modal} for each mode decouples and we obtain non-interacting linear Rossby waves. More generally, the time-evolution equation for the modal amplitudes \eqref{eq:time_modal} allows us to view quasigeostrophic dynamics (with isentropic boundaries) as the nonlinear interaction of vertical modes. The horizontal coupling coefficient
	\begin{equation}
		A_{\vec a \vec b \vec k} = - \unit z \cdot \left(\vec a \times \vec b\right) \, \delta_{\vec a + \vec b, \vec k}
	\end{equation}
	specifies that three waves will interact only if the sum of the horizontal wavevectors of the two incoming waves are equal to the wavevector of the outgoing wave. The vertical coupling coefficient,
	\begin{equation}\label{eq:vertical_coupling}
		\varepsilon_{lmn} = \int_{-H}^0 \phi_l \, \phi_m \, \phi_n \, \mathrm{d}z,
	\end{equation}
	indicates that modal interactions generally depend on the vertical structures of the modes (and hence the stratification). Because the $n=0$ mode is barotropic  [i.e., $\phi_0(z)=1$], we obtain 
	\begin{equation}\label{eq:coulping_barotropic}
		\varepsilon_{0mn} = \delta_{mn},
	\end{equation} 
	which states that, if one wave is barotropic, then the other two waves must have the same vertical mode number (i.e., $m=n$) for an interaction to occur.
	
	To obtain an $(N+1)$-layer model, truncate the series expansion \eqref{eq:expansion_baroclinic} at $n=N$. However, although the original untruncated system conserves total energy and potential enstrophy, there is no reason to expect that truncated models conserve a truncated form of the total energy and potential enstrophy in general. For instance, the conservation of a truncated energy implies that, if a quasigeostrophic state is initialized with energy only in the first $N$ vertical modes, this energy will remain in the first $N$ vertical modes (despite the nonlinear interactions) for all time. In the case of truncations with the baroclinic modes, this ``trapping'' of the initial energy at low modes is a non-trivial consequence of a symmetry in the vertical coupling coefficient, $\varepsilon_{lmn}$.	
	  Multiplying the modal time-evolution equation \eqref{eq:time_modal} by the complex conjugate $\psi_{\vec k n}^*$, taking the real part, and summing over $\vec k$ and $n$ gives the energy equation 
	\begin{equation}\label{eq:energy}
		\d{}{t} \left(\sum_{\vec k n} \frac{1}{2} (k^2+\lambda_n)  \abs{\psi_{\vec k n}}^2  \right)
		+ \sum_{\vec a \vec b \vec k }\sum_{l m n } A_{\vec a \vec b \vec k} \, \varepsilon_{l m n} \, \Re\left\{\psi_{\vec a l} \, q_{\vec k - \vec a \, m} \psi_{\vec k n}^* \right\} = 0,
	\end{equation}
	where $\Re\{A\}$ denotes the real part of $A$.
	After truncating at $n=N$, the nonlinear sum vanishes because it is a contraction between a symmetric tensor, $ \varepsilon_{l m n} \, \Re\left\{\psi_{\vec a l} \, q_{\vec k - \vec a \, m} \psi_{\vec k n}^* \right\}$, and an anti-symmetric tensor, $A_{\vec a \vec b \vec k}$, and so the $N$-truncated energy is conserved. Physically, the truncated energy is conserved because the interaction,
	\begin{equation}
		(\vec a, l) + (\vec b, m) \rightarrow (\vec k, n),
	\end{equation} 
	has the opposite energetic contribution to
	\begin{equation}
		(\vec k, n) + (\vec b, m) \rightarrow (\vec a, l)
	\end{equation}
	in the nonlinear sum in the energy equation \eqref{eq:energy}.
	As we find in chapter \ref{Ch-conc}, this symmetry is lost once we allow for non-isentropic boundaries.
	
	Truncating the modal expansion \eqref{eq:expansion_baroclinic} at $n=1$, using the form of the vertical coupling coefficient with the barotropic mode \eqref{eq:coulping_barotropic}, and transforming back to physical space, we obtain the two-layer quasigeostrophic model \citep[e.g.,][]{smith_scales_2001},
	\begin{align}
		\label{eq:time-barotropic}
		\pd{q_0}{t} + \beta \, \pd{\psi_0}{x} + \J{\psi_0}{q_0} + \J{\psi_1}{q_1} =0, \\
		\label{eq:time-baroclinic}
		\pd{q_1}{t} + \beta \, \pd{\psi_1}{x} + \J{\psi_0}{q_1} + \J{\psi_1}{q_0}+ \varepsilon_{111} \, \J{\psi_1}{q_1} =0.
	\end{align}
	In these two equations, the barotropic potential vorticity is given by $q_0 = \lap \psi_0$, where $\psi_0$ is the barotropic streamfunction, and the first mode baroclinic potential vorticity is given by $q_1 = \left(\lap -\lambda_1\right)\psi_1$, where $\psi_1$ is the first mode baroclinic streamfunction.
	
	\subsection{Geostrophic turbulence with isentropic boundaries}

	Examining the velocity induced by an isolated $n$-th mode potential vorticity anomaly exposes the dynamical distinction between the barotropic mode ($n=0$) and the baroclinic modes ($n>0$). 
	Suppose we have a point potential vorticity anomaly, $q_n \sim \delta(\abs{\vec x})$, where $\delta(\abs{\vec x})$ is the Dirac delta and $\abs{\vec x}$ is the horizontal distance from the anomaly. Then, for $n>0$, the resulting streamfunction is proportional to \citep{polvani_two-layer_1989} 
	\begin{equation}\label{eq:streamfunction_baroclinic}
		\psi_n(\abs{\vec x}) \sim \frac{ \mathrm{e}^{- \abs{\vec x}/L_n}}{\sqrt{\abs{\vec x}/L_n}},
	\end{equation}
	where $L_n = 1/\sqrt{\lambda_n}$ is the $n$-th mode deformation radius. Therefore, the interaction range of this potential vorticity anomaly is determined by the deformation radius; at $\abs{\vec x} \gg L_n$, the velocity induced by the anomaly essentially vanishes. 
	However, in the singular limit $\lambda \rightarrow 0$ (or $L \rightarrow \infty$) of the barotropic mode, we obtain
	\begin{equation}
		\psi_0(\abs{\vec x}) \sim \frac{\mathrm{log}\left(\abs{\vec x}\right)}{2\pi},
	\end{equation}
	which implies an extremely long range velocity field and an infinite interaction range (see figure \ref{F-BTBC}).
	
			\begin{figure}
	\centerline{\includegraphics[width=\textwidth]{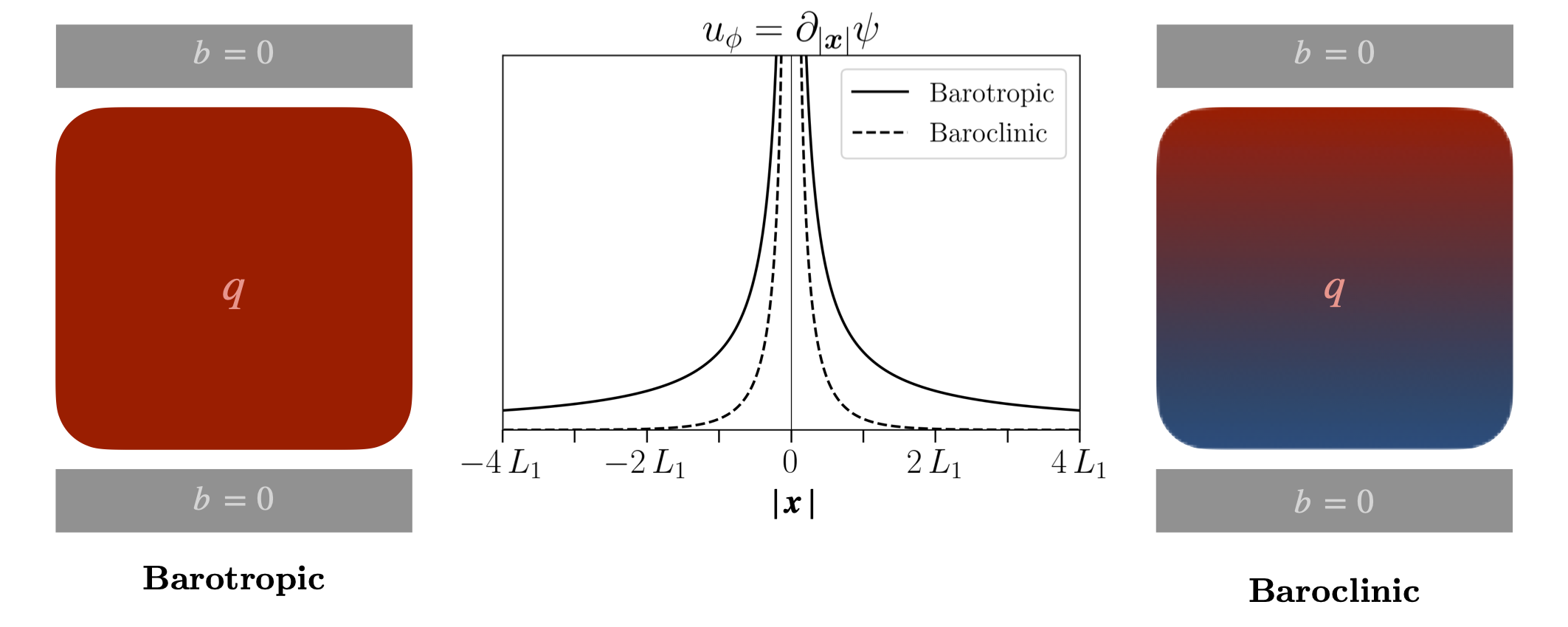}}
	\caption{The azimuthal velocity, $u_\phi = \partial_{\abs{\vec x}} \psi$, for a barotropic point vortex and a baroclinic point vortex.}
	\label{F-BTBC}
  	\end{figure}
	
	 \begin{figure}
	\centerline{\includegraphics[width=\textwidth]{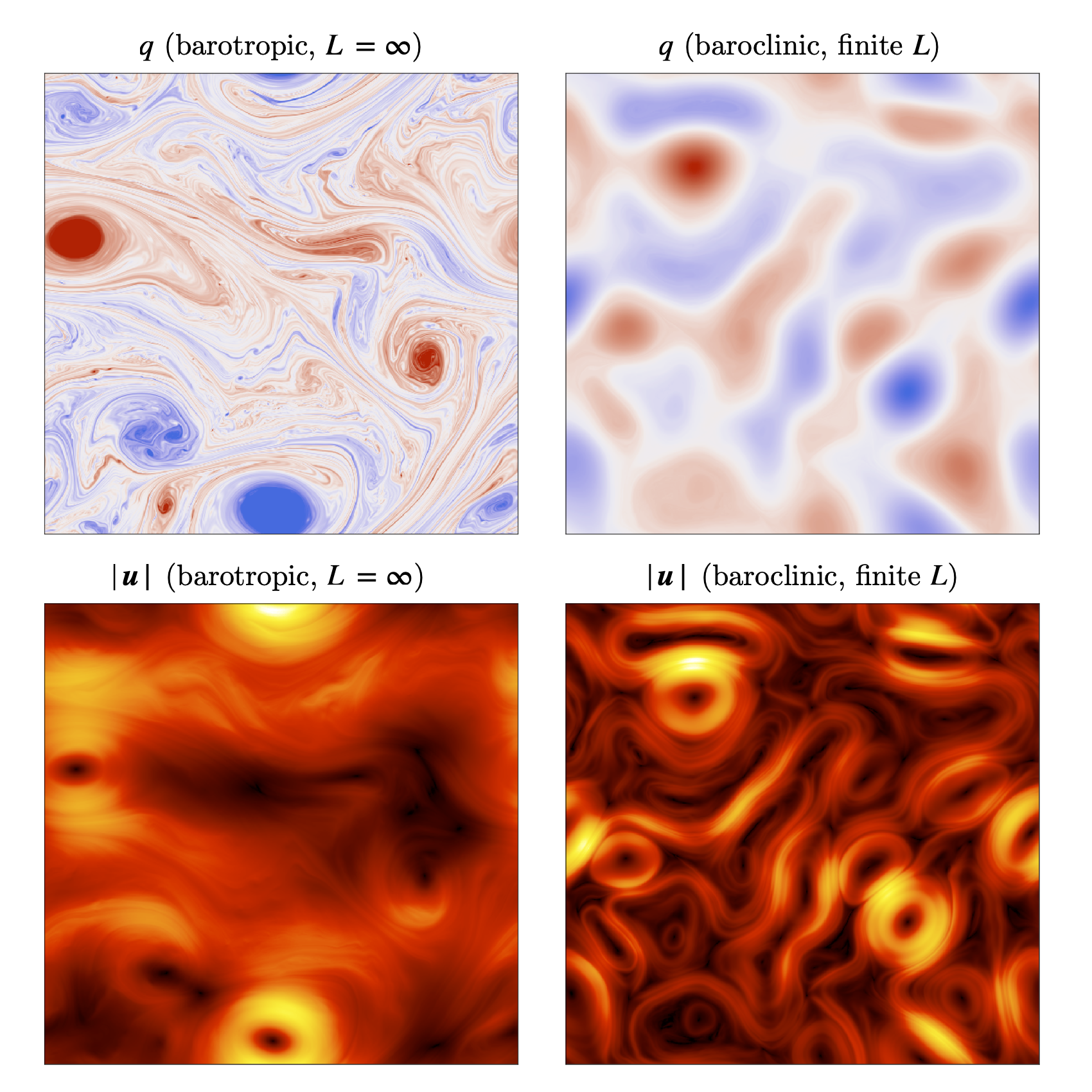}}
	\caption{The potential vorticity (upper row) and horizontal speed (bottom row) in the geostrophic turbulence of a single barotropic ($L=\infty$, left column) and baroclinic mode ($L\neq \infty$, right column), where $L$ is the deformation radius. Both simulations are forced at a horizontal scale equal to one quarter of the domain. The potential vorticity and the horizontal speed are normalized by the maximum value in the snapshot.}
	\label{F-BTBC_snap}
  	\end{figure}

  	As illustrated in figure \ref{F-BTBC_snap}, the interaction range of potential vorticity anomalies modifies the structure of the resulting geostrophic turbulence. In the turbulence of a single barotropic mode ($L_0=\infty$), vorticity anomalies generate long range velocity fields that subject the vorticity field itself to large-scale strain; this straining then leads to the thin vorticity filaments that characterize barotropic turbulence. In contrast, for a baroclinic mode ($L_n \neq \infty$), potential vorticity anomalies generate short range velocity fields that are more efficient at mixing away small-scale inhomogeneities. As a result, the potential vorticity field lacks thin filamentary structures and instead appears spatially diffuse. Moreover, in the case of a single baroclinic mode, the presence of a distinguished length scale (the deformation radius, $L_n$) leads to the emergence of plateaus of homogenized potential vorticity surrounded by kinetic energy ribbons \citep{arbic_coherent_2003}.

	The striking distinction between the dynamics of the barotropic mode and the higher baroclinic modes permits a simplified description of geostrophic turbulence  with isentropic boundaries \citep{rhines_dynamics_1977,salmon_baroclinic_1980}. In this turbulence, large-scale baroclinic instability generates a baroclinic eddy field whose energy cascades to smaller horizontal scales towards the deformation radius and then, at the deformation radius, to larger vertical scales. The barotropic mode is energized by these baroclinic transfers and large-scale quasigeostrophic turbulence resembles a two-dimensional barotropic fluid advecting a nearly passive baroclinic eddy field \citep{larichev_eddy_1995,smith_scales_2002}. 
	Fundamentally, the long reaching velocity fields generated by the barotropic eddies, along with the short range of baroclinic eddies, together allow for the barotropic mode to dominate the time-evolution of the flow and effectively reduces the problem to that of two-dimensional turbulence. Indeed, an elegant parametrization of two-layer quasigeostrophic turbulence has recently been proposed based on the dominance of barotropic vortices \citep{gallet_vortex_2020,gallet_quantitative_2021}.
	
	\subsection{Buoyancy anomalies in the atmosphere}
	
	The atmosphere has no upper boundary. However, the sharp gradient in vertical stratification at the tropopause is a dynamical upper boundary for the troposphere \citep[e.g.,][]{Eady1949,tulloch_theory_2006}. Just as with buoyancy anomalies along a rigid boundary, buoyancy anomalies along a stratification discontinuity induce their own geostrophic velocities that attenuate with vertical distance from the stratification discontinuity \citep{juckes_quasigeostrophic_1994,held_surface_1995}.
	Assuming that the tropopause is a stratification discontinuity between the troposphere and the stratosphere, \cite{juckes_quasigeostrophic_1994} derived a relation between tropopause buoyancy anomalies and vertical displacements of the tropopause; Juckes then showed that this relationship is satisfied in atmospheric general circulation models, indicating the relevance of buoyancy induced flows near the tropopause. 
	
	The dynamics of tropopause buoyancy anomalies was then invoked to account for the buoyancy variance spectra and horizontal kinetic energy spectra observed near the tropopause \citep{nastrom_climatology_1985}. These empirically derived spectra exhibit a steep -3 spectral slope at large horizontal scales ($1000-3000$ km) and a shallower -5/3 spectral slope at smaller horizontal scales (10-200 km). At sufficiently small scales, buoyancy anomalies at the tropopause are expected to have an -5/3 spectral slope in both buoyancy variance and kinetic energy; \cite{tulloch_theory_2006} proposed that the transition to the -3 spectral slope at large horizontal scales occurs because once tropopause buoyancy anomalies are large enough to feel the Earth's surface, their dynamics becomes similar to vorticity anomalies in a barotropic model. 
	
	However, later studies found that atmospheric Rossby numbers are too large for this mechanism to be valid. In regions of slowly varying background stratification, quasigeostrophy is valid if both the Rossby and Froude numbers are much smaller than one. In contrast, near a sharp vertical stratification gradient (like the tropopause), the Rossby and Froude numbers must be smaller than $h/H$ (a more stringent condition because $h/H$ is generally smaller than one), where $h$ is the vertical scale of the sharp stratification gradient and $H$ is the characteristic vertical length scale of the flow \citep{asselin_quasigeostrophic_2016}. By varying the Rossby number in an idealized Boussinesq model of the tropopause, \cite{asselin_boussinesq_2018} showed that tropopause buoyancy anomalies can account for the Nostrum-Gage spectrum only at unrealistically small values of the Rossby number; at more realistic Rossby number values for the atmosphere, the Nostrum-Gage spectrum is best accounted for through unbalanced motion.
	
	In contrast to the relatively large Rossby numbers in the atmosphere, oceanic Rossby numbers remain smaller than one at small horizontal scales. Even at horizontal scales of 10 km, the rotational component of the flow can have a Rossby number  as low as 0.3 \citep{callies_time_2020}. For this reason, we expect boundary buoyancy anomalies to play a more significant role in the ocean, and we focus on oceanic applications for the remainder of the dissertation.

	\subsection{Vertical structure and satellite altimetry in the ocean}
	
	Oceanic observations were first interpreted within the paradigm of layered models and hence in terms of the baroclinic modes. Using satellite altimeter observations, \cite{stammer_global_1997} found correlations between surface eddy scales and the first mode deformation radius and  proposed that the surface altimeter signal is related to processes with a first baroclinic mode vertical structure. Separately, \cite{wunsch_vertical_1997} partitioned the kinetic energy obtained from current meter observations into the baroclinic modes; this partition was justified using the fact that the baroclinic modes are ``complete'' and so can represent any quasigeostrophic state \citep{ferrari_distribution_2010,lacasce_surface_2012,rocha_galerkin_2015}. Wunsch found that most regions are dominated by a combination of the barotropic and first baroclinic modes, and that the surface altimeter signal primarily reflects the first baroclinic mode because of its near surface intensification. These results were supported by the theory and numerical simulations of \cite{smith_scales_2001,smith_scales_2002} who found that, in surface-intensified stratification, energy concentrates in the first baroclinic mode because energy transfers between the baroclinic modes ($n>0$) and the barotropic mode ($n=0$) become less efficient.
	However, despite the claims that the baroclinic modes are complete, this interpretation of the observations neglects the contribution of boundary buoyancy anomalies. 
	
	It was subsequently discovered that boundary buoyancy anomalies can induce significant velocities in the upper ocean and so cannot generally be neglected \citep{lacasce_estimating_2006,lapeyre_dynamics_2006}. Using a numerical model of the North Atlantic, \cite{isernfontanet_three-dimensional_2008} reconstructed the geostrophic velocity field from sea surface temperature in winter; they found spatial correlations  between the reconstructed velocity and model velocity at the ocean's surface exceeds 0.7 over most of the North Atlantic. This correlation implies that, at least in the wintertime North Atlantic, a significant portion of the surface geostrophic flow is induced by surface buoyancy anomalies rather than interior potential vorticity.

	 As a consequence of these findings, \cite{lapeyre_what_2009} questioned the interpretation of the altimeter signal in terms of the baroclinic modes. 
	 Lapeyre noted that the baroclinic modes cannot be complete because they assume vanishing buoyancy anomalies at the lower and upper boundaries; as such, they cannot be used to represent arbitrary quasigeostrophic states but only those with vanishing boundary buoyancy anomalies.
	 Over uniform stratification [i.e., $N(z)=\mathrm{constant}$], geostrophic buoyancy anomalies generate a streamfunction decaying exponentially away from the ocean's surface with a vertical attenuation determined by the magnitude of the local stratification \citep{held_surface_1995}.
 	 Lapyere then appended an upper surface quasigeostrophic mode to the baroclinic modes and partitioned the flow obtained from a numerical ocean model into this expanded set of vertical structures. He found that, over most of the North Atlantic, the surface quasigeostrophic mode dominates, with the only exception being the eastern recirculating branch of the North Atlantic gyre. Lapeyre then concluded that the satellite altimeter signal over the North Atlantic must primarily be due to the surface quasigeostrophic mode rather than the first baroclinic mode.

	Aside from the surface buoyancy induced contribution to the geostrophic velocity, there is yet another interpretation of the vertical structure of ocean eddies. \cite{de_la_lama_vertical_2016} revisited the vertical partition of kinetic energy using a larger current meter dataset than in \cite{wunsch_vertical_1997}. Instead of partitioning the kinetic energy into baroclinic modes, \cite{de_la_lama_vertical_2016} computed the  vertical empirical orthogonal functions; they found that the leading empirical orthogonal function is monotonically decaying from the ocean's surface towards the ocean's bottom boundary. A similar result was found by \cite{wunsch_vertical_1997} who interpreted this vertical structure as the sum of a barotropic and a first baroclinic mode. Instead, \cite{de_la_lama_vertical_2016} noted that this leading empirical orthogonal function resembles the zeroth Rossby wave mode with a vanishing bottom pressure boundary condition (rather than a vanishing bottom buoyancy boundary condition, as in the baroclinic modes, see figure \ref{F-BCSF}). A vanishing bottom pressure boundary condition is expected over steep bottom topography \citep{rhines_edge_1970}; in this limit, to leading order, bottom boundary dynamics decouples from interior dynamics due to propagation of fast dispersive bottom-trapped topographic waves \citep{lacasce_geostrophic_1998,lacasce_geostrophic_2000}. 
	
	 \begin{figure}
	\centerline{\includegraphics[width=\textwidth]{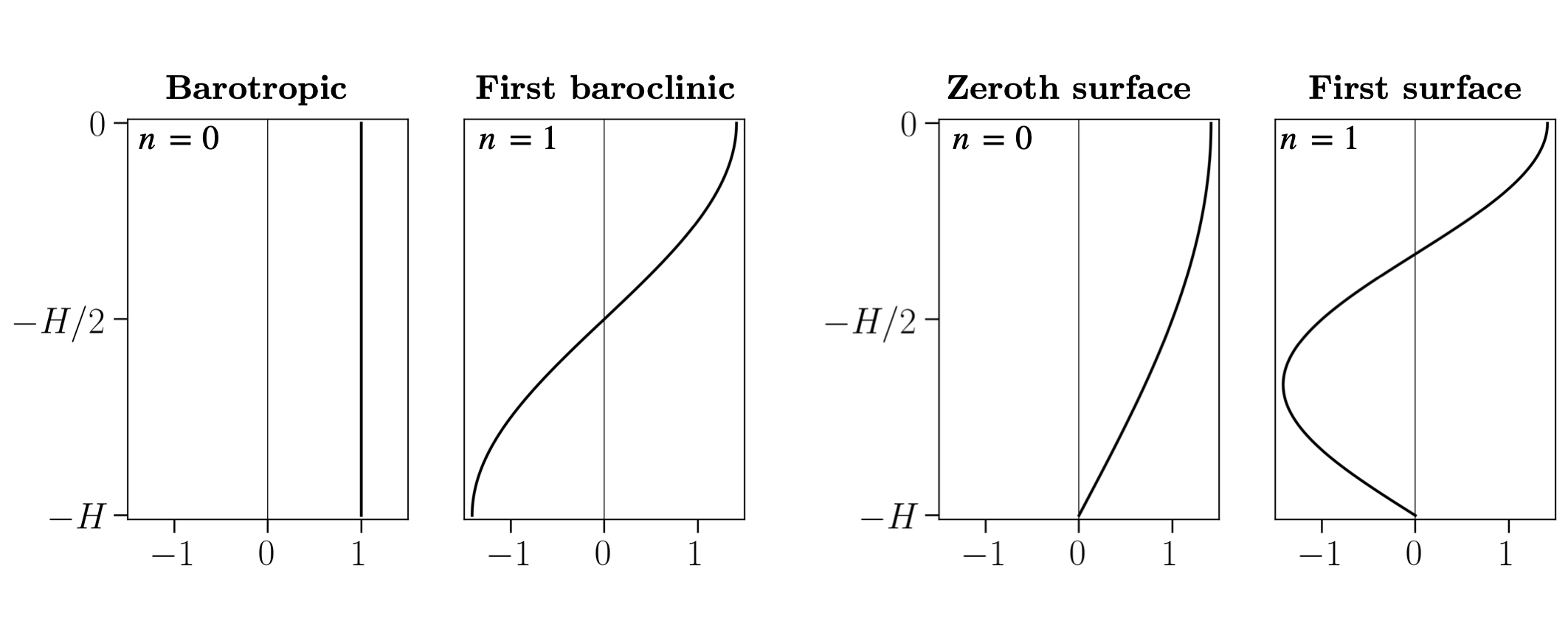}}
	\caption{The zeroth and first baroclinic modes and surface modes in constant stratification.}
	\label{F-BCSF}
  	\end{figure}

	Subsequently, to determine the appropriate bottom boundary condition, \cite{lacasce_prevalence_2017} solved the eigenvalue problem for Rossby wave vertical structure while taking bottom topography into account over the World Ocean. He found that, nearly everywhere in the ocean, bottom topography is steep enough so that the vertical modes nearly vanish at the bottom. Consequently, LaCasce suggested that the surface modes \textemdash{} which he defined as the solution to the eigenvalue problem for Rossby wave vertical structure with a vanishing bottom pressure boundary condition (see figure \ref{F-BCSF}) \textemdash{} are to be preferred to the baroclinic modes almost everywhere in the ocean. LaCasce also suggested that the barotropic mode may not exist in the ocean.  Instead, over steep topography, the vertical inverse cascade is halted at the gravest surface mode, which monotonically decays toward the ocean bottom.	

	\section{Overview of the dissertation}
	
	The debate over the vertical structure of ocean eddies is fundamentally a debate about oceanic geostrophic turbulence.
	What regime of geostrophic turbulence is present in the ocean? 
	Does it consist of long-range barotropic eddies advecting a nearly passive baroclinic flow field? 
	Or does it consist of short-range surface-intensified eddies weakly interacting with bottom-intensified flows? 
	Does the sea surface height measured by satellite altimeters correspond to potential vorticity anomalies and thermocline dynamics (i.e., the first baroclinic mode)? 
	Or does it correspond to surface-trapped motion induced by surface  buoyancy anomalies (i.e., the surface quasigeostrophic mode)?
	
	\begin{figure}
	\centerline{\includegraphics[width=\textwidth]{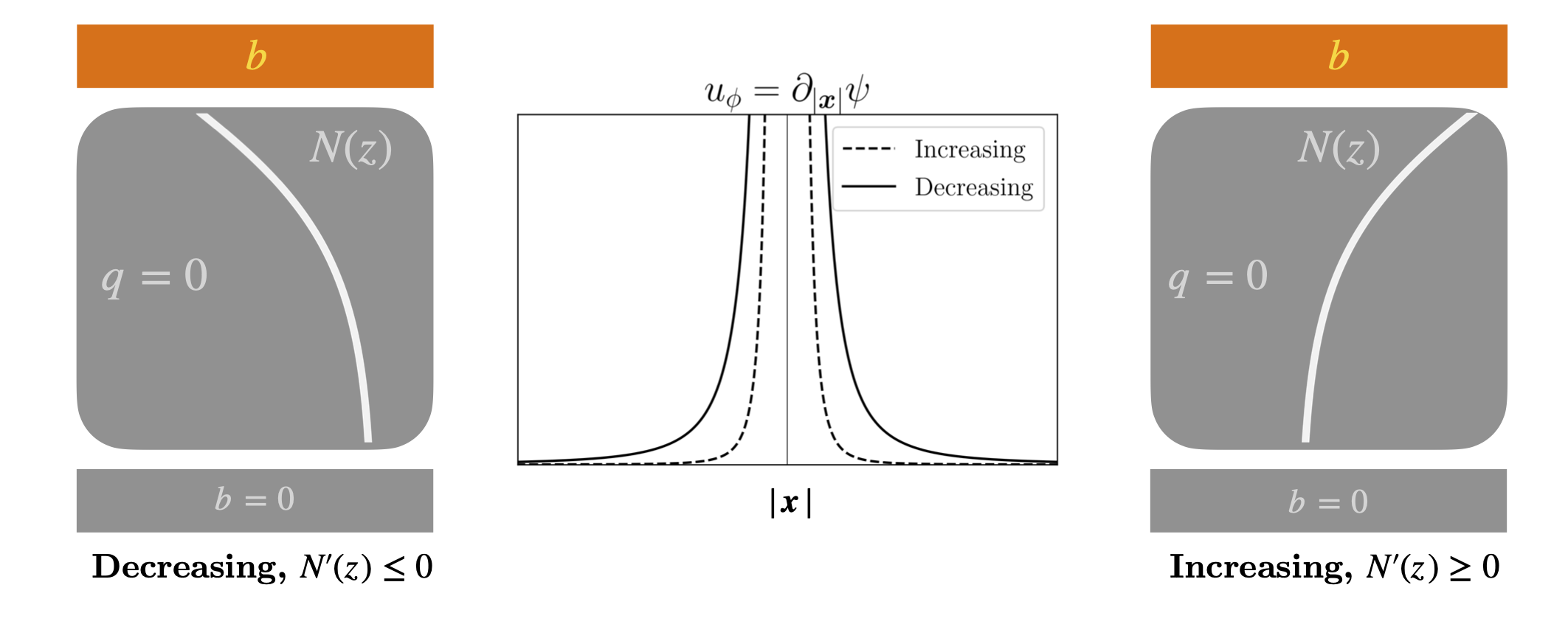}}
	\caption{The azimuthal velocity, $u_\phi = \partial_{\abs{\vec x}} \psi$, for a point buoyancy anomaly at the upper boundary over vertically decreasing and vertically increasing stratification. $N(z)$ is the buoyancy frequency.}
	\label{F-INDE}
  	\end{figure}
	
	\subsection{The first part of the dissertation}
	
	This dissertation contributes to these overarching questions in two ways. The first part of this dissertation, consisting of chapters \ref{Ch-SQG} and \ref{Ch-jets}, develops the geostrophic turbulence theory of boundary buoyancy anomalies in variable stratification. 
	The surface quasigeostrophic model was formulated by \cite{held_surface_1995}; this model describes the dynamics induced by boundary buoyancy anomalies in a uniformly stratified fluid with zero potential vorticity  \citep[although the dynamics induced by boundary buoyancy anomalies has a long history in the atmospheric dynamics literature, e.g.,][]{Charney1947,Eady1949,blumen_uniform_1978}. Chapter 2 extends the surface quasigeostrophic model to account for variable stratification. We find that the vertical stratification controls the interaction range of surface buoyancy anomalies; a surface buoyancy anomaly $b|_{z=0} \sim \delta\left(\abs{\vec x} \right)$  generates an approximate streamfunction of
	\begin{equation}\label{eq:range_sqg}
		\psi\left( \abs{\vec x} \right) \sim \frac{1}{\abs{\vec x}^{2-\alpha}},
	\end{equation}
	for $0<\alpha<2$, where the parameter $\alpha$ is determined by the stratification's vertical structure (figure \ref{F-INDE}). In uniform stratification, we have $\alpha=1$ \citep{pierrehumbert_spectra_1994,held_surface_1995}. However, if the stratification is decreasing towards the upper boundary [$N'(z) \leq 0$, where $N(z)$ is the buoyancy frequency] then we obtain longer range flows with $\alpha >1$.
	 In contrast, over increasing stratification [$N'(z) \geq 0$], we obtain shorter range buoyancy anomalies with $\alpha <1$. 
	In the limit that surface stratification is much larger than deep ocean stratification (e.g., exponential stratification), then surface buoyancy anomalies become extremely local, with an induced streamfunction similar to that induced by a baroclinic mode \eqref{eq:streamfunction_baroclinic}, and with a deformation radius determined by the stratification's approximate $\mathrm{e}$-folding depth. By applying the theory to the North Atlantic, we find an approximate value of $\alpha \approx \nicefrac{3}{2}$ in winter and $\alpha \approx \nicefrac{1}{2}$ in summer. 
	
	\begin{figure}
	\centerline{\includegraphics[width=\textwidth]{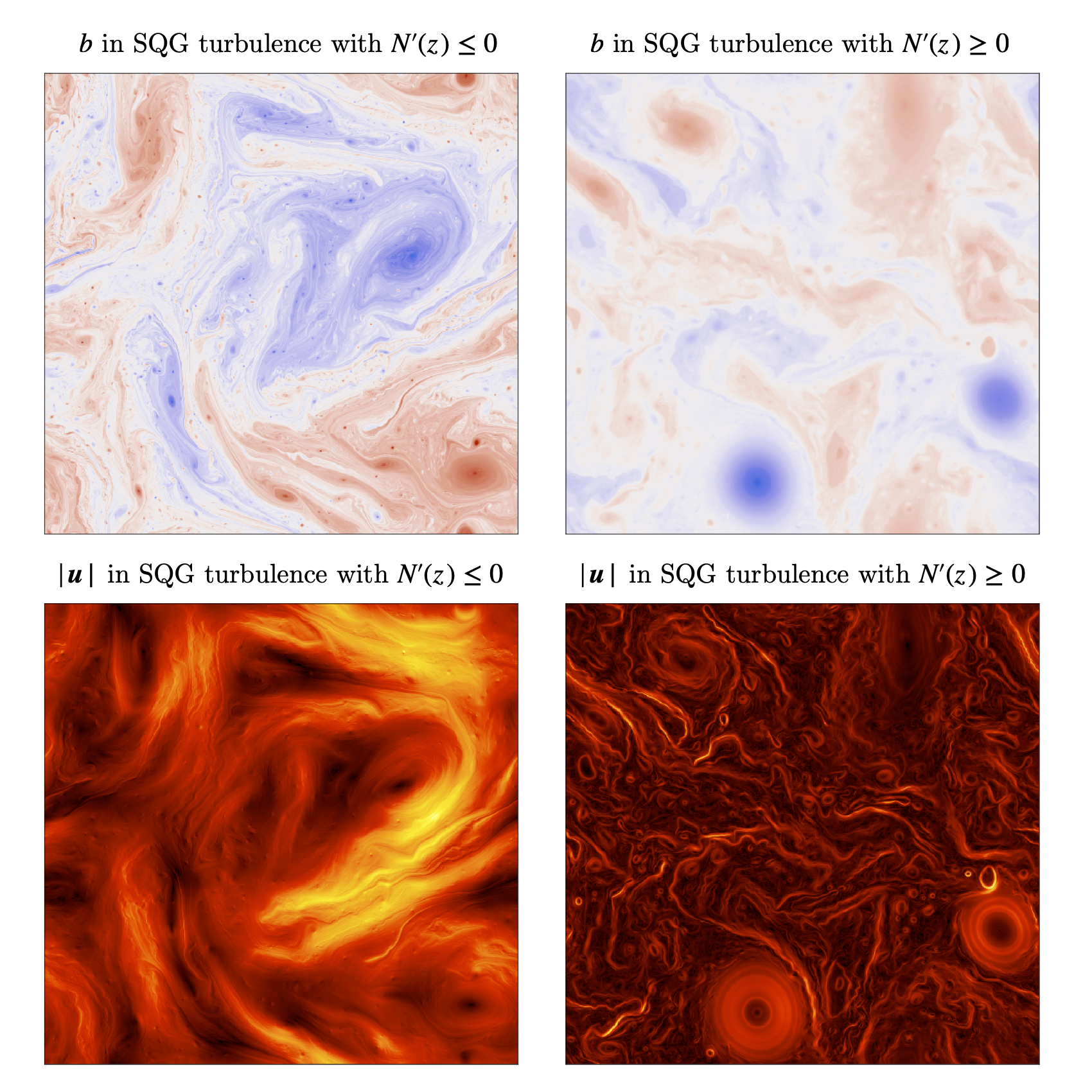}}
	\caption{The upper boundary buoyancy anomaly (upper row) and horizontal speed (bottom row) in the geostrophic turbulence of a surface quasigeostrophic (SQG) mode in decreasing [$N'(z) \leq 0$, left column] and increasing  [$N'(z) \geq 0$, right column] stratification, where $N(z)$ is the buoyancy frequency. Both simulations are forced at a horizontal scale equal to one quarter of the domain. The upper boundary buoyancy and the horizontal speed are normalized by the maximum value in the snapshot.}
	\label{F-SQG_snap}
  	\end{figure}
  	
  	Figure \ref{F-SQG_snap} shows how surface quasigeostrophic turbulence differs over decreasing [$N'(z) \leq 0$] and increasing [$N'(z) \geq 0$]  stratification. Over decreasing stratification, buoyancy anomalies generate long range velocity fields which strain the surface buoyancy field into thin buoyancy filaments. However, unlike the vorticity filaments in barotropic turbulence, these thin buoyancy filaments are unstable to a secondary instability in which they roll up into small scale vortices \citep{pierrehumbert_spectra_1994,held_surface_1995}. As a result, surface quasigeostrophic turbulence over decreasing stratification is characterized by the simultaneous presence of both thin buoyancy filaments along with vortices having a wide range of scales. In contrast, over increasing stratification, buoyancy anomalies generate shorter range velocity fields that are more efficient at mixing away small-scale inhomogeneities. Consequently, the buoyancy field lacks thin buoyancy filaments and instead appears spatially diffuse. 
	
	The dependence of the interaction range on vertical stratification implies that the surface kinetic energy spectrum must also depend on the vertical stratification. 
	There is a considerable body of literature suggesting that over major extratropical currents in winter, especially over regions where mixed-layer instability is active, the surface geostrophic flow observed by satellite altimeters is due to surface buoyancy anomalies \citep{isernfontanet_three-dimensional_2008,lapeyre_what_2009,gonzalez-haro_global_2014,qiu_reconstructability_2016,qiu_reconstructing_2020,miracca-lage_can_2022}. 
	However, uniformly stratified surface quasigeostrophic theory predicts a surface kinetic energy spectrum that is too shallow to be consistent with the spectra found in numerical models and observations \citep{mensa_seasonality_2013,sasaki_impact_2014,callies_seasonality_2015}. 
	Our finding that mixed-layer like stratification steepens the surface kinetic energy spectrum reconciles these two bodies of literature. 
	It also suggests that mixed-layer baroclinic instability acts to energize the surface buoyancy induced portion of the flow so that the wintertime surface velocity is dominated by contributions from surface buoyancy anomalies. In contrast, the summertime mixed-layer is shallow and mixed-layer instability is either weak or non-existent; observations suggest that the surface buoyancy induced velocity no longer dominates the surface geostrophic flow \citep[although it remains a significant component in some locations, see][]{gonzalez-haro_global_2014}.

	In chapter \ref{Ch-jets}, we consider surface quasigeostrophic turbulence in the presence of a meridional buoyancy gradient, which supports the existence of westward propagating surface-trapped Rossby waves \citep{held_surface_1995}. We find that the vertical stratification controls the dispersion of surface-trapped Rossby waves; for a range of horizontal scales, the Rossby wave dispersion relation can be approximated as 
	\begin{equation}
		\omega(\vec k) \approx - \frac{\Lambda \, k_x }{k^\alpha},
	\end{equation}
	where $\Lambda$ is proportional to the vertical shear at the surface, and the parameter $\alpha$ is the same as the $\alpha$ appearing in the streamfunction expression \eqref{eq:range_sqg}.
	Therefore, Rossby wave dispersion is related to the interaction range of surface buoyancy anomalies. 
	Over decreasing stratification, buoyancy anomalies have a longer interaction range (with $\alpha >1$) and we obtain highly dispersive waves, whereas over increasing stratification, buoyancy anomalies have a shorter interaction range (with $\alpha <1$) and Rossby waves are only weakly dispersive. 
	In the limit where the surface stratification is much larger than the deep ocean stratification, then Rossby waves become non-dispersive (i.e., $\alpha \rightarrow 0$).

	\begin{figure}
	\centerline{\includegraphics[width=\textwidth]{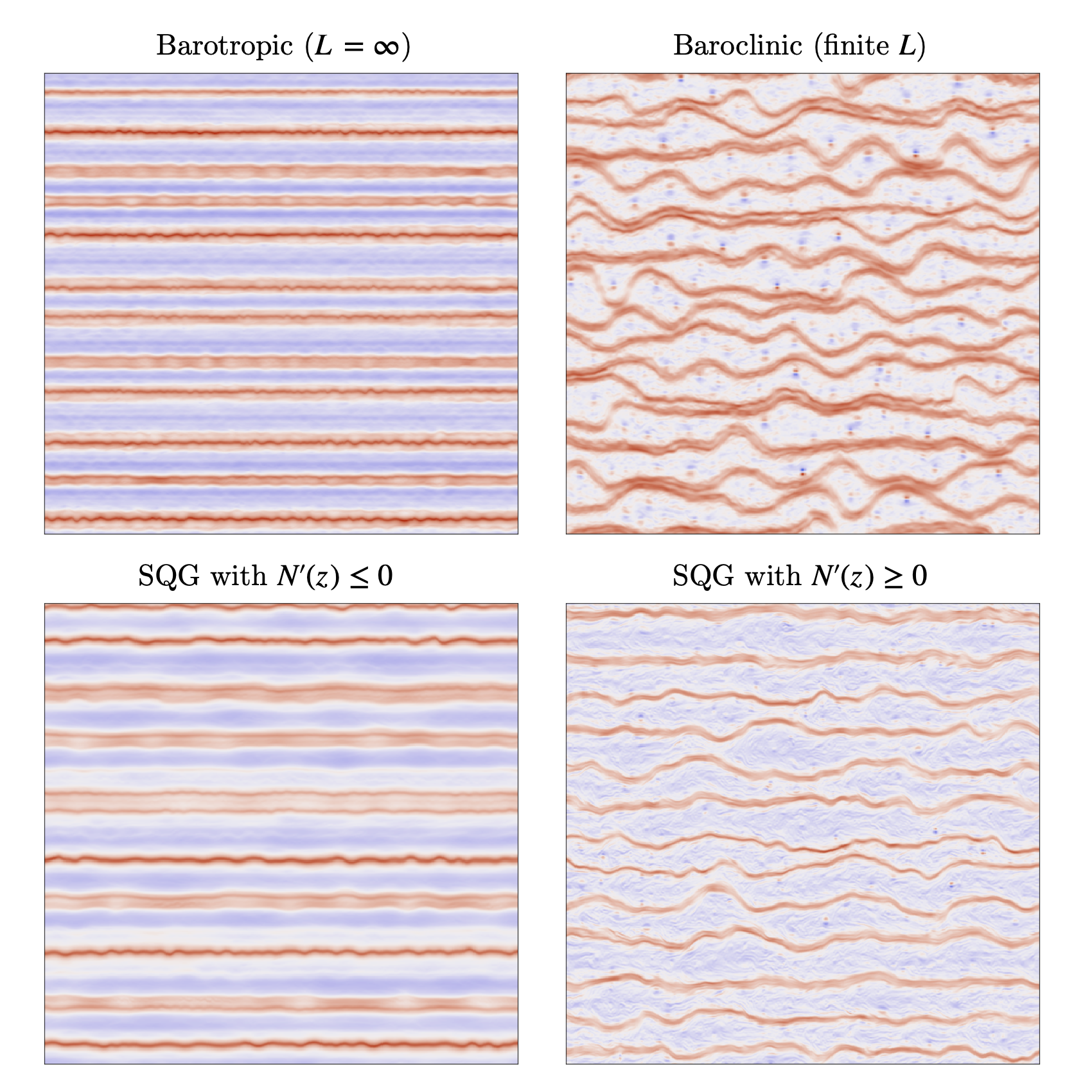}}
	\caption{Jets in two-dimensional turbulence with an infinite and finite deformation radius, $L$, (top row) along with jets in surface quasigeostrophic (SQG) turbulence over decreasing and increasing stratification (bottom row). The zonal velocity is shown, with red values indicating positive (eastward) velocities and blue values indicating negative (westward) velocities. All four simulations are forced at a horizontal scale equal to one eightieth (\nicefrac{1}{80}) of the domain width. The zonal velocity is normalized by its maximum value in each snapshot. }
	\label{F-2D_SQG_jets}
  	\end{figure}

	On a $\beta$-plane, the interaction of Rossby waves with geostrophic turbulence results in latitudinally inhomogeneous mixing that, under certain conditions, spontaneously reorganizes the flow into a staircase structure, consisting of latitudinal zones of homogenized potential vorticity separated by isolated potential vorticity discontinuities \citep{dritschel_multiple_2008,scott_structure_2012,scott_zonal_2019}. In this limit, we obtain sharp eastward jets centred at the potential vorticity discontinuities with westward flows in between. 
	Finite values of the deformation radius result in latitudinally meandering eastward jets having a fixed shape, with the jet width determined by the deformation radius \citep[figure \ref{F-2D_SQG_jets};][]{dritschel_jet_2011,scott_zonal_2019,scott_spacing_2022}.
	In chapter \ref{Ch-jets}, we extend this analysis to a surface quasigeostrophic fluid with a meridional surface buoyancy gradient. Analogously, we find that, under certain conditions, the flow spontaneously reorganizes into a staircase structure consisting of latitudinal zones of homogenized surface buoyancy separated by isolated surface buoyancy discontinuities. 
 	Over decreasing stratification, we obtain straight jets perturbed by highly dispersive, eastward propagating, along jet waves, similar to jets in $\beta$-plane barotropic turbulence. 
 	In contrast, over increasing stratification, we obtain meandering jets whose shape evolves in time due to the westward propagation of weakly dispersive along jet waves (figure \ref{F-2D_SQG_jets}).
	
	\subsection{The second part of the dissertation}
	
	The second part of the dissertation, consisting of chapters \ref{Ch-modes} and \ref{Ch-QG}, concerns an investigation into vertical normal modes in the presence of boundary-confined restoring forces. The ultimate aim of this part of the dissertation is to generalize $N$-layer quasigeostrophic models to account for non-isentropic boundaries. Such a modal truncation would provide a simple model in which to investigate the coupling between boundary-trapped buoyancy induced flows and potential vorticity induced flows in the fluid interior.
	Resolving  boundary buoyancy dynamics in quasigeostrophic models typically requires high vertical resolution near the boundaries \citep{tulloch_note_2009}. Although \cite{tulloch_quasigeostrophic_2009} developed a four-mode model consisting of two surface quasigeostrophic modes nonlinearly coupled to a barotropic and baroclinic mode, because these four modes do not form an orthogonal set, this model does not conserve energy. 
	
	We begin, in chapter \ref{Ch-modes}, by examining the mathematical structure of geophysical waves in the presence of both volume-permeating and boundary-confined restoring forces.  
	If the boundaries are dynamically inert, the resulting eigenvalue problem typically has a Sturm-Liouville form and the properties of such problems are well-known (e.g., the standard baroclinic mode eigenvalue problem with flat boundaries).
	 However, if restoring forces are also present at the boundaries, then the equations of motion contain a time-derivative in the boundary conditions, and this leads to an eigenvalue problem where the eigenvalue correspondingly appears in the boundary conditions.
	 Chapter \ref{Ch-modes} develops the theory of such problems, explores the properties of wave problems with dynamically active boundaries, and provides a precise meaning of what it means for a set of vertical modes to be complete. We then apply the theory to two Boussinesq gravity wave problems as well a Rossby wave problem over topography.
	 
	 Chapter \ref{Ch-QG} then applies the mathematical formalism of chapter \ref{Ch-modes} to obtain all possible discrete normal modes in quasigeostrophic theory that diagonalize the energy and the potential enstrophy. There are two classes of quasigeostrophic normal modes. If the eigenvalue parameter does not appear in the boundary conditions, then we obtain normal modes analogous to the baroclinic modes or the surface modes. That these modes cannot be used to represent every possible quasigeostrophic state can be seen in the following manner. An arbitrary quasigeostrophic state is uniquely determined by specifying the potential vorticity in the fluid interior as well the boundary buoyancy anomalies. However, although we can project the potential vorticity onto the baroclinic modes or the surface modes, we are unable to project the boundary buoyancy onto these modes; either the series expansion does not converge to the buoyancy anomaly or, if it does converge, the resulting series expansion is not differentiable. Either case is physically unacceptable. In contrast, the second class of modes can be used to project an arbitrary potential vorticity profile along with boundary buoyancy anomalies, and the resulting series expansions are differentiable. Consequently, we are able to expand the Bretherton potential vorticity \citep{bretherton_critical_1966} \textemdash{} consisting of $\delta$-sheet potential vorticity contributions at the boundaries \textemdash{} in terms of quasigeostrophic modes. 
	 
	Although the aim behind the analysis of chapters \ref{Ch-modes} and \ref{Ch-QG} was to formulate a modal truncation of the quasigeostrophic equations that accounts for non-isentropic boundaries, we show in chapter \ref{Ch-conc} that no such truncation is possible using discrete quasigeostrophic normal modes. This is because a crucial symmetry in the vertical coupling between vertical modes is lost in the presence of non-isentropic boundaries. As a consequence, if energy is initialized in the first $N$ modes, then energy exchanges are possible with the higher modes, and so, finite modal truncations fail to conserve energy. This argument holds for all possible discrete normal modes that diagonalize the energy and the potential enstrophy (from chapter 5) and so no energy conserving discrete modal truncation for the quasigeostrophic equations is possible in the presence of non-isentropic boundaries. 

%% file: 2-SQG/SQG.tex
\chapter{Surface Quasigeostrophic Turbulence in Variable Stratification}\label{Ch-SQG}

\begin{abstractchapter}
	Numerical and observational evidence indicates that, in regions where mixed-layer instability is active, the surface geostrophic velocity is largely induced by surface buoyancy anomalies. Yet, in these regions, the observed surface kinetic energy spectrum is steeper than predicted by uniformly stratified surface quasigeostrophic theory. By generalizing surface quasigeostrophic theory to account for variable stratification, we show that surface buoyancy anomalies can generate a variety of dynamical regimes depending on the stratification's vertical structure. Buoyancy anomalies generate longer range velocity fields over decreasing stratification and shorter range velocity fields  over increasing stratification. As a result, the surface kinetic energy spectrum is steeper over decreasing stratification than over increasing stratification. An exception occurs if the near surface stratification is much larger than the deep ocean stratification. In this case, we find an extremely local turbulent regime with surface buoyancy homogenization and a steep surface kinetic energy spectrum, similar to equivalent barotropic turbulence. By applying the variable stratification theory to the wintertime North Atlantic, and assuming that mixed-layer instability acts as a narrowband small-scale surface buoyancy forcing, we obtain a predicted surface kinetic energy spectrum between $k^{-4/3}$ and $k^{-7/3}$, which is consistent with the observed wintertime $k^{-2}$ spectrum. We conclude by suggesting a method of measuring the buoyancy frequency's vertical structure using satellite observations.
\end{abstractchapter}

\section{Introduction}\label{S-intro}

\subsection{Geostrophic flow induced by surface buoyancy}

Geostrophic flow in the upper ocean is induced by either interior potential vorticity anomalies, $q$, or surface buoyancy anomalies, $b|_{z=0}$.
At first, it was assumed that the surface geostrophic flow observed by satellite altimeters is due to interior potential vorticity \citep{stammer_global_1997,wunsch_vertical_1997}.
It was later realized, however, that under certain conditions, upper ocean geostrophic flow can be inferred using the surface buoyancy anomaly alone \citep{lapeyre_dynamics_2006,lacasce_estimating_2006}. 
Subsequently, \cite{lapeyre_what_2009} used a numerical ocean model to show that the surface buoyancy induced geostrophic flow dominates the $q$-induced geostrophic flow over a large fraction of the North Atlantic in January.
Lapeyre then concluded that the geostrophic velocity inferred from satellite altimeters in the North Atlantic must usually be due to surface buoyancy anomalies instead of interior potential vorticity.

Similar conclusions have been reached in later numerical studies using the effective surface quasigeostrophic  \citep[eSQG, ][]{lapeyre_dynamics_2006} method. The eSQG method aims to reconstruct three-dimensional velocity fields in the upper ocean: it assumes that surface buoyancy anomalies generate an exponentially decaying streamfunction with a vertical attenuation determined by the buoyancy frequency, as in the uniformly stratified surface quasigeostrophic model \citep{held_surface_1995}.
 Because the upper ocean does not typically have uniform stratification, an "effective" buoyancy frequency is used, which is also intended to account for interior potential vorticity anomalies \citep[][]{lapeyre_dynamics_2006}. 
In practice, however, this effective buoyancy frequency is chosen to be the vertical average of the buoyancy frequency in the upper ocean. 
 \cite{qiu_reconstructability_2016} derived the surface streamfunction from sea surface height in a $\nicefrac{1}{30}^\circ$ model of the Kuroshio Extension region in the North Pacific and used the eSQG method to reconstruct the three-dimensional vorticity field. They found correlations of 0.7-0.9 in the upper 1000 m between the reconstructed and model vorticity throughout the year. This result was also found to hold in a $\nicefrac{1}{48}^\circ$ model with tidal forcing \citep{qiu_reconstructing_2020}.


A clearer test of whether the surface flow is induced by surface buoyancy is to reconstruct the geostrophic flow directly using the sea surface buoyancy or temperature \citep{isernfontanet_potential_2006}. 
This approach was taken by \cite{isernfontanet_three-dimensional_2008} in the context of a $\nicefrac{1}{10}^\circ$ numerical simulation of the North Atlantic. 
When the geostrophic velocity is reconstructed using sea surface temperature, correlations between the reconstructed velocity and the model velocity exceeded 0.7 over most of the North Atlantic in January. 
Subsequently, \cite{miracca-lage_can_2022} used a reanalysis product with a spatial grid spacing of 10 km to reconstruct the geostrophic velocity using both sea surface buoyancy and temperature over certain regions in the South Atlantic. 
The correlations between the reconstructed streamfunctions and the model streamfunction had a seasonal dependence, with correlations of 0.7-0.8 in winter and 0.2-0.4 in summer. 

Observations also support the conclusion that a significant portion of the surface geostrophic flow may be due to surface buoyancy anomalies over a substantial fraction of the World Ocean.
 \cite{gonzalez-haro_global_2014} reconstructed the surface streamfunction using $\nicefrac{1}{3}^\circ$ satellite altimeter data (for sea surface height) and $\nicefrac{1}{4}^\circ$ microwave radiometer data (for sea surface temperature). 
If the surface geostrophic velocity is due to sea surface temperature alone, then the streamfunction constructed from sea surface temperature should be perfectly correlated with the streamfunction constructed from sea surface height. 
The spatial correlations between the two streamfunctions was found to be seasonal. 
For the wintertime Northern hemisphere, high correlations (exceeding 0.7-0.8) are observed near the Gulf Stream and Kuroshio whereas lower correlations (0.5-0.6) are seen in the eastern recirculating branch of North Atlantic and North Pacific gyres [a similar pattern was found by \cite{isernfontanet_three-dimensional_2008} and \cite{lapeyre_what_2009}]. In summer, correlations over the North Atlantic and North Pacific plummet to 0.2-0.5, again with lower correlations in the recirculating branch of the gyres.
In contrast to the strong seasonality observed in the northern hemisphere, correlation over the Southern Ocean typically remain larger than 0.8 throughout the year.

Another finding is that more of the surface geostrophic flow is due to surface buoyancy anomalies in regions with high eddy kinetic energy, strong thermal gradients, and deep mixed layers \citep{isernfontanet_three-dimensional_2008, gonzalez-haro_global_2014, miracca-lage_can_2022}. 
These are the same conditions under which we expect mixed-layer baroclinic instability to be active \citep{boccaletti_mixed_2007,mensa_seasonality_2013,sasaki_impact_2014,callies_seasonality_2015}. Indeed, one model of mixed-layer instability consists of surface buoyancy anomalies interacting with interior potential vorticity anomalies at the base of the mixed-layer \citep{Callies2016}. 
The dominance of the surface buoyancy induced velocity in regions of mixed-layer instability suggests that, to a first approximation, we can think of mixed-layer instability as energizing the surface buoyancy induced part of the flow through vertical buoyancy fluxes and the concomitant release of kinetic energy at smaller scales. 


\subsection{Surface quasigeostrophy in uniform stratification}

The dominance of the surface buoyancy induced velocity suggests that a useful model for upper ocean geostrophic dynamics is the surface quasigeostrophic model \citep{held_surface_1995}, which describes the dynamics induced by surface buoyancy anomalies over uniform stratification.
 The primary difference between surface quasigeostrophic dynamics and two-dimensional barotropic dynamics \citep{kraichnan_inertial_1967} is that surface quasigeostrophic eddies have a shorter interaction range than their two-dimensional barotropic counterparts. 
 One consequence of this shorter interaction range is a flatter kinetic energy spectrum \citep{pierrehumbert_spectra_1994}. Letting $k$ be the horizontal wavenumber, then two-dimensional barotropic turbulence theory predicts a kinetic energy spectrum of $k^{-5/3}$ upscale of small-scale forcing and a $k^{-3}$ spectrum downscale of large-scale forcing \citep{kraichnan_inertial_1967}. 
 If both types of forcing are present, then we expect a spectrum between $k^{-5/3}$ and $k^{-3}$, with the realized spectrum depending on the relative magnitude of small-scale to large-scale forcing \citep{lilly_two-dimensional_1989,maltrud_energy_1991}. 
 In contrast, the corresponding spectra for surface quasigeostrophic turbulence are $k^{-1}$ (upscale of small-scale forcing) and $k^{-5/3}$ (downscale of large-scale forcing) \citep{blumen_uniform_1978}, both of which are flatter than the corresponding two-dimensional barotropic spectra.\footnote{The uniformly stratified geostrophic turbulence theory of \cite{charney_geostrophic_1971} provides spectral predictions similar to the two-dimensional barotropic theory \citep[See ][]{callies_interpreting_2013}.} 
 
The above considerations lead to the first discrepancy between the surface quasigeostrophic model and ocean observations. 
As we have seen, we expect wintertime surface geostrophic velocities near major extratropical currents to be primarily due to surface buoyancy anomalies.
Therefore, the predictions of surface quasigeostrophic theory should hold.
If we assume that mesoscale baroclinic instability acts as a large-scale forcing and that mixed-layer baroclinic instability acts as a small-scale forcing to the upper ocean \citep[we assume a narrowband forcing in both cases, although this may not be the case, see][]{khatri_role_2021}, then we expect a surface kinetic energy spectrum between $k^{-1}$ and $k^{-5/3}$.
However, both observations and numerical simulations of the Gulf Stream and Kuroshio find kinetic energy spectra close to $k^{-2}$ in winter  \citep{sasaki_impact_2014,callies_seasonality_2015,vergara_revised_2019}, which is steeper than predicted.


The second discrepancy relates to the surface transfer function implied by uniformly stratified surface quasigeostrophic theory. The surface transfer function, $\mathcal{F}(\vec k)$, is defined as \citep{isern-fontanet_transfer_2014}
\begin{equation}\label{eq:transfer_func}
	\mathcal{F}(\vec k) = \frac{\hat \psi_{\vec k}}{\hat b_{\vec k}},
\end{equation}
where $\hat \psi_{\vec k}$ and $\hat b_{\vec k}$ are the Fourier amplitudes of the geostrophic streamfunction, $\psi$, and the buoyancy, $b$, at the ocean's surface, and $\vec k$ is the horizontal wavevector. Uniformly stratified surface quasigeostrophic theory predicts an isotropic transfer function $\mathcal{F}(k)\sim k^{-1}$ \citep{held_surface_1995}. Using a $\nicefrac{1}{12}^\circ$ ocean model and focusing on the western coast of Australia, \cite{gonzalez-haro_ocean_2020} confirmed that the transfer function between sea surface temperature the sea surface height is indeed isotropic but found that the transfer function is generally steeper than $k^{-1}$. In another study using a $\nicefrac{1}{16}^\circ$ model of the Mediterranean Sea, \cite{isern-fontanet_transfer_2014} found that the transfer function below 100 km has seasonal dependence closely related to mixed-layer depth, fluctuating between $k^{-1}$ and $k^{-2}$. 

In the remainder of this chapter, we account for these discrepancies by generalizing the uniformly stratified surface quasigeostrophic model \citep{held_surface_1995} to account for variable stratification (section \ref{S-inversion}). Generally, we find that the kinetic energy spectrum in surface quasigeostrophic turbulence depends on the stratification's vertical structure (section \ref{S-turbulence}); we recover the \cite{blumen_uniform_1978} spectral predictions only in the limit of uniform stratification. Stratification controls the kinetic energy spectrum by modifying the interaction range of surface quasigeostrophic eddies, and we illustrate this dependence by examining the turbulence under various idealized stratification profiles (section \ref{S-idealized}). We then apply the theory to the North Atlantic in both winter and summer, and find that the surface transfer function is seasonal, with a $\mathcal{F}(k) \sim k^{-3/2}$ dependence in winter and a $\mathcal{F}(k) \sim k^{-1/2}$ dependence in summer. Moreover, in the wintertime North Atlantic, the theory predicts a surface kinetic energy spectrum between $k^{-4/3}$ and $k^{-7/3}$, which is consistent with both observations and numerical simulations (section \ref{S-ECCO}). Finally, in section 6, we discuss the validity of theory at other times and locations.

\section{The inversion function} \label{S-inversion}

\subsection{Physical space equations}

Consider an ocean of depth $H$ with zero interior potential vorticity $(q=0)$ so that the geostrophic streamfunction satisfies 
\begin{equation}\label{eq:zero_pv}
	\lap \psi  + \pd{}{z}\left(\frac{1}{\sigma^2}\pd{\psi}{z}\right) = 0 \quad \text{for } z\in(-H,0).
\end{equation}
In this equation, $\lap$ is the horizontal Laplacian, $\psi$ is the geostrophic streamfunction, and 
\begin{equation}
	\sigma(z) = N(z)/f,
\end{equation}
where $N(z)$ is the depth-dependent buoyancy frequency and $f$ is the constant local value of the Coriolis frequency. We refer to $\sigma(z)$ as the \emph{stratification} for the remainder of this chapter. The horizontal geostrophic velocity is obtained from $\vec u = \unit z \times \grad \vec \psi$ where $\unit z$ is the vertical unit vector.

The upper surface potential vorticity is given by \citep{bretherton_critical_1966}
\begin{equation}\label{eq:theta_def}
	\theta = - \frac{1}{\sigma_0^2} \pd{\psi}{z}\bigg|_{z=0},
\end{equation}
where $\sigma_0 = \sigma(0)$. The surface potential vorticity is related to the surface buoyancy anomaly through 
\begin{equation}
	b|_{z=0} = -f\,\sigma_0^2\,\theta.
\end{equation}
The time-evolution equation at the upper boundary is given by
\begin{equation}\label{eq:theta_equation}
	\pd{\theta}{t} + \J{\psi}{\theta} = F-D \quad \text{at } z=0,
\end{equation}
where $\J{\theta}{\psi} = \partial_x \theta \, \partial_y \psi - \partial_y \theta \, \partial_x \psi$ represents the advection of $\theta$ by the horizontal geostrophic velocity $\vec u$, $F$ is the buoyancy forcing at the upper boundary, and $D$ is the dissipation. 

We assume a bottom boundary condition of
\begin{equation}
	\psi \rightarrow 0 \text{ as } z\rightarrow-\infty,
\end{equation}
which is equivalent to assuming the bottom boundary, $z=-H$, is irrelevant to the dynamics. In section \ref{S-ECCO}, we find that this assumption is valid in the mid-latitude North Atlantic open ocean at horizontal scales smaller than $\approx 500$ km. We consider alternative boundary conditions in appendix A.

\subsection{Fourier space equations}

Assuming a doubly periodic domain in the horizontal prompts us to consider the Fourier expansion of $\psi$,
\begin{equation}\label{eq:fourier}
	\psi(\vec r, z,t) = \sum_{\vec k} \hat \psi_{\vec k}(t) \, \Psi_{k}(z) \, \e^{\irm \vec k \cdot \vec r},
\end{equation}
where $\vec r = (x,y)$ is the horizontal position vector, $\vec k = (k_x, k_y)$ is the horizontal wavevector,  and $k=\abs{\vec k}$ is the horizontal wavenumber. The wavenumber dependent non-dimensional vertical structure, $\Psi_{k}(z)$, is determined by the boundary-value problem\footnote{To derive the vertical structure equation \eqref{eq:Psi_equation}, substitute the Fourier representation \eqref{eq:fourier} into the vanishing potential vorticity condition \eqref{eq:zero_pv}, multiply by $\mathrm{e}^{-i\vec l \cdot \vec r}$, take an area average, and use the identity $$ \frac{1}{A} \, \int_{A} \mathrm{e}^{\irm \left(\vec k-\vec l\right)\cdot \vec r} \, \mathrm{d} {\vec r} =  \delta_{\vec k,\vec l} $$ where $\delta_{\vec k, \vec l}$ is the Kronecker delta, and $A$ is the horizontal area of the periodic domain.} 
\begin{equation}\label{eq:Psi_equation}
	- \d{}{z} \left(\frac{1}{\sigma^2} \d{\Psi_k}{z} \right) + k^2 \, \Psi_k = 0,
\end{equation}
with upper boundary condition
\begin{equation}\label{eq:Psi_upper}
	\Psi_k(0) = 1
\end{equation}
and bottom boundary condition 
\begin{equation}\label{eq:Psi_lower}
	\Psi_k \rightarrow 0 \text{ as } z\rightarrow-\infty.
\end{equation}
The upper boundary condition \eqref{eq:Psi_upper} is a normalization for the vertical structure, $\Psi_k(z)$, chosen so that
\begin{equation}\label{eq:fourier_zero}
	\psi(\vec r, z=0,t) = \sum_{\vec k} \hat \psi_{\vec k}(t) \, \e^{\irm \vec k \cdot \vec r}.
\end{equation}
Consequently, the surface potential vorticity \eqref{eq:theta_def} is given by
\begin{equation}\label{eq:theta_fourier}
	\theta(\vec r, t) = \sum_{\vec k} \hat \theta_{\vec k}(t) \, \e^{\irm \vec k \cdot \vec r},
\end{equation}
where
\begin{equation}\label{eq:theta_inversion}
	\thetak = - m(k) \, \psik,
\end{equation}
and the inversion function $m(k)$ (with dimensions of inverse length) is defined as
\begin{equation}\label{eq:mk}
	m(k)= \frac{1}{\sigma_0^2} \, \d{\Psi_k(0)}{z}.
\end{equation}
In all our applications, we find the inversion function to be a positive monotonically increasing function of $k$ [i.e., $m(k)>0$ and $\mathrm{d}m/\mathrm{d}k\geq0$]. 
The inversion function is related to the transfer function \eqref{eq:transfer_func} through 
\begin{equation}\label{eq:transfer_inversion}
	\mathcal{F}(k) = \frac{1}{f \, \sigma_0^2 \, m(k)} =  \left[ f \, \d{\Psi_k(0)}{z} \right]^{-1},
\end{equation}
which shows that the transfer function, evaluated at a wavenumber $k$, is related to the characteristic vertical scale of $\Psi_k(z)$.


\subsection{The case of constant stratification}\label{SS-constant_strat}

To recover the well-known case of the uniformly stratified surface quasigeostrophic model \citep{held_surface_1995}, set $\sigma = \sigma_0$. Then  solving the vertical structure equation \eqref{eq:Psi_equation} along with boundary conditions \eqref{eq:Psi_upper} and \eqref{eq:Psi_lower} yields the exponentially decaying vertical structure, 
\begin{equation}
	\Psi_k(z) = \e^{\sigma_0 \, k \, z}.
\end{equation}
Substituting $\Psi_k(z)$ into the definition of the inversion function \eqref{eq:mk}, we obtain
 \begin{equation}
	 m(k) = k/\sigma_0,
\end{equation}
and hence [through the inversion relation \eqref{eq:theta_fourier}] a linear-in-wavenumber inversion relation of 
\begin{equation}
	\hat \theta_{\vec k} = -(k/\sigma_0) \, \hat \psi_{\vec k}.
\end{equation}

In appendix A, we show that $m(k) \rightarrow k/\sigma_0$ as $k\rightarrow \infty$ for arbitrary stratification $\sigma(z)$. Therefore, at sufficiently small horizontal scales, surface quasigeostrophic dynamics behaves as in constant stratification regardless of the functional form of $\sigma(z)$.

\section{Surface quasigeostrophic turbulence}\label{S-turbulence} 

Suppose a two-dimensional barotropic fluid is forced in the wavenumber interval $[k_1,k_2]$. In such a fluid, \cite{kraichnan_inertial_1967} argued that two inertial ranges will form: one inertial range for $k<k_1$ where kinetic energy cascades to larger scales and another inertial range for $k>k_2$ where enstrophy cascades to smaller scales. Kraichnan's argument depends on three properties of two-dimensional vorticity dynamics. First, that there are two independent conserved quantities; namely, the kinetic energy and the enstrophy. Second, that turbulence is sufficiently local in wavenumber space   so that the only available length scale is $k^{-1}$. Third, that the inversion relation between the vorticity and the streamfunction is scale invariant. 

There are two independent conserved quantities in surface quasigeostrophic dynamics, as in Kraichnan's two-dimensional fluid; namely the total energy, $E$, and the surface potential enstrophy, $P$. 
However, the second and third properties of two-dimensional vorticity dynamics do not hold for surface quasigeostrophic dynamics. Even if the turbulence is local in wavenumber space, there are two available length scales at each wavenumber $k$; namely, $k^{-1}$ and $[m(k)]^{-1}$. Moreover, the inversion relation \eqref{eq:theta_inversion} is generally not scale invariant.\footnote{A function $m(k)$ is scale invariant if $m(\lambda k) = \lambda^s m(k)$ for all $\lambda$, where $s$ is a real number. In particular, note that power laws, $m(k) = k^\alpha$, are scale invariant.} Therefore, the arguments in \cite{kraichnan_inertial_1967} do not hold in general for surface quasigeostrophic dynamics.

Even so, in the remainder of this section we show that there is a net inverse cascade of total energy and a net forward cascade of surface potential enstrophy even if there are no inertial ranges in the turbulence. Then we consider the circumstances under which we expect inertial ranges to form. Finally, assuming the existence of an inertial range, we derive the spectra for the cascading quantities. We begin, however, with some definitions.

\subsection{Quadratic quantities}
The two quadratic quantities needed for the cascade argument are the volume-integrated total mechanical energy per mass per unit area, 
\begin{equation}\label{eq:total_energy}
\begin{aligned}
	E &= \frac{1}{2\,A} \int_V \left( \abs{\grad \psi}^2 + \frac{1}{\sigma^{2}} \abs{\pd{\psi}{z}}^2 \right)\mathrm{d}V \\
	  &=  - \frac{1}{2}\,\overline{\psi|_{z=0} \, \theta} = \frac{1}{2} \sum_{\vec k} m(k) \, \abs{\hat \psi_{\vec k}}^2,
\end{aligned}
\end{equation}
and the upper surface potential enstrophy, 
\begin{equation}\label{eq:potential_enstrophy}
	P = \frac{1}{2} \, \overline{\theta^2} = \frac{1}{2} \sum_{\vec k} \left[m(k)\right]^2 \abs{\hat \psi_{\vec k}}^2,
\end{equation}
where the overline denotes an area average over the periodic domain.
Both quantities are time independent in the absence of forcing and dissipation, as can be seen by multiplying the time-evolution equation \eqref{eq:theta_equation} by either $-\psi|_{z=0}$ or $\theta$ and taking an area average.




Two other quadratics we use are the surface kinetic energy
\begin{equation}\label{eq:kinetic_energy}
	K = \frac{1}{2} \overline{\abs{\grad \psi}_{z=0}^2} = \frac{1}{2} \sum_{\vec k} k^2 \, \abs{\hat \psi_{\vec k}}^2
\end{equation}
and the surface streamfunction variance
\begin{equation}\label{eq:stream_var}
	S = \frac{1}{2} \overline{\left(\psi|_{z=0}\right)^2} = \frac{1}{2} \sum_{\vec k} \abs{\hat \psi_{\vec k}}^2.
\end{equation}

The isotropic spectrum $\mathscr{A}(k)$ of a quantity $A$ is defined by
\begin{equation}
	A = \int_{0}^\infty \mathscr{A}(k) \, \di k,
\end{equation}
so that the isotropic spectra of $E, P, K,$ and $S$ are given by $\Esc(k), \Psc(k), \Ksc(k),$ and $\Ssc(k)$. The isotropic spectra are then related by
\begin{equation}\label{eq:variance_relations1}
	\Psc(k) = m(k) \, \Esc(k) = \left[m(k)\right]^2 \, \Ssc(k)
\end{equation}
and 
\begin{equation}\label{eq:variance_relations2}
	\Ksc(k) = k^2 \, \Ssc(k).
\end{equation}

For each spectral density $\mathscr{A}(k)$, there is a time-evolution equation of the form \citep{gkioulekas_new_2007}
\begin{equation}\label{eq:spectral_transfer}
	\pd{\mathscr{A}(k)}{t} + \pd{\Pi_A(k)}{k} =F_A(k) - D_A(k),
\end{equation}
where $\Pi_A(k)$ is the transfer of the spectral density $\mathscr{A}(k)$ from $(0,k)$ to $(k,\infty)$, and $D_A(k)$ and $F_A(k)$ are the dissipation and forcing spectra of $A$, respectively. In an inertial range where $A$ is the cascading quantity, then $\Pi_A(k) = \varepsilon_A$ where $\varepsilon_A$ is a constant and thus $\partial \Pi_A(k)/\partial k = 0$.

\subsection{The inverse and forward cascade}

For a fluid with the variable stratification inversion relation \eqref{eq:theta_inversion} that is forced in the wavenumber interval $[k_1,k_2]$, \cite{gkioulekas_new_2007} prove the following two inequalities for stationary turbulence,
\begin{align}
	\label{eq:gk_1}
	\int_{0}^k \d{m(k')}{k'} \,\Pi_{E}(k') \, \di{k'} < 0, \, \, \text{for all } k>k_2,\\
	\label{eq:gk_2}
	\int_{k}^\infty \d{m(k')}{k'} \, \frac{\Pi_P(k')}{[m(k')]^2} \,\di{k'} > 0, \, \, \text{for all } k<k_1.
\end{align}
These two inequalities do not require the existence of inertial ranges, only that the inversion function $m(k)$ is an increasing function of $k$.
 Therefore, if $\mathrm{d}m(k)/\mathrm{d}k>0$, then there is a net inverse cascade of total energy and a net forward cascade of surface potential enstrophy.

\subsection{When do inertial ranges form?}

The lack of scale invariance along with the presence of two length scales, $k^{-1}$ and $[m(k)]^{-1}$, prevents the use of the \cite{kraichnan_inertial_1967} argument to establish the existence of an inertial range. However, suppose that
in a wavenumber interval, $[k_a,k_b]$, the inversion function takes the power law form
\begin{equation}\label{eq:mk_power}
	m(k) \approx m_\alpha \, k^{\alpha}, 
\end{equation}
where $m_\alpha > 0$ and $\alpha >0$.
Then, in this wavenumber interval, the inversion relation takes the form of the $\alpha$-turbulence inversion relation \citep{pierrehumbert_spectra_1994}, 
\begin{equation}\label{eq:alpha_inv}
		\hat \xi_{\vec k} = - k^{\alpha} \, \hat \psi_{\vec k},
\end{equation}
with $\xi = \theta/m_\alpha$.
The inversion relation \eqref{eq:alpha_inv} is then scale invariant in the wavenumber interval $[k_a,k_b]$. Moreover, $k^{-1}$ is the only available length scale if the turbulence is sufficiently local in wavenumber space. It follows that if the wavenumber interval $[k_a,k_b]$ is sufficiently wide (i.e., $k_a \ll k_b$), then Kraichnan's argument applies to the turbulence over this wavenumber interval and inertial ranges are expected to form. 




\subsection{The Tulloch and Smith (2006) argument}

If we assume the existence of inertial ranges, then we can adapt the cascade argument of \cite{tulloch_theory_2006} to general surface quasigeostrophic fluids to obtain predictions for the cascade spectra.

In the inverse cascade inertial range, we must have $\Pi_E(k)= \varepsilon_E$ where  $\varepsilon_E$ is a constant. Assuming locality in wavenumber space, we have
\begin{equation}\label{eq:epsilon_E}
	\varepsilon_E \sim  \frac{k \, \Esc(k)}{\tau(k)},
\end{equation}
where $\tau(k)$ is a spectrally local timescale\footnote{A spectrally local timescale is appropriate so long as $m(k)$ grows less quickly than $k^2$. Otherwise, a non-local timescale must be used \citep{kraichnan_inertial-range_1971,watanabe_unified_2004}.}. If we further assume that the timescale $\tau(k)$ is determined by the kinetic energy spectrum, $\Ksc(k)$, then dimensional consistency requires 
\begin{equation}\label{eq:turbulent_time}
	\tau(k) \sim \left[k^3 \, \Ksc(k) \right]^{-1/2}.
\end{equation}
Substituting this timescale into equation \eqref{eq:epsilon_E} and using the relationship between the energy spectrum, $\Esc(k)$, and the streamfunction variance spectrum, $\Ssc(k)$, in equations \eqref{eq:variance_relations1} and  \eqref{eq:variance_relations2}, we obtain the total energy spectrum in the inverse cascade inertial range,
\begin{equation}\label{eq:Espec_inv}
	\Esc(k) \sim \varepsilon_E^{2/3} \, k^{-7/3} \, \left[m(k)\right]^{1/3}.
\end{equation}
Analogously, in the forward cascade inertial range, we must have $\Pi_P(k) = \varepsilon_P$ where $\varepsilon_P$ is a constant. 
A similar argument yields the surface potential enstrophy spectrum in the forward cascade inertial range,
\begin{equation}\label{eq:Pspec_forward}
	\Psc(k) \sim \varepsilon_P^{2/3} \, k^{-7/3} \, \left[m(k)\right]^{2/3}.
\end{equation}

The predicted spectra \eqref{eq:Espec_inv} and \eqref{eq:Pspec_forward} are not uniquely determined by dimensional analysis. Rather than assuming that the spectrally local timescale $\tau(k)$ is determined by the kinetic energy spectrum, $\Ksc(k)$, we can assume that $\tau(k)$ is determined by the total energy spectrum, $\Esc(k)$, or the surface potential enstrophy spectrum, $\Psc(k)$.\footnote{These assumptions lead to timescales of  $\tau(k) \sim \left[k^4 \, \Esc(k)\right]^{-1/2}$ and $\tau(k) \sim \left[k^3 \, \Psc(k)\right]^{-1/2}$, respectively.}
Either choice will result in cascade spectra distinct from \eqref{eq:Espec_inv} and \eqref{eq:Pspec_forward}. However, by assuming that the timescale $\tau(k)$ is determined by the kinetic energy spectrum, the resulting cascade spectra agree with the $\alpha$-turbulence predictions of \cite{pierrehumbert_spectra_1994} when the inversion function takes the power law form \eqref{eq:mk_power}.

For later reference, we provide the expressions for the inverse and forward cascade surface kinetic energy spectra.  Using either the inverse cascade spectrum \eqref{eq:Espec_inv} or forward cascade spectrum \eqref{eq:Pspec_forward} along with the relations between the various spectra [equations \eqref{eq:variance_relations1} and \eqref{eq:variance_relations2}], we obtain 
	\begin{equation}\label{eq:inverse_KE}
		\Ksc(k) \sim \varepsilon_E^{2/3} \, k^{-1/3} \, \left[m(k)\right]^{-2/3}
	\end{equation}
	in the inverse cascade and 
	\begin{equation}\label{eq:forward_KE}
			\Ksc(k) \sim \varepsilon_P^{2/3} \, k^{-1/3} \, \left[m(k)\right]^{-4/3}
	\end{equation}
	in the forward cascade. 

Finally, we note that the vorticity spectrum,
	\begin{equation}
		\mathscr{Z}(k) = k^2 \, \mathscr{K}(k),
	\end{equation}
	is an increasing function of $k$ if $m(k)$ is flatter than $k^{5/4}$. In particular, at small scales, we expect $m(k) \sim k$ [section \ref{SS-constant_strat}], implying a vorticity spectrum of $\mathscr{Z}(k) \sim k^{1/3}$. Such an increasing vorticity spectrum implies high Rossby numbers and the breakdown of geostrophic balance at small scales.

\section{Idealized stratification profiles}\label{S-idealized}

In this section we provide analytical solutions for $m(k)$ in the cases of an increasing and decreasing piecewise constant stratification profiles as well as in the case of exponential stratification. These idealized stratification profiles provide intuition for the inversion function's functional form in the case of an arbitrary stratification profile, $\sigma(z)$.

\subsection{Piecewise constant stratification}

	\begin{figure*}
	\centerline{\includegraphics[width=\textwidth]{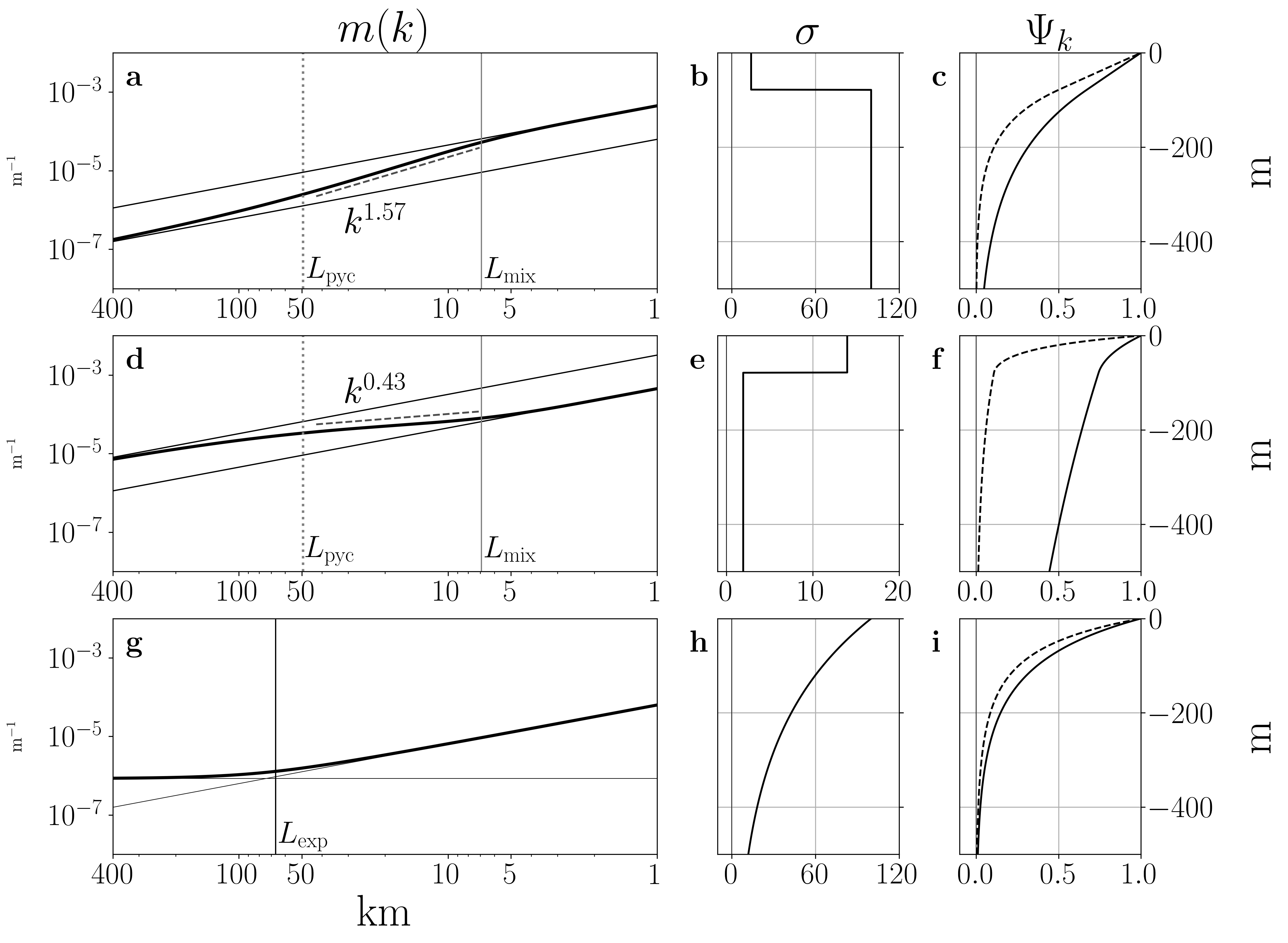}}
	\caption{Log-log plots of the inversion function, $m(k)$ [panels (a), (d), and (g)], for three stratification profiles [panels (b), (e), and (h)] and the resulting streamfunctions at the two horizontal length scales of 50 km (dashed) and 100 km (solid) [for panels (c) and (i)] or 2 km and 10 km [panel (f)]. In the first two inversion function plots [panels (a) and (d)], the thin solid diagonal lines represent the two linear asymptotic states of $k/\sigma_0$ and $k/\sigma_\mathrm{pyc}$. The vertical solid line is the mixed-layer length scale $L_\mathrm{mix}$, given by equation \eqref{eq:mixed_length}, whereas the vertical dotted line is the pycnocline length scale $L_\mathrm{pyc}$, given by equation \eqref{eq:therm_length}. The power $\alpha$, where $m(k)/k^{\alpha} \approx \mathrm{constant}$, is computed by fitting a straight line to the log-log plot of $m(k)$ between $2\pi/L_\mathrm{mix}$ and $2\pi/L_\mathrm{pyc}$. This straight line is shown as a grey dashed line in panels (a) and (d). In panel (g), the thin  diagonal line is the linear small-scale limit, $m(k)\approx k/\sigma_0$, whereas the thin horizontal line is the constant large-scale limit, $m(k) = 2/(\sigma_0^2 \, h_\mathrm{exp})$. Finally, the solid vertical lines in panel (g) indicate the horizontal length scale $L_\mathrm{exp} = 2\pi/k_\mathrm{exp}$ [equation \eqref{eq:kexp}] induced by the exponential stratification. Further details on the stratification profiles are in the text.}
	\label{F-mk_constlinconst}
  	\end{figure*}

	Consider the piecewise constant stratification profile, given by
		\begin{equation}\label{eq:step_strat}
		\sigma(z) = 
		\begin{cases}
			\sigma_0  \, &\text{for } -h < z \leq 0 \\
			\sigma_\mathrm{pyc} \, &\text{for } \infty < z \leq -h.
		\end{cases}
	\end{equation}
	This stratification profile consists of an upper layer of thickness $h$ with constant stratification $\sigma_0$ overlying an infinitely deep layer with constant stratification $\sigma_\mathrm{pyc}$.  If $\sigma_0 < \sigma_\mathrm{pyc}$, then this stratification profile is an idealization of a weakly stratified mixed-layer overlying an ocean of stronger stratification. 
     See panels (b) and (e) in figure \ref{F-mk_constlinconst} for an illustration.  
	
	
	For this stratification profile, an analytical solution is possible, with the solution provided in appendix B. The resulting inversion function is
	\begin{equation}
		m(k) = 
		\frac{k}{\sigma_0} \left[ \frac{\cosh\left(\sigma_0 h k\right) + \left(\frac{\sigma_\mathrm{pyc}}{\sigma_0}\right) \sinh\left(\sigma_0  h k\right) }{\sinh\left(\sigma_0  hk\right) + \left(\frac{\sigma_\mathrm{pyc}}{\sigma_0}\right) \cosh\left(\sigma_0  h k\right) } \right].
	\end{equation}
	At small horizontal scales, $2\pi/k \ll L_\mathrm{mix}$, where 
	\begin{equation}\label{eq:mixed_length}
		L_\mathrm{mix} = 2\,\pi \, \sigma_0 \, h,
	\end{equation}
	the inversion function takes the form $m(k) \approx k/\sigma_0$, as expected from the uniformly stratified theory \citep{held_surface_1995}. At large horizontal scales, $2\pi/k \gg L_\mathrm{pyc}$, where 
	\begin{equation}\label{eq:therm_length}
	 L_\mathrm{pyc} =  2 \, \pi 
	 \begin{cases}
	 		  \, \sigma_\mathrm{pyc} \, h  \, &\text{if } \sigma_0 \leq \sigma_\mathrm{pyc} \\
	 		\sigma_\mathrm{0}^2 \, h /\sigma_\mathrm{pyc}  \, &\text{if } \sigma_0 > \sigma_\mathrm{pyc},
	 \end{cases}
	\end{equation}
	then the inversion function takes the form $m(k) \approx k/\sigma_\mathrm{pyc}$, because at large horizontal scales, the ocean will seem to have constant stratification $\sigma_\mathrm{pyc}$.
	
	The functional form of the inversion function at horizontal scales between $L_\mathrm{mix}$ and $L_\mathrm{pyc}$ depends on whether $\sigma(z)$ is an increasing or decreasing function. If $\sigma(z)$ is a decreasing function, with $\sigma_0 < \sigma_\mathrm{pyc}$, then we obtain a mixed-layer like stratification profile and the inversion function steepens to a super linear wavenumber dependence at these scales. An example is shown in figure \ref{F-mk_constlinconst}(a)-(b). Here, the stratification abruptly jumps from a value of $\sigma_0 \approx 14$ to $\sigma_\mathrm{pyc} = 100$ at $z\approx-79$ m. Consequently, the inversion function takes the form $m(k) \sim k^{1.57}$ between $2\pi/L_\mathrm{pyc}$ and $2\pi/L_\mathrm{mix}$.
	In contrast, if $\sigma_0 > \sigma_\mathrm{pyc}$ then the inversion function flattens to a sublinear wavenumber dependence for horizontal scales between $L_\mathrm{mix}$ and $L_\mathrm{pyc}$. An example is shown in figure \ref{F-mk_constlinconst}(d)-(e), where the stratification abruptly jumps from $\sigma_0 \approx 14$ to $\sigma_\mathrm{pyc} \approx 2$ at $z\approx-79$ m. In this case, the inversion function has a sublinear wavenumber dependence, $m(k) \sim k^{0.43}$, between $2\pi/L_\mathrm{pyc}$ and $2\pi/L_\mathrm{mix}$.
		
	\begin{figure*}
	\centerline{\includegraphics[width=\textwidth]{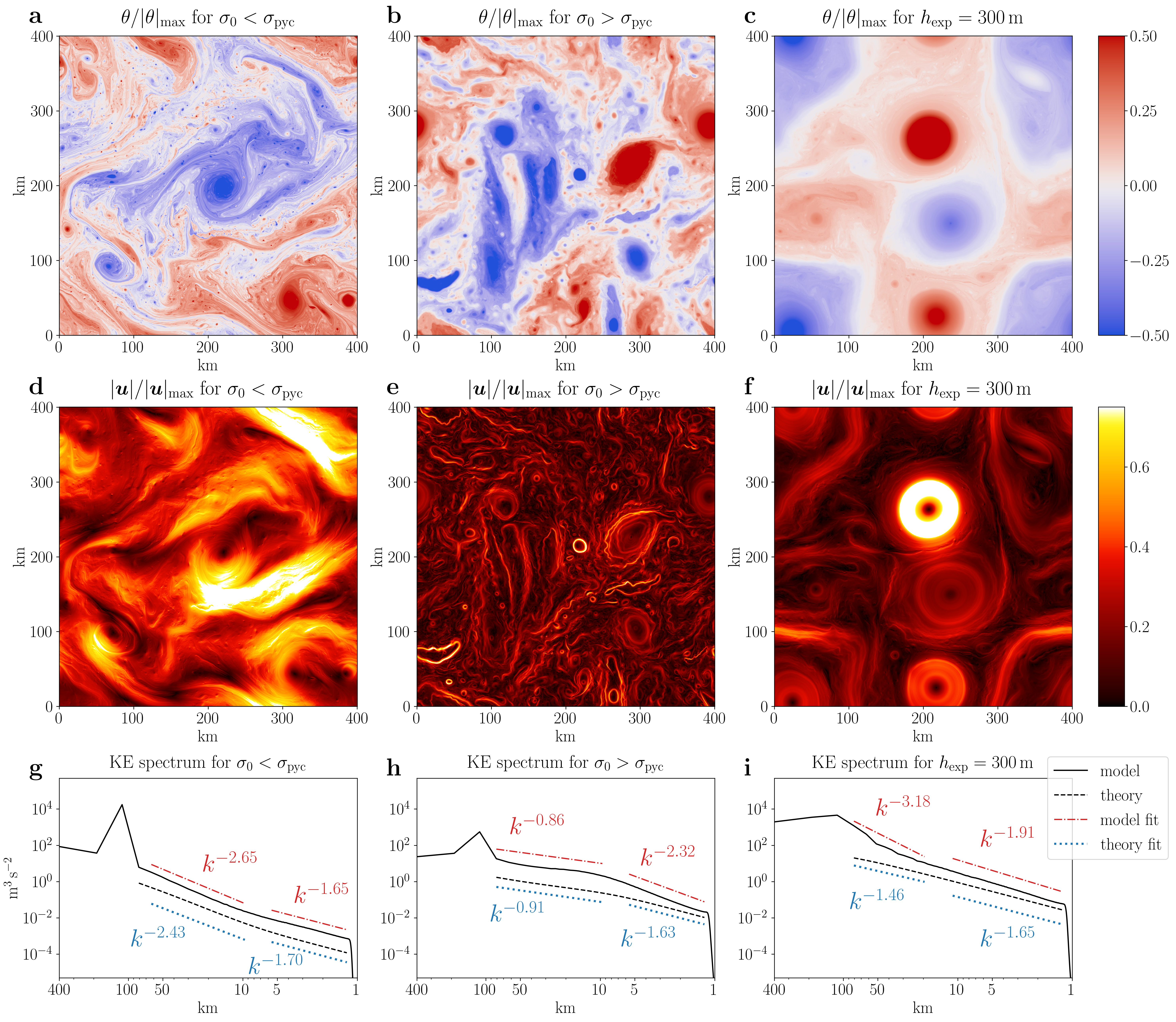}}
	\caption{Results of three pseudo-spectral simulations, forced at approximately 100 km, with $1024^2$ horizontal grid points. See appendix C for a description of the numerical model. The first simulation [panels (a), (d), and (g)] corresponds to the stratification profile and inversion function shown in figure \ref{F-mk_constlinconst}(a)-(b), the second simulations [panels (b), (e), and (h)] corresponds to the stratification profile and inversion function shown in figure \ref{F-mk_constlinconst}(d)-(e), and the third simulation corresponds to the stratification profile and inversion function shown in figure \ref{F-mk_constlinconst}(g)-(h). Plots (a), (b), and (c) are snapshots of the surface potential vorticity, $\theta$, normalized by its maximum value in the snapshot. Plots (d), (e), and (f) are snapshots of the horizontal speed $\abs{\vec u }$ normalized by its maximum value in the snapshot. Plots (g), (h), and (i) show the model kinetic energy spectrum (solid black line) along with the prediction given by equation \eqref{eq:forward_KE} (dashed black line). We also provide linear fits to the model kinetic energy spectrum (dash-dotted red line) and to the predicted spectrum (dotted blue line).}
	\label{F-model_runs_mixed}
	\end{figure*}

	By fitting a power law, $k^\alpha$, to the inversion function, we do not mean to imply that $m(k)$ indeed takes the form of a power law. Instead, the purpose of obtaining the estimated power $\alpha$ is to apply the intuition gained from $\alpha$-turbulence \citep{pierrehumbert_spectra_1994,smith_turbulent_2002,sukhatme_local_2009,burgess_kraichnanleithbatchelor_2015} to surface quasigeostrophic turbulence. In $\alpha$-turbulence, an active scalar $\xi$, defined by the power law inversion relation \eqref{eq:alpha_inv}, is materially conserved in the absence of forcing and dissipation [that is, $\xi$ satisfies the time-evolution equation \eqref{eq:theta_equation} with $\theta$ replaced by $\xi$]. The scalar $\xi$ can be thought of as a generalized vorticity; if $\alpha=2$ we recover the vorticity of two-dimensional barotropic model. If $\alpha=1$, $\xi$ becomes proportional to surface buoyancy in the uniformly stratified surface quasigeostrophic model. To discern how $\alpha$ modifies the dynamics, we consider a point vortex $\xi \sim \delta(r)$, where $r$ is the horizontal distance from the vortex and $\delta(r)$ is the Dirac delta. If $\alpha=2$, we obtain  $\psi(r) \sim \log(r)/2\pi$; otherwise, if $0 < \alpha < 2$, we obtain $\psi(r) \sim - C_\alpha/r^{2-\alpha}$ where $C_\alpha>0$ \citep{iwayama_greens_2010}. 
	Therefore, larger $\alpha$ leads to vortices with a longer interaction range whereas smaller $\alpha$ leads to a shorter interaction range. 
	
	More generally, $\alpha$ controls the spatial locality of the resulting turbulence. In two-dimensional turbulence ($\alpha = 2$), vortices induce flows reaching far from the vortex core and the combined contributions of distant vortices dominates the local fluid velocity. These flows are characterized by thin filamentary $\xi$-structures due to the dominance of large scale strain \citep{watanabe_unified_2004}.
		As we decrease $\alpha$, the turbulence becomes more spatially local, the dominance of large-scale strain weakens, and a secondary instability becomes possible in which filaments roll-up into small vortices; the resulting turbulence is distinguished by vortices spanning a wide range of horizontal scales, as in uniform stratification surface quasigeostrophic turbulence \citep{pierrehumbert_spectra_1994, held_surface_1995}. As $\alpha$ is decreased further the $\xi$ field becomes spatially diffuse because the induced velocity, which now has small spatial scales, is more effective at mixing small-scale inhomogeneities in $\xi$ \citep{sukhatme_local_2009}.	  		
	
	These expectations are confirmed in the simulations shown in figure \ref{F-model_runs_mixed}.
	The simulations are set in a doubly periodic square with side length 400 km and are forced at a horizontal scale of 100 km. Large-scale dissipation is achieved through a linear surface buoyancy damping whereas an exponential filter is applied at small scales. In the case of a mixed-layer like stratification, with $\sigma_0 < \sigma_\mathrm{pyc}$, the $\theta$-field exhibits thin filamentary structures (characteristic of the $\alpha=2$ case) as well as vortices spanning a wide range of horizontal scales (characteristic of the $\alpha=1$ case). In contrast, although the $\sigma_0 > \sigma_\mathrm{pyc}$ exhibits vortices spanning a wide range of scales, no large scale filaments are evident. Instead, we see that the surface potential vorticity is spatially diffuse. These contrasting features are consequences of the induced horizontal velocity field. The mixed-layer like case has a velocity field dominated by large-scale strain, which is effective at producing thin filamentary structures. In contrast the velocity field in the $\sigma_0> \sigma_\mathrm{pyc}$ case consists of narrow meandering currents, which are effective at mixing away small-scale inhomogeneities.

	Both the predicted [equation \eqref{eq:forward_KE}] and diagnosed surface kinetic energy spectra are plotted in figure \ref{F-model_runs_mixed}. In the $\sigma_0 > \sigma_\mathrm{pyc}$ case, the predicted and diagnosed spectrum are close, although the diagnosed spectrum is steeper at large scales \citep[a too steep spectrum is also observed in the $\alpha=1$ and $\alpha=2$ cases, see][]{schorghofer_energy_2000}. In the $\sigma_0 < \sigma_\mathrm{pyc}$ case, the large-scale spectrum agrees with the predicted spectrum. However, at smaller scales, the model spectrum is significantly steeper.
	
	\begin{figure*}
	\centerline{\includegraphics[width=\textwidth]{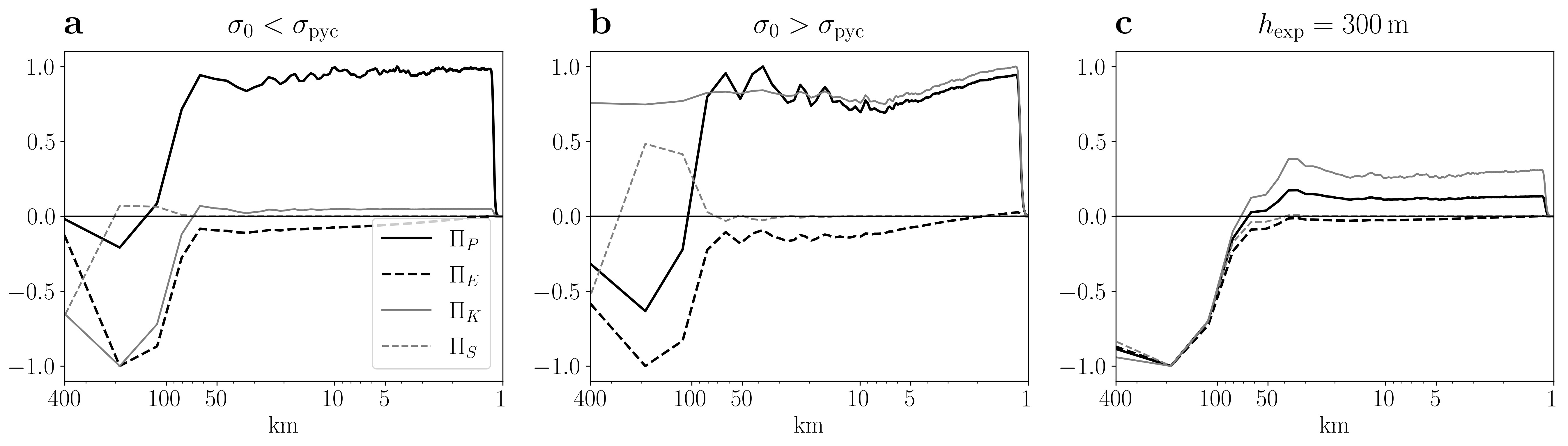}}
	\caption{Transfers of spectral densities, $\Pi_A$, for $A=E,P, K,S$ [see equation \eqref{eq:spectral_transfer}] normalized by their absolute maximum for the three simulations in figure \ref{F-model_runs_mixed}. }
	\label{F-model_runs_transfer}
	\end{figure*}
	
	The derivation of the predicted spectra in section \ref{S-turbulence} assumed the existence of an inertial range, which in this case means $\Pi_P(k)=$ constant. To verify whether this assumption holds, we show in figure \ref{F-model_runs_transfer} the transfer of the spectral densities $\Esc(k),\, \Psc(k),\, \Ksc(k)$ and $\Ssc(k)$. 
	In the mixed-layer like case, with $\sigma_0 < \sigma_\mathrm{pyc}$, an approximate inertial range forms with some deviations at larger scales. However, in the $\sigma_0 > \sigma_\mathrm{pyc}$ case, $\Pi_P$ is an increasing function at small scales, which indicates that the spectral density of surface potential enstrophy, $\Psc(k)$, is diverging at these scales. That is, at small scales, there is a depletion of $\Psc(k)$ and this depletion is causing the steepening of the kinetic energy spectrum at small-scales in figure \ref{F-model_runs_mixed}.
		
	\subsection{An exponentially stratified ocean}\label{SS-exp}
	
	Now consider the exponential stratification profile
	\begin{equation}\label{eq:sigma_exp}
		\sigma = \sigma_0 \, \mathrm{e}^{z/h_\mathrm{exp}}.
	\end{equation}
	Substituting the stratification profile \eqref{eq:sigma_exp} into the vertical structure equation \eqref{eq:Psi_equation} with boundary conditions \eqref{eq:Psi_upper} and \eqref{eq:Psi_lower} yields the vertical structure
	\begin{equation}\label{eq:Psi_exp}
		\Psi_k(z) = \mathrm{e}^{z/h_\mathrm{exp}} \, \frac{ I_1\left(\mathrm{e}^{z/h_\mathrm{exp}} \sigma_0 \, h_\mathrm{exp} \, k \right) }{I_1\left(\sigma_0 \, h_\mathrm{exp} \, k \right)},
	\end{equation}
	where $I_n(z)$ is the modified Bessel function of the first kind of order $n$.
	
	To obtain the inversion function, we substitute the vertical structure \eqref{eq:Psi_exp} into the definition of the inversion function \eqref{eq:mk} to obtain
	\begin{equation}\label{eq:mk_exp}
		m(k) = \frac{1}{\sigma_0^2 h_\mathrm{exp}} \, + \,   \frac{k}{2\sigma_0} \left[  \frac{ I_0\left( \sigma_0 h_\mathrm{exp}  k \right) }{I_1\left(\sigma_0  h_\mathrm{exp} k \right)} + \frac{ I_2\left( \sigma_0  h_\mathrm{exp}  k \right) }{I_1\left(\sigma_0  h_\mathrm{exp} k \right)}  \right].
	\end{equation}
	In the small-scale limit $k\gg 1/\left(\sigma_0\,h_\mathrm{exp}\right)$, the inversion function becomes $m(k) \approx k/\sigma_0$ as in constant stratification surface quasigeostrophic theory. In contrast, the large-scale limit $k\ll  1/\left(\sigma_0\,h_\mathrm{exp}\right)$ gives
	\begin{equation}\label{eq:mk_exp_large}
		m(k) \approx \frac{h_\mathrm{exp}}{4} \left( k_\mathrm{exp}^2 + k^2 \right),
	\end{equation}
	where $k_\mathrm{exp}$ is given by
	\begin{equation}\label{eq:kexp}
		k_\mathrm{exp} = \frac{2\, \sqrt{2}}{ \sigma_0 \, h_\mathrm{exp}}.
	\end{equation}
	As $k/k_\mathrm{exp} \rightarrow 0$, the inversion function asymptotes to a constant value and the vertical structure becomes independent of the horizontal scale $2\pi/k$, with $\Psi_k \rightarrow \Psi_0$ where
	\begin{equation}
	    \Psi_0(z) = \mathrm{e}^{2z/h_\mathrm{exp}}.
	\end{equation}
	Further increasing the horizontal scale no longer modifies $\Psi_k(z)$ and so vertical structure is arrested at $\Psi_0$.
 	
 	An example with $h_\mathrm{exp}=300$ m and $\sigma_0 = 100$ is shown in figure \ref{F-mk_constlinconst}(g)-(i). At horizontal scales smaller than $L_\mathrm{exp}=2\pi/k_\mathrm{exp}$, the inversion function rapidly transitions to the linear small-scale limit of $m(k)\approx k/\sigma_0$. In contrast, at horizontal scales larger than $L_\mathrm{exp}$, the large-scale approximation \eqref{eq:mk_exp_large} holds, and at sufficiently large horizontal scales, the inversion function asymptotes to constant value of $m(k) = h_\mathrm{exp}\,k_\mathrm{exp}^2/4$.
	
	The inversion relation implied by the inversion function \eqref{eq:mk_exp_large} is
	\begin{equation}\label{eq:theta_shallow}
		\hat \theta_{\vec k} \approx - \frac{h_\mathrm{exp}}{4} \left( k_\mathrm{exp}^2 + k^2 \right) \hat \psi_{\vec k},
	\end{equation}
	which is isomorphic to the inversion relation in the equivalent barotropic model \citep{larichev_weakly_1991}, with $k_\mathrm{exp}$ assuming the role of the deformation wavenumber. 
	Using the relations between the various spectra [equations \eqref{eq:variance_relations1} and \eqref{eq:variance_relations2}] with an inversion function of the form $m(k)\approx m_0 +m_1 k^2$, we obtain $\Esc(k) \approx m_0 \, \Ssc(k) + m_1 \, \Ksc(k) $ and $\Psc(k) \approx m_0^2 \, \Ssc(k) + 2\, m_0 \, m_1 \Ksc(k)$; solving for $\Ssc(k)$ and $\Ksc(k)$ then yields
\begin{equation}
	\Ssc(k) \approx \frac{2 \, m_0 \Esc(k) - \Psc(k)}{m_0^2},
\end{equation}
and
\begin{equation}
	\Ksc(k) \approx  \frac{\Psc(k) - m_0 \, \Esc(k)}{m_0\,m_1}.
\end{equation}
The inverse cascade of total energy then implies an inverse cascade of surface streamfunction variance, $S$; conversely, the forward cascade of surface potential enstrophy implies a forward cascade of surface kinetic energy, $K$. Moreover, using an argument analogous to that in \cite{larichev_weakly_1991}, we find that 
\begin{equation}
	\Ssc(k) \sim k^{-11/3}
\end{equation}
 in the inverse cascade inertial range whereas 
 \begin{equation}\label{eq:equiv_KE}
 	 \Ksc(k) \sim k^{-3}
 \end{equation}
 in the forward cascade inertial range.
	
	The implied dynamics are extremely local; a point vortex, $\theta(r) \sim \delta(r)$, leads to an exponentially decaying streamfunction, $\psi(r) \sim \exp(-k_\mathrm{exp}r)/\sqrt{k_\mathrm{exp}r}$ \citep{polvani_two-layer_1989}. Therefore, as for the $\sigma_0>\sigma_\mathrm{pyc}$ case above, we expect a spatially diffuse surface potential vorticity field and no large-scale strain. However, unlike the $\sigma_0>\sigma_\mathrm{pyc}$ case, the presence of a distinguished length scale, $L_\mathrm{exp}$, leads to the emergence of plateaus of homogenized surface potential vorticity surrounded by kinetic energy ribbons \citep{arbic_coherent_2003}. Both of these features can be seen in figure \ref{F-model_runs_mixed}.
	
	 The $k^{-3}$ surface kinetic energy spectrum \eqref{eq:equiv_KE} is only expected to hold at horizontal scales larger than $\sigma_0 \, h_\mathrm{exp}$; at smaller scales we should recover the $k^{-5/3}$ spectrum expected from uniformly stratified surface quasiogeostrophic theory.  Figure \ref{F-model_runs_mixed}(i) shows that there is indeed a steepening of the kinetic energy spectrum at horizontal scales larger than 20 km, although the model spectrum is somewhat steeper than the predicted $k^{-3}$. Similarly, although the spectrum flattens at smaller scales, the small-scale spectrum is also slightly steeper than the predicted $k^{-5/3}$.
	
	We can also examine the spectral transfer functions of $\Psc(k)$ and $\Ksc(k)$. At large-scales, we expect an inertial range in surface kinetic energy, so $\Pi_K(k) = $ constant, whereas at small scales, we expect an inertial range in surface potential enstrophy, so $\Pi_P(k)= $ constant. However, figure \ref{F-model_runs_transfer}(c) shows that although both $\Pi_K(k)$ and $\Pi_P(k)$ become approximately flat at small scales, we observe significant deviations at larger scales.
	
	\subsection{More general stratification profiles}
	
	These three idealized cases provide intuition for how the inversion function behaves for an arbitrary stratification profile, $\sigma(z)$. Generally, if $\sigma(z)$ is decreasing over some depth, then the inversion function will steepen to a super linear wavenumber dependence over a range of horizontal wavenumber whose vertical structure function significantly impinges on these depths. A larger difference in stratification between these depths leads to a steeper inversion function.
	 Analogously, if $\sigma(z)$ is increasing over some depth, then the inversion function will flatten to a sublinear wavenumber dependence, with a larger difference in stratification leading to a flatter inversion function.  
	Finally, if $\sigma(z)$ is much smaller at depth than near the surface, the inversion function will flatten to become approximately constant, and we recover an equivalent barotopic like regime, similar to the exponentially stratified example.
		
	\section{Application to the ECCOv4 ocean state estimate}\label{S-ECCO}
	
	We now show that, over the mid-latitude North Atlantic, the inversion function is seasonal at horizontal scales between 1-100 km, transitioning from $m(k) \sim k^{3/2}$ in winter to $m(k)\sim k^{1/2}$ in summer.	To compute the inversion function $m(k)$, we obtain the stratification profile $\sigma(z)=N(z)/f$ at each location from the Estimating the Circulation and Climate of the Ocean version 4 release 4 \citep[ECCOv4,][]{forget_ecco_2015} state estimate. We then compute $\Psi_k(z)$ using the vertical structure equation \eqref{eq:Psi_equation} and then use the definition of the inversion function \eqref{eq:mk} to obtain $m(k)$ at each wavenumber $k$.
		
	\subsection{The three horizontal length-scales}
	
	In addition to $L_\mathrm{mix}$ and $L_\mathrm{pyc}$ [defined in equations \eqref{eq:mixed_length} and \eqref{eq:therm_length}], we introduce the  horizontal length scale, $L_H$, the full-depth horizontal scale, defined by
	\begin{equation}\label{eq:LH}
		L_H = 2 \, \pi \, \sigma_\mathrm{ave} \, H,
	\end{equation}
	where $\sigma_\mathrm{ave}$ is the vertical average of $\sigma$ and $H$ is the local ocean depth. The bottom boundary condition becomes important to the dynamics at horizontal scales larger than $\approx L_H$.
	
	\begin{figure*}
		\centerline{\includegraphics[width=\textwidth]{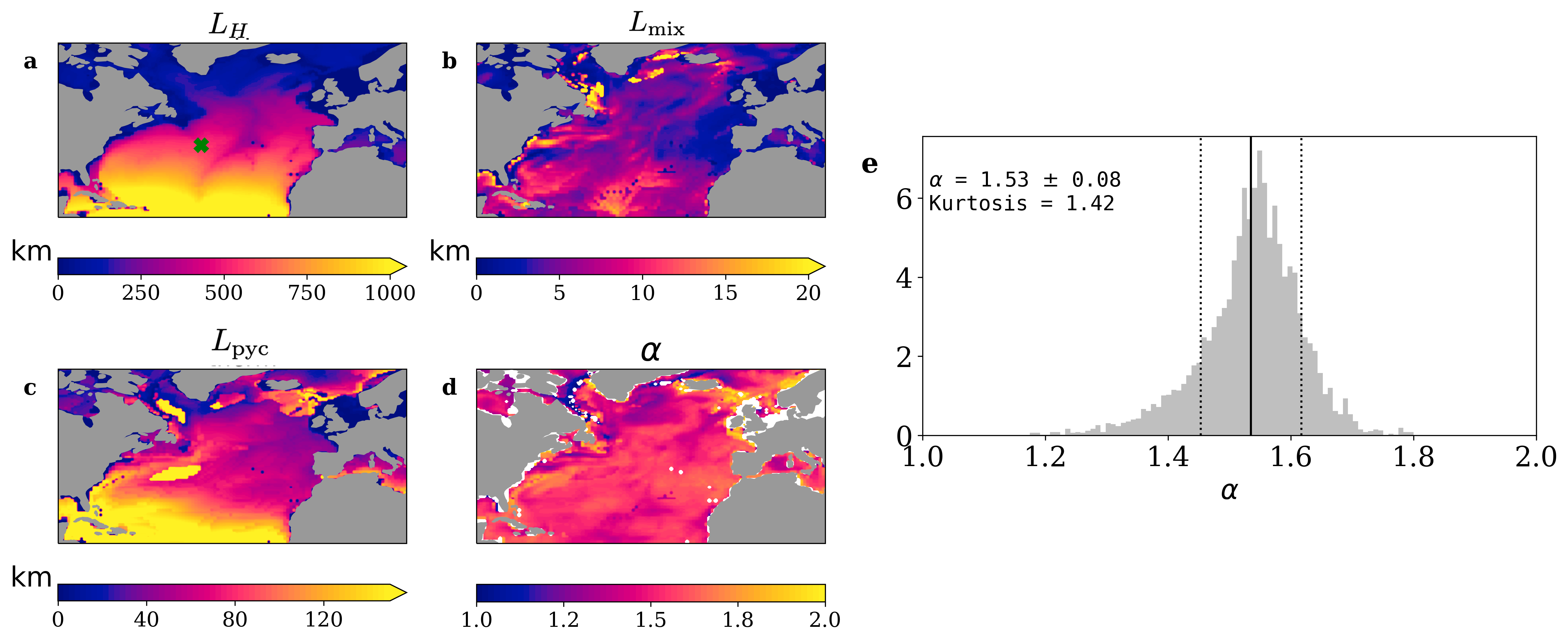}}
		\caption{Panels (a), (b), and (c) show the horizontal length scales $L_H$, $L_\textrm{mix}$, and $L_\textrm{pyc}$ as computed from 2017 January mean ECCOv4 stratification profiles, $\sigma(z) = N(z)/f$, over the North Atlantic. The green `x' in panel (a) shows the location chosen for the inversion functions of figure \ref{F-mk_realistic_small} and the model simulations of figure \ref{F-model_runs_ecco}. Panel (d) shows $\alpha$, defined by $m(k)/k^{\alpha} \approx \mathrm{constant}$, over the North Atlantic. We compute $\alpha$ by fitting a straight line to a log-log plot of $m(k)$ between $2\pi/L_\mathrm{mix}$ and $2\pi/L_\mathrm{pyc}$. Panel (e) is a histogram of the computed values of $\alpha$ over the North Atlantic. We exclude from this histogram grid cells with $L_H < 150$ km; these are primarily continental shelves and high-latitude regions. In these excluded regions, our chosen bottom boundary condition \eqref{eq:no-slip} may be influencing the computed value of $\alpha$.}
		\label{F-alpha_obs}
	\end{figure*}
	
	\begin{figure*}
		\centerline{\includegraphics[width=\textwidth]{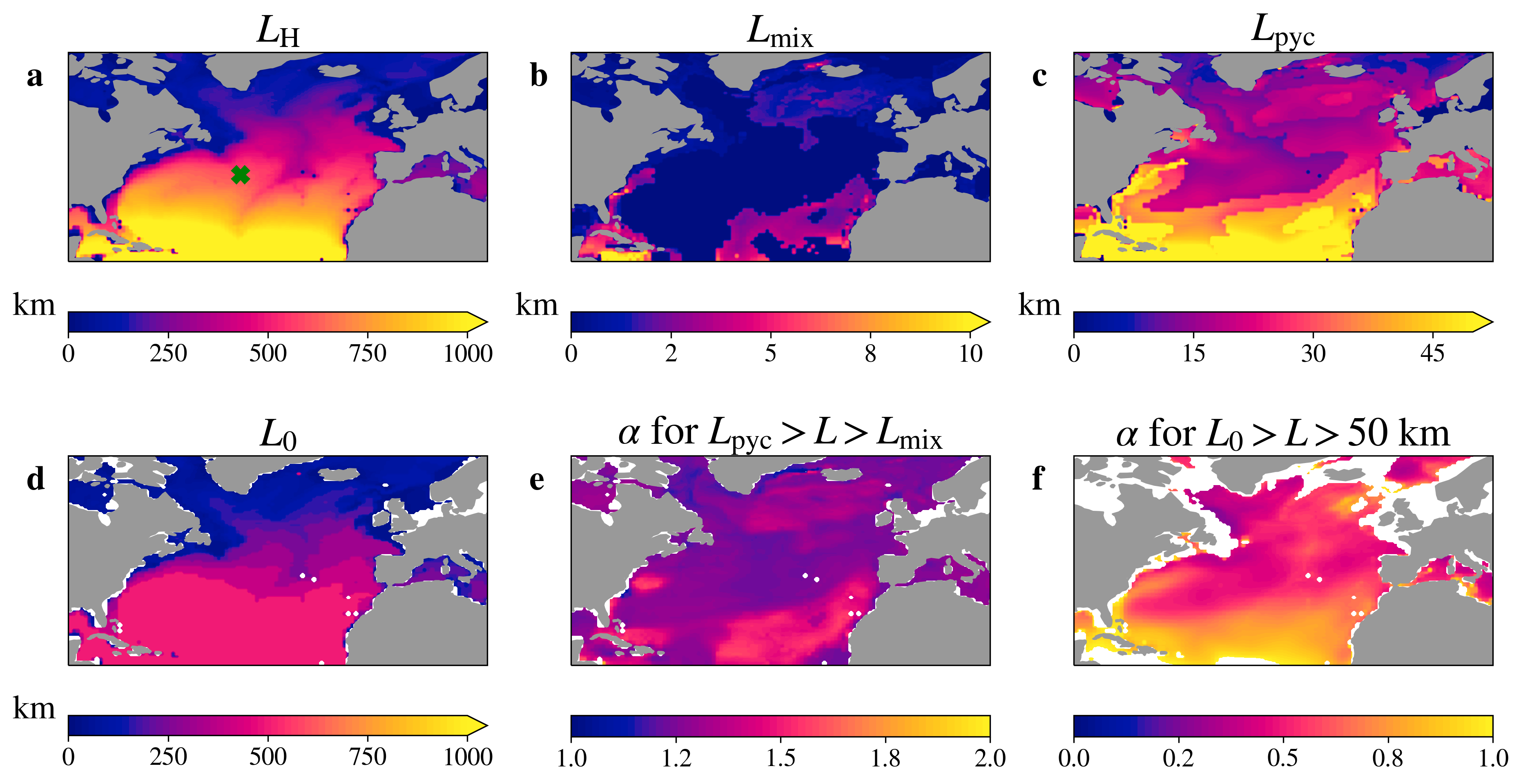}}
		\caption{Panels (a), (b), (c), and (e) are as in figure \ref{F-alpha_obs}(a)-(d), but computed from 2017 July mean stratification profiles. The calculation of $L_0$ in panel (d) is explained in the text. In panel (f), we show $\alpha$ but measured between $2\pi/(50 \, \mathrm{km})$ and $2\pi/L_0$. }
		\label{F-NorthAtlanticJul}
	\end{figure*} 
	
	We compute all three length scales using ECCOv4 stratification profiles over the North Atlantic, with results displayed in figures \ref{F-alpha_obs}(a)-(c) and \ref{F-NorthAtlanticJul}(a)-(c) for January and July, respectively. To compute the mixed-layer horizontal scale, $L_\mathrm{mix} = 2 \, \pi \sigma_0 \, h_\mathrm{mix}$, we set $\sigma_0$ equal to the stratification at the uppermost grid cell. The mixed-layer depth, $h_\mathrm{mix}$, is then defined as follows. We first define the pycnocline stratification $\sigma_\mathrm{pyc}$ to be the maximum of $\sigma(z)$. The mixed-layer depth $h_\mathrm{mix}$ is then the depth at which $\sigma(-h_\mathrm{mix})= \sigma_0 + \left(\sigma_\mathrm{pyc}-\sigma_0\right)/4$. Finally, the pycnocline horizontal scale, $L_\mathrm{pyc}$, is computed as $L_\mathrm{pyc}=2\, \pi \, \sigma_\mathrm{pyc} \, h_\mathrm{pyc}$, where $h_\mathrm{pyc}$ is the depth of the stratification maximum $\sigma_\mathrm{pyc}$. 

	Figures \ref{F-alpha_obs}(a) and \ref{F-NorthAtlanticJul}(a) show that $L_H$ is not seasonal, with typical mid-latitude open ocean values between $400-700$ km. On continental shelves, as well as high-latitudes, $L_H$ decreases to values smaller than $200$ km. As we approach the equator, the full-depth horizontal scale $L_H$ becomes large due to the smallness of the Coriolis parameter.
	
	Constant stratification surface quasigeostrophic theory is only valid at horizontal scales smaller than $L_\mathrm{mix}$. Figure \ref{F-alpha_obs}(b) shows that the wintertime $L_\mathrm{mix}$ is spatially variable with values ranging between $1-15$ km. In contrast, figure \ref{F-NorthAtlanticJul}(b) shows that the summertime $L_\mathrm{mix}$ is less than 2 km over most of the midlatitude North Atlantic.
	
	Finally, we expect to observe a superlinear inversion function for horizontal scales between $L_\mathrm{mix}$ and $L_\mathrm{pyc}$. The latter, $L_\mathrm{pyc}$, is shown in figures \ref{F-alpha_obs}(c) and \ref{F-NorthAtlanticJul}(c). Typical mid-latitude values range between $70-110$ km in winter but decrease to $15-30$ km in summer.
		
	\subsection{The inversion function at a single location}
	Before computing the form of the inversion function over the North Atlantic, we focus on a single location. However, we must first address what boundary conditions to use in solving the vertical structure equation $\eqref{eq:Psi_equation}$ for $\Psi_k(z)$. We cannot use the infinite bottom boundary condition \eqref{eq:Psi_lower} because the ocean has a finite depth. However,  given that figures \ref{F-alpha_obs}(a) and \ref{F-NorthAtlanticJul}(a) show that the bottom boundary condition should not effect the inversion function at horizontal scales smaller than 400 km in the mid-latitude open ocean (in the North Atlantic), we choose to use the no-slip bottom boundary condition
	\begin{align}\label{eq:no-slip}
		\Psi_k(-H) = 0.
	\end{align}
	The alternate free-slip boundary condition
	\begin{align}\label{eq:free-slip}
		\d{\Psi_k(-H)}{z} = 0
	\end{align}
	gives qualitatively identical results for horizontal scales smaller than 400 km, which are the scales of interest in this study [see appendix A for the large-scale limit of $m(k)$ under these boundary conditions]\footnote{ The no-slip boundary condition \eqref{eq:no-slip} is appropriate over strong bottom friction \citep{arbic_baroclinically_2004} or steep topography \citep{lacasce_prevalence_2017} whereas the free-slip boundary condition \eqref{eq:free-slip} is appropriate over a flat bottom.}.
	
	\begin{figure*}
	\centerline{\includegraphics[width=\textwidth]{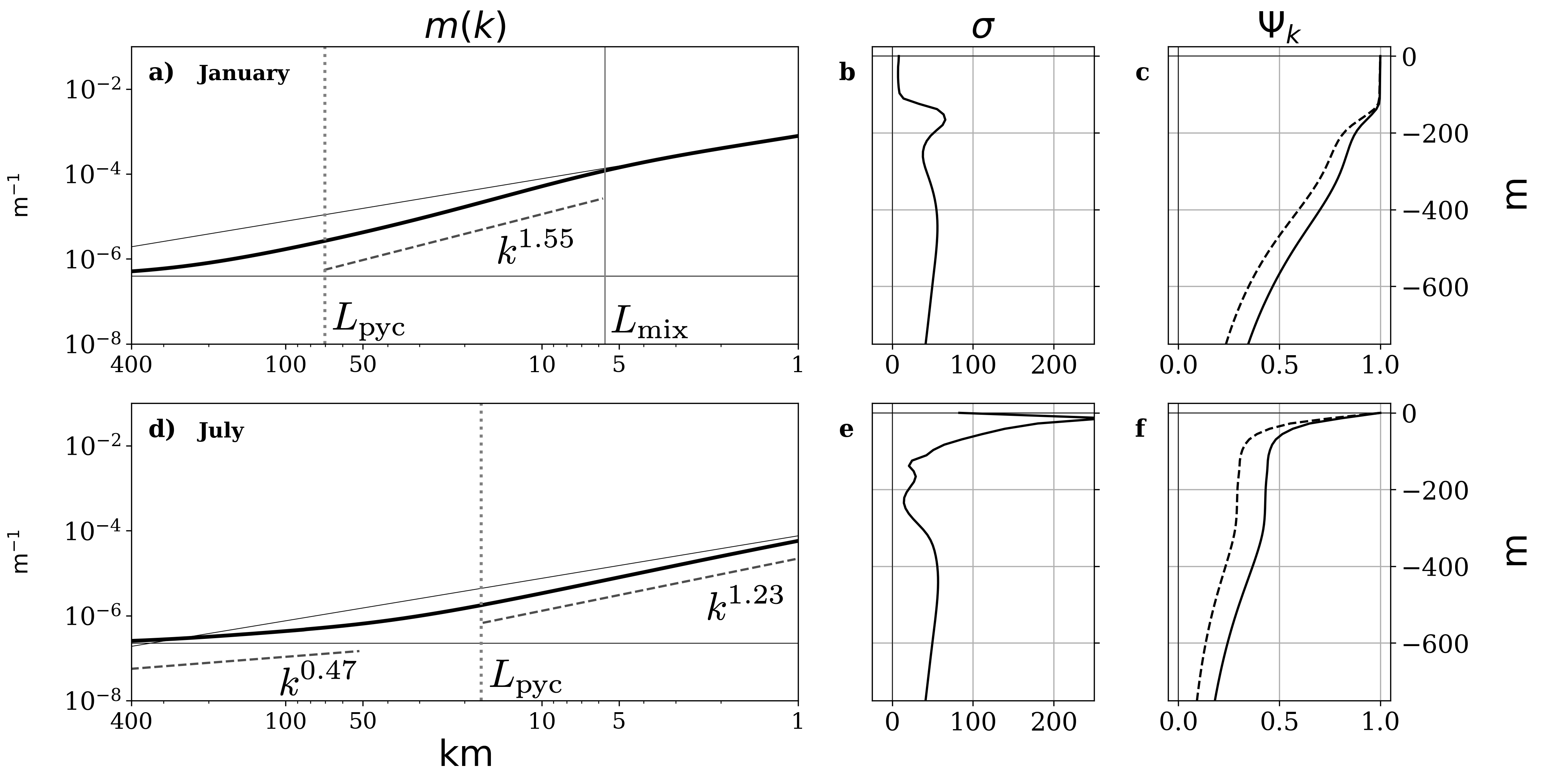}}
	\caption{As in figure \ref{F-mk_constlinconst} but for the mid-latitude North Atlantic location  ($38^\circ$ N, $45^\circ$ W)  in January [(a)-(c)] and July [(d)-(f)]. This location is marked by a green `x' in figure \ref{F-alpha_obs}(a). Only the upper 750 m of the stratification profiles and vertical structures are shown in panels (b), (c), (e) and (f).}
	\label{F-mk_realistic_small}
  	\end{figure*}

  	\begin{figure*}
	\centerline{\includegraphics[width=\textwidth]{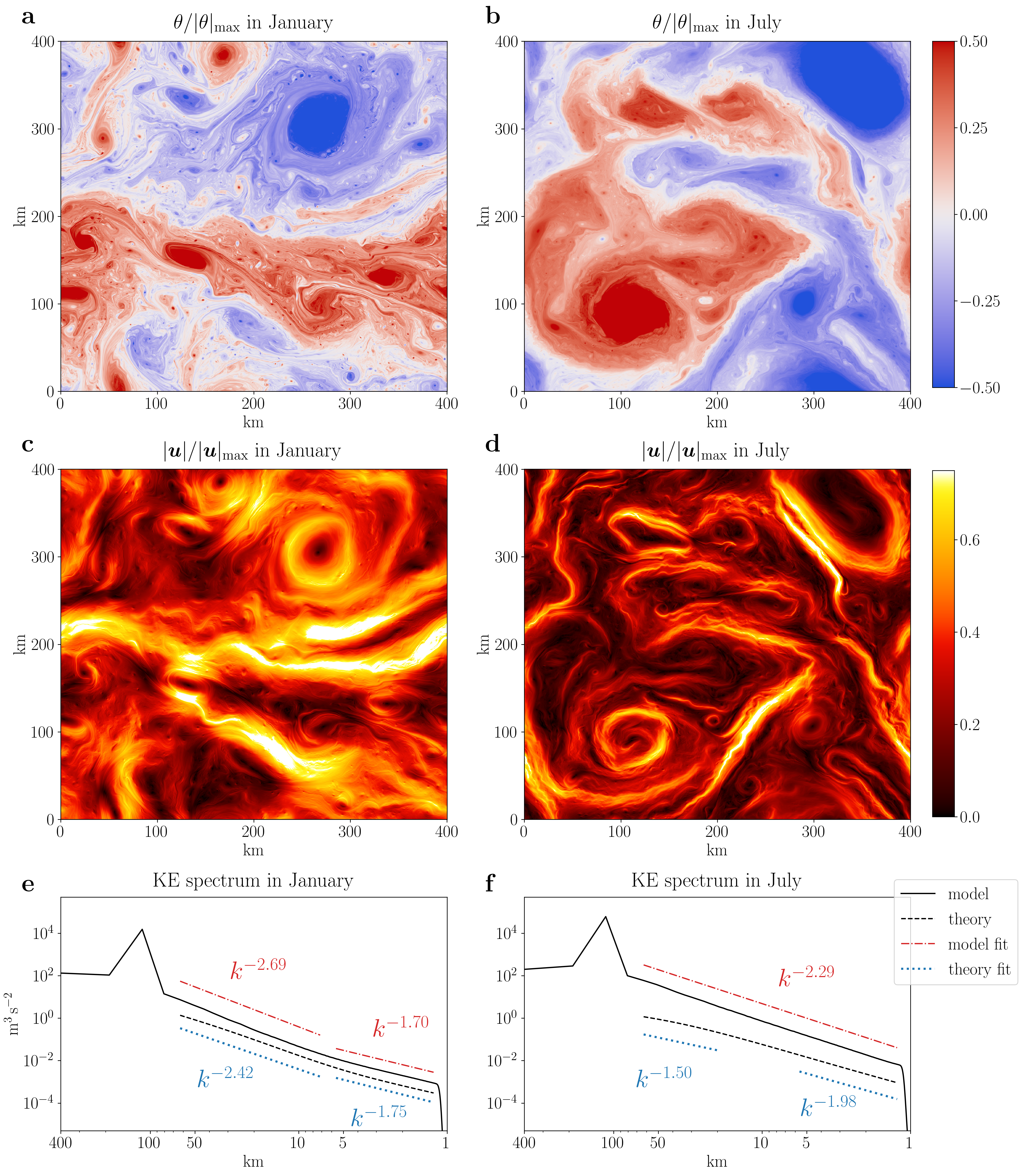}}
	\caption{Two pseudo-spectral simulations differing only in the chosen stratification profile $\sigma(z) = N(z)/f$. Both simulations use a monthly averaged 2017 stratification at the mid-latitude North Atlantic location ($38^\circ$ N,$45^\circ$ W) [see the green 'x' in figure \ref{F-alpha_obs}a] in January [(a), (c), (e)] and July [(b), (d), (f)]. The stratification profiles are obtained from the Estimating the Circulation and Climate of the Ocean version 4 release 4 \citep[ECCOv4,][]{forget_ecco_2015} state estimate. Otherwise as in figure \ref{F-model_runs_mixed}.}
	\label{F-model_runs_ecco}
 	\end{figure*}
	
	Figure \ref{F-mk_realistic_small} shows the computed inversion function in the mid-latitude North Atlantic at ($38^\circ$ N, $45^\circ$ W)  [see the green `x' in figure \ref{F-alpha_obs}(a)]. In winter, at horizontal scales smaller than $L_\mathrm{mix} \approx 5$ km, we recover the linear $m(k) \approx k/\sigma_0$ expected from constant stratification surface quasigeostrophic theory. However, for horizontal scales between  $L_\mathrm{mix} \approx 5$ km and $L_\mathrm{pyc} \approx 70$ km, the inversion function, $m(k)$, becomes as steep as a $k^{3/2}$ power law. Figure \ref{F-model_runs_ecco} shows a snapshot of the surface potential vorticity and the geostrophic velocity from a surface quasigeostrophic model using the wintertime inversion function. The surface potential vorticity snapshot is similar to the idealized mixed-layer snapshot of figure \ref{F-model_runs_mixed}(a), which is also characterized by $\alpha \approx 3/2$ (but at horizontal scales between 7-50 km). Both simulations exhibit a preponderance of small-scale vortices as well as thin filaments of surface potential vorticity. As expected, the kinetic energy spectrum [figure \ref{F-model_runs_ecco}(e)] transitions from an $\alpha\approx3/2$ regime to an $\alpha = 1$ regime near $L_\mathrm{mix} = 5$ km. Moreover, as shown in figure \ref{F-model_runs_real_transfer}, an approximate inertial range is evident between the forcing and dissipation scales.	
	
	\begin{figure*}
	\centerline{\includegraphics[width=27pc]{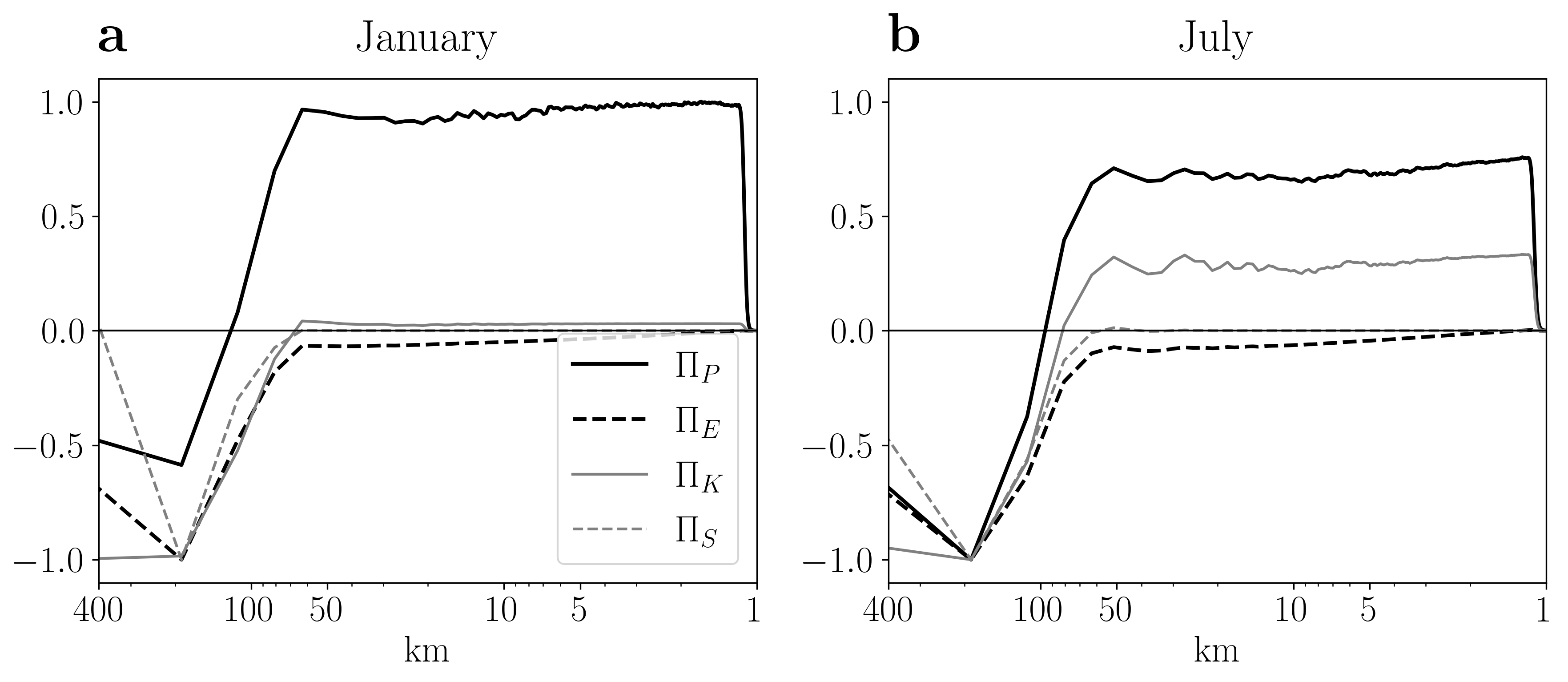}}
	\caption{Transfers of spectral densities, $\Pi_A$, for $A=E,P,K,S$ [see equation \eqref{eq:spectral_transfer}] normalized by their absolute maximum for the two simulations in figure \ref{F-model_runs_ecco}.}
	\label{F-model_runs_real_transfer}
	\end{figure*}

	In summer, the mixed-layer horizontal scale, $L_\mathrm{mix}$, becomes smaller than 1 km and the pycnocline horizontal scale, $L_\mathrm{pyc}$, decreases to 20 km. We therefore obtain a super linear regime, with $m(k)$ as steep as $k^{1.2}$, but only for horizontal scales between 1-20 km. Thus, although there is a range of wavenumbers for which $m(k)$ steepens to a super linear wavenumber dependence in summer, this range of wavenumbers is narrow, only found at small horizontal scales, and the steepening is much less pronounced than in winter. 
	At horizontal scales larger than $L_\mathrm{pyc}$, the summertime inversion function flattens, with the $m(k)$ increasing like a $k^{1/2}$ power law between 50-400 km. This flattening is due to the largely decaying nature of ocean stratification below the stratification maximum.

	As expected from a simulation with a sublinear inversion function at large scales, the surface potential vorticity appears spatially diffuse [figure \ref{F-model_runs_ecco}(d)] and comparable to the $\sigma_0 > \sigma_\mathrm{pyc}$ and the exponential simulations [figure \ref{F-model_runs_mixed}(b)-(c)]. However, despite having a sublinear inversion function, the July simulations is dynamically more similar to the exponential simulation rather than the  $\sigma_0 > \sigma_\mathrm{pyc}$ simulation. The July simulation displays approximately homogenized regions of surface potential vorticity surrounded by surface kinetic energy ribbons, as well as the steeper surface kinetic energy spectrum associated with these features. As a result, the surface kinetic energy spectrum does not follow the predicted spectrum \eqref{eq:forward_KE}.

	\subsection{The inversion function over the North Atlantic}
	
	We now present power law approximations to the inversion function $m(k)$ over the North Atlantic in winter and summer. In winter, we obtain the power $\alpha$, where $m(k)/k^\alpha \approx \mathrm{constant}$, by fitting a straight line to $m(k)$ on a log-log plot between $2\pi/L_\mathrm{mix}$ and $2\pi/L_\mathrm{pyc}$. A value of $\alpha =1$ is expected for constant stratification surface quasigeostrophic theory. A value of $\alpha = 2$ leads to an inversion relation similar to two-dimensional barotropic dynamics. 
	 However, in general, we emphasize that $\alpha$ is simply a crude measure of how quickly $m(k)$ is increasing; we do not mean to imply that $m(k)$ in fact takes the form of a power law. Nevertheless, the power $\alpha$ is useful because, as $\alpha$-turbulence suggests (and the simulations in section \ref{S-idealized} confirm), the rate of increase of the inversion function measures the spatial locality of the resulting flow. 

	\begin{figure}
	\centerline{\includegraphics[width=19pc]{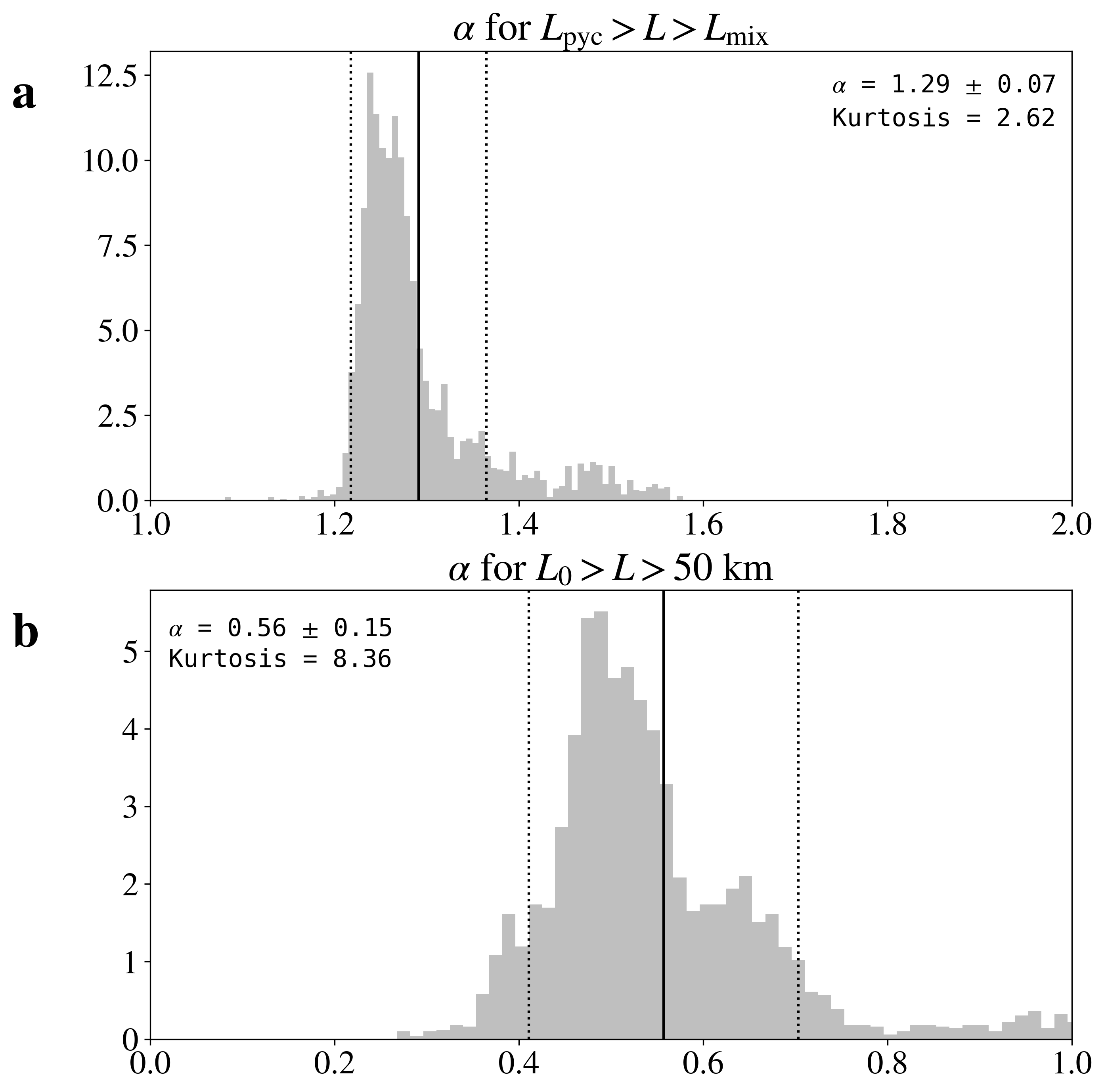}}
	\caption{Panel (a) is as in figure \ref{F-alpha_obs}(e), but with the additional restriction that $L_H < 750$ km to filter out the non-seasonal equatorial region. In panel (b), we instead plot $\alpha$ as obtained by fitting a straight line to a log-log plot of $m(k)$ between $2\pi/(50 \, \mathrm{km})$ and $2\pi/L_0$ with the same restrictions as in panel (a).}
	\label{F-alpha_dist_jul}
	\end{figure}
	
	Figure \ref{F-alpha_obs}(d) shows that we generally have $\alpha \approx 3/2$ in the wintertime open ocean. Deviations appear at high-latitudes (e.g., the Labrador sea and southeast of Greenland) and on  continental shelves where we find regions of low $\alpha$. However, both of these regions have small values of $L_H$ so that our chosen no-slip bottom boundary condition \eqref{eq:no-slip} may be influencing the computed $\alpha$ there. 
	
	A histogram of the computed values of $\alpha$ [figure \ref{F-alpha_obs}(e)] confirms that $\alpha \approx 1.53 \pm 0.08$ in the wintertime mid-latitude open ocean. This histogram only includes grid cells with $L_H >150$ km, which ensures that the no-slip bottom boundary condition \eqref{eq:no-slip} is not influencing the computed distribution.
	
	An inversion function of $m(k) \sim k^{3/2}$ implies a surface kinetic energy spectrum of $k^{-4/3}$ upscale of small-scale forcing [equation \eqref{eq:inverse_KE}] and a spectrum of $k^{-7/3}$ downscale of large-scale forcing [equation \eqref{eq:forward_KE}].  
	As we expect wintertime surface buoyancy anomalies to be forced both by large-scale baroclinic instability and by small-scale mixed-layer baroclinic instability, the realized surface kinetic energy spectrum should be between $k^{-4/3}$ and $k^{-7/3}$. Such a prediction is consistent with the finding that North Atlantic geostrophic surface velocities are mainly due to surface buoyancy anomalies \citep{lapeyre_what_2009,gonzalez-haro_global_2014}  and observational evidence of a $k^{-2}$ wintertime spectrum \citep{callies_seasonality_2015}. 
	
	 The universality of the $m(k) \sim k^{3/2}$ regime over the mid-latitudes is expected because it arises from a mechanism universally present over the mid-latitude ocean in winter; namely, the deepening of the mixed-layer. However, a comment is required on why this regime also appears at low latitudes where we do not observe deep wintertime mixed-layers. At low latitudes, the $m(k) \sim k^{3/2}$ regime emerges because there is a large scale separation between $L_\mathrm{mix}$ and $L_\mathrm{pyc}$. The smallness of the low latitude Coriolis parameter $f$ cancels out the shallowness of the low latitude mixed-layer depth resulting in values of $L_\mathrm{mix}$ comparable to the remainder of the mid-latitude North Atlantic, as seen in figure \ref{F-alpha_obs}(b). However, no similar cancellation occurs for $L_\mathrm{pyc}$ which reaches values of $\approx 500$ km due to the smallness of the Coriolis parameter $f$ at low latitudes. As a consequence, there is a non-seasonal $m(k)\sim k^{3/2}$ regime at low latitudes for horizontal scales between $10-500$ km.
	 
	 The analogous summertime results are presented in figure \ref{F-NorthAtlanticJul}(e) and figure \ref{F-alpha_dist_jul}(a). Near the equator, we obtain values close to $\alpha \approx 3/2$, as expected from the weak seasonality there. In contrast, the midlatitudes generally display $\alpha \approx 1.2-1.3$ but this superlinear regime is only present at horizontal scales smaller than  $L_\mathrm{pyc} \approx 20-30$ km. Figure \ref{F-alpha_dist_jul}(a) shows a histogram of the measured $\alpha$ values but with the additional restriction that $L_H < 750$ km to filter out the near equatorial region (where $\alpha \approx 3/2$). 
	 	 
	The summertime inversion function shown in figure \ref{F-mk_realistic_small}(d) suggests that the inversion function flattens at horizontal scales larger than 50 km, with $m(k)$ increasing like a $k^{1/2}$ power law.
	  We now generalize this calculation to the summertime midlatitude North Atlantic by fitting a straight line to $m(k)$ on a log-log plot between $2\pi/(50 \, \mathrm{km})$ and $2\pi/L_0$ where $L_0$ is defined by
	 \begin{equation}
	 	m\left(\frac{2\,\pi}{L_0} \right) = m_0 = \left[ \int_{-H}^0  \sigma^2(s) \di s \right] ^{-1}
	 \end{equation}
	 and $m_0$ is defined by the second equality. In this case, we solve for $m(k)$ using the free-slip boundary condition \eqref{eq:free-slip}. We made this choice because $m(k)$ must cross $m_0$ in the large-scale limit if we apply the free-slip boundary condition \eqref{eq:free-slip}. In contrast, $m(k)$ asymptotes to $m_0$ from above if we apply the no-slip boundary condition \eqref{eq:no-slip}. See appendix A for more details. In any case, if we use the free-slip boundary condition \eqref{eq:free-slip}, then  $L_0$ is a horizontal length scale at which the flattening of $m(k)$ ceases and $m(k)$ instead begins to steepen in order to attain the required $H\, k^2$ dependence at large horizontal scales [see equation \eqref{eq:mk_free}]. Over the mid-latitudes North Atlantic, $L_0$ has typical values of 200-500 km [figure \ref{F-NorthAtlanticJul}(d)].
	 
	 When $\alpha$ is measured between 50 km and $L_0$, we find typical midlatitude values close to $\alpha \approx 1/2$ [figure \ref{F-NorthAtlanticJul}(f)]. A histogram of these $\alpha$ values is provided in figure \ref{F-alpha_dist_jul}(b), where we only consider grid cells satisfying $L_H > 150$ km and $L_H < 750$ km (the latter condition filters out near equatorial grid cells). The distribution is broad with a mean of $\alpha = 0.56 \pm 0.15$ and a long tail of $\alpha > 0.8$ values. Therefore, $m(k)$ flattens considerably in response to the decaying nature of summertime upper ocean stratification. It is not clear, however, whether the resulting dynamics will be similar to the $\sigma_0 > \sigma_\mathrm{pyc}$ case or the exponentially stratified case in section \ref{S-idealized}. As we have seen, the summertime simulation (in figure \ref{F-model_runs_ecco}) displayed characteristics closer to the idealized exponential case than the $\sigma_0 > \sigma_\mathrm{pyc}$ case. Nevertheless, the low summertime values of $\alpha$ indicate that buoyancy anomalies generate shorter range velocity fields in summer than in winter.
	 
	 \cite{isern-fontanet_transfer_2014} and \cite{gonzalez-haro_ocean_2020} measured the inversion function empirically, through equation \eqref{eq:transfer_func}, and found that the inversion function asymptotes to a constant at large horizontal scales (270 km near the western coast of Australia and 100 km in the Mediterranean Sea). They suggested this flattening is due to the dominance of the interior quasigeostrophic solution at large scales \citep[a consequence of equation 29 in][]{lapeyre_dynamics_2006}. 
	 We instead suggest this flattening is intrinsic to surface quasigeostrophy. In our calculation the inversion function does not become constant at horizontal scales smaller than 400 km. However, if the appropriate bottom boundary condition is the no-slip boundary condition \eqref{eq:no-slip}, then the inversion asymptotes to a constant value at horizontal scales larger than $L_H$ (appendix A).

	 \section{Discussion and conclusion}\label{S-dicussion}
	 
	 As reviewed in the introduction, surface geostrophic velocities over the Gulf Stream, the Kuroshio, and the Southern Ocean are primarily induced by  surface buoyancy anomalies in winter \citep{lapeyre_what_2009, isern-fontanet_diagnosis_2014, gonzalez-haro_global_2014, qiu_reconstructability_2016, miracca-lage_can_2022}.
	  However, the kinetic energy spectra found in observations and numerical models are too steep to be consistent with uniformly stratified surface quasigeostrophic theory \citep{blumen_uniform_1978,callies_interpreting_2013}.
	  By generalizing surface quasigeostrophic theory to account for variable stratification, we have shown that surface buoyancy anomalies can generate a  variety of dynamical regimes depending on the stratification's vertical structure. 
	  Buoyancy anomalies generate longer range velocity fields over decreasing stratification [$\sigma'(z)\leq0$] and shorter range velocity fields over increasing stratification  [$\sigma'(z)\geq0$]. As a result, the surface kinetic energy spectrum is steeper over decreasing stratification than over increasing stratification. An exception occurs if there is a large difference between the surface stratification and the deep ocean stratification (as in the exponential stratified example of section \ref{S-idealized}). In this case, we find regions of approximately homogenized surface buoyancy surrounded by kinetic energy ribbons \citep[similar to][]{arbic_coherent_2003} and this spatial reorganization of the flow results in a steep kinetic energy spectrum. By applying the variable stratification theory to the wintertime North Atlantic and assuming that mixed-layer instability acts as a narrowband small-scale surface buoyancy forcing, we find that the theory predicts a surface kinetic energy spectrum between $k^{-4/3}$ and $k^{-7/3}$, which is consistent with the observed wintertime $k^{-2}$ spectrum \citep{sasaki_impact_2014,callies_seasonality_2015,vergara_revised_2019}. There remains the problem that mixed-layer instability may not be localized at a certain horizontal scale but is forcing the surface flow at a wide range of scales \citep{khatri_role_2021}. In this case we suggest that the main consequence of this broadband forcing is again to flatten the $k^{-7/3}$ spectrum.	  
	  
	  Over the summertime North Atlantic, buoyancy anomalies generate a more local velocity field and the surface kinetic energy spectrum is flatter than in winter. This contradicts the $k^{-3}$ spectrum found in observations and numerical models \citep{sasaki_impact_2014,callies_seasonality_2015}. However, observations also suggest that the surface geostrophic velocity is no longer dominated by the surface buoyancy induced contribution, suggesting the importance of interior potential vorticity for the summertime surface velocity \citep{gonzalez-haro_global_2014,miracca-lage_can_2022}. As such, the surface kinetic energy predictions of the present model, which neglects interior potential vorticity, are not valid over the summertime North Atlantic.
	  
	
	The situation in the North Pacific is broadly similar to that in the North Atlantic. In the Southern Ocean, however, the weak depth-averaged stratification leads to values of $L_H$ close to 150-200 km. As such, the bottom boundary becomes important at smaller horizontal scales than in the North Atlantic. Regardless of whether the appropriate bottom boundary condition is no-slip \eqref{eq:no-slip} or free-slip \eqref{eq:free-slip}, in both cases, the resulting inversion function implies a steepening to a $k^{-3}$ surface kinetic energy spectrum (appendix A). The importance of the bottom boundary in the Southern Ocean may explain the observed steepness of the surface kinetic energy spectra [between $k^{-2.5}$ to $k^{-3}$ \citep{vergara_revised_2019}] even though the surface geostrophic velocity seems to be largely due to surface buoyancy anomalies throughout the year \citep{gonzalez-haro_global_2014}.
	
	The claims made in this chapter can be explicitly tested in a realistic high-resolution ocean model; this can be done by finding regions where the surface streamfunction as reconstructed from sea surface height is highly correlated to the surface streamfunction as reconstructed from sea surface buoyancy \citep[or temperature, as in][]{gonzalez-haro_global_2014}. Then, in regions where both streamfunctions are highly correlated, the theory predicts that the inversion function, as computed from the stratification [equation \eqref{eq:mk}], should be identical to the inversion function computed through the surface streamfunction and buoyancy fields [equations \eqref{eq:transfer_func} and \eqref{eq:transfer_inversion}]. Moreover, in these regions, the model surface kinetic energy spectrum must be between the inverse cascade and forward cascade kinetic energy spectra [equations \eqref{eq:inverse_KE} and \eqref{eq:forward_KE}].
		
	 Finally the vertical structure equation \eqref{eq:Psi_equation} along with the inversion relation \eqref{eq:theta_fourier} between $\thetak$ and $\psik$ suggest the possibility of measuring the buoyancy frequency's vertical structure, $N(z)$, using satellites observations. This approach, however, is limited to regions where the surface geostrophic velocity is largely due to surface buoyancy anomalies. By combining satellite measurements of sea surface temperature and sea surface height, we can use the inversion relation \eqref{eq:theta_fourier} to solve for the inversion function. Then we obtain $N(z)$ by solving the inverse problem for the vertical structure equation \eqref{eq:Psi_equation}. How practical this approach is to measuring the buoyancy frequency's vertical structure remains to be seen.
	 			 	 	 
\begin{subappendices}

\section{The small- and large-scale limits}\label{A-small_large}
	   	 
	\subsection{The small-scale limit}
	Let $h$ be a characteristic vertical length scale associated with $\sigma(z)$ near $z=0$. Then, in the small-scale limit, $k\, \sigma_0 \, h \gg 1$, the infinite bottom boundary condition \eqref{eq:Psi_lower} is appropriate.
	With the substitution
	\begin{align}
		\Psi(z) = \sigma(z) \, P(z),
	\end{align}
	we transform the vertical structure equation \eqref{eq:Psi_equation} into a Schrödinger equation
	\begin{equation}\label{eq:app-P_equation}
		\dd{P}{z} = \left[- \frac{1}{\sigma}\dd{\sigma}{z} + 2 \left(\frac{1}{\sigma}\d{\sigma}{z}\right)^2 + k^2 \, \sigma^2 \right] P,
	\end{equation}
	with a lower boundary condition
	\begin{equation}
		\sigma \, P \rightarrow 0 \quad \text{ as } \quad z\rightarrow -\infty.
	\end{equation} 
	In the limit $k\, \sigma_0 \, h \gg 1$, the solution to the Schrödinger equation equation \eqref{eq:app-P_equation} is given by
	 \begin{equation}
	  	\Psi_k(z) \approx \sqrt{\frac{\sigma(z)}{\sigma_0}} \, \exp\left({k\, \int_0^z \sigma(s) \di{s}}\right).
	  \end{equation}
	  On substituting $\Psi_k(z)$ into the definition of the inversion function \eqref{eq:mk}, we obtain $m(k) \approx k/\sigma_0$ to leading order in $(k\sigma_0 h)^{-1}$. Therefore, the inversion relation in the small-scale limit coincides with the familiar inversion relation of constant stratification surface quasigeostrophic theory \citep{blumen_uniform_1978,held_surface_1995}.
	  
	 \subsection{The large-scale free-slip limit}
	Let $k_H= 2\pi/L_H$, where the horizontal length scale $L_H$ is defined in equation \eqref{eq:LH}. Then, in the large-scale limit, $k/k_H\ll 1$, we assume a solution of the form
	\begin{equation}\label{eq:Psi_series_large}
		\Psi_k(z) = \Psi_k^{(0)}(z) + \left(\frac{k}{ k_H}\right)^2 \, \Psi_k^{(1)}(z) + \cdots.
	\end{equation}
	Substituting the series expansion \eqref{eq:Psi_series_large} into the vertical structure equation \eqref{eq:Psi_equation} and applying the free-slip bottom boundary condition \eqref{eq:free-slip} yields
	\begin{equation}\label{eq:Psi_free_sol}
		\Psi_k(z) \approx A \left[1 + k^2 \int_{-H}^z \sigma^2(s) \, \left(s + H\right) \di\,s + \cdots \right],
	\end{equation}
	where $A$ is a constant determined by the upper boundary condition \eqref{eq:Psi_upper}.
	To leading order in $k/k_H$, the large-scale vertical structure is independent of depth.
	
	Substituting the solution \eqref{eq:Psi_free_sol} into the definition of the inversion function \eqref{eq:mk} gives
	\begin{equation}\label{eq:mk_free}
		m(k) \approx H \, k^2.
	\end{equation}
	Therefore, over a free-slip bottom boundary, the large-scale dynamics resemble two-dimensional vorticity dynamics, generalizing the result of \cite{tulloch_theory_2006} to arbitrary stratification $\sigma(z)$.

	 \subsection{The large-scale no-slip limit}

	 Substituting the expansion \eqref{eq:Psi_series_large} into the vertical structure equation \eqref{eq:Psi_equation} and applying the no-slip lower boundary condition \eqref{eq:no-slip} yields
	\begin{align}\label{eq:Psi_no_sol}
		\Psi_k (z) \approx B \Bigg[ \int_{-H}^z \, \sigma^2(s) \, \di \, s \, + k^2 \, \int_{-H}^z \sigma^2(s_3) \int_{-H}^{s_3} \int_{-H}^{s_2} \sigma^2(s_1) \, \di s_1 \, \di s_2 \, \di s_3 \Bigg],
	\end{align}
	where $B$ is a constant determined by the upper boundary condition \eqref{eq:Psi_upper}. Substituting the solution \eqref{eq:Psi_no_sol} into the definition of the inversion function \eqref{eq:mk} gives
	\begin{equation}\label{eq:mk_no}
		m(k) \approx m_1 \left(k_\sigma^2 +  k^2 \right),
	\end{equation}
	where $k_\sigma = \sqrt{m_0/m_1}$ is analogous to the deformation wavenumber, the constant $m_0$ is given by
	\begin{equation}
		m_0 = \left[ \int_{-H}^0  \sigma^2(s) \di s \right] ^{-1}.
	\end{equation}
	and $m_1$ is some constant determined by integrals of $\sigma(z)$. If $\sigma(z)$ is positive then both $m_0$ and $m_1$ are also positive. Therefore, over a no-slip bottom boundary, the large-scale dynamics resemble those of the equivalent barotropic model.

\section{Inversion function for piecewise constant stratification}

We seek a solution to the vertical structure equation \eqref{eq:Psi_equation} for the piecewise constant stratification \eqref{eq:step_strat} with upper boundary condition \eqref{eq:Psi_upper} and the infinite lower boundary condition \eqref{eq:Psi_lower}. The solution has the form 
\begin{equation}\label{eq:Psi_sharp}
	\Psi_k(z) = \cosh\left(\sigma_0 \, k\,z\right) + a_2 \sinh\left(\sigma_0 \, k\, z\right),
\end{equation}
for $-h \leq z \leq 0$, and
\begin{equation}
	\Psi_k(z) = a_3\,  e^{\sigma_\mathrm{pyc} k(z+h)},
\end{equation}
for $-\infty < z < -h$. 
To determine $a_2$ and $a_3$, we require $\Psi_k(z)$ to be continuous across $z=-h$ and that its derivative satisfy
\begin{equation}
	\frac{1}{\sigma_0^2} \, \d{\Psi_k(-h^+)}{z} = \frac{1}{\sigma_\mathrm{pyc}^2} \, \d{\Psi_k(-h^-)}{z},
\end{equation}
where the $-$ and $+$ superscripts indicate limits from the below and above respectively. Solving for $a_2$ and substituting equation \eqref{eq:Psi_sharp} into the definition of the inversion function \eqref{eq:mk} then yields $m(k)$.

\section{The numerical model}

We solve the time-evolution equation \eqref{eq:theta_equation} using the pseudo-spectral \texttt{pyqg} model \citep{abernathey_pyqgpyqg_2019}. To take stratification into account, we use the inversion relation \eqref{eq:theta_inversion}. Given a stratification profile $\sigma(z)$ from ECCOv4, we first interpolate the ECCOv4 stratification profile with a cubic spline onto a vertical grid with 350 vertical grid points. We then numerically solve the vertical structure equation \eqref{eq:Psi_equation}, along with boundary conditions \eqref{eq:Psi_upper} and either \eqref{eq:no-slip} or \eqref{eq:free-slip}, and obtain the vertical structure at each wavevector $\vec k$. Using the definition of the inversion function \eqref{eq:mk} then gives $m(k)$.

 We apply a large-scale forcing, $F$, between the (non-dimensional) wavenumbers $3.5<k<4.5$ in all our simulations, corresponding to horizontal length scales 88 - 114 km. Otherwise, the forcing $F$ is as in \cite{smith_turbulent_2002}. The dissipation term can be written as
 \begin{equation}
     D = r_d \, \theta + \mathrm{ssd}
 \end{equation}
 where $r_d$ is a damping rate and $\mathrm{ssd}$ is small-scale dissipation. Small-scale dissipation is through an exponential surface potential enstrophy filter as in \cite{arbic_coherent_2003}. 

\end{subappendices}


%% file: 3-Jets/Jets.tex

\chapter{The Buoyancy Staircase Limit in Surface Quasigeostrophic Turbulence}\label{Ch-jets}

\begin{abstractchapter}
	Surface buoyancy gradients over a quasigeostrophic fluid permit the existence of surface-trapped Rossby waves.
The interplay of these Rossby waves with surface quasigeostrophic turbulence results in latitudinally inhomogeneous mixing that, under certain conditions, culminates in a surface buoyancy staircase: a meridional buoyancy profile consisting of mixed-zones punctuated by sharp buoyancy gradients, with eastward jets centred at the sharp gradients and weaker westward flows in between. 
In this article, we investigate the emergence of this buoyancy staircase limit in surface quasigeostrophic turbulence and we examine the dependence of the resulting dynamics on the vertical stratification.
Over decreasing stratification [$\mathrm{d}N(z)/\mathrm{d}z\leq0$, where $N(z)$ is the buoyancy frequency], we obtain flows with a longer interaction range (than in uniform stratification) and highly dispersive Rossby waves. In the staircase limit, we find straight jets that are perturbed by eastward propagating along jet waves, similar to two-dimensional barotropic $\beta$-plane turbulence. In contrast, over increasing stratification [$\mathrm{d}N(z)/\mathrm{d}z\geq0$], we obtain flows with shorter interaction range and weakly dispersive Rossby waves. In the staircase limit, we find sinuous jets with large latitudinal meanders whose shape evolves in time due to the westward propagation of weakly dispersive along jet waves.
These along jet waves have larger amplitudes over increasing stratification than over decreasing stratification, and, as a result, the ratio of domain-averaged zonal to meridional speeds is two to three times smaller over increasing stratification than over decreasing stratification. Finally, we find that, for a given Rhines wavenumber, jets over increasing stratification are closer together than jets over decreasing stratification.
\end{abstractchapter}

\section{Introduction}

Perturbations to a barotropic (i.e., depth-independent) fluid with a background potential vorticity gradient, $\beta>0$, propagate westward as Rossby waves.
In a turbulent flow, the non-linear interplay between Rossby waves and turbulence results in the latitudinally inhomogeneous mixing of potential vorticity, which, through a positive dynamical feedback, spontaneously reorganizes the flow into one characterized by eastward jets  \citep{dritschel_multiple_2008}.
 The ultimate limit of such inhomogeneous mixing, which can be achieved for sufficiently large values of $\beta$, is a potential vorticity staircase: a piecewise constant potential vorticity profile consisting well-mixed regions separated by isolated discontinuities, with eastward jets centred at the  discontinuities and westward flows in between \citep{danilov_scaling_2004,dunkerton_barotropic_2008,scott_structure_2012,scott_zonal_2019}. 

Analogously, a buoyancy gradient at the surface of a quasigeostrophic fluid supports the existence of surface-trapped Rossby waves that are less dispersive than their barotropic counterparts \citep{held_surface_1995,lapeyre_surface_2017}. The purpose of this chapter is to investigate the formation of zonal jets in the presence of a background surface buoyancy gradient and to examine the realizability of surface buoyancy staircases in the surface quasigeostrophic model.
 Although the present study is the first to systematically investigate the emergence of surface quasigeostrophic jets, there are previous studies which make use of the uniformly stratified surface quasigeostrophic model with a background buoyancy gradient. 
 These include \cite{smith_turbulent_2002}, who derive the dependence of the diffusion coefficient of a passive tracer in the presence a background buoyancy gradient. 
 Another is \cite{sukhatme_local_2009}, who, in their investigation of $\alpha$-turbulence models with a background gradient, note that, because of the decreased interaction range, surface quasigeostrophic jets in uniform stratification should be narrower than their counterparts in the barotropic model. 
 Finally, \cite{lapeyre_surface_2017} demonstrates that jets can indeed form in the uniformly stratified surface quasigeostrophic model.
 
We also investigate how surface quasigeostrophic jets depend on the underlying vertical stratification. Chapter \ref{Ch-SQG} shows that the vertical stratification modifies the interaction range of vortices in the surface quasigeostrophic model. Suppose we have an infinitely deep fluid governed by the time-evolution of geostrophic buoyancy anomalies at its upper boundary. Then if the stratification is decreasing [$N'(z)\leq0$, where $N(z)$ is buoyancy frequency] towards the fluid's surface (that is, the upper boundary), then the interaction range is longer than in the uniformly stratified model and the resulting turbulence is characterized by thin buoyancy filaments \textemdash{} analogous to the thin vorticity filaments in two-dimensional barotropic turbulence. Conversely, if the stratification is increasing [$N'(z)\geq0$] towards the surface, then the interaction range is shorter than in uniform stratification, and the buoyancy field appears spatially diffuse and lacks thin filamentary structures. In this chapter, we find  that the interaction range is related to Rossby wave dispersion: flows with a longer interaction range have more dispersive Rossby waves whereas flows with a shorter interaction range have less dispersive Rossby waves. One of our aims is to characterize the dependence of surface quasigeostrophic jets on the functional form of the vertical stratification.  

There are two motivations behind the present work. The first is its potential relevance to the upper ocean. Buoyancy anomalies at the ocean's surface are governed by the surface quasigeostrophic model \citep{lapeyre_dynamics_2006,lacasce_estimating_2006,isernfontanet_potential_2006}. Both numerical \citep{isernfontanet_three-dimensional_2008,lapeyre_what_2009,qiu_reconstructability_2016,qiu_reconstructing_2020,miracca-lage_can_2022} as well as observational \citep{gonzalez-haro_global_2014} studies indicate that a significant fraction of the surface geostrophic velocity is induced by sea surface buoyancy anomalies, especially over wintertime extratropical currents. Moreover, upper ocean turbulence has been found to be anisotropic \citep{maximenko_observational_2005,scott_zonal_2008}, with significant differences in anisotropy between major extratropical currents and other regions in the ocean \citep{wang_anisotropy_2019}. However, our neglect of the planetary $\beta$ effect, as well our assumption of vanishing interior potential vorticity, may limit the direct relevance of this study to the upper ocean.

The second motivation is that the variable stratification surface quasigeostrophic model is a simple two-dimensional model in which we can investigate how jet dynamics depend on the stratification's vertical structure. 
Another such model is the equivalent barotropic model for which the deformation radius represents the rigidity of the free surface. Small values of the deformation radius lead to a pliable free surface allowing a significant degree of horizontal divergence. 
The resulting flow then has an exponentially short interaction range, with a horizontal attenuation on the order of the deformation radius \citep{polvani_two-layer_1989}, and with approximately non-dispersive Rossby waves. Consequently, for a finite deformation radius, we obtain jets whose width is on the order of the deformation radius with a fixed meandering shape \citep{scott_spacing_2022}. 
In contrast, for the variable stratification surface quasigeostrophic model, rather than just specifying a constant (i.e., the deformation wavenumber), one instead has to specify the stratification's functional form, $N(z)$. 
Over decreasing stratification $[N'(z) <0]$, because of the longer interaction range and the more dispersive waves, we obtain jets similar to the two-dimensional barotropic model. 
Conversely, over increasing stratification $[N'(z)>0]$, the shorter interaction range along with the weakly dispersive waves lead to sinuous jets whose shape evolves in time through the propagation of weakly dispersive along jet waves. 
Moreover, because of these along jet waves, a smaller fraction of the total energy is contained in the zonal mode over increasing stratification (with a shorter interaction range) than over decreasing stratification (with a longer interaction range). 

The remainder of this chapter is organized as follows.
 Section \ref{S-theory} introduces the variable stratification surface quasigeostrophic model and shows how the stratification's vertical structure controls both the interaction range of point vortices as well as the dispersion of surface-trapped Rossby waves. 
 Then, in section \ref{S-staircase_theory}, we introduce two wavenumbers, $\keps$ and $k_r$, whose ratio, $\keps/k_r$, forms the key non-dimensional parameter of this study; here, $\keps$ is a wavenumber depending on the energy injection rate whereas $k_r$ is a wavenumber depending on surface damping rate. This non-dimensional number is a generalization of the non-dimensional number used in previous studies \citep{danilov_rhines_2002,sukoriansky_arrest_2007,scott_structure_2012}. 
 By considering an idealized buoyancy staircase, we also investigate  how the Rhines wavenumber relates to the jet spacing under decreasing, increasing, and uniform stratification.
 Section \ref{S-numerical} then presents numerical experiments detailing the emergence of the staircase limit as $\keps/k_r$ is increased for various stratification profiles. 
 In addition, we also present experiments where we fix the external parameters and vary the vertical stratification alone.
 Finally, we conclude in section \ref{S-conclusion_jets}.

\section{The interaction range and wave dispersion}\label{S-theory}

\subsection{Equations of motion}
 Consider an infinitely deep fluid with zero interior potential vorticity. The geostrophic streamfunction, $\psi$, then satisfies
 \begin{equation}\label{eq:zero_pv}
  \pd{}{z}\left(\frac{1}{\sigma^2} \pd{\psi}{z} \right) + \lap \psi = 0
 \end{equation}
 in the fluid interior, $z \in (-\infty,0)$. The horizontal Laplacian is denoted by $\lap = \partial_x^2 + \partial_y^2$ and the non-dimensional stratification is given by 
 \begin{equation}
 	 \sigma(z) = N(z)/f,
 \end{equation}
 where $N(z)$ is the buoyancy frequency and $f$ is the constant local value of the Coriolis parameter.
  Time-evolution is  determined by the material conservation of surface potential vorticity \citep{bretherton_critical_1966},
   \begin{equation}\label{eq:theta_inversion}
 	\theta = -\frac{1}{\sigma_0^2} \pd{\psi}{z}\Big|_{z=0},
  \end{equation}
  at the upper boundary, $z=0$, where $\sigma_0 = \sigma(0)$. Explicitly, the time-evolution equation is
 \begin{equation}\label{eq:time-evolution}
 	\pd{\theta}{t} + J(\psi,\theta) + \Lambda \, \partial_x \theta = F-D,
 \end{equation}
 at $z=0$, where $J(\psi,\theta) = \partial_x \psi \, \partial_y \theta - \partial_x \theta \, \partial_y \, \psi$ represents the advection of $\theta$ by the geostrophic velocity, $\vec u = \unit z \times \grad \psi$. The frequency, $\Lambda$, is given by
 \begin{equation}
 	\Lambda = \frac{1}{\sigma_0^2} \d{U}{z}\Big|_{z=0},
 \end{equation}
 where $U(z)$ is a background zonal geostrophic flow. Without loss of generality, we have assumed that $U(0)=0$ in the time-evolution equation \eqref{eq:time-evolution} to eliminate a constant advective term. The dissipation, $D$, consists of linear damping and small-scale dissipation,
 \begin{equation}
 	D = r \, \theta + \mathrm{ssd},
 \end{equation}
 where $r$ is the damping rate.
 The forcing, $F$, and the small-scale dissipation, $\mathrm{ssd}$, are described in section \ref{S-numerical}.
 
  The surface buoyancy anomaly, $b|_{z=0}$, is related to the surface potential vorticity, $\theta$, through 
  \begin{equation}
  	 b|_{z=0} = - f\,\sigma_0^2 \, \theta.
  \end{equation}
  Therefore, the time-evolution equation \eqref{eq:time-evolution} equivalently states that surface buoyancy anomalies are materially conserved in the absence of forcing and dissipation. In addition, the frequency, $\Lambda$, corresponds to a meridional buoyancy gradient,
 \begin{equation}\label{eq:dBdy}
 	\d{B}{y}\Big|_{z=0} = - f \, \sigma_0^2 \, \Lambda,
 \end{equation}
 where $B(y,z) $ is the buoyancy field that is in geostrophic balance with background zonal velocity, $U(z)$.
 
 If we further assume a doubly periodic domain in the horizontal, then we can expand the streamfunction as
 \begin{equation}\label{eq:fourier_psi}
 	\psi(\vec r, z,t) = \sum_{\vec k} \hat \psi_{\vec k}(t) \, \Psi_k(z)\, \mathrm{e}^{\mathrm{i} \vec k \cdot \vec{x}},
 \end{equation}
 where $\vec x=(x,y)$ is the horizontal position vector, $z$ is the vertical coordinate, $\vec k = (k_x,k_y)$ is the horizontal wavevector, $k=\abs{\vec k}$ is the horizontal wavenumber, and $t$ is the time coordinate. The non-dimensional wavenumber-dependent vertical structure, $\Psi_k(z)$, is determined by the boundary value problem (chapter \ref{Ch-SQG})
 \begin{equation}\label{eq:vertical_structure}
 	-\d{}{z} \left(\frac{1}{\sigma^2}\d{\Psi_k}{z}\right) + k^2 \, \Psi_k(z) = 0,
 \end{equation}
with the upper boundary condition
\begin{equation}\label{eq:upper}
	\Psi_k(0)=1,
\end{equation}
 and lower boundary condition
 \begin{equation}\label{eq:lower}
 	\Psi_k \rightarrow 0 \quad \textrm{as } \quad z \rightarrow -\infty.
 \end{equation}
 The upper boundary condition \eqref{eq:upper} is a normalization for the vertical structure, $\Psi_k(z)$, chosen so that 
 \begin{equation}
 	 \psi(\vec r, z=0,t) = \sum_{\vec k} \hat \psi_{\vec k}(t) \, \mathrm{e}^{\mathrm{i} \vec k \cdot \vec{x}}.
 \end{equation}
  The corresponding Fourier expansion of the surface potential vorticity is given by
  \begin{equation}
  	\theta(\vec r, t) = \sum_{\vec k} \hat \theta_{\vec k}(t) \, \mathrm{e}^{\mathrm{i} \vec k \cdot \vec{x}},
  \end{equation}
  where 
  \begin{equation}\label{eq:inversion}
  	\hat \theta_{\vec k} = - m(k) \, \hat \psi_{\vec k},
  \end{equation}
  and the function $m(k)$ is given by
  \begin{equation}\label{eq:mk}
  	m(k) = \frac{1}{\sigma_0^2} \d{\Psi_k(0)}{z}.
  \end{equation}
  The function $m(k)$ relates $\hat \theta_{\vec k}$ to $\hat \psi_{\vec k}$ in the Fourier space inversion relation \eqref{eq:inversion} and so we call $m(k)$ the \emph{inversion function}.
  
	To recover the well-known case of the uniformly stratified quasigeostrophic model \citep{held_surface_1995}, set $\sigma(z) = \sigma_0$. Then the vertical structure equation \eqref{eq:vertical_structure} along with boundary conditions \eqref{eq:upper} and \eqref{eq:lower} yield the exponentially decaying vertical structure $\Psi_k(z) = \exp\left({\sigma_0\,k\,z}\right)$. On substituting $\Psi_k(z)$ into equation \eqref{eq:mk}, we obtain a linear inversion function 
	\begin{equation}\label{eq:mk_linear}
			m(k) = \frac{k}{\sigma_0}
	\end{equation}
	 and hence [from the inversion relation \eqref{eq:inversion}] a linear-in-wavenumber inversion relation $\hat \theta_{\vec k} = - (k/\sigma_0) \, \hat \psi_{\vec k}$.
	
	\subsection{The inversion function and spatial locality}
	
	\begin{figure}
  \centerline{\includegraphics[width=1.\columnwidth]{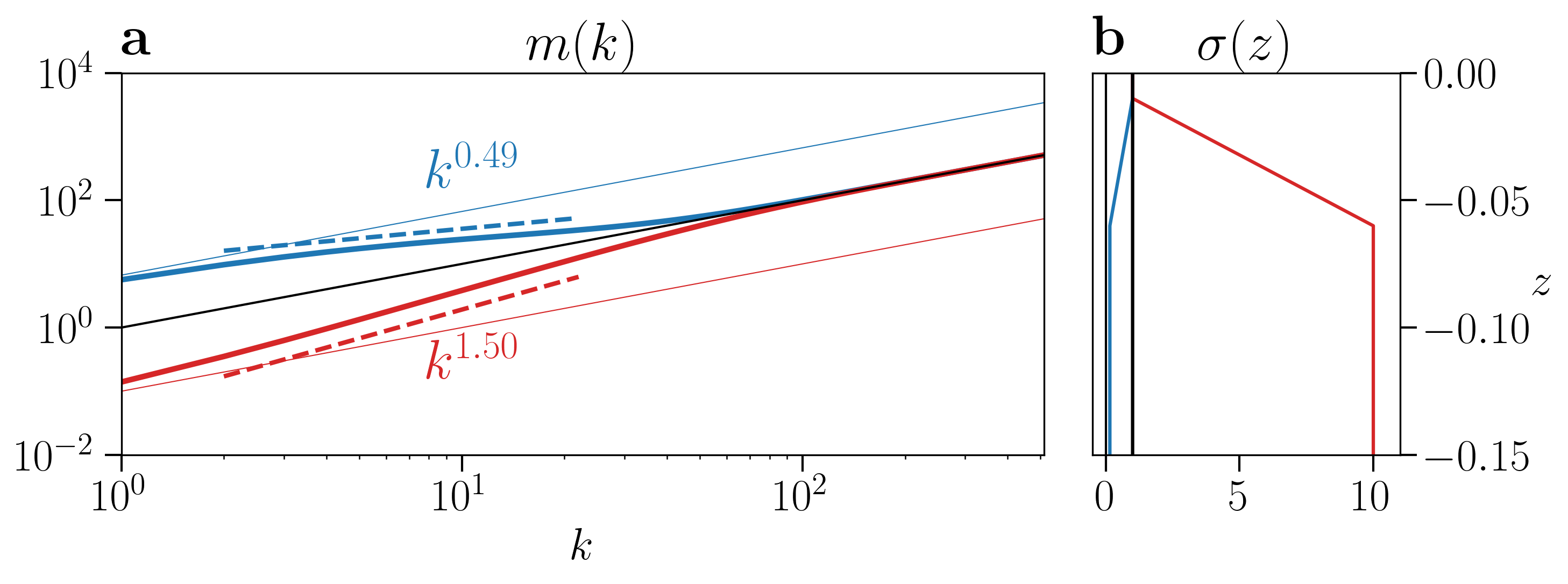}}
  \caption{The inversion functions, $m(k)$ [in panel (a)] for two stratification profiles [panel (b)] given by the piecewise stratification profile \eqref{eq:sigma_piece}. One stratification profile is increasing [$\sigma'(z)\geq 0$, blue], with $\sigma_0=1$, $\sigma_\mathrm{pyc}=0.15$, $h_\mathrm{mix} = 0.01$, and $h_\mathrm{lin}=0.05$. The other stratification profile is decreasing [$\sigma'(z)\leq 0$, red] with $\sigma_0=1$, $\sigma_\mathrm{pyc}=10$, $h_\mathrm{mix} = 0.01$, and $h_\mathrm{lin}=0.05$. The thin black line is given by $k/\sigma_0$ where $\sigma_0=1$, whereas the blue and red lines are given by $k/\sigma_\mathrm{pyc}$ with $\sigma_\mathrm{pyc}=0.15$ for the thin blue line and $\sigma_\mathrm{pyc}=10$ for the thin red line.}
  \label{F-mk_sigma}
	\end{figure}
	
	The inversion function $m(k)$, which is determined by the stratification's vertical structure, controls the spatial locality of the resulting turbulence. We illustrate this point with the following piecewise stratification profile,
	\begin{equation}\label{eq:sigma_piece}
		\sigma(z) = 
		\begin{cases}
			\sigma_0  \quad &\text{for} \quad -h_\mathrm{mix} < z < 0\\
			\sigma_0 + \Delta \sigma \left(\frac{z+h_\mathrm{mix}}{h_\mathrm{lin}}\right) \quad &\text{for} \quad -(h_\mathrm{mix} + h_\mathrm{lin}) < z < -h_\mathrm{mix}\\
			\sigma_\mathrm{pyc} \quad &\text{for}  \quad -\infty < z < -(h_\mathrm{mix} + h_\mathrm{lin}) ,
		\end{cases}	
	\end{equation}
	where $\Delta \sigma =\sigma_0 - \sigma_\mathrm{pyc}$. 
	At small horizontal scales, where $k\gg k_s$, and 
	\begin{equation}
		k_s=1/\left(\sigma_0\,h_\mathrm{mix}\right),
	\end{equation}
	then $m(k) \approx k/\sigma_0$, as in the uniformly stratified model of \cite{held_surface_1995}. Likewise, in the large-scale limit, where $k \ll k_\mathrm{pyc}$, and 
	\begin{equation}
		k_\mathrm{pyc} = 
		\begin{cases}
			1/\left(\sigma_\mathrm{pyc} \, h_\mathrm{mix}\right) \quad &\text{for} \quad \Delta \sigma \leq 0 \\
			\sigma_\mathrm{pyc}/\left(\sigma_0^2 \, h_\mathrm{mix}\right) \quad &\text{for} \quad \Delta \sigma>0,
		\end{cases}
	\end{equation}
	then $m(k) \approx k/\sigma_\mathrm{pyc}$. However, for wavenumbers between $k_\mathrm{pyc} \lesssim k \lesssim k_s$, the inversion function takes an approximate power law form
	\begin{equation}\label{eq:mk_power}
		m(k) \approx m_0 \, k^{\alpha},
	\end{equation}
	where $m_0>0$ and $\alpha \geq 0$.
	The power $\alpha$ depends on the ratio $\sigma_\mathrm{pyc}/\sigma_0$ between the deep and surface stratification. If the stratification decreases towards the surface [$\sigma^\prime(z) \leq 0$, or $\sigma_\mathrm{pyc}/\sigma_0>1$] then $\alpha > 1$, with $\sigma_\mathrm{pyc}/\sigma_0 \rightarrow \infty$ sending $\alpha \rightarrow 2$. In contrast, if the stratification increases towards the surface [$\sigma^\prime(z) \geq 0$, or $\sigma_\mathrm{pyc}/\sigma_0 < 1$] then $\alpha <1$, with  $\sigma_\mathrm{pyc}/\sigma_0 \rightarrow 0$ sending $\alpha \rightarrow 0$.	
	Thus, for wavenumbers $k_\mathrm{pyc} \lesssim k \lesssim k_s$, the inversion relation \eqref{eq:inversion} has the approximate form
	\begin{equation}
		\hat \xi_{\vec k} = - k^\alpha \, \hat \psi_{\vec k},
	\end{equation}
	where $\hat \xi_{\vec{k}}= \hat \theta_{\vec{k}}/m_0 $, which is the inversion relation for $\alpha$-turbulence \citep{pierrehumbert_spectra_1994,smith_turbulent_2002,sukhatme_local_2009}. Figure \ref{F-mk_sigma} provides two examples, one with decreasing stratification (with $\alpha \approx 1.50$) and another with increasing stratification (with $\alpha \approx 0.49$). 
	
	To see how the parameter $\alpha$ modifies the resulting dynamics, consider a point vortex at the origin, given by $\xi =  \delta(\abs{\vec x})$, where $\abs{\vec x}$ is the horizontal distance from the vortex centre, and $\delta(\abs{\vec x})$ is the Dirac delta. If $\alpha=2$, then the streamfunction induced by the point vortex is logarithmic, $\psi(\abs{\vec x}) = \log(\abs{\vec x})/(2\pi)$. If $0 < \alpha < 2$, then $\psi(\abs{\vec x}) = -C_\alpha/\abs{\vec x}^{2-\alpha}$ where $C_\alpha>0$ is a constant \citep{iwayama_greens_2010}. Smaller $\alpha$ leads to  vortices with velocities decaying more quickly with the horizontal distance $\abs{\vec x}$, and hence a shorter interaction range. Thus, the vertical stratification modifies the relationship between a surface buoyancy anomaly and its induced velocity field: a surface buoyancy anomaly over decreasing stratification [$\sigma'(z)\leq 0$] generates a longer range velocity field than an identical buoyancy anomaly over increasing stratification [$\sigma'(z)\geq 0$].

	\subsection{Wave dispersion in variable stratification}
	
	The background gradient term, $\Lambda$, in the time-evolution equation \eqref{eq:time-evolution} allows for the propagation of surface-trapped Rossby waves. Substituting a wave solution of the form $\psi(x,z,t) =  \Psi_{k}(z) \, \exp \left[{\mathrm{i} \left(\vec k \cdot \vec r - \omega t \right)}\right]$, where the vertical structure $\Psi_k(z)$ satisfies the boundary value problem \eqref{eq:vertical_structure}\textendash \eqref{eq:lower}, into the time-evolution equation \eqref{eq:time-evolution} yields the angular frequency 
	\begin{equation}\label{eq:dispersion}
		\omega(\vec k) = - \frac{\Lambda\,k_x}{m(k)}.
	\end{equation}
	Given the relationship \eqref{eq:dBdy} between the  meridional surface buoyancy gradient $\mathrm{d}B/\mathrm{d}y|_{z=0}$ and the frequency $\Lambda$, a poleward decreasing buoyancy gradient ($f\mathrm{d}{B}/\mathrm{d}{y}<0$) implies westward propagating $(\omega<0)$ Rossby waves.
	
	The dispersion relation \eqref{eq:dispersion} shows that Rossby wave dispersion is coupled to the flow's interaction range and hence the stratification's vertical structure. If we approximate the inversion function as a power law \eqref{eq:mk_power} between $k_\mathrm{pyc} \lesssim k \lesssim k_s$, then the zonal phase speed, $c=\omega/k_x$, becomes $c\sim 1/k^\alpha$. Therefore, at these horizontal scales, Rossby waves are more dispersive over decreasing stratification (with $\alpha >1$) than over increasing stratification (with $\alpha <1$). In the limit that $\sigma_0 \gg \sigma_{\mathrm{pyc}}$ in which $\alpha \rightarrow 0$, then $c \approx $ constant, and so Rossby waves become non-dispersive.
	
	\section{From edge waves to surface-trapped jets}\label{S-staircase_theory}
	
	The emergence of jets in barotropic $\beta$-plane turbulence is due to two properties of the potential vorticity \citep{dritschel_multiple_2008,scott_zonal_2019}. The first is the resilience of strong latitudinal potential vorticity gradients to mixing \citep[i.e., "Rossby wave elasticity", ][]{dritschel_multiple_2008}. Regions with weak latitudinal potential vorticity gradients are preferentially mixed, weakening the gradient in these regions and enhancing the gradient in regions where the latitudinal potential vorticity gradient is already strong \citep{dritschel_jet_2011}. The ultimate limit of such latitudinally inhomogeneous mixing is a potential vorticity staircase \citep{danilov_barotropic_2004,dritschel_multiple_2008,scott_structure_2012}, which consists of uniform regions of potential vorticity punctuated by sharp potential vorticity gradients. The second property is that, through potential vorticity inversion, strong (positive) latitudinal gradients in potential vorticity correspond to eastward jets. Therefore, inverting a potential vorticity staircase produces a flow with eastward zonal jets centred at the sharp frontal zones, with weaker westward flows in between \citep{scott_zonal_2019}.

	
	However, the limit of a  potential vorticity staircase is only achieved for sufficiently large values of the non-dimensional number $\keps/\krh$ \citep{scott_structure_2012}, which is a ratio of the forcing intensity wavenumber, $\keps$, to the Rhines wavenumber, $\krh$. The forcing intensity wavenumber is given by \citep{maltrud_energy_1991}
	\begin{equation}
		\keps = (\beta^3/\varepsilon_{\mathcal{K}})^{1/5},
	\end{equation}
	 where $\varepsilon_{\mathcal{K}}$ is the kinetic energy injection rate in the barotropic model, and is obtained by setting the turbulent strain rate equal to the Rossby wave frequency \citep{vallis_generation_1993}. The Rhines wavenumber is given by \citep{rhines_waves_1975}
	 \begin{equation}\label{eq:rhines_2d}
		 \krh = \sqrt{\beta/U_\mathrm{rms}},
	 \end{equation}
	  where $U_\mathrm{rms}$ is the rms velocity. 
	 \cite{scott_structure_2012} found that the ratio $\keps/\krh$ controls the structure of zonal jets in barotropic $\beta$-plane turbulence; as $\keps/\krh$ is increased, the zonal jet strength increases and the potential vorticity gradient at the jet core becomes larger, with the staircase limit approached as $\keps/\krh \sim O(10)$.
		  
	 Jet formation in surface quasigeostrophic turbulence proceeds similarly, with the surface buoyancy (which is proportional to $\theta$) taking the role of the potential vorticity and the frequency, $\Lambda$, taking the role of the potential vorticity gradient, $\beta$. In this section, we first derive a non-dimensional number analogous to $\keps/\krh$ for surface quasigeostrophy. Then we consider how vertical stratification (and the non-locality parameter $\alpha$) modifies jet structure in the buoyancy staircase limit, as well as how it modifies the relationship between the Rhines wavenumber and the jet spacing.
	 
	 Before proceeding, we comment on two differences between two-dimensional barotropic turbulence and its surface quasigeostrophic counterpart. First, in the absence of forcing and dissipation, the kinetic energy,
	  \begin{equation}\label{eq:kinetic}
	  	\mathcal{K} = - \frac{1}{2}\, \overline{\psi \lap \psi } = \frac{1}{2} \, \overline{\abs{u}^2},
	  \end{equation}
	  is a conserved constant in two-dimensional barotropic turbulence (the overline denotes an area average). With a constant kinetic energy injection rate, $\varepsilon_\mathcal{K}$, and a linear damping rate, $r$, the equilibrium kinetic energy is $\mathcal{K}=\varepsilon_{\mathcal{K}}/2r$. By definition, the rms velocity is given by 
	    $U_\mathrm{rms} = \sqrt{2\mathcal{K}}$. Combining this expression with the definition of the kinetic energy \eqref{eq:kinetic} and substituting into the definition of the Rhines wavenumber \eqref{eq:rhines_2d} yields a Rhines wavenumber expressed in terms of external parameters alone,  
	    \begin{equation}
	    	\krh = \beta^{1/2} (r/\varepsilon_\mathcal{K})^{1/4}.
	    \end{equation}
	   In contrast, in surface quasigeostrophy, the total energy,
	   \begin{equation}
	  	\mathcal{E} = - \frac{1}{2} \overline{\psi|_{z=0}\, \theta },
	  \end{equation}
	  is a conserved constant in the absence of forcing and dissipation and there is no general relationship between the rms velocity, $U_\mathrm{rms}$, and the equilibrium total energy, $\mathcal{E}=\varepsilon/2r$, where $\varepsilon$ is the total energy injection rate in the surface quasigeostrophic model. Therefore, we are not generally able to express the Rhines wavenumber in terms of the external parameters $\varepsilon$, $\Lambda$, and $r$.
	  Second, because $\mathcal{E}$ and $\mathcal{K}$ have different dimensions, the kinetic energy injection in the barotropic model, $\varepsilon_{\mathcal{K}}$, has different dimensions than the total energy injection rate in the surface quasigeostrophic model, $\varepsilon$. In particular, $\varepsilon$ has dimensions of $L^2/T^3$.
	  
	  

	  \subsection{The forcing intensity wavenumber}

	  	  
	  To obtain the forcing intensity wavenumber, $\keps$, we compare the Rossby wave frequency \eqref{eq:dispersion} to the turbulent strain rate, $\omega_s(k)$.
	  If the inversion function is not approximately constant (i.e., $\alpha \neq 0$) then the strain rate is (chapter \ref{Ch-SQG})
	  \begin{equation}\label{eq:strain}
	  	\omega_s(k) \sim \varepsilon^{1/3}\, k^{4/3} \,\left[m(k)\right]^{-1/3}.
	  \end{equation}
	  In particular, if $m(k) = m_0 \,k^{\alpha}$, then $\omega_s(k) \sim m_0^{1/3}\varepsilon^{1/3}\, k^{\left(4-\alpha\right)/3}$. Setting the absolute value of the Rossby wave frequency for waves with $k = k_x$ equal to the turbulent strain rate \eqref{eq:strain} yields the condition
	  \begin{equation}\label{eq:keps_condition}
	  	\keps \left[m(\keps)\right]^{2} \sim \frac{\abs{\Lambda}^3}{\varepsilon}.
	  \end{equation}
	 A solution to this equation always exists because $\mathrm{d}m/\mathrm{d}k \geq 0$. If the inversion function takes the power law form \eqref{eq:mk_power}, then we obtain
	 \begin{equation}\label{eq:keps_alpha}
	 	\keps = \left(\frac{\abs{\Lambda}^3}{m_0^2\,\varepsilon}\right)^{1/\left(2\alpha+1\right)},
	 \end{equation}
	 which is equivalent to a wavenumber derived in \cite{smith_turbulent_2002}.

	\subsection{The damping rate wavenumber and the Rhines wavenumber}

	Suppose the inversion function takes an approximate power law form, $m(k) \approx m_0 \, k^\alpha$, near the energy containing wavenumbers.
    Then the generalization of the Rhines wavenumber at these wavenumbers is 
	\begin{equation}\label{eq:rhines}
		\krh = \left(\frac{\Lambda}{m_0 \, U_\mathrm{rms}}\right)^{1/\alpha}.
	\end{equation}
	 However, unlike in two-dimensional barotropic turbulence where $U_\mathrm{rms}= \sqrt{2\mathcal{K}} = \sqrt{\varepsilon_{\mathcal{K}}/r}$, we do not have a general relationship between $U_\mathrm{rms}$ and the external parameters $r$ and $\varepsilon$ in surface quasigeostrophic turbulence. To obtain a second wavenumber that depends on the damping rate, $r$, we follow \cite{smith_turbulent_2002}. From dimensional considerations, the energy spectrum at small wavenumbers is
	\begin{equation}\label{eq:Espectrum}
		E_\Lambda(k) \sim \Lambda^2 \, k^{-(\alpha+3)}/m_0.
	\end{equation}
	Then, defining $k_r$ as the wavenumber at which the inverse cascade halts, we obtain
	\begin{equation}
		\frac{\varepsilon}{2r} \approx  \int_{k_r}^\infty E(k) \, \mathrm{d}k \approx \left(\frac{\Lambda^2 /m_0}{\alpha+2}\right) k_r^{-(\alpha+2)},
	\end{equation}
	where the second equality follows because the integral is dominated by its peak at low wavenumbers. Solving for $k_r$ and neglecting any non-dimensional coefficients, we obtain
	\begin{equation}\label{eq:kr}
		k_r = \left(\frac{\Lambda^2 \, r}{m_0 \, \varepsilon}\right)^{1/(\alpha+2)}.
	\end{equation}
	Note that the damping rate wavenumber, $k_r$, has the same dependence on $\Lambda$, $\varepsilon$, and $r$ as the Rhines wavenumber, $\krh$, only if $\alpha=2$. 
	
			 
	 \subsection{Surface potential vorticity inversion} 
	
	A perfect surface potential vorticity staircase consists of mixed zones of halfwidth $b$, where $\mathrm{d}\theta/\mathrm{d}y  = -\Lambda$, separated by jump discontinuities at which $\mathrm{d}\theta/\mathrm{d}y = \infty$. We find it more conveniant to work with the relative surface potential vorticity, $\theta$, rather than the total surface potential vorticity, $\theta + \Lambda \, y$. In this case, if the total surface potential vorticity,  $\theta + \Lambda \, y$, is a perfect staircase with step width $2b$, then the relative surface potential vorticity, $\theta$, is a $2b$-periodic sawtooth wave. 
	
	Our first question is whether such a staircase is possible for general $m(k)$. To answer this question, we consider the velocity field induced by a jump discontinuity in $\theta$. For a jump discontinuity in an infinite domain,
	 \begin{equation}
		\theta =  
		\begin{cases}
			\Delta \theta  \quad  &\text{for} \quad 0< y< \infty\\
			0 \quad &\text{for} \quad  -\infty < y < 0 ,
		\end{cases}
	\end{equation}
	the zonal velocity is given by
	\begin{equation}
		u  = \frac{\Delta \theta}{2\pi} \int_{-\infty}^\infty \frac{\mathrm{e}^{\mathrm{i} \, k_y y}}{m\left(\abs{k_y}\right)} \mathrm{d}k_y.
	\end{equation}
	If $m(k)=m_0\,k^\alpha$, then this expression is proportional to $\abs{y}^{\alpha -1}$ if $\alpha \neq 1$ and logarithmic otherwise, and so the zonal velocity diverges at $y=0$ if $\alpha \leq 1$.  Consequently, we expect that a perfect staircase should not be possible over constant or increasing stratification due to the divergence of the zonal velocity at a jump discontinuity.

	
	\begin{figure}
  \centerline{\includegraphics[width=0.9\columnwidth]{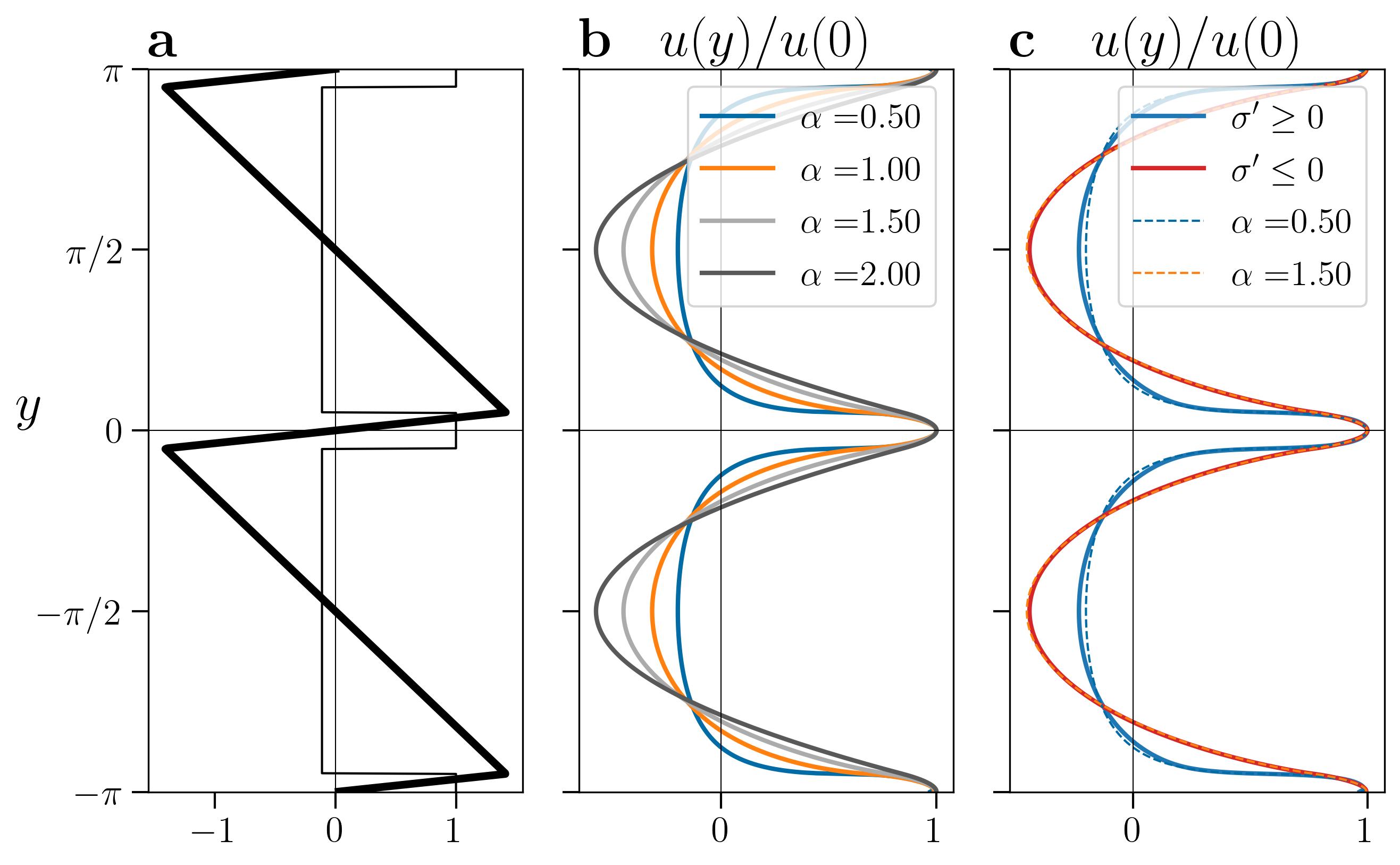}}
  \caption{Panel (a) shows a sloping sawtooth function (thick black line) along with its derivative (thin black line). Panel (b) shows the normalized zonal velocity induced by the sloping sawtooth function in panel (a) for various values of the parameter $\alpha$. Panel (c) shows the normalized zonal velocity induced by the sawtooth function in (a) in the increasing (blue line) and decreasing stratifications (red line) shown in figure \ref{F-mk_sigma}.}
  \label{F-sawtooth}
	\end{figure}
	
	We therefore consider the more general case of a sloping staircase, where there is a finite frontal zone of width $2a$ between the mixed zones. In this case, $\theta$ is a $2(a+b)$-periodic sloping sawtooth wave (see figure \ref{F-sawtooth}),  and is given by the periodic extension of
	\begin{equation}\label{eq:sloping_sawtooth}
		\theta = \Lambda
		\begin{cases}
			- \left[y-{(a+b)} \right] \quad &\text{for} \quad  {a}< y < {a+b} \\
			 \frac{\, b}{a} y \quad &\text{for} \quad \abs{y}\leq {a}\\
			- \left[y+ {(a+b)}\right] \quad &\text{for} \quad - {(a+b)}<y<- {a}.
		\end{cases}
	\end{equation}
	The meridional gradient $\mathrm{d}\theta/\mathrm{d}y$ is then a piecewise constant $2(a+b)$-periodic function
	\begin{equation}
		\d{\theta}{y} = \Lambda
		 \begin{cases}
			-1  \quad &\text{for} \quad  {a}< y < {a+b} \\
			 \frac{\, b}{a} \quad &\text{for} \quad \abs{y}\leq {a}\\
			-1  \quad &\text{for} \quad  -{(a+b)}<y<-{a}.
		\end{cases}
	\end{equation}
 	Therefore the gradient in the frontal zones exceeds the gradient in the mixed zones by a factor of $b/a$, which approaches infinity as $b/a \rightarrow \infty$ in the sawtooth wave limit.
 	
	The zonal velocity, $u=-\partial_y \psi$, is obtained by using the inversion relation \eqref{eq:inversion} to solve for the streamfunction. Alternatively, taking the meridional derivative of surface potential vorticity \eqref{eq:theta_inversion} gives
 	\begin{equation}
		\pd{\theta}{y} = \frac{1}{\sigma_0^2} \pd{u}{z}\Big|_{z=0}.
	\end{equation}
	Then in Fourier space [$\partial_y \rightarrow \mathrm{i} k_y$ and $\sigma_0^{-2}\partial_z|_{z=0} \rightarrow m(k)$] we obtain
	\begin{equation}
		\hat u_{\vec k} = \frac{1}{m(k)} \left( \mathrm{i}\,  k_y \, \hat \theta_{\vec k} \right),
	\end{equation}
	which shows that the induced zonal velocity is obtained by smoothing $\mathrm{d}\theta/\mathrm{d}y$ by the function $m(k)$. An immediate consequence is that the east-west asymmetry in the zonal velocity is fundamentally due to the east-west asymmetry in the gradient $\mathrm{d}\theta/\mathrm{d}y$.

	Figure \ref{F-sawtooth} shows an example of sloping sawtooth $\theta$ profile along with the induced zonal velocities. 
	For a power law inversion function, $m(k) = m_0 k^{\alpha}$, the parameter $\alpha$ modifies the zonal velocity in two ways. 
	First, in more local flows (with smaller $\alpha$), the zonal velocity decays more rapidly away from the jet centre, as expected. 
	Second, the degree of smoothing increases with $\alpha$, and so more local regimes (with smaller $\alpha$) are more east-west asymmetric, with the ratio $\abs{u_\mathrm{min}}/u_\mathrm{max}$ taking smaller values for smaller $\alpha$.
	Figure \ref{F-velocity_sawtooth}(b) shows $\abs{u_\mathrm{min}}/u_\mathrm{max}$ as a function of $a/b$ for $\alpha \in \{1/2,\,1,\,3/2,\,2\}$. For $\alpha = 2$, we obtain $\abs{u_\mathrm{min}}/u_\mathrm{max} \rightarrow 1/2$ in the limit $a/b \rightarrow 0$ so that eastward jets are only twice as strong as westward flows in the perfect staircase limit \citep{danilov_scaling_2004,dritschel_multiple_2008}. At $\alpha = 3/2$, we find $\abs{u_\mathrm{min}}/u_\mathrm{max} \approx 0.29$ in the $a/b \rightarrow 0$ limit so that eastward jets are now more than three time as strong as westward flows. Once $\alpha \leq 1$, then the maximum jet velocity diverges as $\alpha \rightarrow 0$ [figure \ref{F-velocity_sawtooth}(a)] and so $\abs{u_\mathrm{min}}/u_\mathrm{max} \rightarrow 0$ as $a/b \rightarrow 0$.
	
	If $m(k)$ is not a power law, then the results are similar so long as $m(k)$ can be approximated by a power law at small wavenumbers. 
	Figure \ref{F-sawtooth} shows the induced velocity for the inversion functions computed from idealized stratifications profiles (shown in figure \ref{F-mk_sigma}). Because these inversion functions can be approximated by power laws $m(k) \approx k^{0.49}$ and $m(k) \approx k^{1.50}$ at small wavenumbers, the induced velocity fields nearly coincide with the velocity fields computed from power law inversion functions with $\alpha =0.5$ and $\alpha = 1.5$.
	
	\begin{figure}
  \centerline{\includegraphics[width=0.75\columnwidth]{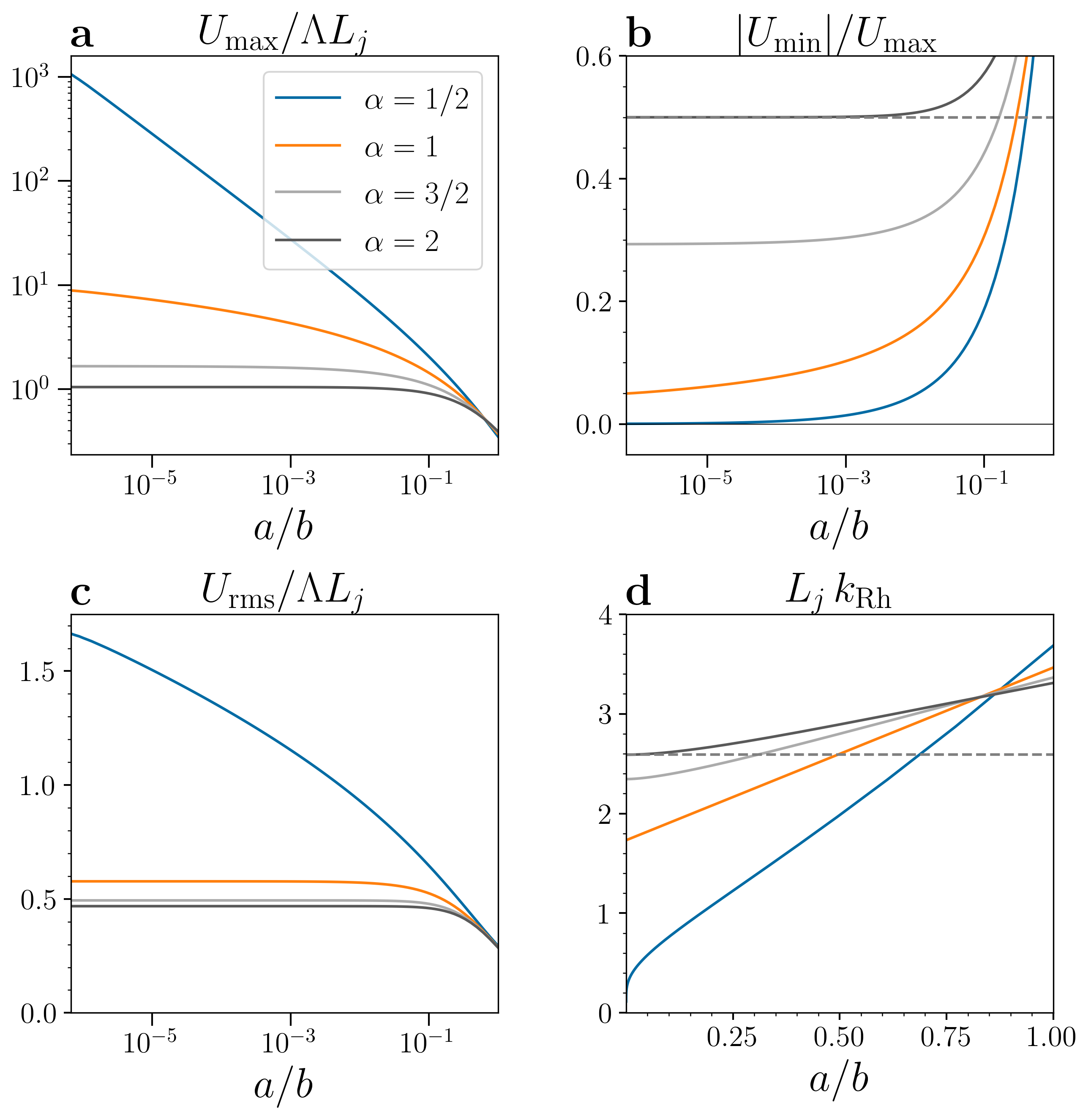}}
  \caption{Properties of zonal velocity profiles induced by sloping sawtooth profiles \eqref{eq:sloping_sawtooth} of $\theta$ as a function of the non-dimensional frontal zone width $a/b$ separating the mixed zones for four values of $\alpha$. Panel (a) shows the maximum zonal velocity, panel (b) shows the ratio of westward speed to eastward speed, panel (c) shows the rms zonal velocity, and panel (d) shows the product $L_j \krh$ where $L_j=a+b$ is the halfwidth separation (the distance between $U_\mathrm{mix}$ and $U_\mathrm{min}$) and $\krh$ is the Rhines wavenumber \eqref{eq:rhines}. }
  \label{F-velocity_sawtooth}
\end{figure}
	 	
 	Finally, we examine how the Rhines wavenumber, $\krh$, relates to jet spacing. Let
 	\begin{equation}
 			 L_j=a+b 
 	\end{equation}
	 be the half-separation between the jets, i.e., the half distance between consecutive zonal velocity maxima.
	For two-dimensional barotropic turbulence (i.e., the $\alpha = 2$ case), we have $L_j = 45^{1/4}/\krh \approx 2.59/\krh$
 	in the staircase limit \citep[i.e, for $a/b \rightarrow 0$,][]{dritschel_multiple_2008,scott_structure_2012}. This result is found by solving for the zonal velocity induced by a staircase with halfwidth $L_j=b$, taking the rms of the zonal velocity, and then substituting into the definition of the generalized Rhines wavenumber \eqref{eq:rhines}. As figure \ref{F-velocity_sawtooth}(d) shows, because the velocity field induced by a perfect staircase depends on the inversion function, $m(k)$, the relationship between $L_j$ and $\krh$ also depends on the inversion function. For $m(k) = k^{3/2}$, an analogous calculation gives $L_j  \approx 2.35/\krh$ in the staircase limit. For $\alpha = 1$, even though the maximum velocity diverges at $a/b \rightarrow 0$, the rms velocity asymptotes to a constant value, and so we obtain a half jet-separation of $L_j  \approx 1.73/\krh$ (figure \ref{F-velocity_sawtooth}). Finally in the $\alpha = 1/2$ case, although the rms speed has not converged by $a/b = 10^{-6}$, the product $L_j \, \krh$ is approaching values close to zero.
 		 
\section{Numerical Simulations}\label{S-numerical}
	
	\subsection{The numerical model}
	
	We use the \texttt{pyqg} pseudo-spectral model \citep{abernathey_pyqgpyqg_2019} which solves the time-evolution equation \eqref{eq:time-evolution} in a square domain with side length $L=2\pi$.  Time-stepping is through a third-order Adam-Bashforth scheme with small-scale dissipation achieved through a scale-selective exponential filter \citep{smith_turbulent_2002,arbic_coherent_2003},
	\begin{equation}
		\mathrm{ssd} =
		 \begin{cases}
			1 \quad \text{for} \quad  k \leq k_0 \\
			e^{-a(k-k_0)^4} \quad  \text{for} \quad k > k_0,
		\end{cases}
	\end{equation}
	with $a = 23.6$ and $k_0 = 0.65 k_\mathrm{Nyq}$ where $k_\mathrm{Nyq}=\pi$ is the Nyquist wavenumber. The forcing is isotropic, centred at wavenumber $k_f = 80$, and normalized so that the energy injection rate is $\varepsilon = 1$ \citep[see appendix B in][]{smith_turbulent_2002}. However, the effective energy injection rate, $\varepsilon_\mathrm{eff}$, is smaller than $\varepsilon$ due to dissipation.  To determine $\varepsilon_\mathrm{eff}$ from numerical simulations, we use $ \varepsilon_\mathrm{eff}= 2\, r \, \mathcal{E}$ where $\mathcal{E}$ is the equilibrated total energy diagnosed from the model. In what follows, we report values of $\keps/k_r$ using $\varepsilon_\mathrm{eff}$ instead of $\varepsilon$. The model is integrated forward in time until at least  $t=5/r$ to allow the fluid to reach equilibrium. All model runs use $1024^2$ horizontal grid points. 
	
	\subsection{For what values of $\keps/\kr$ do jets form?}\label{SS-sims}
	
	  \begin{figure}
  \centerline{\includegraphics[width=1\columnwidth]{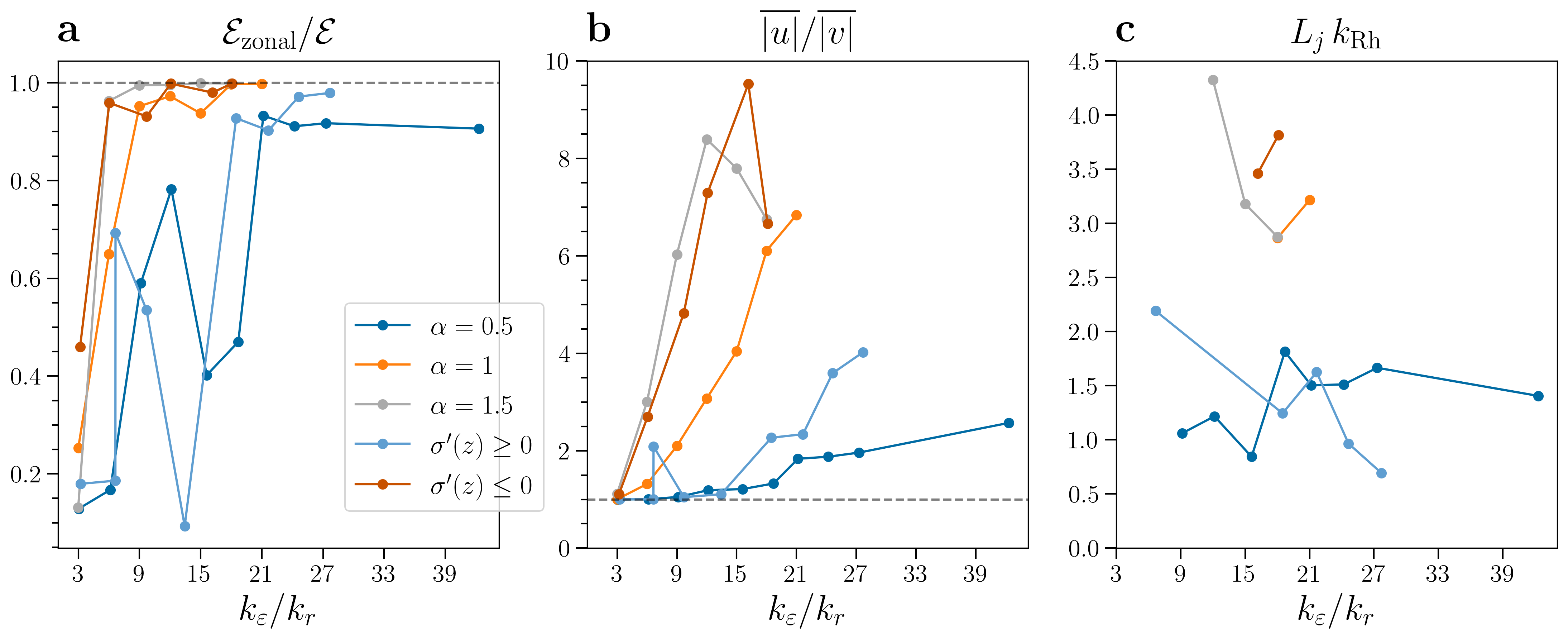}}
  \caption{Diagnostics from five series of simulations as a function of the non-dimensional number $\keps/k_r$. The first three series of simulations have inversion function $m(k)=k^\alpha$ with $\alpha \in \{\nicefrac{1}{2},\, 1,\, \nicefrac{3}{2}\}$. For the other two series, the inversion functions are shown in figure \ref{F-mk_sigma_alt}. Panel (a) shows the ratio of energy in the zonal mode to total energy. Panel (b) shows the ratio of domain averaged zonal speed to domain averaged meridional speed. Panel (c) shows the ratio of westward zonal speed to eastward zonal speed. Panel (d) shows the relationship between the halfwidth jet spacing, $L_j$, and the Rhines wavenumber, $\krh$.}
  \label{F-krh_velocity}
   \end{figure}
	
	For our first set of simulations, we vary $\keps/k_r$ over the values shown in figure \ref{F-krh_velocity}. We do so by fixing $k_r=8$ and varying $\keps$. For a given value of $\keps$, we choose $\Lambda$ and $r$ so as to maintain $\kr = 8$ (the energy injection rate, $\varepsilon$, is fixed at unity for all model runs). Given $\kr$ and $\keps$, we rearrange the definition of $k_r$ \eqref{eq:kr} to solve for $\gamma= r \, \Lambda^2$,
	\begin{equation}
		\gamma = m_0 \, \varepsilon \, k_r^{\alpha+2},
	\end{equation}
	then solve for $r$ in the implicit equation \eqref{eq:keps_condition} for $\keps$ ,
	\begin{equation}
		 r = \frac{\gamma}{\left( \varepsilon \, \keps \, [m(\keps)]^2 \right)^{2/3}},
	\end{equation}
	and finally use the definition $\gamma = r\, \Lambda^2$ to solve for $\Lambda$.

	 
	\subsubsection{Power law inversion functions}
	We first describe the results from three series of simulations with power law inversion functions, $m(k)=k^\alpha$, with $\alpha \in \{\nicefrac{1}{2},\, 1, \, \nicefrac{3}{2} \}$.
	Summary diagnostics from these simulations are shown in figure \ref{F-krh_velocity}.	
	In panel (a), we observe that the ratio of energy in the zonal mode to total energy, $\mathcal{E}_\mathrm{zonal}/\mathcal{E}$, increases with $\keps/k_r$, and that the majority of the total energy is in the zonal mode for sufficiently large $\keps/k_r$. For a fixed $\keps/k_r$, more of the total energy is zonal in more non-local flows (with larger $\alpha$) than in more local flows (with smaller $\alpha$); for $\alpha = \nicefrac{3}{2}$, we have $\mathcal{E}_\mathrm{zonal}/\mathcal{E} \approx 1$ by $\keps/k_r \approx 6$ as compared to $\keps/k_r \approx 12$ for $\alpha = 1$. Moreover, for $\alpha = \nicefrac{1}{2}$, we find that $\mathcal{E}_\mathrm{zonal}/\mathcal{E}$ asymptotes to approximately 0.9 once $\keps/k_r \approx 18$ with little subsequent change for larger values of $\keps/k_r$. In panel (b), we observe a striking contrast in the ratio $\overline{\abs{u}}/\overline{\abs{v}}$ between different values of $\alpha$ (the overline denotes a domain average). For $\alpha = \nicefrac{3}{2}$, the domain averaged zonal speed, $\overline{\abs{u}}$, is approximately eight times larger than the domain averaged  meridional speed, $\overline{\abs{v}}$, for large $\keps/k_r$.  In contrast, for $\alpha = \nicefrac{1}{2}$, $\overline{\abs{u}}$ only exceeds  $\overline{\abs{v}}$ by a multiple of two for large $\keps/k_r$.
		
	\begin{figure}
  \centerline{\includegraphics[width=\textwidth]{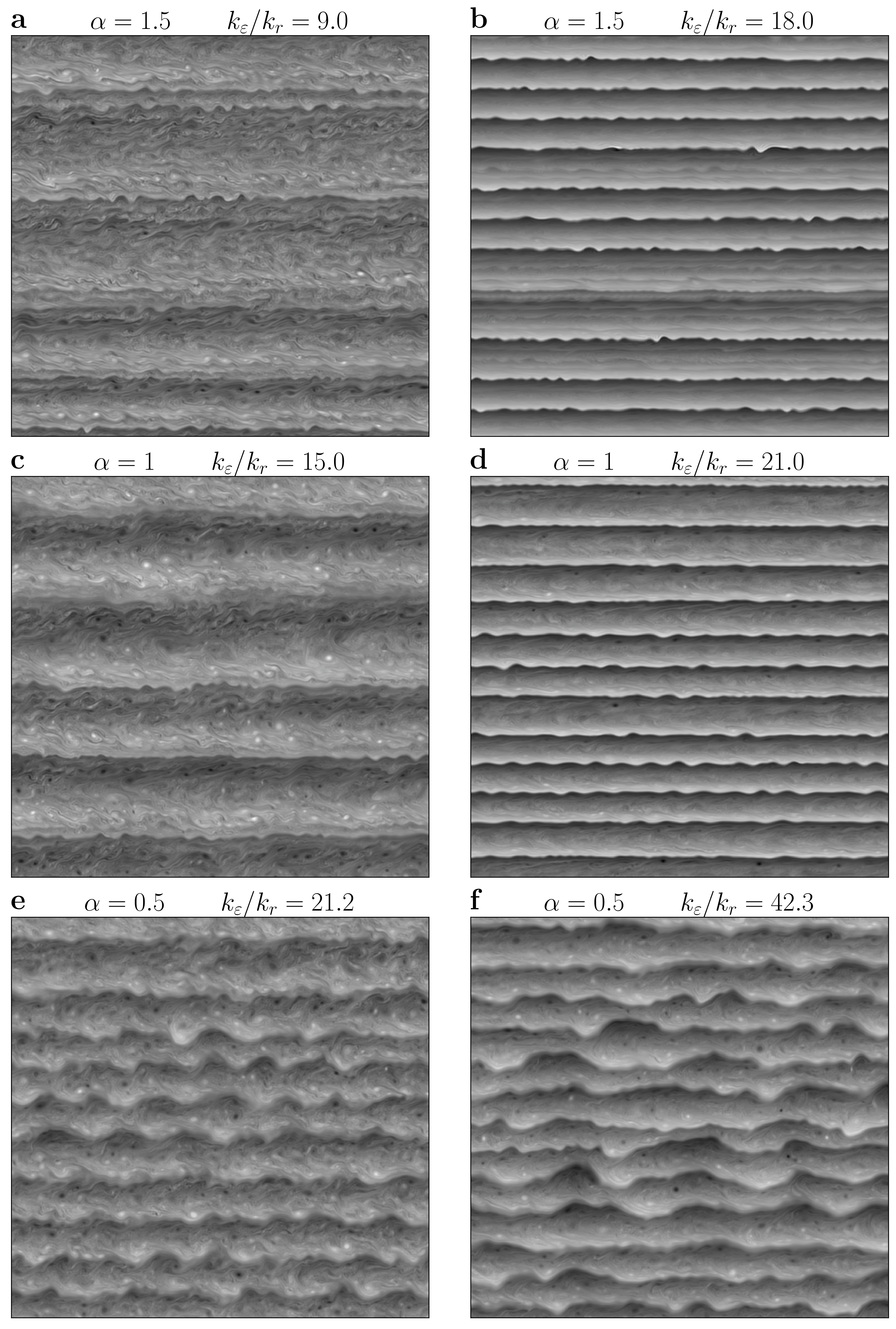}}
  \caption{Snapshots of the relative surface potential vorticity, $\theta$, for simulations with power law inversion functions, $m(k)=k^\alpha$. In each snapshot, the $\theta$ field is normalized by its maximum value in the snapshot. Only one quarter of the domain is shown (i.e., $512^2$ grid points). }
  \label{F-alpha_jets}
   \end{figure}
   
   	\begin{figure}
  \centerline{\includegraphics[width=1\columnwidth]{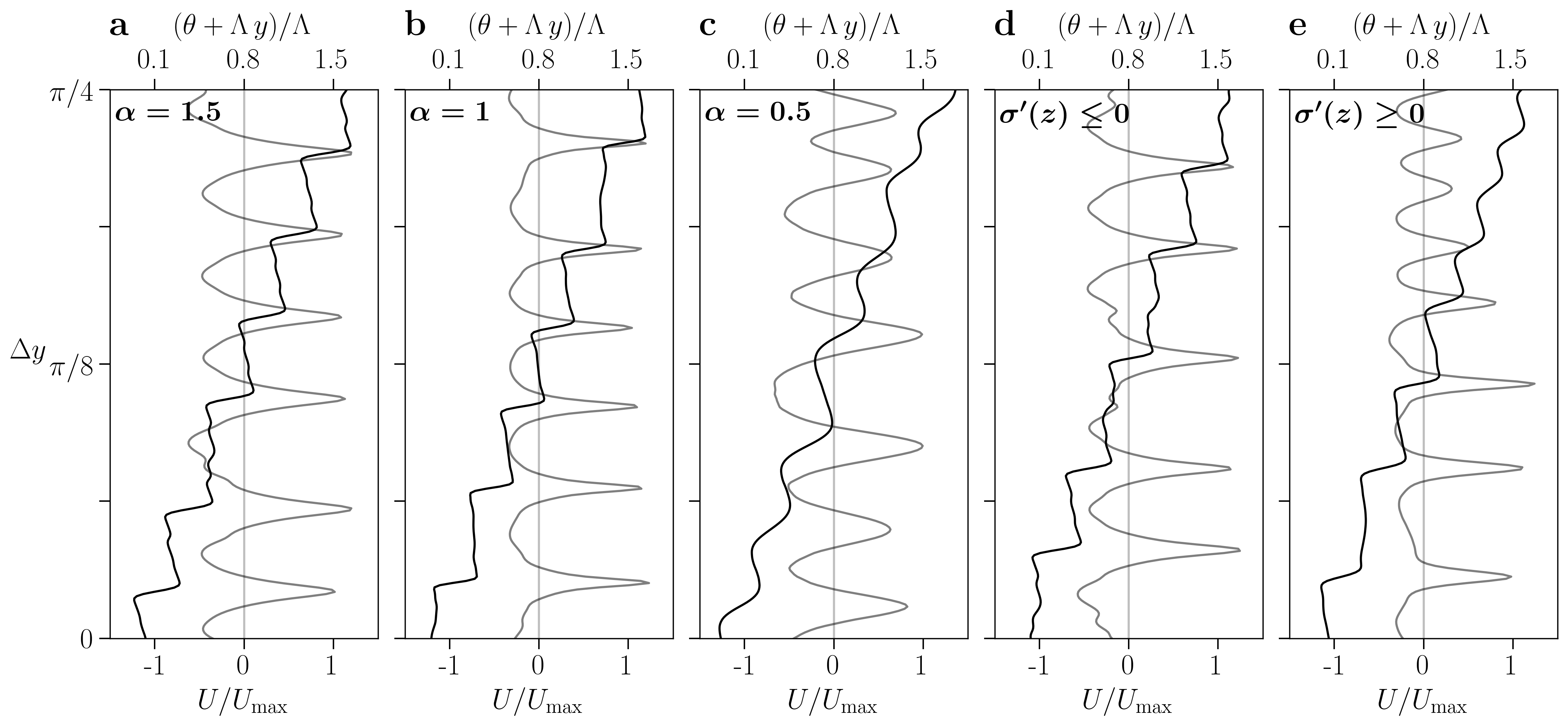}}
  \caption{The zonal mean total surface potential vorticity, $\theta + \Lambda \, y$, in black and the zonal mean zonal velocity, $U$, in grey. }
  \label{F-model_staircase}
	\end{figure}
 
	Next, we examine the jet structure for different $\alpha$ as a function of $\keps/k_r$. Figure \ref{F-alpha_jets} shows $\theta$-snapshots from model runs with $m(k)=k^\alpha$. For each value of $\alpha$, two model runs are shown: one where jets have just become visible in the $\theta$-snapshot and another with the largest value of $\keps/k_r$, which we expect to be  closest to the staircase limit. The jets are visible in these snapshots as the regions with strong gradients. Because these are $\theta$-snapshots rather than $(\theta+\Lambda\,y)$-snapshots, the $(\theta + \Lambda \, y)$-staircase is instead a $\theta$-sawtooth, and the mixed zones between the jets are approximately linear in $\theta$. 
	 We confirm this to be the case in figure \ref{F-model_staircase}, where the zonal averages of the total surface potential vorticity, $\theta+\Lambda \, y$, and the zonal velocity are shown. For the  $\alpha = \nicefrac{3}{2}$ and $\alpha =1$ cases, we observe an approximate staircase structure with nearly uniform mixed zones separated by frontal zones, and with jets centred at sharp $\theta$ gradients. As expected from the idealized staircases of section \ref{S-staircase_theory}, close to the staircase limit, the $\alpha=1$ jets are narrower than the $\alpha=\nicefrac{3}{2}$ jets, and the ratio of maximum westward speed to maximum eastward speed, $|U_\mathrm{min}|/|U_\mathrm{max}|$, is smaller at $\alpha=1$ than at $\alpha=\nicefrac{3}{2}$.
	 

	 In contrast to the $\alpha = \nicefrac{3}{2}$ and the $\alpha = 1$ series, the $\alpha=\nicefrac{1}{2}$ series approaches the staircase limit slowly with $\keps/k_r$. 
	The $\alpha=\nicefrac{1}{2}$ staircase remains smooth even at $\keps/k_r = 42$ [figure \ref{F-model_staircase}(c)]. The ratio of frontal zone width to mixed zone width, $a/b$, is between $0.5$ and  $0.65$ for $\alpha=\nicefrac{1}{2}$ jets. In contrast, this ratio is between $0.15$ and $0.2$ for the $\alpha =\nicefrac{3}{2}$ and $\alpha = 1$ jets.
	 In part, the broadness of the $\alpha=\nicefrac{1}{2}$ frontal zones is a consequence of zonal averaging in the presence of large amplitude undulations. However, it is evident from the $\theta$-snapshots of figure \ref{F-alpha_jets} that the $\alpha=\nicefrac{1}{2}$ frontal zones are indeed broader than the $\alpha=\nicefrac{3}{2}$ and $\alpha=1$ frontal zones [e.g., compare panels (a) and (d) with (f) in figure \ref{F-alpha_jets}], even without zonal averaging. 
	 
	
	 We now examine how the generalized Rhines wavenumber, $\krh$, relates to the jet spacing.
	 From figure \ref{F-velocity_sawtooth}(d), a ratio of $a/b \approx 0.2$ leads to a $L_j \, \krh \approx 2.2$ for $\alpha=\nicefrac{3}{2}$ and $L_j \, \krh \approx 2.0$ for $\alpha=1$. But as figure \ref{F-krh_velocity}(d) shows, we find values closer to $L_j \, \krh \approx 3$ for both of these cases. In contrast, for the $\alpha =\nicefrac{1}{2}$ jets, figure \ref{F-velocity_sawtooth}(d) predicts $1.98 \lesssim L_j \, \krh \lesssim 2.5$ for the observed range of $0.5 \lesssim a/b \lesssim 0.65$, but we find $L_j \, \krh \approx 1.5$ for $\keps/ k_r \geq 18$, which is smaller than predicted.  
	 
	 	 
	Returning to figure \ref{F-alpha_jets}, we observe that there are undulations along the jets, with smaller values of $\alpha$ corresponding to larger amplitude undulations. 
	 These undulations propagate as waves and are less dispersive for smaller $\alpha$, propagating eastward for $\alpha=\frac{3}{2}$, westward for $\alpha=\nicefrac{1}{2}$,  and are nearly stationary for $\alpha = 1$. Moreover, the waves in the $\alpha=\nicefrac{1}{2}$ case maintain their shape as they propagate for a significant fraction of the domain, although they eventually disperse or merge with other along jet waves. 
	 That we obtain larger amplitude along jet undulations for smaller $\alpha$ is a consequence of the more local inversion operator \eqref{eq:inversion} at smaller $\alpha$. A jet in a highly local flow (with small $\alpha$) is ``a coherent structure that hangs together strongly while being easy to push sideways'' \citep[][in the context of equivalent barotropic jets]{mcintyre_potential-vorticity_2008}. However, although both an equivalent barotropic jet and an $\alpha=\nicefrac{1}{2}$ jet exhibit large meridional undulations, the undulations in the equivalent barotropic case are frozen in place \citep[because of a vanishing group velocity at large scales, ][]{mcintyre_potential-vorticity_2008} and so the equivalent barotropic jet behaves like a meandering river with a fixed shape. In contrast, the $\alpha=\nicefrac{1}{2}$ jet behaves like a flexible string whose shape evolves in time with the propagation of weakly dispersive waves. Another difference between the two cases is that an equivalent barotropic jet has a width given by the deformation radius. In contrast, there is no analogous characteristic scale for $\alpha=\nicefrac{1}{2}$ jets and, in principle, the jets should become infinitely thin as $\keps/k_r \rightarrow \infty$.
	 	 	
	\begin{figure}
  \centerline{\includegraphics[width=.9\columnwidth]{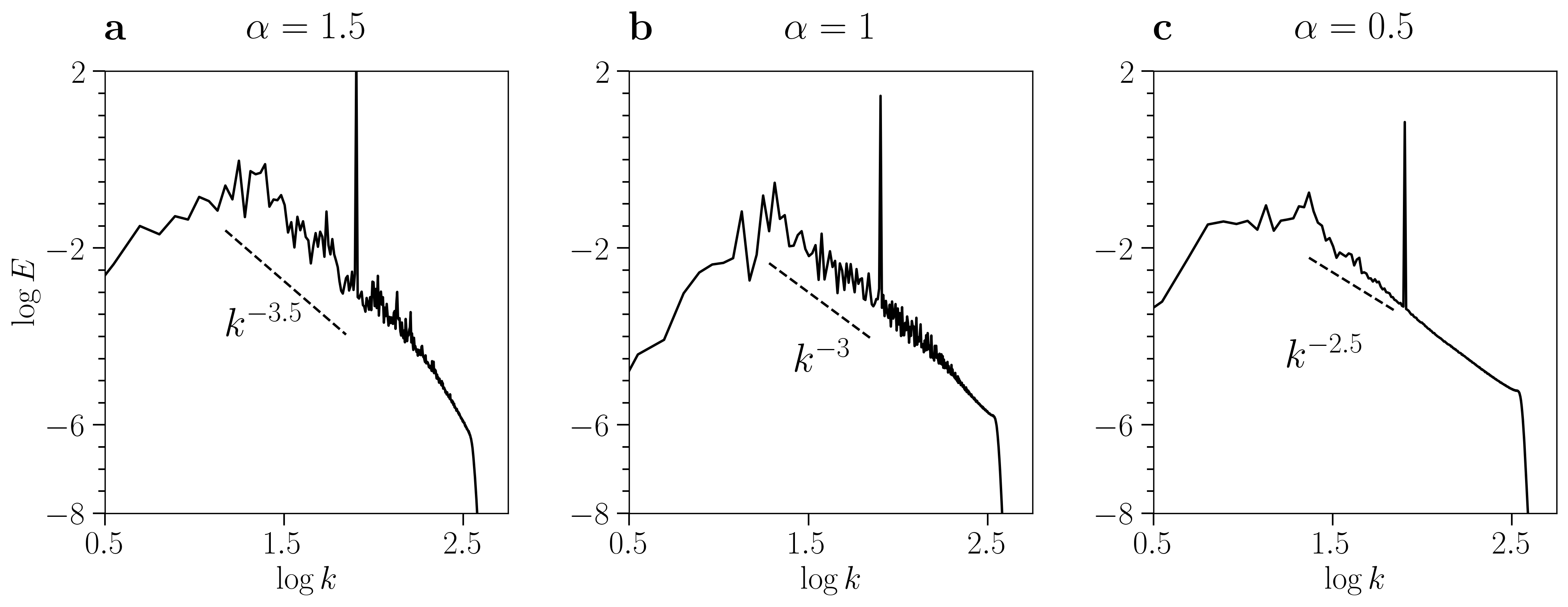}}
  \caption{The total energy spectrum, $E(k)$, as a function of the wavenumber, $k=k_x^2+k_x^2$, for three simulations with power law inversion functions, $m(k)=k^\alpha$. The values of $\keps/k_r$ are 18.0 for panel (a), 21.0 for panel (b), and 42.3 for panel (c).}
  \label{F-energy_spectrum}
	\end{figure}

	Energy spectra for the three power law simulations are shown in figure \ref{F-energy_spectrum}. The energy spectrum obtained from dimensional analysis  \eqref{eq:Espectrum} gives a $k^{-\alpha - 3}$ wavenumber dependence, which leads to the familiar $k^{-5}$ spectrum for beta-plane barotropic turbulence ($\alpha=2$). Although early investigations \citep{chekhlov_effect_1996,huang_anisotropic_2000,danilov_barotropic_2004} found a $k^{-5}$ spectrum in barotropic $\beta$-plane turbulence, \cite{scott_structure_2012} instead found a shallower $k^{-4}$ spectrum in the staircase limit \citep[suggested earlier by][]{danilov_barotropic_2004,danilov_scaling_2004}, which they explained as a consequence of the sharp discontinuities of the staircase. Generalizing their argument to the present case, a one dimensional $\theta(y)$ series with discontinuities implies a Fourier series with coefficients decaying as $k^{-1}$, leading to a $\theta^2$ spectrum of $k^{-2}$, and hence an energy spectrum 
	\begin{equation}
		E(k) \sim k^{-2}\left[m(k)\right]^{-1}. 
	\end{equation}
	If $m(k) \sim k^{\alpha}$, then we obtain a spectrum $E(k) \sim k^{-\alpha - 2}$, which yields the $k^{-4}$ spectrum observed in \cite{scott_structure_2012}, where $\alpha=2$. For $\alpha=\nicefrac{3}{2}$, $\alpha=1$, and $\alpha=\nicefrac{1}{2}$, the predicted spectrum is proportional to $k^{-3.5}$, $k^{-3}$, and $k^{-2.5}$, respectively.
	The diagnosed spectra shown in figure \ref{F-energy_spectrum} are consistent with these shallow spectra, instead of energy spectrum  \eqref{eq:Espectrum} obtained from dimensional considerations.

	\subsubsection{Inversion functions from $\sigma(z)$}
	
	\begin{figure}
  \centerline{\includegraphics[width=0.9\textwidth]{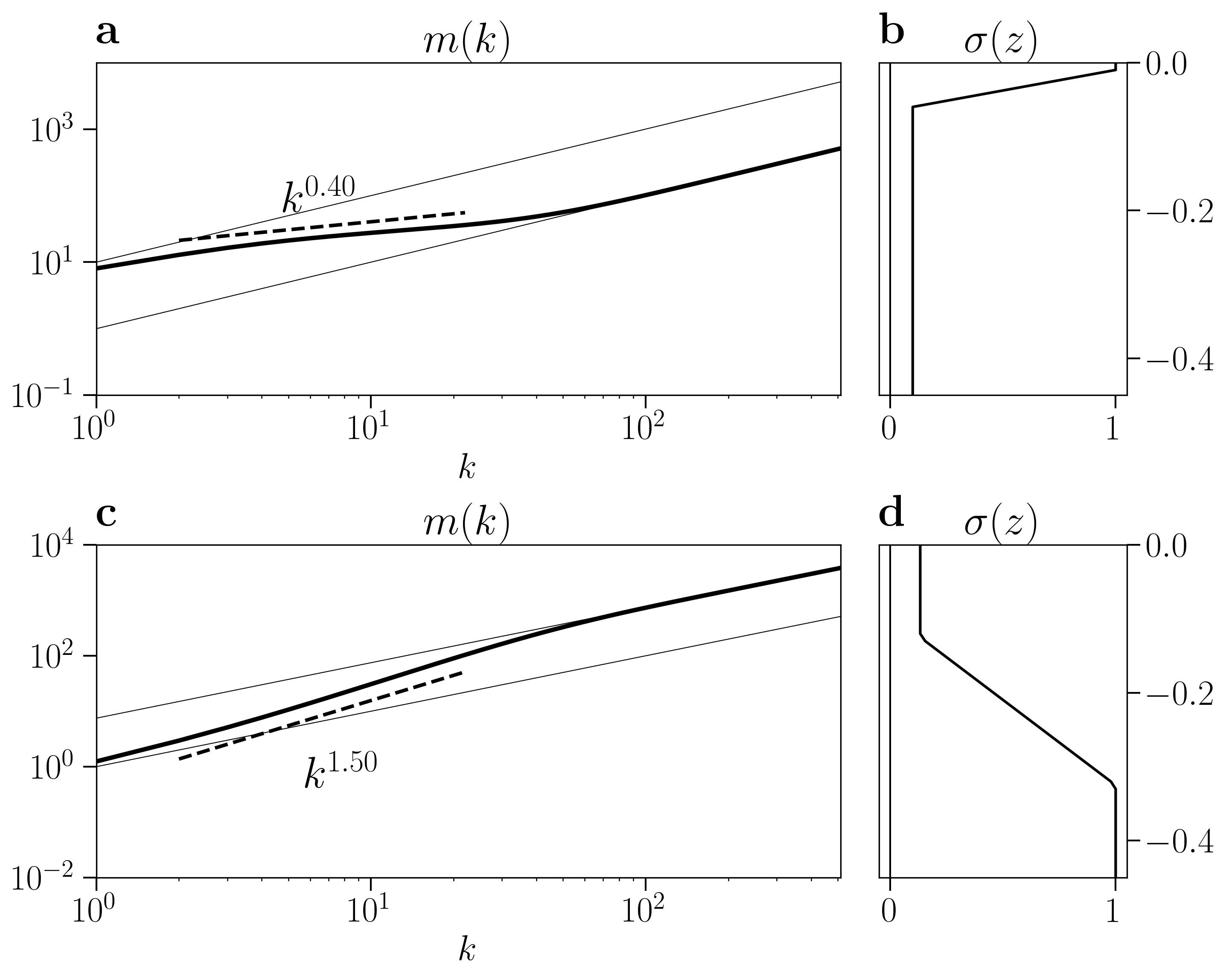}}
  \caption{Inversion functions  [panel (a) and (c)] along with their corresponding stratification profiles [panels (b) and (d), respectively]. The stratification profiles are given by the piecewise function \eqref{eq:sigma_piece}. For panel (a), we have $\sigma_0 = 1$, $\sigma_\mathrm{pyc}=0.1$, $h_\mathrm{mix}=0.01$, and $h_\mathrm{lin}=0.05$. For panel (c), we have $\sigma_0 = 0.133$, $\sigma_\mathrm{pyc}=1$, $h_\mathrm{mix}= 0.125$, and $h_\mathrm{lin}=0.2$. The thin grey lines in panels (a) and (c) are given by $k/\sigma_0$ and $k/\sigma_\mathrm{pyc}$.   }
  \label{F-mk_sigma_alt}
   \end{figure}
   
   We also ran two series of simulations where we specified a piecewise stratification profile \eqref{eq:sigma_piece}, and then obtained $m(k)$ by solving the boundary value problem \eqref{eq:vertical_structure}\textendash \eqref{eq:lower} at each wavenumber. The stratification profiles and the resulting inversion functions are shown in figure \ref{F-mk_sigma_alt}.
	 One case consists of an increasing stratification  profile [$\sigma^\prime(z) \geq 0$] with $\sigma_0 = 1$, $\sigma_\mathrm{pyc} = 0.1$, $h_\mathrm{mix} =0.01$ and $h_\mathrm{lin}=0.05$. The resulting $m(k)$ is approximately linear for $k \gtrsim 70$ and transitions to an approximate sub-linear wavenumber dependence $m(k) \sim k^{0.40}$ for wavenumbers $5 \lesssim k \lesssim 50$. The second case consists of a decreasing stratification profile [$\sigma^\prime(z) \leq 0$] with $\sigma_0 = 0.13$, $\sigma_\mathrm{pyc} = 1$, $h_\mathrm{mix} =0.125$ and $h_\mathrm{lin}=0.2$. The resulting $m(k)$ is approximately linear at wavenumbers $k \gtrsim 60$ and transitions to an approximate super linear wavenumber dependence $m(k) \sim k^{1.50}$ between $3 \lesssim k \lesssim 60$.
	 
	As seen in figure \ref{F-krh_velocity}, the $\sigma'(z) \leq 0$ case is similar to the $\alpha=\nicefrac{3}{2}$ case, with the various diagnostics close to the $\alpha=\nicefrac{3}{2}$ counterpart. In contrast, there are significant differences between the $\sigma'(z) \geq 0$ simulations and the $\alpha=\nicefrac{1}{2}$ simulations.
	 In the  $\sigma'(z) \geq 0$ series, the ratio of energy in the zonal mode to total energy continues to increase as $\keps/k_r$ is increased, whereas it asymptotes to a constant in the $\alpha=\nicefrac{1}{2}$ series. 
	Moreover, the ratio of domain average zonal speed to domain averaged meridional speed, $\overline{|u|}/\overline{|v|}$, is generally larger in the $\sigma' \geq 0$ series than in the $\alpha=\nicefrac{1}{2}$ series.  Finally, for the largest values of $\keps/k_r$, the product $L_j\, \krh$ reaches smaller values in the $\sigma' \geq 0$ simulations than in the $\alpha=\nicefrac{1}{2}$ simulations.

	\begin{figure}
  \centerline{\includegraphics[width=\textwidth]{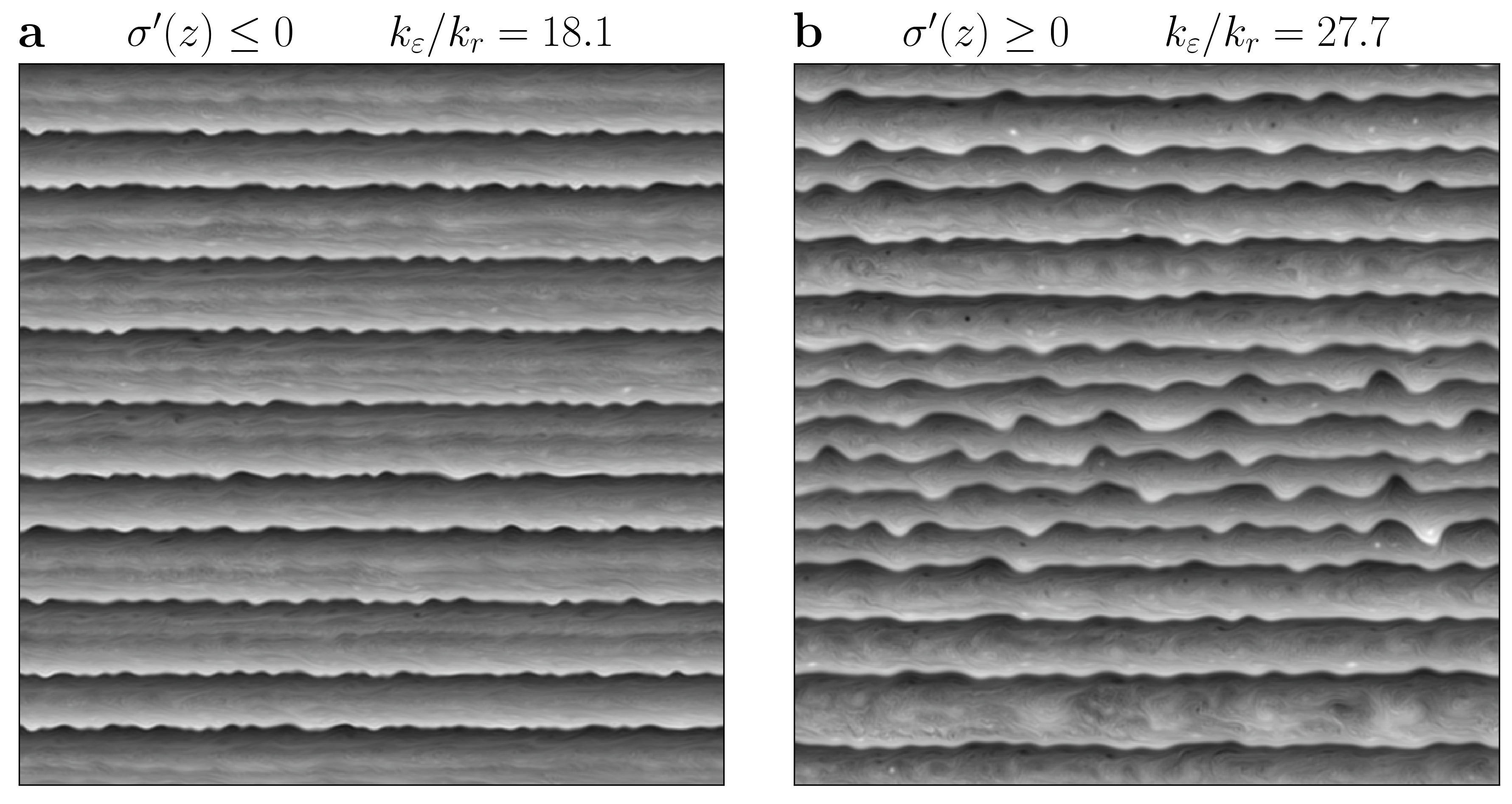}}
  \caption{Snapshots of the relative surface potential vorticity, $\theta$, normalized by its maximum value in the snapshot,  for simulations with inversion functions shown in figure \ref{F-mk_sigma_alt}.  Only one quarter of the domain is shown (i.e., $512^2$ grid points). }
  \label{F-alpha_jets_mixed}
   \end{figure}
   
   These differences can be explained by the snapshots of figure \ref{F-alpha_jets_mixed} as well as the zonal averages of figure \ref{F-model_staircase}. As expected from the model diagnostics, both the snapshots and the zonal average from the $\sigma^\prime \leq 0$  simulation are qualitatively similar to the $\alpha=\nicefrac{3}{2}$ simulation. In contrast, the $\sigma^\prime \geq 0$ snapshot is evidently closer to the staircase limit than the $\alpha=\nicefrac{1}{2}$ snapshot: the mixed zones are more homogeneous and the frontal zones are sharper. The zonal average of the $\sigma^\prime \geq 0$ simulation in figure \ref{F-model_staircase} also shows how the $\sigma^\prime \geq 0$ simulation is closer to the staircase limit than the $\alpha=\nicefrac{1}{2}$ simulation, although, again, zonal averaging in the presence of large amplitude undulations is artificially smoothing the jets. Therefore, the differences in the diagnostics between the $\sigma^\prime \geq 0$ series and the $\alpha=\nicefrac{1}{2}$ series stem from the more rapid approach (i.e., at smaller $\keps/k_r$) of the $\sigma^\prime \geq 0$ series to the staircase limit.

   \subsection{Simulations with fixed parameters} 
   
   The dependence of the non-dimensional number $\keps/k_r$ on the external parameters $\varepsilon, \Lambda$, and $r$ depends on the functional form of $m(k)$. For example, if $m(k) \sim k^\alpha$, then 
   \begin{equation}\label{eq:J_power}
   	 \keps/{k_r} = |\Lambda|^\frac{4-\alpha}{(2\alpha+1)(\alpha+2)}\, \varepsilon^\frac{\alpha-1}{(1+2\alpha)(\alpha+2)} \, r^\frac{-1}{\alpha+2}.
   \end{equation}
	Because the forcing intensity wavenumber, $\keps$, is obtained by solving the implicit equation for $\keps$ \eqref{eq:keps_condition}, an analogous expression for $\keps/k_r$ is not possible for general $m(k)$. However, at sufficiently large $\keps$, the inversion function asymptotes to $m(\keps) \approx \keps/\sigma_0$ and so, using $\alpha$-turbulence expression for $\keps$ \eqref{eq:keps_alpha} with $\alpha=1$, we obtain
	\begin{equation}\label{eq:J_strat}
		\keps/k_r \approx  |\Lambda|^\frac{\alpha}{\alpha+2} \, \varepsilon^\frac{1-\alpha}{3\alpha + 6} \, r^\frac{-1}{\alpha+2}\, m_0^\frac{1}{\alpha+2} \, \sigma_0^{2/3}
	\end{equation}
	for large $\keps$, where $\alpha$ is the approximate power law dependence of $m(k)$ near $k_r$. 
	
	Therefore, simulations with identical $\keps/k_r$ but distinct inversion functions cannot be directly compared because they have different values of $\Lambda$ and $r$. 
	Here, we investigate how the stratification modifies jet structure as all other parameters are held fixed.
	We therefore run two additional series of simulations with the stratification profiles and inversion functions shown in figure \ref{F-mk_sigma}. 
	The stratification profiles were chosen so that they both have identical stratification at the upper boundary. One case corresponds to an increasing stratification profile, $\sigma^\prime \geq 0$, with an approximate power law dependence of $m(k) \sim k^{0.49}$ at small wavenumbers.  The second case consists of a decreasing stratification profile, $\sigma^\prime \leq 0$, with a $m(k) \sim k^{1.50}$ at small wavenumbers.
	Aside from the different stratification profiles, these two series of simulations are run under the same conditions as the constant stratification ($\alpha=1$) simulations of section \ref{SS-sims}, with identical values of $\Lambda$, $\varepsilon$, and $r$.
	
	\begin{figure}
  \centerline{\includegraphics[width=1\columnwidth]{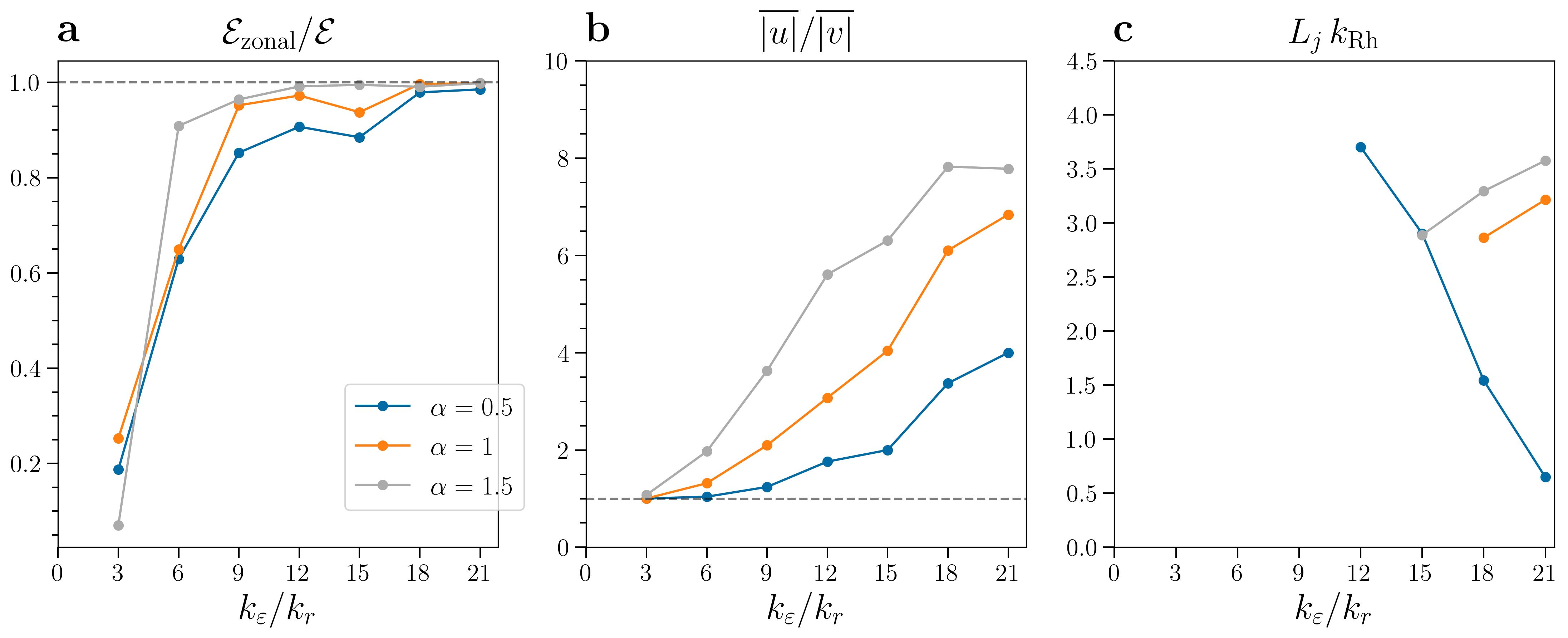}}
  \caption{As in figure \ref{F-krh_velocity}, but the $\sigma'\geq 0$ and $\sigma'\leq 0$ series now only differ from the $\sigma'=0$ (i.e., $\alpha=1$) series only in the vertical stratification (and hence the inversion function).}
  \label{F-krh_velocity_fixed}
   \end{figure}
	
	Summary diagnostics are shown in figure \ref{F-krh_velocity_fixed}. We see that, at a fixed value of $\Lambda$ and $r$, more of the total energy is in the zonal mode in the $\sigma'(z) \leq 0$ simulation than in the constant stratification simulation, which in turn is larger than the $\sigma'(z) \geq 0$ simulation (and similarly for the ratio of area averaged zonal to meridional speeds, $\overline{\abs{u}}/\overline{\abs{v}}$). 
	Therefore, increased non-locality (larger  $\alpha$) promotes anisotropy in the velocity field and leads to larger zonal velocities relative to meridional velocities. Indeed, figure \ref{F-alpha_jets_fixed} shows $\theta$ snapshots from these simulations; the more local, $\sigma'\geq 0$, simulations have larger meridional undulations along the jets.
	Moreover, compared to the $\keps/k_r= 15$ constant stratification simulation in figure \ref{F-alpha_jets}(c), the $\sigma'(z) \leq 0$ simulation in figure \ref{F-alpha_jets_fixed}(a) is closer to the staircase limit whereas the frontal zones in the $\sigma'(z) \geq 0$ simulation [figure \ref{F-alpha_jets_fixed}(c)] remain broad. Finally, we show values of the product $L_j \, \krh$, relating the Rhines wavenumber to the half spacing between the jets,  in figure \ref{F-krh_velocity_fixed}(c). These values are similar to those in shown in figure \ref{F-krh_velocity}(c).
	 
   \begin{figure}
  \centerline{\includegraphics[width=\textwidth]{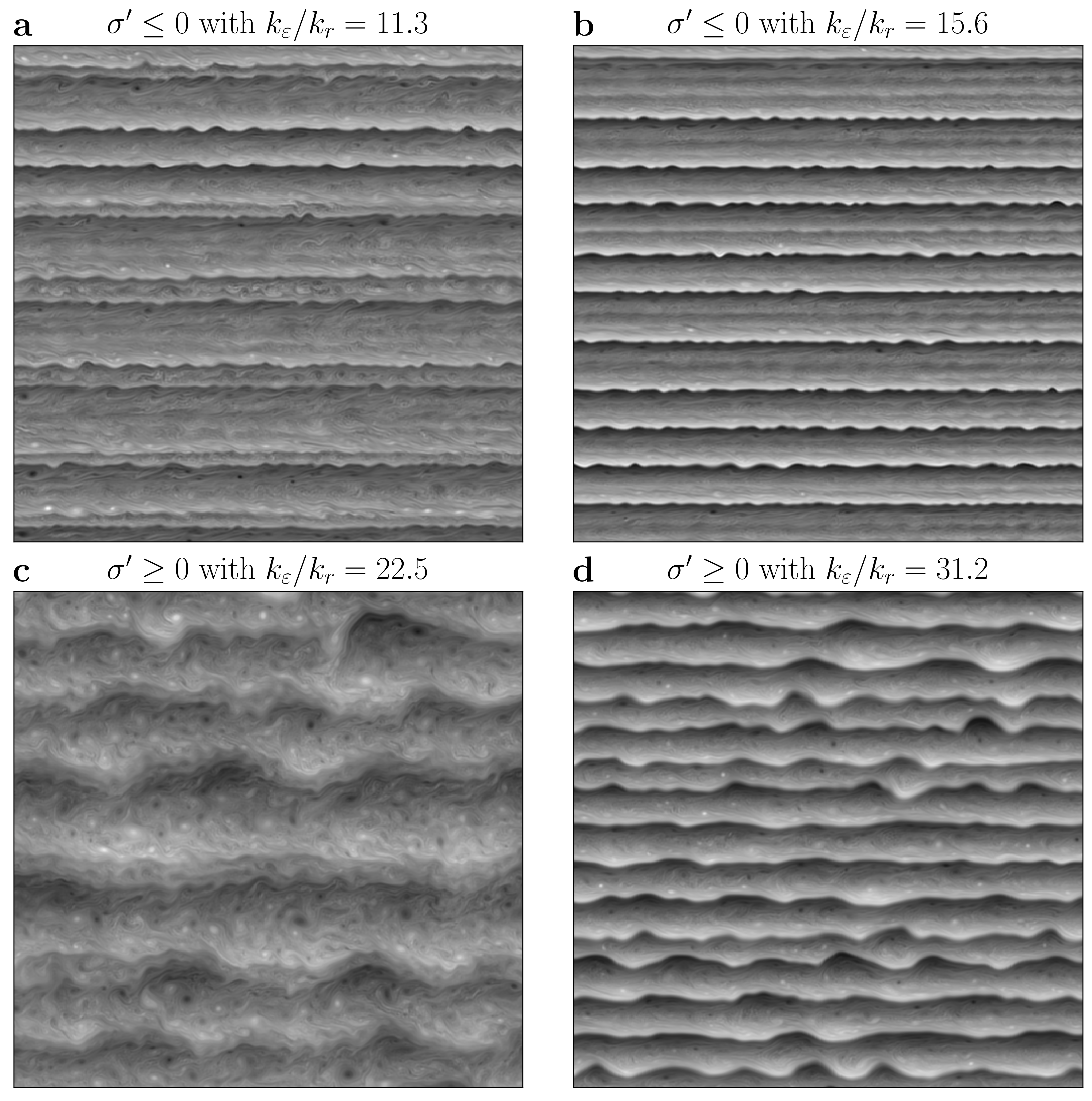}}
  \caption{Snapshots of relative surface potential vorticity, $\theta$, where $\theta$ is normalized by its maximum value in the snapshot. Panels (a) and (c) are from simulations with identical $\Lambda$, $r$, and $\varepsilon$ as the $\alpha=1$ simulation shown in figure \ref{F-alpha_jets}(c), whereas panels (b) and (d) are from simulations with identical $\Lambda$, $r$, and $\varepsilon$ as the $\alpha=1$ simulation shown in figure \ref{F-alpha_jets}(d). Only one quarter of the domain is shown (i.e., $512^2$ grid points).}
  \label{F-alpha_jets_fixed}
   \end{figure}

\section{Conclusion} \label{S-conclusion_jets}

We have examined the emergence of staircase-like buoyancy structures in surface quasigeostrophic turbulence with a mean background buoyancy gradient.
We found that the stratification's vertical structure controls the locality of the inversion operator and the dispersion of surface-trapped Rossby waves. 
As we go from decreasing stratification profiles [$\sigma'(z) \leq 0$] to increasing stratification profiles [$\sigma'(z) \geq 0$], the inversion operator becomes more local and Rossby wave less dispersive. 
In all cases, we find that the non-dimensional ratio, $\keps/k_r$, controls the extent of inhomogeneous buoyancy mixing. Larger $\keps/k_r$ correspond to sharper buoyancy gradients at jet centres with larger peak jet velocities that are separated by more homogeneous mixed-zones.
Moreover, we found that the staircase limit is reached at smaller $\keps/k_r$ in more non-local flows; the staircase limit is reached by $\keps/k_r \approx 15$ for our $\sigma \leq 0$ simulations compared to $\keps/k_r \approx 25$ for our $\sigma \geq 0$ simulations.

In addition, once the staircase limit is reached, the dynamics of the jets depends on the locality of the inversion operator and, hence, on the stratification's vertical structure. In flows with a more non-local inversion operator [or decreasing stratification, $\sigma'(z) \leq 0$], we obtain straight jets that are perturbed by dispersive, eastward propagating, along jet waves. In contrast, for more local flows [or over increasing stratification, $\sigma'(z) \geq 0$], we obtain jets with latitudinal meanders on the order of the jet spacing. The shape of these jets evolves in time as these meanders propagate westwards as weakly dispersive waves. 

The inversion operator's locality is also reflected in two more aspects of the dynamics. First, the domain-averaged zonal speed exceeds the domain-averaged meridional speed by approximately a factor of eight in our most non-local simulations, whereas this ratio is merely two in our most local simulations. This observation is consistent with the fact that jets are narrower and exhibit larger latitudinal meanders in more local flows. Second, for a given Rhines wavenumber, jets in more local flows are closer together. Indeed, we found $L_j \, \krh \approx 3-4$ in our most non-local simulations, where $L_j$ is the jet half spacing, as compared to $L_j \, \krh \approx 0.5-1.5$ in our most local simulations.

Several open questions remain. First, we have not examined the dynamics of the along jet waves. As we observed, these waves propagate eastwards in our most non-local simulations [with $\sigma'(z) \leq 0$] but westwards for our most local simulations [with $\sigma'(z) \geq 0$]. These waves are not described by the dispersion relation \eqref{eq:dispersion}; rather, the relevant model is that of freely propagating edge waves along a buoyancy discontinuity \citep{mcintyre_potential-vorticity_2008}. However, the difficulty here is that a jump discontinuity in the buoyancy field results in infinite velocities over constant or increasing stratification. In addition, the relationship of the along jet waves in the staircase limit to the non-linear zonons found by \cite{sukoriansky_nonlinear_2008} remains unclear.

The divergence of the velocity at a buoyancy discontinuities raises a second question. Is there a limit to how close the staircase limit can be approached?
In barotropic dynamics, the velocity remains finite at a jump continuity in the vorticity, and, in this case, \cite{scott_structure_2012} report that a vorticity staircase case can be approached arbitrarily. Whether this result continues to hold for arbitrarily sharp buoyancy gradients and arbitrarily large zonal velocities is not clear. Because the rms velocity seems to converge for arbitrarily sharp staircases, even for the most local inversion relations we considered, there may not be any energetic reason precluding arbitrarily sharp buoyancy gradients.

Finally, there remains the question of how relevant these results are for the upper ocean, which, in addition to surface buoyancy gradients, has interior potential vorticity gradients as well. In particular, our neglect of the $\beta$-effect limits the direct relevance of this model to the upper ocean. Whether surface buoyancy staircases can emerge under more realistic oceanic conditions requires further investigation.

%% file: 4-Modes/Modes.tex
\chapter{Normal Modes With Boundary Dynamics in Geophysical Fluids}\label{Ch-modes}

\begin{abstractchapter}
	Three-dimensional geophysical fluids support both internal and boundary-trapped waves. To obtain the normal modes in such fluids we must solve a differential eigenvalue problem for the vertical structure (for simplicity, we only consider horizontally periodic domains). If the boundaries are dynamically inert (e.g., rigid boundaries in the Boussinesq internal wave problem, flat boundaries in the quasigeostrophic Rossby wave problem) the resulting eigenvalue problem typically has a Sturm-Liouville form and the properties of such problems are well-known. However, when restoring forces are also present at the boundaries, then the equations of motion contain a time-derivative in the boundary conditions and this leads to an eigenvalue problem where the eigenvalue correspondingly appears in the boundary conditions. In certain cases, the eigenvalue problem can be formulated as an eigenvalue problem in the Hilbert space $L^2\oplus \C$ and this theory is well-developed. Less explored is the case when the eigenvalue problem takes place in a Pontryagin space, as in the Rossby wave problem over sloping topography. This article develops the theory of such problems and explores the properties of wave problems with dynamically-active boundaries. The theory allows us to solve the initial value problem for quasigeostrophic Rossby waves in a region with sloping bottom (we also apply the theory to two Boussinesq problems with a free-surface). For a step-function perturbation at a dynamically-active boundary, we find that the resulting time-evolution consists of waves present in proportion to their projection onto the dynamically-active boundary. 
\end{abstractchapter}

\section{Introduction}\label{S-intro}

An important tool in the study of wave motion near a stable equilibrium is the separation of variables. When applicable, this elementary technique transforms a linear partial differential equation into an ordinary differential eigenvalue problem for each coordinate \cite[e.g.,][]{hillen_partial_2012}. Upon solving the differential eigenvalue problems, one obtains the normal modes of the physical system. The normal modes are the fundamental wave motions for the given restoring forces, each mode represents an independent degree of freedom in which the physical system can oscillate, and any solution of the wave problem may be written as a linear combination of these normal modes. 

To derive the normal modes, we must first linearize the dynamical equations of motion about some equilibrium state. We then encounter linearized restoring forces of two kinds: 
\begin{itemize}
	\item[\emph{1.}] volume-permeating forces experienced by fluid particles in the interior, and
	\item[\emph{2.}] boundary-confined forces only experienced by fluid particles at the boundary.
\end{itemize} 
Examples of volume-permeating forces include the restoring forces resulting from continuous density stratification and continuous volume potential vorticity gradients. These restoring forces respectively result in internal gravity waves \citep{sutherland_internal_2010} and Rossby waves \citep{Vallis2017}. Examples of boundary-confined restoring forces include the gravitational force at a free-surface (i.e., at a jump discontinuity in the background density), forces arising from gradients in surface potential vorticity \citep{schneider_boundary_2003}, and the molecular forces giving rise to surface tension. These restoring forces respectively result in surface gravity waves \citep{sutherland_internal_2010}, topographic/thermal waves \citep{hoskins_use_1985}, and capillary waves \citep{lamb_hydrodynamics_1975}. 

In the absence of  boundary-confined restoring forces, we can often apply Sturm-Liouville theory \cite[e.g.,][]{hillen_partial_2012,zettl_sturm-liouville_2010} to the resulting eigenvalue problem. We thus obtain a countable infinity of waves whose vertical structures form a basis of $L^2$, the space of square-integrable functions (see section \ref{math-section}), and, given some initial vertical structure, we know how to solve for the subsequent time-evolution as a linear combination for linearly independent waves. Moreover, a classic result of Sturm-Liouville theory is that the $n$th mode has $n$ internal zeros.

In the presence of boundary-confined restoring forces, the governing equations have a time-derivative in the boundary conditions. The resulting eigenvalue problem correspondingly contains the eigenvalue parameter in the boundary conditions. Sturm-Liouville theory is inapplicable to such problems.

In this chapter, we present a general method for solving these problems by delineating a generalization of Sturm-Liouville theory. Some consequences of this theory are the following. There is a countable infinity of waves whose vertical structures form a basis of $L^2\oplus\C^s$, where $s$ is the number of dynamically-active boundaries; thus, each boundary-trapped wave, in mathematically rigorous sense, provides an additional degree of freedom to the problem. The modes satisfy an orthogonality relation involving boundary terms, the modes may have a negative norm, and the modes may have finite jump discontinuities at dynamically-active boundaries (although the \emph{solutions} are always continuous, see section  \ref{S-Boussinesq-rot}). When negative norms are possible (as in quasigeostrophic theory), there is a new expression for the Fourier coefficients that one must use to solve initial value problems [see equation \eqref{expansion}]. We can also expand boundary step-functions (representing some boundary localized perturbation) as a sum of modes. Moreover, the $n$th mode may not have $n$ internal zeros; indeed, depending on physical parameters in the problem, two or three linearly independent modes with an identical number of internal zeros may be present.  

We also show that the eigenfunction expansion of a function is term-by-term differentiable, with the derivative series converging uniformly on the whole interval, regardless of the boundary condition the function satisfies at the dynamically-active boundaries. This property is in contrast with a traditional Sturm-Liouville eigenfunction expansion where the term-by-term derivative converges uniformly only if the function satisfies the same boundary condition as the eigenfunctions.

We apply the theory to three geophysical wave problems. The first is that of a Boussinesq fluid with a free-surface; we find that the $n$th mode has $n$ internal zeros. The second example is that of a rotating Boussinesq fluid with a free-surface where we assume that the stratification suppresses rotational effects in the interior but not at the upper boundary. We find that there are two linearly independent modes with $M$ internal zeros, where the integer $M$ depends on the ratio of the Coriolis parameter to the horizontal wavenumber, and that the eigenfunctions have a finite jump discontinuity at the upper boundary. The third application is to a quasigeostrophic fluid with a sloping lower boundary. We find that modes with an eastward phase speed have a negative norm whereas modes with a westward phase speed have a positive norm (the sign of the norm has implications for the relative phase of a wave and for series expansions). Moreover, depending on the propagation direction, there can be two linearly independent modes with no internal zeros. For all three examples, we outline the properties of the resulting series expansions and provide the general solution. We also consider the time-evolution resulting from a vertically localized perturbation at a dynamically-active boundary; we idealize such a perturbation as a boundary step-function. The step-function perturbation induces a time-evolution in which the amplitude of each constituent wave is proportional to the projection of that wave onto the boundary. \label{qg-time-evolution}

To our knowledge, most of the above results cannot be found in the literature [however, the gravity wave orthogonality relation has been noted before, e.g., \cite{gill_atmosphere-ocean_2003} and \cite{kelly_vertical_2016} for the hydrostatic case and  \cite{olbers_internal_1986} and \cite{early_fast_2020} for the non-hydrostatic case]. For instance, we provide the only solution to the initial value problem for Rossby waves over topography in the literature [equation \eqref{qg-time-evolution}]. Moreover, many of the properties we discuss arise in practical problems in physical oceanography. The number of internal zeros of Rossby waves is also a useful quantity in observational physical oceanography [e.g., \cite{clement_vertical_2014} and \cite{de_la_lama_vertical_2016}]. In addition, the question of whether the quasigeostrophic baroclinic modes are complete is a controversial one. 
\cite{lapeyre_what_2009} has suggested that the baroclinic modes are incomplete because they assume a vanishing surface buoyancy anomaly. Consequently, \cite{smith_surface-aware_2012} address this issue by deriving an $L^2\oplus \C^2$ basis for quasigeostrophic theory. Yet many authors, citing completeness theorems from Sturm-Liouville theory, insist that the baroclinic modes are indeed complete and can represent all quasigeostrophic states \citep{ferrari_distribution_2010,lacasce_surface_2012,rocha_galerkin_2015}. This chapter shows that, by including boundary-confined restoring forces, we obtain a set of modes with additional degrees-of-freedom. These degrees-of-freedom manifest in the behaviour of eigenfunction expansions at the boundaries.
In addition, the distinction between $L^2$ and $L^2\oplus \C^s$ bases that we present here is useful for equilibrium statistical mechanical calculations where one must decompose fluid motion onto a complete set of modes \citep{bouchet_statistical_2012, venaille_catalytic_2012}.

The plan of the chapter is the following. We formulate the mathematical theory in section \ref{math-section}. We then apply the theory to the two Boussinesq wave problems, in section \ref{S-Boussinesq}, and to the quasigeostrophic wave problem, in section \ref{S-QG}. We consider the time-evolution of a localized perturbation at a dynamically-active boundary in section \ref{S-step-forcing}. We then conclude in section  \ref{S-conclusion_modes}.

\section{The eigenvalue problem} \label{math-section}
	
In this section, we outline the theory of the differential eigenvalue problem,
\begin{align} \label{EigenDiff_modes}
	-(p \, \phi')' + q \, \phi &= \lambda \, r \, \phi \quad \text{for } \quad z\in \left(z_1,z_2\right)\\
	\label{EigenB1}
	- \left[a_1 \, \phi(z_1) - b_1 \, (p \, \phi')(z_1)\right] &= \lambda \left[c_1 \, \phi(z_1) - d_1 \, (p \, \phi')(z_1)\right] \\
	\label{EigenB2}
	- \left[a_2 \, \phi(z_2) - b_2 \, (p \, \phi')(z_2)\right] &= \lambda \left[c_2 \, \phi(z_2) - d_2 \, (p \, \phi')(z_2)\right], 	   	
\end{align}
where $p^{-1},q$, and $r$ are real-valued integrable functions; $a_i,b_i,c_i$, and $d_i$ are real numbers with $i\in\{1,2\}$; and where $\lambda \in \C$ is the eigenvalue parameter. We further assume that $p>0$ and $r>0$, that $p$ and $r$ are twice continuously differentiable, that $q$ is continuous,  and that $(a_i,b_i)\neq (0,0)$ for $i\in\{1,2\}$.
The system of equations \eqref{EigenDiff_modes}\textendash\eqref{EigenB2} is an eigenvalue problem for the eigenvalue $\lambda \in \C$ and differs from a regular Sturm-Liouville problem in that $\lambda$ appears in the boundary conditions \eqref{EigenB1} and \eqref{EigenB2}.  That is, setting $c_{i}= d_{i} = 0$ recovers the traditional Sturm-Liouville problem. The presence of $\lambda$ as part of the boundary condition leads to some fundamentally new mathematical features that are the subject of this section and fundamental to the physics of this chapter. 

It is useful to define the two boundary parameters
    \begin{equation}\label{Di}
    	D_i = (-1)^{i+1} \left(a_i \, d_i - b_i \, c_i\right) \quad i=1,2. 
    \end{equation}
Just as the function $r$ acts as a weight for the interval $(z_1,z_2)$ in traditional Sturm-Liouville problems, the constants $D_i^{-1}$ will play analogous roles for the boundaries $z=z_i$ when $D_i\neq 0$.

\subsubsection*{Outline of the mathematics}

The right-definite case, when the $D_i \geq 0$ for $i\in\{1,2\}$, is well-known in the mathematics literature; most of the right-definite results in this section are due to \cite{evans_non-self-adjoint_1970}, \cite{walter_regular_1973}, and \cite{fulton_two-point_1977}. In contrast, the left-definite case, defined below, is much less studied. In this section, we generalize the right-definite results of \cite{fulton_two-point_1977} to the left-definite problem as well as provide an intuitive formulation \cite[in terms of functions rather than vectors, for a vector formulation see][]{fulton_two-point_1977} of the eigenvalue problem. 

In section \ref{S-Op-Form} we state the conditions under which we obtain real eigenvalues and a basis of eigenfunctions. We proceed, in section  \ref{S-Prop-of-Eig}, to explore the properties of eigenfunctions and eigenfunction expansions. Finally, in section \ref{S-oscillation}, we discuss oscillation properties of the eigenfunctions. Additional properties of the eigenvalue problem are found in appendix \ref{S-additional-properties} and a literature review, along with various technical proofs, is found in appendix \ref{S-math-app}.

\subsection{Formulation of the problem} \label{S-Op-Form}
\subsubsection{The function space of the problem}
We denote by $L^2$ the Hilbert space of square-integrable ``functions" $\phi$ on the interval $(z_1,z_2)$ satisfying 
\begin{equation}
	\intz \abs{\phi}^2 \, r \, \mathrm{d}z < \infty.
\end{equation}
To be more precise, the elements of $L^2$ are not functions but rather equivalence classes of functions \cite[e.g.,][ section I.3]{reed_methods_1980}. Two functions, $\phi$ and $\psi$, are equivalent in $L^2$ (i.e., $\phi=\psi$ in $L^2$) if they agree in a mean-square sense on $[z_1,z_2]$,
\begin{equation}\label{L2-equal}
	\intz \abs{\phi(z) - \psi(z)}^2 \, r \, \mathrm{d}z = 0.
\end{equation}
Significantly, we can have $\phi=\psi$ in $L^2$ but $\phi\neq \psi$ pointwise. 
	
Furthermore, as a Hilbert space, $L^2$ is endowed with a positive-definite inner product
\begin{equation}\label{L2-inner}
	\inner{\phi}{\psi}_\sigma = \intz \ov{\phi}\, \psi \, \mathrm{d}\sigma = \intz \ov{\phi} \, \psi \ r\, \mathrm{d}z,
\end{equation}
where the symbol $\ov{\ }$ denotes complex conjugation and the measure $\sigma$ associated $L^2$ induces a differential element $\mathrm{d}\sigma = r\,\mathrm{d}z$ (see appendix \ref{S-additional-properties}). The positive-definiteness is ensured by our assumption that $r>0$ (i.e., $\inner{\phi}{\phi}_{\sigma} > 0$ for $\phi \neq 0$ when $r>0$).

It is well-known that traditional Sturm-Liouville problems [i.e., equations \eqref{EigenDiff_modes}\textendash\eqref{EigenB2} with $c_i=d_i=0$ for $i=1,2$] are eigenvalue problems in some subspace of $L^2$ \citep{debnath_introduction_2005}. For the more general case of interest here, the eigenvalue problem occurs over a ``larger'' function space denoted by $L^2_\mu$ which we construct in appendix \ref{S-additional-properties}.

Let the integer $s\in\{0,1,2\}$ denote the number of $\lambda$-dependent boundary conditions and let $S$ denote the set
\begin{equation}\label{S-set}
	S = \{j \ | \ j \in \{1,2\} \text{ and } (c_j,d_j) \neq (0,0) \}.
\end{equation}
$S$ is one of $\emptyset,\{1\},\{2\},\{1,2\}$ and $s$ is the number of elements in the set $S$. In appendix \ref{S-additional-properties}, we show that $L^2_\mu$ is isomorphic to the space $L^2\oplus \C^s$ and is thus ``larger'' than $L^2$ by $s$ dimensions. 

We denote elements of $L^2_\mu$ by upper case letters $\Psi$; we define $\Psi(z)$ for $z\in[z_1,z_2]$ by
\begin{equation}\label{L2_mu-element}
	\Psi(z) = 
	\begin{cases}
		\Psi(z_i) \quad & \textrm{at } z=z_i, \textrm{ for }\, i \in S, \\
		\psi(z) \quad & \textrm{otherwise},
	\end{cases}
\end{equation}
where $\Psi(z_i) \in \C$ are constants, for $i\in S$, and the corresponding lower case letter $\psi$ denotes an element of $L^2$. 
Two elements $\Phi$ and $\Psi$ of $L^2_\mu$ are equivalent in $L^2_\mu$ if and only if
\begin{itemize}
	\item [\emph{1.}] $\Phi(z_i) = \Psi(z_i)$ for $i\in S$, and
	\item [\emph{2.}] $\phi(z)$ and $\psi(z)$ are equivalent in $L^2$ [i.e., as in equation \eqref{L2-equal}].
\end{itemize}
Here, $\Phi$, as an element of $L^2_\mu$, is defined as in equation \eqref{L2_mu-element}. The primary difference between $L^2$ and $L^2_\mu$ is that $L^2_\mu$ discriminates between functions that disagree at $\lambda$-dependent boundaries.

The measure $\mu$ associated with $L^2_\mu$ (see appendix \ref{S-additional-properties}) induces a differential element
\begin{equation} \label{dmu}
	\mathrm{d}\mu(z) = \left[ r(z) + \sum_{i\in S} D_i^{-1} \, \delta(z-z_i) \right] \mathrm{d}z,
\end{equation}
where $\delta(z)$ is the Dirac delta. The induced inner product on $L^2_\mu$ is
\begin{equation}\label{dmu_inner}
	\inner{\Phi}{\Psi} = \intz \ov{\Phi} \, {\Psi} \, \mathrm{d}\mu = \intz \ov {\Phi} \, {\Psi} \, r \, \mathrm{d}z + \sum_{i\in S} D_i^{-1} \, \ov{\Phi(z_i)} \, \Psi(z_i).
\end{equation}
If $D_i > 0$ for $i\in S$ then this inner product is positive-definite and $L^2_\mu$ is a Hilbert space. However, this is not the case in general.

Let $\kappa$ denote the number of negative $D_i$ for $i\in S$ (the possible values are $\kappa=0,1,2$). Then $L^2_\mu$ has a $\kappa$-dimensional subspace of elements $\Psi$ satisfying 
\begin{equation}
	\inner{\Psi}{\Psi} < 0.
\end{equation}
This makes $L^2_\mu$ a Pontryagin space of index $\kappa$ \citep{bognar_indefinite_1974}. If $\kappa=0$ then $L^2_\mu$ is again a Hilbert space. In the present case, $L^2_\mu$ also has an infinite-dimensional subspace of elements $\psi$ satisfying 
\begin{equation}
	\inner{\Psi}{\Psi} > 0.
\end{equation}

\subsubsection{Reality and completeness} \label{S-real-complete}

In appendix \ref{S-eigen-in-L2-mu}, we reformulate the eigenvalue problem \eqref{EigenDiff_modes}\textendash\eqref{EigenB2} as an eigenvalue problem of the form,
\begin{equation}\label{Op-Eigenproblem}
	\Lc \, \Phi = \lambda \, \Phi,
\end{equation}
in a subspace of $L^2_\mu$, where $\Lc$ is a linear operator and $\Phi$ an element of $L^2_\mu$. We also define the notions of right- and left-definiteness that are required for the reality and completeness theorem below. The following two propositions can be considered to define right- and left-definiteness for applications of the theory. Both propositions are obtained through straightforward manipulations (see appendix \ref{S-additional-properties}).

\begin{proposition}[Criterion for right-definiteness]\label{right-definite-prop}
	The eigenvalue problem \eqref{EigenDiff_modes}\textendash\eqref{EigenB2} is right-definite if $r>0$ and $D_i>0$ for $i\in S$.
\end{proposition}

\begin{proposition}[Criterion for left-definiteness] \label{left-definite-prop}
	The eigenvalue problem \eqref{EigenDiff_modes}\textendash\eqref{EigenB2} is left-definite if the following conditions hold:
	\begin{itemize}
		\item [(i)] the functions $p,q$ satisfy $p>0, q \geq 0$,
		\item [(ii)] for the $\lambda$-dependent boundary conditions, we have 
			\begin{equation}
				\frac{a_i \, c_i}{D_i} \leq 0, \quad \frac{b_i \, d_i}{D_i} \leq 0, \quad (-1)^i \frac{a_i \, d_i}{D_i} \geq 0  \quad \text{for } i \in S.
			\end{equation}
		\item [(iii)] for the $\lambda$-independent boundary conditions, we have 
			\begin{equation}
				b_i = 0  \quad \text{or} \quad  (-1)^{i+1}\frac{a_i}{b_i} \geq 0 \quad \text{ if } b_i \neq 0   \quad \text{ for } i \in \{1,2\}\setminus S.
			\end{equation}
	\end{itemize}
\end{proposition}
The notions of right and left-definiteness are not mutually exclusive. Namely, a problem can be neither right- or left-definite; both right- and left-definite; only right-definite; or only left-definite. In this chapter, we always assume that $p>0$ and $r>0$.


The reality of the eigenvalues and the completeness of the eigenfunctions in the space $L^2_\mu$ is given by the following theorem.
\begin{theorem}[Reality and completeness]\label{real-complete}

	Suppose the eigenvalue problem \eqref{EigenDiff_modes}\textendash\eqref{EigenB2} is either right-definite or left-definite. Moreover, if the problem is not right-definite, we assume that $\lambda=0$ is not an eigenvalue. Then the eigenvalue problem  \eqref{EigenDiff_modes}\textendash\eqref{EigenB2} has a countable infinity of real simple eigenvalues $\lambda_n$ satisfying 
	\begin{equation}
		\lambda_0 < \lambda_1 < \dots < \lambda_n < \dots \rightarrow \infty,
	\end{equation}
	with corresponding eigenfunctions $\Phi_n$. Furthermore, the set of eigenfunctions $\{\Phi_n\}_{n=0}^\infty$ is a complete orthonormal basis for $L^2_\mu$ satisfying
	\begin{equation}
		\inner{\Phi_m}{\Phi_n} = \pm \delta_{mn}.
	\end{equation} 
\end{theorem}
\begin{proof}
	See appendix \ref{S-real-proof}.
\end{proof}
Recall that $\kappa$ denotes the number of negative $D_i$ for $i\in S$. We then have the following corollary of the proof of theorem \ref{real-complete}.

\begin{proposition}\label{eigenvalue-sign}
	Suppose the eigenvalue problem \eqref{EigenDiff_modes}\textendash\eqref{EigenB2} is left-definite and that $\lambda=0$ is not an eigenvalue. Then there are $\kappa$ negative eigenvalues and their eigenfunctions satisfy
	\begin{equation}
		\inner{\Phi}{\Phi} < 0.
	\end{equation}
	The remaining eigenvalues are positive and their eigenfunctions satisfy
	\begin{equation}
		\inner{\Phi}{\Phi}>0.
	\end{equation}
\end{proposition}
In other words, proposition \ref{eigenvalue-sign} states that we have the relationship
\begin{equation}
	\lambda_n \inner{\Phi_n}{\Phi_n} > 0
\end{equation}
for left-definite problems.

\subsection{Properties of the eigenfunctions} \label{S-Prop-of-Eig}
For the remainder of section \ref{math-section}, we assume that the eigenvalue problem \eqref{EigenDiff_modes}\textendash\eqref{EigenB2} satisfies the requirements of theorem \ref{real-complete}.

\subsubsection{Eigenfunction expansions}

The eigenvalue problem \eqref{EigenDiff_modes}\textendash\eqref{EigenB2} has \emph{eigenfunctions} $\{\Phi_n\}_{n=0}^\infty$ as well as corresponding \emph{solutions} $\{\phi_n\}_{n=0}^\infty$. In other words, while the $\phi_n$ are the solutions to the differential equation defined by equations \eqref{EigenDiff_modes}\textendash\eqref{EigenB2} with $\lambda=\lambda_n$, the eigenfunctions required by the operator formulation of the problem [equation \eqref{Op-Eigenproblem}] are $\Phi_n$. The functions $\Phi_n$ and $\phi_n$ are related by equation \eqref{L2_mu-element}, with the boundary values $\Phi_n(z_i)$ of $\Phi_n$ determined by
\begin{equation}\label{discontinuity}
	\Phi_n(z_i) = \left[c_i \, \phi(z) - d_i \, (p \, \phi')(z)\right] \quad \text{for } i\in S.
\end{equation}
Thus, while the solutions $\phi_n$ are continuously differentiable over the closed interval $[z_1,z_2]$, the eigenfunctions $\Phi_n$ are continuously differentiable over the open interval $(z_1,z_2)$ but generally have finite jump discontinuities at the $\lambda$-dependent boundaries.
The eigenfunctions $\Phi_n$ are continuous in the closed interval $[z_1,z_2]$ only if $c_i=1$ and $d_i=0$ for $i\in S$. In this case, the eigenfunctions $\Phi_n$ coincide with the solutions $\phi_n$ on the closed interval $[z_1,z_2]$.

The boundary conditions of the eigenvalue problem \eqref{EigenDiff_modes}\textendash\eqref{EigenB2} are not unique. One can multiply each boundary condition by an arbitrary constant to obtain an equivalent problem. To uniquely specify the eigenfunctions in physical applications, the boundary coefficients $\{a_i,b_i,c_i,d_i\}$ of equations \eqref{EigenDiff_modes}\textendash\eqref{EigenB2} must be chosen so that $r \, \mathrm{d}z$ has the same dimensions as $D_i^{-1} \, \delta(z-z_i) \, \mathrm{d}z$ [recall that $\delta(z)$ has the dimension of inverse length]. In the quasigeostrophic problem, we must also invoke continuity and set $c_i=1$.

Since $\{\Phi_n\}_{n=0}^\infty$ is a basis for $L^2_\mu$, then any $\Psi \in L^2_\mu$ may be expanded in terms of the eigenfunctions \citep[][thereom IV.3.4]{bognar_indefinite_1974},
\begin{equation}\label{expansion}
	\Psi = \sum_{n=0}^\infty \frac{\inner{\Psi}{\Phi_n}}{\inner{\Phi_n}{\Phi_n}} \, \Phi_n.
\end{equation}
We emphasize that the above equality is an equality in $L^2_\mu$ and not a pointwise equality [see the discussion following equation \eqref{L2_mu-element}]. Some properties of $L^2_\mu$ expansions are given in appendix \ref{S-more-eigen-expansions}.

An important property that distinguishes the basis $\{\Phi_n\}_{n=0}^\infty$ of $L^2_\mu$ from an $L^2$ basis is its ``sensitivity'' to function values at boundary points $z=z_i$ for $i\in S$. See section \ref{S-step-forcing} for a physical application.

A natural question is whether the basis $\{\Phi_n\}_{n=0}^\infty$ of $L^2_\mu$ is also a basis of $L^2$. Recall that the set $\{\Phi_n\}_{n=0}^\infty$ is a basis of $L^2$ if every element $\psi \in L^2$ can be written \emph{uniquely} in terms of the functions $\{\Phi_n\}_{n=0}^\infty$. However, in general, this is not true. If $s>0$, the $L^2_\mu$ basis $\{\Phi_n\}_{n=0}^\infty$ is overcomplete in $L^2$ \citep{walter_regular_1973,russakovskii_matrix_1997}.

\subsubsection{Uniform convergence and term-by-term differentiability}\label{S-uniform}

Along with the eigenfunction expansion \eqref{expansion} in terms of the eigenfunctions $\{\Phi_n\}_{n=0}^\infty$, we also have the expansion
\begin{equation}\label{expansion-phi}
	\sum_{n=0}^\infty \frac{\inner{\Psi}{\Phi_n}}{\inner{\Phi_n}{\Phi_n}} \, \phi_n
\end{equation}
in terms of the solutions $\phi_n$. The two expansions differ in their behaviour at $\lambda$-dependent boundaries, $z=z_i$ for $i\in S$, but are otherwise equal. In particular, the $\Phi_n$ eigenfunction expansion \eqref{expansion} must converge to $\Psi(z_i)$ at $z=z_i$ for $i\in S$ as this equality is required for $\Psi$ to be equal to the series expansion \eqref{expansion} in $L^2_\mu$ [see the discussion following equation \eqref{L2_mu-element}]. Some properties of both expansions are given in appendix \ref{S-pointwise}. In particular, theorem \ref{pointwise} shows that the $\phi_n$ solution series \eqref{expansion-phi} does not generally converge to $\Psi(z_i)$ at $z=z_i$.

\begin{figure}
  \includegraphics[width=1.\columnwidth]{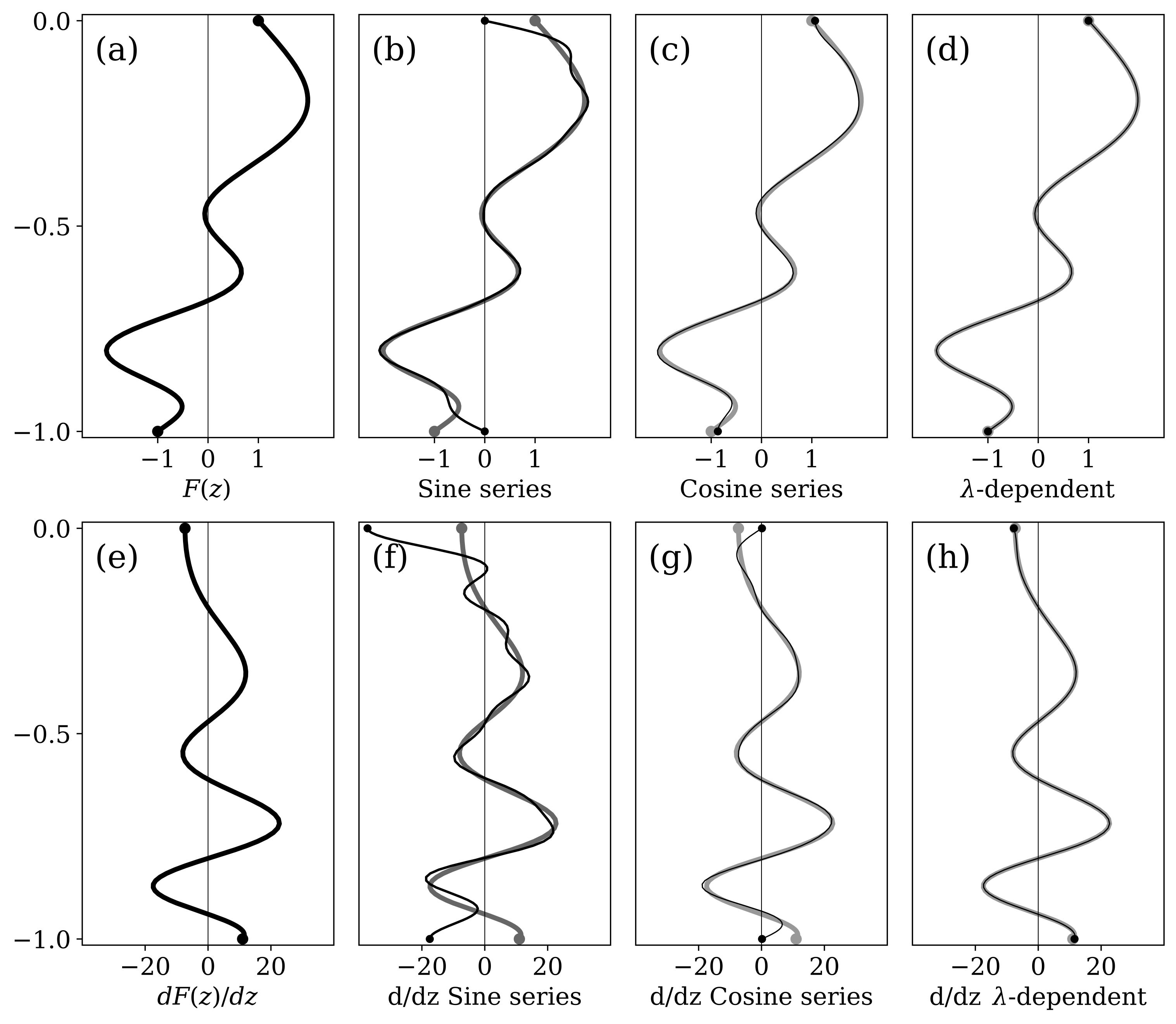}
  \caption{Convergence to a function $F(z) = 1 + 2z +(3/2)\sin(2\pi z)\cos(\pi^2z^2 + 3)$ for $z \in[-1,0]$, shown in panel (a), by various eigenfunction expansions of $-\phi''= \lambda \, \phi$ with fifteen terms, as discussed in section \ref{S-uniform}. Panel (b) shows the Fourier sine expansion of $F$. Since the sine eigenfunctions vanish at the boundaries $z=-1,0$, the series expansion will not converge to $F$ at the boundaries. Panel (c) shows the cosine expansion of $F$ which converges uniformly to $F$ on the closed interval $[-1,0]$. Panel (d) shows an expansion with boundary coefficients in equations \eqref{EigenB1}\textendash\eqref{EigenB2} given by $(a_1,b_1,c_1,d_1) = (-0.5,-5,1,0)$ and $(a_2,b_2,c_2,d_2) =  (0.5,-5,1,0)$. Since the $c_i=1$ and $d_i=0$, then $\Phi_n = \phi_n$ and the series expansions \eqref{expansion} and \eqref{expansion-phi} coincide. As with the cosine series, the expansion converges uniformly to $F$ on $[-1,0]$. The derivative of $F$ is shown in panel (e). Panel (f) show the derivative of the sine series expansion. In panel (g), we show the differentiated cosine series which does not converge to the derivative $F'$ at the boundaries $z=z_1,z_2$. In contrast, in panel (h), the differentiated series obtained from a problem with $\lambda$-dependent boundary conditions converges uniformly to the derivative $F'$.}
  \label{F-convergence}
\end{figure}

The following theorem is of central concern for physical applications. 

\begin{theorem}[Uniform convergence]\label{uniform}
 Let $\psi$ be a twice continuously differentiable function on $[z_1,z_2]$ satisfying all $\lambda$-independent boundary conditions of the eigenvalue problem \eqref{EigenDiff_modes}\textendash\eqref{EigenB2}. Define the function $\Psi$ on $z\in[z_1,z_2]$ by 
	\begin{equation}
	\Psi(z) = 
	\begin{cases}
		c_i \, \psi(z) - d_i \, (p \, \psi')(z) \quad & \textrm{at } z=z_i, \textrm{ for }\, i \in S, \\
		\psi(z) \quad & \textrm{otherwise}.
	\end{cases}
	\end{equation}
	Then 
	\begin{equation}\label{uniform-two-series}
		\psi(z) = \sum_{n=0}^\infty \frac{\inner{\Psi}{\Phi_n}}{\inner{\Phi_n}{\Phi_n}} \, \phi_n(z)  \quad \textrm{and} \quad \psi'(z) = \sum_{n=0}^\infty \frac{\inner{\Psi}{\Phi_n}}{\inner{\Phi_n}{\Phi_n}} \, \phi_n'(z) 
	\end{equation}
	with both series converging uniformly and absolutely on $[z_1,z_2]$.
\end{theorem}
\begin{proof}
	See appendix \ref{S-A-left-fulton}.
\end{proof}

If $c_i=1$ and $d_i=0$ for $i\in S$ then we can replace $\Phi_n$ by $\phi_n$ and $\Psi$ by $\psi$ in equation \eqref{uniform-two-series}.

In addition, if both boundary conditions of the eigenvalue problem \eqref{EigenDiff_modes}\textendash\eqref{EigenB2} are $\lambda$-dependent, then both  expansions in equation \eqref{uniform-two-series} converge uniformly on $[z_1,z_2]$ regardless of the boundary conditions $\psi$ satisfies. As discussed in appendix \ref{S-pointwise}, for traditional Sturm-Liouville expansions, an analogous result holds only if $\psi$ satisfies the same boundary conditions as the eigenfunctions. Figure \ref{F-convergence} contrasts the convergence behaviour of such a problem (with continuous eigenfunctions, so $c_i=1$ and $d_i=0$ for $i\in S$) with the convergence behaviour of sine and cosine series. All numerical solutions in this chapter are obtained using a pseudo-spectral code in Dedalus \citep{burns_dedalus_2020}. 

\begin{figure}
  \centerline{\includegraphics[width=0.8\columnwidth]{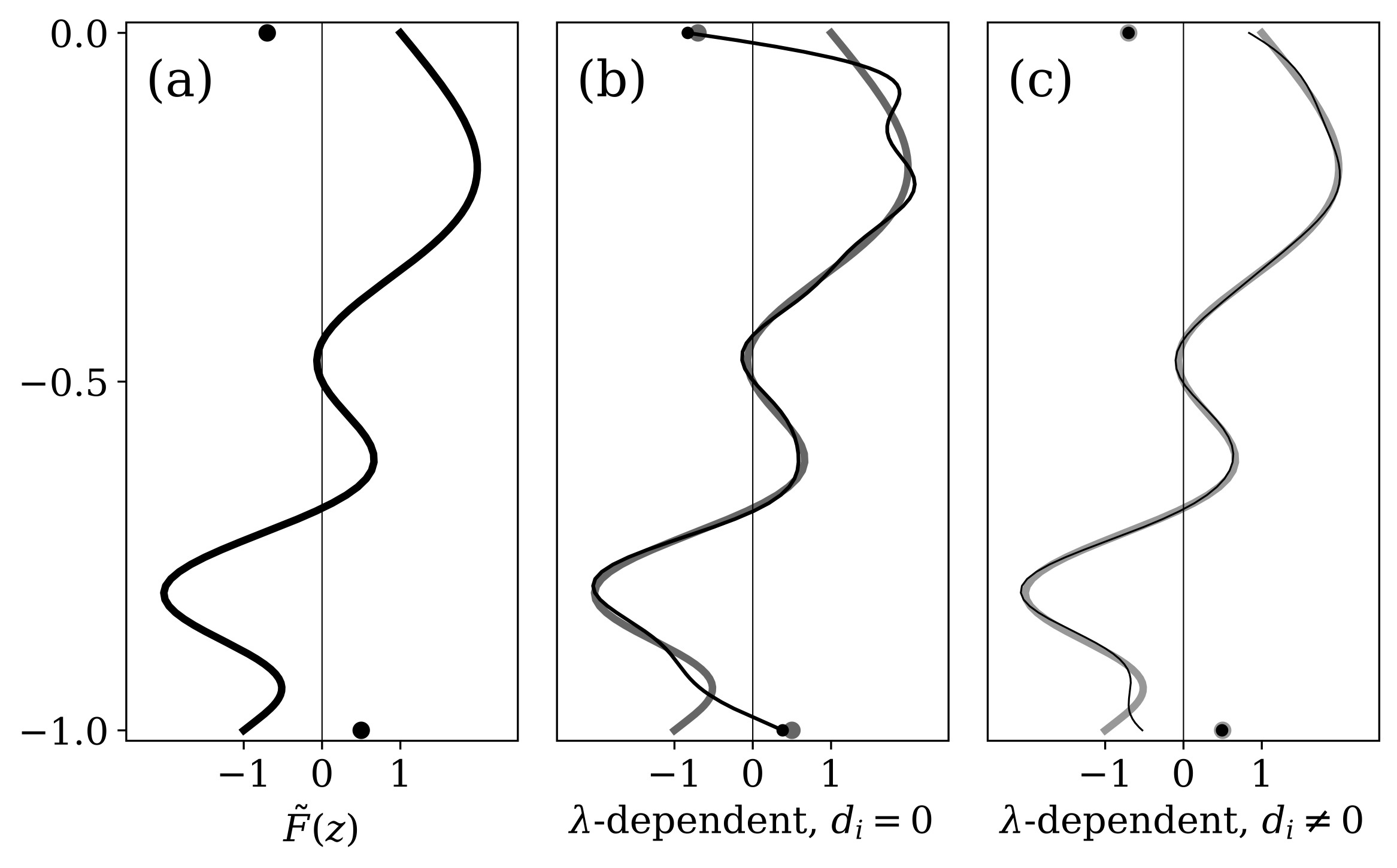}}
  \caption{Convergence to a function $\tilde F$ with finite jump discontinuities at the boundaries by two eigenfunction expansions (with $\lambda$-dependent boundary conditions) of $-\phi'' = \lambda\,\phi$ with fifteen terms, as discussed in section \ref{S-uniform}. The function $\tilde F(z)$ is defined by $\tilde F(z) = F(z)$ for $z\in(z_1,z_2)$ where $F(z)$ is the function defined in figure \ref{F-convergence}, $F(-1) = 0.5$ at the lower boundary, and $F(0)=-0.7$ at the upper boundary. The function $\tilde F$ is shown in panel (a). In panel (b), the boundary coefficients in equations \eqref{EigenB1}\textendash\eqref{EigenB2} are given by  $(a_1,b_1,c_1,d_1) = (-0.5,-5,1,0)$ and $(a_2,b_2,c_2,d_2) =  (0.5,-5,1,0)$ as in figure \ref{F-convergence}. In panel (c), the boundary coefficients are  $(a_1,b_1,c_1,d_1) = (-0.5,-5,1,0.1)$ and $(a_2,b_2,c_2,d_2) =  (0.5,-5,1,-0.1)$. The $\Phi_n$ expansion \eqref{expansion} and the $\phi_n$ expansion \eqref{expansion-phi} are not generally equal at the boundaries $z=-1,0$; this figure shows the $\Phi_n$ expansion. The $\Phi_n$ series \eqref{expansion} converges pointwise to $\tilde F$ on $[-1,0]$, however, the convergence will not be uniform if $d_i=0$ for $i\in S$, as in panel (b). The boundary values of the $\Phi_n$ series \eqref{expansion} are shown with a black dot. In panel (b), the eigenfunctions $\Phi_n$ are continuous and a large number of terms are required for the series to converge to the discontinuous function $\tilde F$. Panel (c) shows that the discontinuous eigenfunction $\Phi_n$ have almost converged to the $\tilde F$\textemdash including at the jump discontinuities; the black dot in panel (c) overlap with the grey dots, which represent the boundary values of $\tilde F$. Although the $\phi_n$ series \eqref{expansion-phi} converges to $\tilde F$ in the interior $(-1,0)$, the $\phi_n$ series does not generally converge to $\tilde F$ at the boundaries but instead converges to the values given in theorem \ref{pointwise}.} 
  \label{F-convergence-disc}
\end{figure}

Another novel property of the eigenfunction expansions is that we obtain pointwise convergence to functions that are smooth in the interior of the interval, $(z_1,z_2)$, but have finite jump discontinuities at $\lambda$-dependent boundaries (see appendix \ref{S-pointwise}). If $d_i\neq 0$ for $i\in S$, the convergence is even uniform \cite[][corollary 2.1]{fulton_two-point_1977}. Figure \ref{F-convergence-disc} illustrates the convergence behaviour for eigenfunction expansions with $\lambda$-dependent boundary conditions in the two cases $d_i=0$ and $d_i\neq 0$. Note the presence of Gibbs-like oscillations in the case $d_i=0$ shown in panel (b). Although the $\Phi_n$ eigenfunction series \eqref{expansion} converges pointwise to the discontinuous function, the $\phi_n$ solution series \eqref{expansion-phi} converges to the values given in theorem \ref{pointwise} at the $\lambda$-dependent boundaries. The ability of these series expansions to converge to functions with boundary jump discontinuities is related to their ability to expand distributions in the \citet{bretherton_critical_1966} ``$\delta$-function formulation'' of a problem.

\subsection{Oscillation theory}\label{S-oscillation}

Recall that for regular Sturm-Liouville problems [i.e., equations \eqref{EigenDiff_modes}\textendash\eqref{EigenB2} with $c_i=d_i=0$] we obtain a countable infinity of real simple eigenvalues, $\lambda_n$, that may be ordered as 
\begin{equation}\label{eig-seq}
	\lambda_0 < \lambda_1 < \lambda_2 < \dots \rightarrow \infty,
\end{equation}
with associated eigenfunctions $\phi_n$. The $n$th eigenfunction $\phi_n$ has $n$ internal zeros in the interval $(z_1,z_2)$ so that no two eigenfunctions have the same number of internal zeros.

However, once the eigenvalue $\lambda$ appears in the boundary conditions, there may be up to $s+1$ linearly independent eigenfunctions with the same number of internal zeros. The crucial parameters deciding the number of zeros is $-b_i/d_i$ for $i\in S$, where $b_i$ and $d_i$ are the boundary coefficients appearing in the boundary conditions \eqref{EigenB1}\textendash\eqref{EigenB2}. The following lemma outlines the possibilities when only one boundary condition is $\lambda$-dependent.

\begin{lemma}[Location of double oscillation count]\label{extra-oscillation}

Suppose that $s=1$, $i\in S$, and let $\kappa$ be the number of negative $D_i$ for the eigenvalue problem \eqref{EigenDiff_modes}\textendash\eqref{EigenB2}. We have the following possibilities.
\begin{itemize}
	\item[(i)] Right-definite, $d_i \neq 0$: The eigenfunction $\Phi_n$ corresponding to the eigenvalue $\lambda_n$ has $n$ internal zeros if $\lambda_n < -b_i/d_i$ and $n-1$ internal zero if $-b_i/d_i \leq \lambda_n$.
	\item[(ii)] Right-definite, $d_i = 0$:  The $n$th eigenfunction has $n$ internal zeros.
	\item[(iii)] Left-definite: If $\kappa=0$ then all eigenvalues are positive, the problem is right-definite, and  either (i) or (ii) applies.
	Otherwise, if $\kappa=1$, then the eigenvalues may be ordered as 
	\begin{equation}
		\lambda_0 < 0 < \lambda_1 < \lambda_2 < \dots  \rightarrow \infty.
	\end{equation}
	Both eigenfunctions $\Phi_0$ and $\Phi_1$ have no internal zeros. The remaining eigenfunctions $\Phi_n$, for $n>1$, have $n-1$ internal zeros.
\end{itemize}
\end{lemma}
\begin{proof}
	Parts (i), (ii) and (ii) are due to \cite{linden_leightons_1991}, \cite{binding_sturmliouville_1994}, and \cite{binding_left_1999}, respectively.
\end{proof}

When both boundary conditions are $\lambda$-dependent, the situation is similar. See \cite{binding_sturmliouville_1994} and \cite{binding_left_1999} for further discussion.

\section{Boussinesq gravity-capillary waves} \label{S-Boussinesq}

Consider a rotating Boussinesq fluid on an $f$-plane with a reference Boussinesq density of $\rho_0$. The fluid is subject to a constant gravitational acceleration $g$ in the downwards, $-\unit z$, direction, and to a surface tension $T$ \citep[with dimensions of force per unit length, see][]{lamb_hydrodynamics_1975} at its upper boundary. The upper boundary of the fluid, given by $z=\eta$, is a free-surface defined by the function $\eta(\vec x,t)$, where $\vec x = \unit x \, x + \unit y \, y$ is the horizontal position vector. The lower boundary of the fluid is a flat rigid surface given by $z=-H$. The fluid region is periodic in both horizontal directions $\unit x$ and $\unit y$. 

\subsection{Linear equations of motion}

The governing equations for infinitesimal perturbations about a background state of no motion, characterized by a prescribed background density of $\rho_B = \rho_B(z)$, are 
	\begin{align} \label{grav-cap-1}
	\partial_t^2  \lap  w + f_0^2 \, \partial_z^2  w + N^2 \, \lap_z w = 0 \quad &\text{for } z \in \left(-H,0\right) \\
	\label{grav-cap-2}
	w = 0 \quad &\text{for } z=-H \\
	\label{grav-cap-3}
	-\partial_t^2 \partial_z w - f_0^2 \, \partial_z w + g_b \, \lap_z  w - \tau \, \laptwo_z  w = 0 \quad &\text{for } z=0,
\end{align}
where $w$ is the vertical velocity, $f_0$ is the constant value of the Coriolis frequency, the prescribed buoyancy frequency $N^2$ is given by
\begin{equation}\label{N2}
	N^2(z) = - \frac{g}{\rho_0} \d{\rho_B(z)}{z},
\end{equation}
the acceleration $g_b$ is the effective gravitational acceleration at the upper boundary
\begin{equation}
	g_b = -\frac{g}{\rho_0} \left[\rho_a - \rho_B(0-)\right]
\end{equation}
where $\rho_a$ is the density of the overlying fluid, and the parameter $\tau$ is given by
\begin{equation}
	\tau = \frac{T}{\rho_0}
\end{equation}
where $T$ is the surface tension. The three-dimensional Laplacian is denoted $\lap = \partial_x^2 + \partial_y^2 + \partial_z^2$, the horizontal Laplacian is denoted by $\lap_z = \partial_x^2 + \partial_y^2$, and the horizontal biharmonic operator is given by $\laptwo_z = \lap_z \, \lap_z$. See equation (1.37) in \cite{dingemans_water_1997} for the surface tension term in \eqref{grav-cap-3}. The remaining terms in equation \eqref{grav-cap-1}\textendash\eqref{grav-cap-3} are standard \citep{gill_atmosphere-ocean_2003}. Consistent with our assumption that $\eta(\vec x,t)$ is small, we evaluate the upper boundary condition at $z=0$ in equation \eqref{grav-cap-3}.
 
\subsection{Non-rotating Boussinesq fluid}\label{S-Boussinesq-non-rot} 

\begin{figure}
  \centerline{\includegraphics[width=1.\columnwidth]{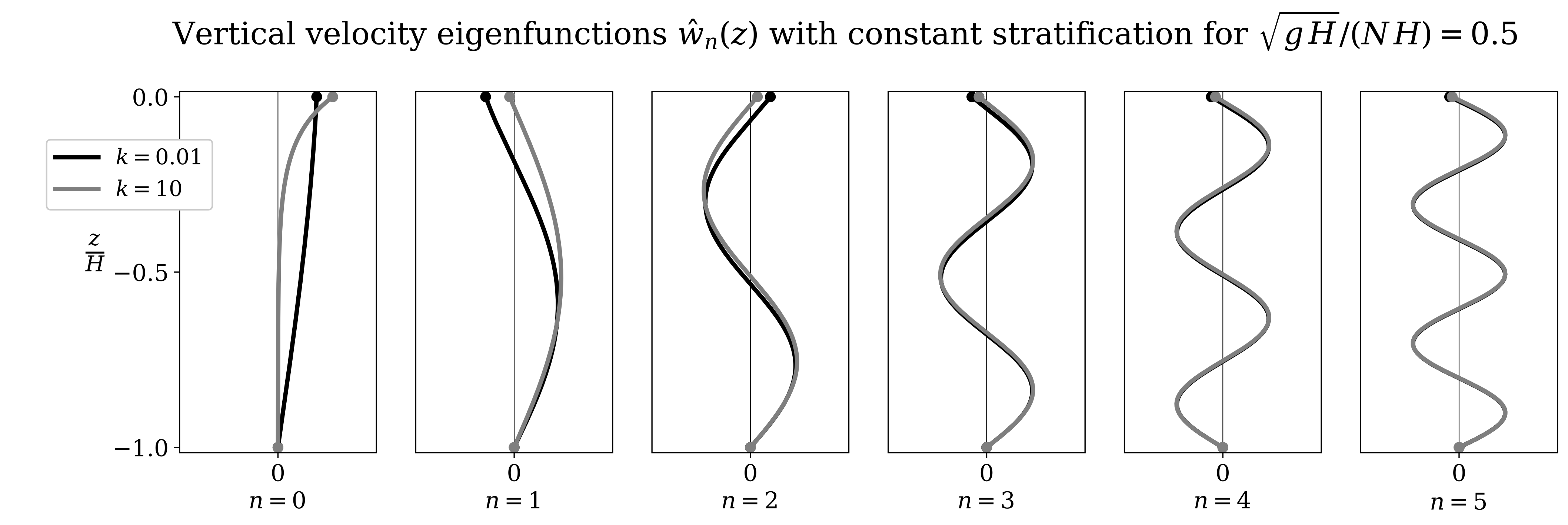}}
  \caption{The vertical velocity eigenfunctions $\hat W_n = \hat w_n$ of the non-rotating Boussinesq eigenvalue problem \eqref{non-rot-SL1}\textendash\eqref{non-rot-SL3} for two distinct wavenumbers with constant stratification, as discussed in section \ref{S-Boussinesq-non-rot}. For both wavenumbers, the $n$th eigenfunction has $n$ internal zeros as in regular Sturm-Liouville theory. The zeroth mode ($n=0$) corresponds to a surface gravity wave and is trapped to the upper boundary for large horizontal wavenumbers. In contrast to the internal wave problem with a rigid lid, the modes $\hat w_n$ now depend on the horizontal wavenumber $k$ through the boundary condition \eqref{non-rot-SL3}, however, this dependence is weak for $n\gg1$, as can be observed in this figure; for $n>2$, the modes for $k=0.01$ (in black)  and for $k=10$ (in grey) nearly coincide. The horizontal wavenumbers $k$ are non-dimensionalized by $H$.}
  \label{F-nonrot-eigenfunctions}
\end{figure}

We assume wave solutions of the form
\begin{equation}\label{w-wave}
	w(\vec x, z ,t) = \hat w(z) \, \mathrm{e}^{\mathrm{i} \left( \vec k \cdot \vec x - \omega t\right)}
\end{equation}
where $\vec k = \unit x \, k_x + \unit y \, k_y$ is the horizontal wavevector and $\omega$ is the angular frequency. Substituting the wave solution \eqref{w-wave} into equations \eqref{grav-cap-1}\textendash\eqref{grav-cap-3} and setting $f_0=0$ yields
\begin{align}
	\label{non-rot-SL1}
	-\hat w'' + k^2 \, \hat w = \sigma^{-2} \, N^2 \, \hat w \quad &\text{for } z \in (-H,0) \\
	\label{non-rot-SL2}
	\hat w = 0 \quad &\text{for } z = -H \\
	\label{non-rot-SL3}
	(g_b + \tau \, k^2)^{-1} \hat w ' =  \sigma^{-2} \, \hat w \quad &\text{for } z = 0,
\end{align}
where $\sigma = \omega/k$ is the phase speed and $k = \abs{\vec k}$ is the horizontal wavenumber. Equations \eqref{non-rot-SL1}\textendash\eqref{non-rot-SL3} are an eigenvalue problem for the eigenvalue $\lambda = \sigma^{-2}$.

\subsubsection*{Definiteness \& the underlying function space}

Equations \eqref{non-rot-SL1}\textendash\eqref{non-rot-SL3} form an eigenvalue problem with one $\lambda$-dependent boundary condition, namely, the upper boundary condition \eqref{non-rot-SL3}. The underlying function space is then
\begin{equation} \label{gravity-functionspace}
	L^2_\mu \cong L^2 \oplus \C.
\end{equation}
We write $\hat W_n$ for the eigenfunctions and $\hat w_n$ for the solutions of the eigenvalue problem \eqref{non-rot-SL1}\textendash\eqref{non-rot-SL3} [see the paragraph containing equation \eqref{discontinuity}]. The eigenfunctions $\hat W_n$ are related to the solutions $\hat w_n$ by equation \eqref{L2_mu-element} with boundary values $\hat W_n(0)$ given by equation \eqref{discontinuity}. However, since $c_2 = 1$ and $d_2=0$ in equation \eqref{non-rot-SL3} [compare with equations \eqref{EigenDiff_modes}\textendash\eqref{EigenB2}] then $\hat W_n = \hat w_n$ on the closed interval $[-H,0]$; thus, the solutions $w_n$ are also the eigenfunctions.

By theorem \ref{real-complete}, the eigenfunctions $\{\hat w_n \}_{n=0}^\infty$ form an  orthonormal basis of $L^2_\mu$. For functions $\varphi$ and $\phi$, the inner product is
\begin{align}\label{grav-inner}
	\inner{\varphi}{\phi} &= \frac{1}{N_0^2 \, H} \left[ \int_{-H}^0 \varphi \, \phi \, N^2 \, \textrm{d}z + (g_b + \tau \, k^2) \varphi(0) \, \phi(0)\right]
\end{align}
obtained from equations \eqref{dmu_inner} and equation \eqref{Di}; we have introduced the factor $1/(N^2_0\, H)$ in the above expression for dimensional consistency in eigenfunction expansions ($N^2_0$ is a typical value of $N^2$). Orthonormality is then given by 
\begin{align}\label{non-rot-ortho}
	\delta_{mn} &= \inner{\hat w_m}{\hat w_n}
\end{align}
and we have chosen the solutions $\hat w_n$ to be non-dimensional (so the Kronecker delta is non-dimensional as well).

One verifies that the eigenvalue problem \eqref{non-rot-SL1}\textendash\eqref{non-rot-SL3} is right-definite using proposition \ref{right-definite-prop} and left-definite using proposition \ref{left-definite-prop}.
Right-definiteness implies that $L^2_\mu$, with the inner product \eqref{grav-inner}, is a Hilbert space. That is, all eigenfunctions $\hat w_n$ satisfy
\begin{equation}
	\inner{\hat w_n}{\hat w_n} > 0.
\end{equation} 
Left-definiteness, along with proposition \ref{eigenvalue-sign}, ensures that all eigenvalues $\lambda_n = \sigma_n^{-2}$ are positive. Indeed, the phase speeds $\sigma_n$ satisfy
\begin{equation}\label{grav-phase-speeds}
	\sigma_0^2 > \sigma_1^2 > \dots > \sigma^2_n > \dots  \rightarrow 0.
\end{equation}

\subsubsection*{Properties of the eigenfunctions}

By lemma \ref{extra-oscillation}, the $n$th eigenfunction $\hat w_n$ has $n$ internal zeros in the interval $(-H,0)$. See figure \ref{F-nonrot-eigenfunctions} for an illustration of the first six eigenfunctions.

The eigenfunctions $\{\hat w_n \}_{n=0}^\infty$ are complete in $L^2$ but do not form a basis in $L^2$; in fact, the basis is overcomplete in $L^2$. The presence of a free-surface provides an additional degree of freedom over the usual  rigid-lid $L^2$ basis of internal wave eigenfunctions. Indeed, the $n=0$ wave in figure \ref{F-nonrot-eigenfunctions} corresponds to a surface gravity wave, while the remaining modes are internal gravity waves (with some surface motion).

\subsubsection*{Expansion properties}

Given a twice continuously differentiable function $\chi (z)$ satisfying $\chi(-H)=0$, then, from theorem \ref{uniform}, we have
\begin{equation}
	\chi(z) = \sum_{n=0}^\infty \inner{\chi}{\hat w_n} \, \hat w_n(z) \quad \textrm{and} \quad \chi'(z) = \sum_{n=0}^\infty \inner{\chi}{\hat w_n}\, w_n'(z),
\end{equation}
with both series converging uniformly on $[-H,0]$ (note that $\chi$ is not required to satisfy any particular boundary condition at $z=0$). If $\chi$ is the vertical structure at time $t=0$ (and at some wavevector $\vec k$) and we assume $\partial_t w(\vec x,z,t=0)=0$, then the subsequent time-evolution is given by
\begin{equation}
	 w(\vec x, z,t) = \sum_{n=0}^\infty \inner{\chi}{\hat w_n} \, w_n(z) \, \cos\left(\sigma_n k t\right) \mathrm{e}^{\mathrm{i}\vec k \cdot \vec x} .
\end{equation}

\subsubsection*{The $f$-plane hydrostatic problem}

Suppose we have hydrostatic gravity waves on an $f$-plane with free surface at the upper boundary, as in \cite{kelly_vertical_2016}. The appropriate inner product is obtained by setting $\tau = 0$ in the inner product \eqref{grav-inner}. All the above results on the eigenfunctions of gravity-capillary waves carry over to the hydrostatic $f$-plane problem provided we set 
\begin{equation}
    \sigma^2 = \frac{\omega^2 - f_0^2}{k^2}.
\end{equation}

\subsection{A Boussinesq fluid with a rotating upper boundary} \label{S-Boussinesq-rot}

\begin{figure}
  \centerline{\includegraphics[width=1.\columnwidth]{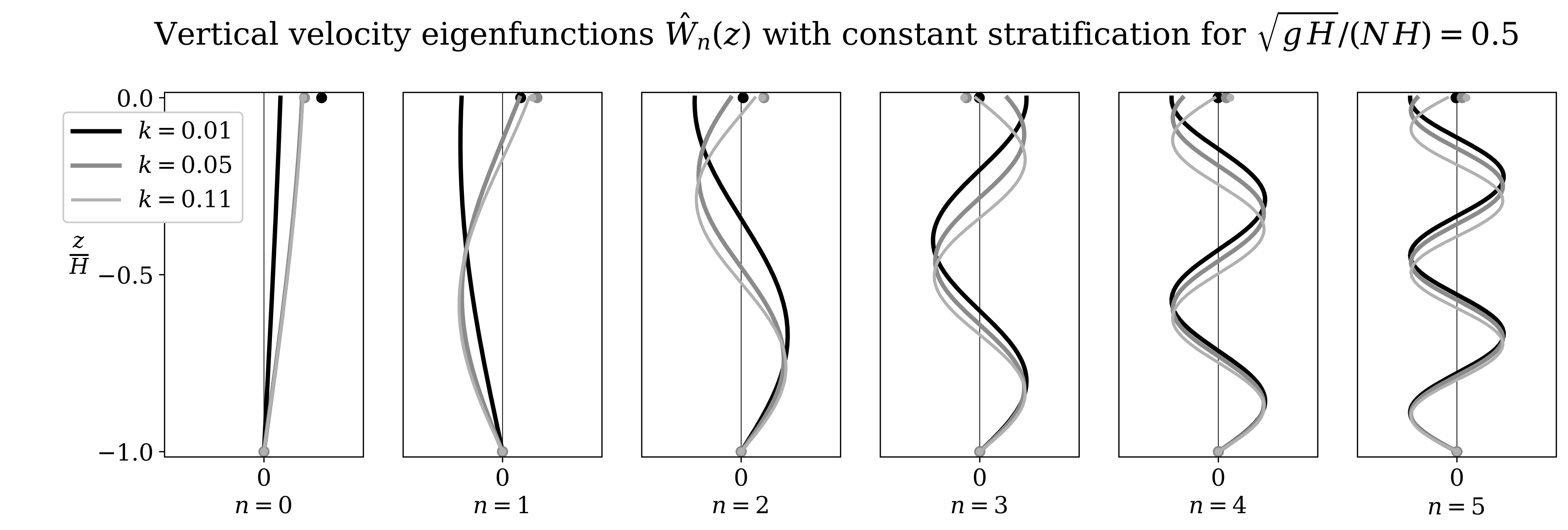}}
  \caption{The vertical velocity eigenfunctions $\hat W_n$ of a Boussinesq fluid with a rotating upper boundary\textemdash eigenvalue problem \eqref{rot-bound-SL1}\textendash\eqref{rot-bound-SL3}. This figure is discussed in section \ref{S-Boussinesq-rot}. The wavenumbers $k$ in the figure are non-dimensionalized by the depth $H$. The dots represent the values of the eigenfunctions at the boundaries. Note that the eigenfunctions have a finite jump discontinuity at $z=0$. For $k \, H=0.01$ (given by the black line) there are two modes with no internal zeros. As $k$ increases, we obtain two modes with one internal zero (at $k\, H =0.05$, the thick grey line) and then two modes with three internal zeros (at $k\, H=0.11$, the thin grey line).}
  \label{F-rotating-eigenfunctions}
\end{figure}
Although this next example is not geophysically relevant, it has the curious property that the resulting eigenfunctions are discontinuous.

Let $N^2_0$ be a typical value of $N^2(z)$. Consider the situation where $f_0^2/N_0^2 \ll 1$ but
\begin{equation}\label{rot-grav-scaling}
	\frac{g_b + \tau \, k^2}{f_0^2 \, H} \sim O(1).
\end{equation}
Accordingly, we may neglect the Coriolis parameter in the interior equation \eqref{grav-cap-1} but not at the upper boundary condition \eqref{grav-cap-3}. Substituting the wave solution \eqref{w-wave} into equations \eqref{grav-cap-1}\textendash\eqref{grav-cap-3} yields
\begin{align}
	\label{rot-bound-SL1}
	-\hat w'' + k^2 \, \hat w = \sigma^{-2} \, N^2 \, \hat w \quad &\text{for } z \in (-H,0) \\
	\label{rot-bound-SL2}
	\hat w = 0 \quad &\text{for } z = -H \\
	\label{rot-bound-SL3}
	(g_b + \tau \, k^2)^{-1} \hat w ' =  \sigma^{-2} \left[ \hat w  + \frac{f_0^2}{k^2} (g_b + \tau \, k^2)^{-1} \hat w' \right]\quad &\text{for } z = 0,
\end{align}
where $\sigma=\omega/k$ is the phase speed. Equations \eqref{rot-bound-SL1}\textendash\eqref{rot-bound-SL3} form an eigenvalue problem for the eigenvalue $\lambda = \sigma^{-2}$.

\subsubsection*{Definiteness \& the underlying function space}

As in the previous case, the eigenvalue problem is both right-definite and left-definite, the underlying function space $L^2_\mu$ is given by equation \eqref{gravity-functionspace}, and the appropriate inner product is equation \eqref{grav-inner}.  By right-definiteness, the space $L^2_\mu$, equipped with the inner product \eqref{grav-inner}, is a Hilbert space; thus, all eigenfunctions $\hat W_n$ satisfy
 \begin{equation}
 	\inner{\hat W_m}{\hat W_n} > 0.
 \end{equation}
 By theorem \ref{real-complete}, all eigenvalues $\lambda_n = \sigma_n^{-2}$ are real and the corresponding eigenfunctions $\{\hat W_n\}_{n=0}^\infty$ form an orthonormal basis of the Hilbert space $L^2_\mu$. By proposition \ref{eigenvalue-sign}, all eigenvalues $\lambda_n = \sigma_n^{-2}$ are positive and satisfy equation \eqref{grav-phase-speeds}.

\subsubsection*{Boundary jump discontinuity of the eigenfunctions}

The main difference between the previous non-rotating problem \eqref{non-rot-SL1}\textendash\eqref{non-rot-SL3} and the above problem \eqref{rot-bound-SL1}\textendash\eqref{rot-bound-SL3} is that, in the present problem, if $f_0 \neq 0$ then $d_2 \neq 0$ [see equation \eqref{EigenB2}]. Thus, by equation \eqref{discontinuity}, the eigenfunctions $\hat W_n$ generally have a jump discontinuity at the upper boundary $z=0$ (see figure \ref{F-rotating-eigenfunctions}) and so are not equal to the solutions $\hat w_n$. The eigenfunctions $\hat W_n$ are defined by $\hat W_n(z) = \hat w_n(z)$ for $z \in[-H,0)$ and
\begin{equation}
	\hat W_n(0) = \hat w_n(0) +  \frac{f_0^2}{k^2} (g_b + \tau \, k^2)^{-1} \, \hat w_n'(0)
\end{equation}
[see equation \eqref{discontinuity}]. It is not difficult to show that
\begin{equation}
	\hat W_n(0) \approx 0 \quad \text{ for } n \text{ sufficiently large,}
\end{equation}
as can be seen in figure \ref{F-rotating-eigenfunctions}. 

Physical motion is given by the solutions $\hat w_n$ which are continuous over the closed interval $[-H,0]$. The jump discontinuity in the eigenfunctions $\hat W_n$ does not correspond to any physical motion; instead, the eigenfunctions $\hat W_n$ are convenient mathematical aids used to obtain eigenfunction expansions in the function space $L^2_\mu$.

\subsubsection*{Number of internal zeros of the eigenfunctions} 

Another consequence of $d_2 \neq 0$ is that by, lemma \ref{extra-oscillation}, there are two distinct solutions $\hat w_M$ and $\hat w_{M+1}$ with the same number  of internal zeros (i.e., $M$) in the interval $(-H,0)$. Noting that  
\begin{equation}
	-\frac{b_2}{d_2} = \frac{k^2}{f_0^2}
\end{equation}
the integer $M$ is determined by 
\begin{equation}
	\sigma_0^2 > \sigma_1^2>\dots > \sigma_M^2 > \frac{f_0^2}{k^2} \geq \sigma^2_{M+1} > \dots > 0.
\end{equation}
A smaller $f_0$ or a larger $k$ implies a larger $M$ and hence that $\hat w_M$ and $\hat w_{M+1}$ have a larger number of internal zeros, as shown in figure \ref{F-rotating-eigenfunctions}.

\subsubsection*{Expansion properties}

As in the previous problem, the eigenfunctions are complete in $L^2_\mu$ but overcomplete in $L^2$ due to the additional surface gravity-capillary wave.

Given a twice continuously differentiable function $\chi(z)$ satisfying $\chi(-H)=0$, we define the discontinuous function $X(z)$ by
\begin{equation}
	X(z) = 
	\begin{cases}
		\chi(z) &\quad \textrm{for } z\in [-H,0) \\
		\chi(0) + \frac{f_0^2}{k^2}  \left( g_b + \tau \, k^2 \right)^{-1} \chi'(0)  &\quad \textrm{for } z = 0
	\end{cases}
\end{equation}
as in theorem \ref{uniform}. Then, by theorem \ref{uniform}, we have the expansions
\begin{equation}
	\chi(z) = \sum_{n=0}^\infty \inner{X}{\hat W_n} \, \hat w_n(z) \quad \textrm{and} \quad \chi'(z) = \sum_{n=0}^\infty \inner{X}{\hat W_n}\, w_n'(z).
\end{equation}
Moreover, if $\chi(z)$ is the vertical structure at $t=0$ (and at some wavevector $\vec k$) and we assume $\partial_t w(\vec x,z,t=0)=0$, then the subsequent time-evolution is given by
\begin{equation}
  w(\vec x, z,t) = \sum_{n=0}^\infty \inner{X}{\hat W_n} \, \hat w_n(z) \, \cos\left(\sigma_n k t\right) \,\mathrm{e}^{\mathrm{i}\vec k \cdot \vec x}.
\end{equation}

\section{Quasigeostrophic waves} \label{S-QG}

\subsection{Linear equations}

Linearizing the quasigeostrophic equations about a quiescent background state with an infinitesimally sloping lower boundary, at $z=-H$, and a rigid flat upper boundary, at $z=0$, renders
\begin{align}
	\label{linear-q}
	\partial_t \left[ \lap_z \psi +  \partial_z \left( S^{-1} \, \partial_z \psi \right) \right]  +  \unit z \cdot \left( \grad_z \psi \times \grad_z f \right) &=0 \quad \text{for } z\in(-H,0)\\
	\label{linear-r1}
	\partial_t \left( S^{-1} \, \partial_z \psi \right)  +  \unit z \cdot \left( \grad_z \psi \times f_0 \, \grad_z h \right) &=0 \quad \text{for } z=-H\\
	\label{linear-r2}
	\partial_t \left( S^{-1} \, \partial_z \psi \right) &= 0 \quad \text{for } z=0.
\end{align}
See \cite{rhines_edge_1970}, \cite{charney_oceanic_1981}, \cite{straub_dispersive_1994} for details.
The streamfunction $\psi$ is defined through $\vec u = \unit z \times \grad_z \psi$ where $\vec u$ is the horizontal velocity and $\grad_z = \unit x \, \partial_x + \unit y \, \partial_y$ is the horizontal Laplacian. The stratification parameter $S$ is given by
\begin{equation}
	S(z) = \frac{N^2(z)}{f_0^2},
\end{equation}
where $N^2$ is the buoyancy frequency and $f_0$ is the reference Coriolis parameter. The latitude dependent Coriolis parameter $f$ is defined by
\begin{equation}
	f(y) = f_0 + \beta \, y.
\end{equation}
Finally, $h(\vec x)$ is the height of the topography at the lower boundary and is a linear function of the horizontal position vector $\vec x$. Consistent with quasigeostrophic theory, we assume that topography $h$ is small and so we evaluate the lower boundary condition at $z=-H$ in equation \eqref{linear-r1}.

\subsection{The streamfunction eigenvalue problem} \label{S-QG-stream}

We assume wave solutions of the form
\begin{equation}\label{QG-psi-wave}
	\psi(\vec x, z,t) = \hat \psi (z) \, \mathrm{e}^{\mathrm{i}(\vec k \cdot \vec x - \omega t)}
\end{equation}
where $\vec k = \unit x \, k_x + \unit y \, k_y$ is the horizontal wavevector and $\omega$ is the angular frequency. 
 
We denote by $\Delta \theta_f$ the angle between the horizontal wavevector $\vec k$ and the gradient of Coriolis parameter $\grad_z f$,
\begin{equation} \label{theta_f}
	\sin{(\Delta \theta_f)} = \frac{1}{k \, \beta } \, \unit z \cdot \left(\vec k \times  {\grad_z f }\right),
\end{equation}
where $k= \abs{\vec k}$ is the horizontal wavenumber. Positive angles are measured counter-clockwise relative to $\vec k$. Thus, $\Delta \theta_f>0$ indicates that $\vec k$ points to the right of $\grad_z f$ while $\Delta \theta_f<0$ indicates that $\vec k$ points to the left of $\grad_z f$.

We define the topographic parameter $\alpha$ by
\begin{equation}
	\alpha = \abs{f_0 \, \grad_z h}.
\end{equation}
In analogy with $\Delta \theta_f$, we define the angle $\Delta \theta_h$ by
\begin{equation} \label{theta_i}
	\sin{(\Delta \theta_h)} =  \frac{1}{k \, \alpha} \, \unit z \cdot \left( \vec k \times  {f_0 \grad_z h} \right) 
\end{equation}
with a similar interpretation assigned to $\Delta \theta_h>0$ and $\Delta \theta_h<0$.

Substituting the wave solution \eqref{QG-psi-wave} into the linear quasigeostrophic equations \eqref{linear-q}\textendash \eqref{linear-r2} and assuming that $\alpha \, \sin(\Delta\theta_h) \neq 0$, $\omega \neq 0$, and $k\neq 0$, we obtain
\begin{align}
		\label{QG-psi-l-eigen1}
	- (S^{-1} \, \hat \psi')' + k^2 \, \hat \psi  = \lambda  \, \hat \psi \quad &\text{for } z\in(-H,0)\\
	\label{QG-psi-l-eigen2}
	-\frac{\beta}{\alpha} \, \frac{\sin{(\Delta \theta_f)}}{ \sin{(\Delta \theta_h)}} \, S^{-1} \, \hat \psi' =  \lambda \,  \psi  \quad &\text{for } z=-H \\
	\label{QG-psi-l-eigen3}
	S^{-1}\, \psi' = 0 \quad &\text{for } z=0,
\end{align}
where we have defined the eigenvalue $\lambda$ by
\begin{equation}\label{qg-dispersion}
	\lambda = -\frac{k \, \beta \, \sin{(\Delta \theta_f)}}{\omega}.
\end{equation}
Since $k\neq 0$ then $\lambda=0$ is not an eigenvalue. The above problem \eqref{QG-psi-l-eigen1}\textendash\eqref{QG-psi-l-eigen3} was recently considered in \cite{lacasce_prevalence_2017}.

\begin{figure}
  \centerline{\includegraphics[width=1.\columnwidth]{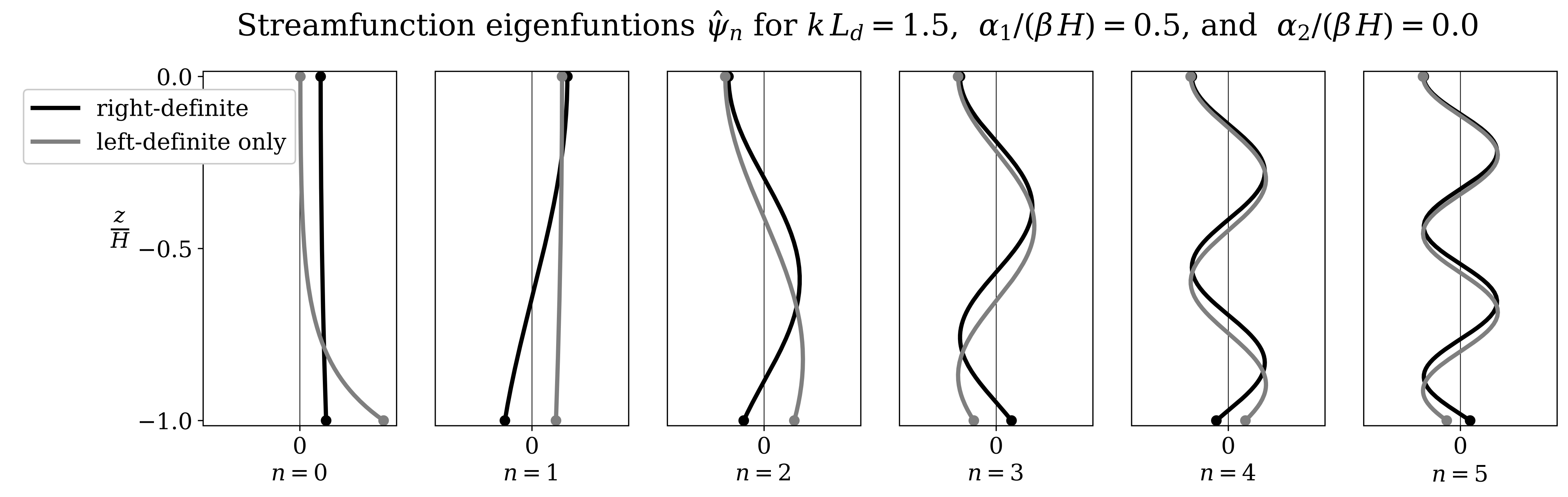}}
  \caption{The streamfunction eigenfunctions $\hat \psi_n$ of the quasigeostrophic eigenvalue problem with a sloping bottom from section \ref{S-QG-stream}. Two cases are shown. The first is with $\Delta \theta_f = -90^\circ$ and $\Delta \theta_1 = -30^\circ$ and is both right-definite and left-definite. The second is with $\Delta \theta_f = -45^\circ$ and $\Delta \theta_1 = 15^\circ$ and is only left-definite. In the right-definite case the $n$th eigenfunction has $n$ internal zero whereas in the left-definite only case there are two eigenfunctions ($n=0,1$) with no internal zeros.}
  \label{F-QG-psi}
\end{figure}

\subsubsection*{Definiteness \& the underlying function space}

The eigenvalue problem has one $\lambda$-dependent boundary condition and so the underlying function space is 
\begin{equation}
	L^2_\mu \cong L^2 \oplus \C.
\end{equation}
The appropriate inner product is obtained from equations \eqref{dmu_inner} and \eqref{Di}
\begin{align}\label{inner-qg}
	\inner{\varphi}{\phi} =  \frac{1}{H} \left[ \int_{-H}^0 \varphi \, \phi \, \mathrm{d}z + \frac{\alpha}{\beta} \, \frac{\sin\left(\Delta \theta_h\right)}{\sin\left(\Delta \theta_f\right)} \, \varphi(-H) \, \phi(-H) \right]
\end{align} 
where we have introduced the factor $1/H$ for dimensional consistency in eigenfunction expansions. By proposition \ref{right-definite-prop}, the problem is right-definite for horizontal wavevectors $\vec k$ satisfying
\begin{equation}\label{qg-right-definite-condition}
	\frac{\sin{(\Delta \theta_h)}}{\sin{(\Delta \theta_f)}} >0
\end{equation}
and, in such cases, $L^2_\mu$ equipped with the inner product \eqref{dmu_inner} is a Hilbert space. However, $L^2_\mu$ is not a Hilbert space for all wavevectors $\vec k$. By proposition \ref{left-definite-prop}, the problem is left-definite for all wavevectors $\vec k$ and so $L^2_\mu$, equipped with the inner product \eqref{dmu_inner}, is generally a Pontryagin space.

We write $\hat \Psi_n$ for the eigenfunctions and $\hat \psi_n$ for the solutions of equations \eqref{QG-psi-l-eigen1}\textendash \eqref{QG-psi-l-eigen3}. The eigenfunctions $\hat \Psi_n$ are related to the solutions $\hat \psi_n$ by \eqref{L2_mu-element} with boundary values $\hat \Psi_n(0)$ given by equation \eqref{discontinuity}. However, since $c_1 = 1$ and $d_1=0$ in equation \eqref{QG-psi-l-eigen2} [compare with equations \eqref{EigenDiff_modes}\textendash\eqref{EigenB2}] then $\hat \Psi_n = \hat \psi_n$ on the closed interval $[-H,0]$. Thus, the solutions $\psi_n$ are also the eigenfunctions.

With theorem \ref{real-complete}, we deduce that all eigenvalues $\lambda_n$ are real and the corresponding eigenfunctions $\{\hat \psi_n\}_{n=0}^\infty$ form an orthonormal basis for $L^2_\mu$. Orthonormality is defined with respect to the inner product given by equation \eqref{inner-qg} and takes the form
\begin{equation}
	\pm \delta_{mn} = \inner{\hat \psi_m}{\hat \psi_n}
\end{equation}
where we have taken the eigenfunctions $\hat \psi_m$ and $\hat \psi_n$ to be non-dimensional.

\subsubsection*{Properties of the eigenfunctions}

By lemma \ref{extra-oscillation}, the number of internal zeros of the eigenfunctions $\{\hat \psi_n\}_{n=0}^\infty$ depends on the propagation direction and hence [by equation \eqref{qg-right-definite-condition}] on the definiteness of the problem (see figure \ref{F-QG-psi}):
\begin{itemize}
	\item [\emph{1.}] if the problem is right-definite then the $n$th eigenfunction has $n$ internal zeros,
	\item [\emph{2.}] if the problem is not right-definite then both $\psi_0$ and $\psi_1$ have no internal zeros; the remaining eigenfunctions $\psi_n$, for $n>1$, have $n-1$ internal zeros.
\end{itemize}

As the problem is left-definite for all wavevectors $\vec k$, we can use proposition \ref{eigenvalue-sign} to determine the sign of the eigenvalues. 
Proposition \ref{eigenvalue-sign} informs us that
\begin{equation}
	\lambda_n \inner{\hat \psi_n}{\hat \psi_n} > 0.
\end{equation}
In the first case, when the problem is right-definite, all eigenvalues are positive and all eigenfunctions $\hat \psi_n$ satisfy $\inner{\hat \psi_n}{\hat \psi_n} >0$.
In the second case, when the problem is only left-definite, then there is one negative eigenvalue $\lambda_0$ and the corresponding eigenfunction $\hat \psi_0$ satisfies $\inner{\hat \psi_0}{\hat \psi_0} < 0$.
The remaining eigenvalues are positive and their corresponding eigenfunctions satisfy 
$ \inner{\hat \psi_n}{\hat \psi_n} >0$.
In fact, from equation \eqref{qg-dispersion}, we see that waves with $\inner{\hat \psi_n}{\hat \psi_n} >0 $ have westward phase speeds $\omega_n /k <0$ while waves with $\inner{\hat \psi_n}{\hat \psi_n} <0 $ have eastward phase speeds $\omega_n/k >0$.

\subsubsection*{Expansion properties}

The eigenfunctions $\{\hat \psi_n\}_{n=0}^\infty$ are complete in $L^2_\mu$ but overcomplete in $L^2$. Physically, there is now an additional eigenfunction corresponding to a topographic Rossby wave ($n=0$ in figure \ref{F-QG-psi}).

Given a twice continuously differentiable function $\phi(z)$ satisfying $\phi'(0)=0$, then from theorem \ref{uniform}, we have
\begin{equation}
	\phi(z) = \sum_{n=0}^\infty \frac{\inner{\phi}{\hat \psi_n}}{\inner{\hat \psi_n}{\hat \psi_n}} \, \hat \psi_n(z) \quad \textrm{and} \quad \phi'(z) = \sum_{n=0}^\infty \frac{\inner{\phi}{\hat \psi_n}}{\inner{\hat \psi_n}{\hat \psi_n}} \, \hat \psi'_n(z),
\end{equation}
with both series converging uniformly on $[-H,0]$ (note that $\phi$ is not required to satisfy any particular boundary condition at $z=-H$). If the vertical structure at time $t=0$  (and at some wavevector $\vec k$) is given by $\phi$, then the subsequent time-evolution is given by
\begin{equation}\label{qg-time-evolution}
	 \psi(\vec x, z,t) = \sum_{n=0}^\infty \frac{\inner{\phi}{\hat \psi_n}}{\inner{\hat \psi_n}{\hat \psi_n}} \, \hat \psi_n(z) \, \cos\left(\omega_n t\right) \, \mathrm{e}^{\mathrm{i}\vec k \cdot \vec x},
\end{equation}
where the angular frequency $\omega_n$ is given by equation \eqref{qg-dispersion}.

\section{A localized perturbation at the boundary}\label{S-step-forcing}

We now consider a localized perturbation at a dynamically-active boundary; we idealize such a perturbation by a boundary step-function $\Theta_i$ (for $i \in S$) given by
	\begin{equation}\label{Theta}
		\Theta_i(z) =
		\begin{cases}
			1 \quad \text{if } z=z_i\\
			0 \quad \text{otherwise.}
		\end{cases}
	\end{equation}
	Using equation \eqref{expansion}, the series expansion of $\Theta_i$ is found to be
	\begin{equation}\label{theta-expansion}
		\Theta_i 
		=  \frac{1}{D_i} \sum_{n=0}^\infty \frac{\, \Phi_n(z_i)}{\inner{\Phi_n}{\Phi_n}} \, \Phi_n(z).
	\end{equation}
	
For the non-rotating Boussinesq problem of section \ref{S-Boussinesq-non-rot}, a step-function perturbation with amplitude $w_0$ (at some wavevector $\vec k$) yields the time-evolution
\begin{equation}
	w(\vec x, z,t) = w_0 \left(\frac{g_b + \tau \, k^2}{N_0^2\, H}\right) \sum_{n=0}^\infty \hat w_n(0) \, \hat w_n(z) \, \cos\left(\sigma_n k t \right) \,\mathrm{e}^{\mathrm{i}\vec k \cdot \vec x}.
\end{equation}
Analogously, for the quasigeostrophic problem of section \ref{S-QG-stream}, a step-function perturbation with amplitude $\psi_0$ (at some wavevector $\vec k$) yields the time-evolution
\begin{equation}
	\psi(\vec x, z,t) = \psi_0 \left[\frac{\alpha \, \sin(\Delta \theta_h) }{H\, \beta \, \sin(\Delta \theta_f)}\right] \sum_{n=0}^\infty \frac{\hat \psi_n(-H)}{\inner{\hat\psi_n}{\hat \psi_n}} \, \hat \psi_n(z) \, \cos\left(\omega_n t\right) \,\mathrm{e}^{\mathrm{i} \vec k \cdot \vec x}.
\end{equation}
That both the above series converge to a step-function at $t=0$ (and $\vec x = \vec 0$) is confirmed by theorem \ref{expansion-theorem} along with theorem 2 in \citet{fulton_two-point_1977}.

We thus see that a step-function perturbation induces wave motion with an amplitude that is proportional to the boundary-confined restoring force (at wavevector $\vec k$). Moreover, the amplitude of each constituent wave in the resulting motion is proportional to the projection of that wave onto the dynamically-active boundary.

\section{Summary and conclusions}\label{S-conclusion_modes}

We have developed a mathematical framework for the analysis of three-dimensional wave problems with dynamically-active boundaries (i.e., boundaries where time derivatives appear in the boundary conditions). The resulting waves have vertical structures that depend on the wavevector $\vec k$: For Boussinesq gravity waves, the dependence is only through the wavenumber $k$ whereas the dependence for quasigeostrophic Rossby waves is on both the wavenumber $k$ and the propagation direction $\vec k/k$. Moreover, the vertical structures of the waves are complete in a space larger than $L^2$, namely, they are complete in $L^2_\mu \cong L^2 \oplus \C^s$ where $s$ is the number of dynamically active boundaries (and the number of boundary-trapped waves). Each dynamically active boundary contributes an additional boundary-trapped wave and hence an additional degree of freedom to the problem. Mathematically, the presence of boundary-trapped waves allows us to expand a larger collection of functions (with a uniformly convergent series) in terms of the modes. The resulting series are term-by-term differentiable and the differentiated series converges uniformly. In fact, the normal modes have the intriguing property converging pointwise to functions with finite jump discontinuities at the boundaries, a property related to their ability to expand distributions in the \cite{bretherton_critical_1966} ``$\delta$-function formulation'' of a physical problem. By considering a step-function perturbation at a dynamically-active boundary, we find that the subsequent time-evolution consists of waves whose amplitude is proportional to their projection at the dynamically-active boundary. Within the mathematical formulation is a qualitative oscillation theory relating the number of internal zeros of the eigenfunctions to physical quantities; indeed, for the quasigeostrophic problem, the number of zeros of the topographic Rossby wave depends on the propagation direction while, for the rotating Boussinesq problem, the ratio of the Coriolis parameter to the horizontal wavenumber determines at which integer $M$ we obtain two modes with $M$ zeros.

Our results also clarify the difference between the traditional quasigeostrophic baroclinic modes and the the $L^2\oplus \C^2$ eigenfunctions of \cite{smith_surface-aware_2012}. Namely, the series expansion of a function in terms of the \cite{smith_surface-aware_2012} eigenfunctions has a term-by-term derivative that converges uniformly over the whole interval regardless of the boundary conditions satisfied by the function. In contrast, an eigenfunction expansion in terms of the baroclinic modes only has this property if the function satisfies the same boundary conditions as the baroclinic modes. One consequence is the following. Suppose we expand an arbitrary quasigeostrophic state, with boundary buoyancy anomalies, in terms of the baroclinic modes. The presence of these boundary buoyancy anomalies implies that this state does not satisfy the same boundary conditions as the baroclinic modes. The resulting series expansion in term of the baroclinic modes is then not differentiable at the boundaries. We are thus unable to recover the value of the boundary buoyancy anomalies from the series expansion and so we have lost information in the expansion process.  This loss of information does not occur with  $L^2\oplus \C^2$ expansions. 

Normal mode decompositions of quasigeostrophic motion play an important role in physical oceanography \cite[e.g.,][]{wunsch_vertical_1997,lapeyre_what_2009,lacasce_prevalence_2017}. Other applications include the extension of equilibrium statistical mechanical calculations \cite[e.g.,][]{bouchet_statistical_2012,venaille_catalytic_2012} to three-dimensional systems with dynamically-active boundaries. Moreover, the mathematical framework developed here is useful for the development of weakly non-linear wave turbulence theories \cite[e.g.,][]{fu_nonlinear_1980,smith_scales_2001,scott_wave_2014} in systems with both internal and boundary-trapped waves.

\begin{subappendices}

\section{Additional properties of the eigenvalue problem}\label{S-additional-properties}

\subsection{Construction of $L^2_\mu$}\label{L2-mu-construction}

First, define the weighted Lebesgue measure $\sigma$ by
\begin{equation}
	\sigma([a,b]) = \int_a^b r \, \mathrm{d}z \quad \text{where } a,b\in[z_1,z_2].
\end{equation}
The measure $\sigma$ induces the differential element
\begin{equation}
	\mathrm{d}\sigma(z) = r(z)\,\mathrm{d}z
\end{equation}
and is the measure associated with $L^2$ [see equations \eqref{L2-equal} and \eqref{L2-inner}].

Now, for $i\in S$ [see equation \eqref{S-set}], define the pure point measure $\nu_i$ by \cite[e.g.,][section I.4, example 2]{reed_methods_1980}
\begin{equation}
	\nu_i([a,b]) = 
		\begin{cases}
			D_i^{-1} \ \ \ &\text{if } z_i \in [a,b] \\
			0 \ \ \ &\text{otherwise,}
		\end{cases}
\end{equation}
where $D_{i}$ is the combination of boundary condition coefficients given by equation \eqref{Di}. 
The pure point measure $\nu_i$ induces the differential element 
\begin{equation}
	\mathrm{d}\nu_i(z) = D_i^{-1}\, \delta(z-z_i) \, \mathrm{d}z,
\end{equation}
where $\delta(z)$ is the Dirac distribution.

Consider now the space $L^2_{\nu_i}$ of ``functions'' $\phi$ satisfying  
\begin{equation}
	\abs{ \intz \abs{\phi}^2\, \mathrm{d}\nu_i} = \abs{D_i^{-1}} \, \intz \abs{\phi}^2 \, \delta(z-z_i) \, \mathrm{d}z = \abs{D_i^{-1}} \,\abs{\phi(z_i)}^2 < \infty.
\end{equation}
Elements of $L^2_{\nu_i}$ are not functions, but rather equivalence classes of functions. Two functions, $\phi$ and $\psi$, on the interval $[z_1,z_2]$ are equivalent in $L^2_{\nu_i}$ if $\phi(z_i) = \psi(z_i)$. In particular, $L^2_{\nu_i}$ is a one-dimensional vector space and is hence isomorphic to the field of complex numbers $\C$
\begin{equation}
	L^2_{\nu_i} \cong \C.
\end{equation}

Now define the measure $\mu$ by
\begin{equation}\label{mu}
	\mu = \sigma + \sum_{i\in S} \nu_i
\end{equation}
with an induced differential element of
\begin{equation} \label{dmu-appendix}
	\mathrm{d}\mu(z) = \left[ r(z) + \sum_{i\in S} D_i^{-1} \, \delta(z-z_i) \right] \mathrm{d}z.
\end{equation}
Then $L^2_\mu$ is the space of equivalence classes of functions that are square-integrable with respect to the measure $\mu$.

Since the measures $\sigma$ and $\nu_i$, for $i\in S$, are mutually singular, we have \cite[][section II.1, example 5]{reed_methods_1980}
\begin{equation}\label{isomorphism}
	L^2_\mu \cong L^2 \oplus \sum_{i\in S} L^2_{\nu_i} \cong L^2 \oplus \C^s
\end{equation}
from which we see that $L^2_\mu$ is ``larger'' by $s$ dimensions.

\subsection{The eigenvalue problem in $L^2_\mu$}\label{S-eigen-in-L2-mu}

We construct here an operator formulation of \eqref{EigenDiff_modes}\textendash\eqref{EigenB2} as an eigenvalue problem in the Pontryagin space $L^2_\mu$.

Define the differential operator $\ell$ acting on a function $\phi$ by
\begin{equation}
	\ell \, \phi = \frac{1}{r} \, \left[  (p \, \phi')' - q \, \phi \right].
\end{equation}
We also define the following boundary operators for $i \in S$, 
\begin{align}
	\Bc_i \phi &= \left[a_i \, \phi(z_i) - b_i \, (p \, \phi')(z_i)\right]\\
	\Cc_i  \phi &= \left[c_i \, \phi(z_i) - d_i \, (p \, \phi')(z_i)\right].
\end{align}
Let $\Phi$ be an element of $L^2_\mu$, as in equation \eqref{L2_mu-element}, with boundary values $\Phi(z_i) = \Cc_i \phi$ for $i\in S$ and equal to $\phi$ elsewhere.
We then define the operator $\Lc$, acting on functions $\Phi$, by
\begin{equation}
	\Lc \, \Phi = 
	\begin{cases}
		-\ell \, \phi  \quad &\text{for } z\in(z_1,z_2)\\
		-\Bc_i \,  \phi  \quad &\text{for } z=z_i \text{ where }i\in S
	\end{cases}
\end{equation}
with a domain $D(\Lc)\subset L^2_\mu$ defined by
\begin{equation}
\begin{aligned}
	D(\Lc) = \{ \Phi \in L^2_\mu \ | & \ \phi \text{ is continuously differentiable, } \ell \, \phi \in L^2,   \ \Phi(z_i) = \Cc_i \, \phi \\\text{ for } i\in S  & \text{ and } \Bc_i \phi = 0 \text{ for } i \in \{1,2\}\setminus S  \}.
\end{aligned}
\end{equation}
Recall that $S$ contains indices of the $\lambda$-dependent boundary conditions, and therefore, $\{1,2\}\setminus S$ contains the indices of the $\lambda$-independent boundary conditions. 

Then, on the subspace $D(\Lc)$ of $L^2_\mu$, the eigenvalue problem \eqref{EigenDiff_modes}\textendash\eqref{EigenB2} may be written as 
\begin{equation}\label{Op-Eigenproblem-app}
	\Lc \, \Phi = \lambda \, \Phi.
\end{equation} 
As shown in \cite{russakovskii_operator_1975,russakovskii_matrix_1997}, $\Lc$ is a self-adjoint operator in the space $L^2_\mu$.

There is a natural quadratic form $Q$, induced by the eigenvalue problem \eqref{EigenDiff_modes}\textendash\eqref{EigenB2}, given by
\begin{equation}\label{quadratic_form}
	Q(\Phi,\Psi) = \inner{\Phi}{\Lc \, \Psi}.
\end{equation}
For elements $\Phi,\Psi \in D(\Lc)$, we obtain
\begin{equation}
\begin{aligned}
		Q(\Phi,\Psi)  &= \intz \left[p \, \ov{\phi'} \, \psi' + q \, \ov{\phi} \, \psi \right] \, \mathrm{d}z 
	+ \sum_{i\in\{1,2\}\setminus S} (-1)^{i+1} \, \frac{a_i}{b_i} \,\ov{\phi(z_i)} \, \psi(z_i) \, \\
	&\quad - \sum_{i\in S} \frac{1}{D_i}
	\ov{
	\left(
	\begin{matrix}
		\psi(z_i) \\
		-(p \, \psi')(z_i)
	\end{matrix}
	\right)
	} \cdot
	\left(
	\begin{matrix}
		a_i \, c_i & a_i \, d_i \\
		a_i \, d_i & b_i \, d_i
	\end{matrix}
	\right)
	\left(
	\begin{matrix}
		\phi(z_i) \\
		-(p \, \phi')(z_i)
	\end{matrix}
	\right)
\end{aligned}
\end{equation}
for $b_i \neq 0$ for $i \in \{1,2\} \setminus S$. If $b_i=0$ for $i \in \{1,2\} \setminus S$ then we replace the term $a_i/b_i$ with zero. 

To develop the reality and completeness theorem \ref{real-complete}, we provide the following definitions.

\begin{definition}[Right-definite]
	The eigenvalue problem \eqref{EigenDiff_modes}\textendash\eqref{EigenB2} is said to be right-definite if $L^2_\mu$ is a Hilbert space or, equivalently, if
	\begin{equation}
		\inner{\Phi}{\Phi} > 0 
	\end{equation} 
	for all non-zero $\Phi \in L^2_\mu$.
\end{definition}

\begin{definition}[Left-definite]
	The eigenvalue problem \eqref{EigenDiff_modes}\textendash\eqref{EigenB2} is said to be left-definite if 
	\begin{equation}\label{left-definite-inequality}
		Q(\Phi,\Phi) \geq 0
	\end{equation}
	for all $\Phi \in D(\Lc)$.
\end{definition}

One can then prove propositions \ref{right-definite-prop} and \ref{left-definite-prop} through straightforward manipulations.

\subsection{Properties of eigenfunction expansions}\label{S-more-eigen-expansions}

The following theorem features some of the novel properties of the basis $\{\Phi_n\}_{n=0}^\infty$ of $L^2_\mu$. Theorem \ref{expansion-theorem} below is a generalization of a theorem first formulated, in the right-definite case, by \citet{walter_regular_1973} and \cite{fulton_two-point_1977}.

\begin{theorem}[Eigenfunction expansions]\label{expansion-theorem} Let $\{\Phi_n\}_{n=0}^\infty$ be the set of eigenfunctions of the eigenvalue problem \eqref{EigenDiff_modes}\textendash\eqref{EigenB2}. Then the following properties hold.
\begin{itemize}
	\item[(i)] Null series:  For $i \in S$, we have 
	\begin{equation}\label{null}
		0 = D_i^{-1} \, \sum_{n=0}^\infty \frac{1}{\inner{\Phi_n}{\Phi_n}} \, \Phi_n(z_i) \, \phi_n(z) 
	\end{equation}
	with equality in the sense of $L^2$.
	\item[(ii)] Unit series: For $i \in S$, we have
	\begin{equation}\label{unit}
		1 = D_i^{-1} \, \sum_{n=0}^\infty \frac{1}{\inner{\Phi_n}{\Phi_n}} \, \abs{\Phi_n(z_i)}^2.
	\end{equation}
	\item[(iii)] $L^2$-expansion: Let $\psi \in L^2$, then
	\begin{equation}\label{L2-expansion}
		\psi =  \sum_{n=0}^\infty \frac{1}{\inner{\Phi_n}{\Phi_n}} \,\left(\intz \ov{\psi} \, \phi_n \, r \, \mathrm{d}z \right) \, \phi_n.
	\end{equation}
	with equality in the sense of $L^2$.
	\item[(iv)] Interior-boundary orthogonality: Let $\psi \in L^2$, then for $i \in S$, we have
	\begin{equation}\label{int-bound-orth}
		0  = \sum_{n=0}^\infty \frac{1}{\inner{\Phi_n}{\Phi_n}}  \left(\intz \ov{\psi} \, \phi_n \, r \, \mathrm{d}z \right) \Phi_n(z_i).
	\end{equation}
\end{itemize}
\end{theorem}
\begin{proof}
	The proof is similar to the proof of corollary 1.1 in \cite{fulton_two-point_1977}.
\end{proof}

\subsection{Pointwise convergence and Sturm-Liouville series}\label{S-pointwise}

Theorem 3 in \cite{fulton_two-point_1977} states that the $\Phi_n$ series expansion \eqref{expansion} behaves like a Fourier series in the interior of the interval $(z_1,z_2)$ (see appendix \ref{S-math-app} for why this theorem applies in the left-definite case). Since the expansions \eqref{expansion} and \eqref{expansion-phi} in terms of $\Phi_n$ and $\phi_n$ are equal in the interior, then the above theorem applies to the $\phi_n$ series \eqref{expansion-phi} as well. It is at the boundaries points, $z=z_1,z_2$, where the novel behaviour of the series expansions \eqref{expansion} and \eqref{expansion-phi} appears.

For traditional Sturm-Liouville expansions [with eigenfunctions of problem \eqref{EigenDiff_modes}-\eqref{EigenB2} with $c_i,d_i = 0$ for $i=1,2$], eigenfunction expansions behave like the analogous Fourier series on $[z_1,z_2]$ [page 16 in \cite{titchmarsh_eigenfunction_1962}  or chapter 1, section 9, in \cite{levitan_introduction_1975}]. In particular, for a twice continuously differentiable function $\psi$, the eigenfunction expansion of $\psi$ converges uniformly to $\psi$ on $[z_1,z_2]$ so long as the eigenfunctions $\phi_n$ do not vanish at the boundaries. If the eigenfunctions vanish at one of the boundaries, then we only obtain uniform convergence if $\psi$ vanishes at the corresponding boundary as well \citep[][section 22]{brown_fourier_1993}. Under these conditions, the resulting expansion will be differentiable in the interior of the interval, $(z_1,z_2)$, but not at the boundaries $z=z_1,z_2$ [see chapter 8, section 3, in \citet{levitan_introduction_1975} for the equiconvergence of differentiated Sturm-Liouville series with Fourier series and see section 23 in \citet{brown_fourier_1993} for the convergence behaviour of differentiated Fourier series].

Returning to the case of eigenfunction expansions for the eigenvalue problem \eqref{EigenDiff_modes}\textendash\eqref{EigenB2} with $\lambda$-dependent boundaries, the following theorem provides pointwise (as well as uniform, in the case $d_i\neq 0$) convergence conditions for the $\phi_n$ series \eqref{expansion-phi}.

\begin{theorem}[Pointwise convergence]\label{pointwise}
	Let $\psi$ be a twice continuously differentiable function on the interval $[z_1,z_2]$ satisfying any $\lambda$-independent boundary conditions in the eigenvalue problem \eqref{EigenDiff_modes}\textendash\eqref{EigenB2}. Define the function $\Psi$ on $[z_1,z_2]$ by
	\begin{equation}
		\Psi(z) = 
		\begin{cases}
			\Psi(z_i) \quad &\textrm{at } z=z_i, \textrm{ for } i \in S,\\
			\psi(z) \quad &\textrm{otherwise}.
		\end{cases}
	\end{equation}
	where $\Psi(z_i)$ are constants for $i\in S$ (the $\lambda$-dependent boundaries). Then we have the following.
	\begin{itemize}
		\item [(i)] If $d_i \neq 0$ for $i \in S$, then the $\phi_n$ series expansion \eqref{expansion-phi} converges uniformly to $\psi(z)$ on the closed interval $[z_1,z_2]$,
		\begin{equation}
			\sum_{n=0}^\infty \frac{\inner{\Psi}{\Phi_n}}{\inner{\Phi_n}{\Phi_n}} \, \phi_n(z) = \psi(z).
		\end{equation} 
		 Furthermore, for the differentiated series, we have
			\begin{equation}
			\sum_{n=0}^\infty \frac{\inner{\Psi}{\Phi_n}}{\inner{\Phi_n}{\Phi_n}} \, \phi_n'(z) = 
			\begin{cases}
				\left( c_i \, \psi(z_i)- \Psi(z_i) \right)/ d_i \quad &\textrm{at } z=z_i, \textrm{ for } i \in S\\
				\psi'(z) \quad &\textrm{otherwise}.
			\end{cases}
			\end{equation} 
		\item [(ii)] If $d_i =0$, then we have
			\begin{equation}
				\sum_{n=0}^\infty \frac{\inner{\Psi}{\Phi_n}}{\inner{\Phi_n}{\Phi_n}} \, \phi_n=
				\begin{cases}
					\Psi(z_i)/c_i \quad &\textrm{at } z=z_i, \textrm{ for } i \in S\\
					\psi(z) \quad &\textrm{otherwise}.
				\end{cases}
			\end{equation}
	\end{itemize}
\end{theorem}
\begin{proof}
	This theorem is a generalization of corollary 2.1 in \cite{fulton_two-point_1977}. We provide the extension of the corollary to the left-definite problem in appendix \ref{S-A-left-fulton}.
\end{proof}

The $\Phi_n$ series \eqref{expansion} converges to $\Psi(z_i)$ at $z=z_i$ 
for $i \in S$ (i.e., at $\lambda$-dependent boundaries) but otherwise behaves as in theorem \ref{pointwise}. 

\section{Literature survey and mathematical proofs}\label{S-math-app}

\subsection{Literature survey}\label{S-lit-review}

There is an extensive literature associated with the eigenvalue problem \eqref{EigenDiff_modes}\textendash\eqref{EigenB2} with $\lambda$-dependent boundary conditions \citep[see][ and citations within]{schafke_s-hermitesche_1966,fulton_two-point_1977}. One can use the $S$-hermitian theory of \cite{schafke_s-hermitesche_1965,schafke_s-hermitesche_1966,schafke_s-hermitesche_1968} to show that one obtains real eigenvalues when the problem is either right-definite or left-definite (see section \ref{math-section}) but completeness results in $L^2_\mu$ are unavailable in this theory.

The right-definite theory is well-known \citep{evans_non-self-adjoint_1970,walter_regular_1973,fulton_two-point_1977}. In particular, \cite{fulton_two-point_1977} applies the residue calculus techniques of \cite{titchmarsh_eigenfunction_1962} to the right-definite problem and, in the process, extends some well-known properties of Fourier series to eigenfunction expansions associated with \eqref{EigenDiff_modes}\textendash\eqref{EigenB2}. A recent Hilbert space approach to the right-definite problem, in the context of obtaining a projection basis for quasigeostrophic dynamics, is given by \cite{smith_surface-aware_2012}.

The left-definite problem is less examined. As we show in this chapter, the eigenvalue problem is naturally formulated in a Pontryagin space, and, in such a setting, one can prove, in the left-definite case, that the eigenvalues are real and that the eigenfunctions form a basis for the underlying function space. We prove this result, stated in theorem \ref{real-complete}, in appendix \ref{S-real-proof}.

With these completeness results, we may apply the residue calculus techniques of \cite{titchmarsh_eigenfunction_1962} to extend the results of \cite{fulton_two-point_1977} to the left-definite problem. Indeed, \cite{fulton_two-point_1977} uses a combination of Hilbert space methods as well as residue calculus techniques to prove various convergence results for the right-definite problem. However, only theorem 1 of \cite{fulton_two-point_1977} makes use of Hilbert space methods. If we extend Fulton's theorem 1 to the left-definite problem, then all the results of \cite{fulton_two-point_1977} will apply equally to the left-definite problem. A left-definite analogue of theorem 1 of \cite{fulton_two-point_1977}, along with its proof, is given in appendix \ref{S-A-left-fulton}.

\subsection{A Pontryagin space theorem} \label{S-A-Pontryagin}

A Pontryagin space $\Pi_\kappa$, for a finite non-negative integer $\kappa$, is a Hilbert space with a $\kappa$-dimensional subspace of elements satisfying 
\begin{align}
	\inner{\phi}{\phi}<0.
\end{align}
An introduction to the theory of Pontryagin spaces can be found in \cite{iohvidov_spectral_1960} as well as in the monograph of \cite{bognar_indefinite_1974}. Another resource is the monograph of \cite{azizov_linear_1989} on linear operators in indefinite inner product spaces.

Pontryagin spaces admit a decomposition 
\begin{equation}
	\Pi_{\kappa} = \Pi^+ \oplus \Pi^-
\end{equation}
into orthogonal subspaces $(\Pi^+, +\inner{\cdot}{\cdot})$ and $(\Pi^-, -\inner{\cdot}{\cdot})$. Moreover, one can associate with a Pontryagin space $(\Pi_\kappa,\inner{\cdot}{\cdot})$ a corresponding Hilbert space $(\Pi, \inner{\cdot}{\cdot{}}_+)$ where the positive-definite inner product $\inner{\cdot}{\cdot}_+$ is defined by
\begin{equation}\label{induced_hilbert_inner}
	\inner{\phi}{\psi}_+ = \inner{\phi_+}{\psi_+} - \inner{\phi_-}{\psi_-},  \quad \phi,\psi \in \Pi,
\end{equation}
where $\phi=\phi_++\phi_-$ and $\psi=\psi_++\psi_-$, with $\phi_{\pm},\psi_{\pm}\in \Pi^{\pm}$ \citep{azizov_linear_1981}. 

As a prerequisite to proving theorem \ref{real-complete}, we require the following.

\begin{theorem}[Positive compact Pontryagin space operators]\label{positive-compact}

	Let $\Ac$ be a positive compact operator in a Pontryagin space $\Pi_\kappa$ and suppose that $\lambda=0$ is not an eigenvalue. Then all eigenvalues are real and the corresponding eigenvectors form an orthonormal basis for $\Pi_\kappa$. There are precisely $\kappa$ negative eigenvalues and the remaining eigenvalues are positive. Moreover, positive eigenvalues have positive eigenvectors and negative eigenvalues have negative eigenvectors.
\end{theorem}
\begin{proof}
	By theorem VII.1.3 in \citet{bognar_indefinite_1974} the eigenvalues are all real. Moreover, since $\lambda=0$ is not an eigenvalue, then all eigenspaces are definite \citep[][theorem VII.1.2]{bognar_indefinite_1974} and hence all eigenvalues are semi-simple \citep[][lemma II.3.8]{bognar_indefinite_1974}.

	Since $\Ac$ is a compact operator and $\lambda=0$ is not an eigenvalue, then the span of the generalized eigenspaces is dense in $\Pi_\kappa$ \cite[][lemma 4.2.14]{azizov_linear_1989}. Since all eigenvalues are semi-simple, then all generalized eigenvectors are eigenvectors and so the span of the eigenvectors is dense in $\Pi_\kappa$. Orthogonality of eigenvectors can be shown as in a Hilbert space.
	
	Let $\lambda$ be an eigenvalue and $\phi$ the corresponding eigenvector. By the positivity of $\Ac$, we have
	\begin{equation}
		\inner{\Ac \, \phi}{\phi} = \lambda \inner{\phi}{\phi} \geq 0.
	\end{equation}
	Since all eigenspaces are definite, it follows that positive eigenvectors must correspond to positive eigenvalues and negative eigenvectors must correspond to negative eigenvalues.
	
	Finally, by theorem IX.1.4 in \cite{bognar_indefinite_1974}, any dense subset of $\Pi_\kappa$ must contain a negative-definite $\kappa$ dimensional subspace. Consequently, there are $\kappa$ negative eigenvectors and hence $\kappa$ negative eigenvalues.
\end{proof}

\subsection{Proof of theorem \ref{real-complete}} \label{S-real-proof}

\begin{proof}
The proof for the left-definite case is essentially the standard proof \citep[e.g.,][section 5.10]{debnath_introduction_2005} with theorem \ref{positive-compact} substituting for the Hilbert-Schmidt theorem. We give a general outline nonetheless.

 First, it is well-known that $\Lc$ is self-adjoint in $L^2_\mu$ \citep[e.g.,][]{russakovskii_operator_1975,russakovskii_matrix_1997}. Since $\lambda =0$ is not an eigenvalue, then the inverse operator $\Lc^{-1}$ exists and is an integral operator on $L^2_\mu$. For an explicit construction, see section 4 in \citet{walter_regular_1973}, \cite{fulton_two-point_1977}, and \cite{hinton_expansion_1979}. The eigenvalue problem for $\Lc$, equation \eqref{Op-Eigenproblem}, is then equivalent to
\begin{equation}
	\Lc^{-1} \, \phi = \lambda^{-1} \, \phi
\end{equation}
and both problems have the same eigenfunctions.

The operator $\Lc^{-1}$ is a positive compact operator and so satisfies the requirements of theorem \ref{positive-compact}. Application of theorem \ref{positive-compact} to $\Lc^{-1}$ then assures that all eigenvalues $\lambda_n$ are real, the eigenfunctions form an orthonormal basis for $L^2_\mu$, and the sequence of eigenvalues $\{\lambda_n\}_{n=0}^\infty$ is countable and bounded from below.

The claim that the eigenvalues are simple is verified in \citet{binding_left_1999} for the left-definite problem. Alternatively, an argument similar to that of \cite{fulton_two-point_1977} and \citep[][page 12]{titchmarsh_eigenfunction_1962} can be made to prove the simplicity of the eigenvalues.
\end{proof}

\subsection{Extending Fulton (1977) to the left-definite problem}\label{S-A-left-fulton}

The following is a left-definite analogue of theorem 1 in \cite{fulton_two-point_1977}. The proof is almost identical to the right-definite case \citep{fulton_two-point_1977, hinton_expansion_1979} with minor modifications. Essentially, since $\inner{\Psi}{\Psi}$ can be negative, we must replace these terms in the inequalities below with the induced Hilbert space inner product $\inner{\Psi}{\Psi}_+$ given by equation \eqref{induced_hilbert_inner}. Our $L^2_\mu$ Green's functions $G$ corresponds to $\tilde G$ in \citet{hinton_expansion_1979}.

\begin{theorem}[A left-definite extension of Fulton's theorem 1]
Let $\Psi \in L^2_\mu$ be defined on the interval $[z_1,z_2]$ by
	\begin{equation}
		\Psi(z) = 
		\begin{cases}
			\Psi(z_i) \quad &\textrm{at } z=z_i, \textrm{ for } i \in S,\\
			\psi(z) \quad &\textrm{otherwise},
		\end{cases}
	\end{equation}
where $\psi \in L^2$ and $\Psi(z_i)$ are constants for $i \in S$. The eigenfunctions $\Phi_n$ are defined similarly (see section \ref{math-section}). 
\begin{itemize}
	\item[(i)] Parseval formula: For $\Psi \in L^2_\mu$, we have
	\begin{equation}\label{parseval-formula}
		\inner{\Psi}{\Psi} = \sum_{n=0}^\infty \frac{\abs{\inner{\Psi}{\Phi_n}}^2}{\inner{\Phi_n}{\Phi_n}}.
	\end{equation}
	\item[(ii)] For $\Psi \in D(\Lc)$, we have
	\begin{equation} \label{expansion-fulton-theorem1}
		\Psi = \sum_{n=0}^\infty \frac{\inner{\Psi}{\Phi_n}}{\inner{\Phi_n}{\Phi_n}}  \Phi_n.
	\end{equation}
	with equality in the sense of $L^2_\mu$. Moreover, we have 
	\begin{equation}\label{expansion-int-series}
		\psi = \sum_{n=0}^\infty \frac{\inner{\Psi}{\Phi_n}}{\inner{\Phi_n}{\Phi_n}}  \phi_n,
	\end{equation}
	which converges uniformly and absolutely for $z \in[z_1,z_2]$ and may be differentiated term-by-term, with the differentiated series converging uniformly and absolutely to $\psi'$ for $z \in [z_1,z_2]$. The boundaries series
	\begin{equation}\label{boundary_series_appendix}
		\Psi(z_i)  = \sum_{n=0}^\infty \frac{\inner{\Psi}{\Phi_n}}{\inner{\Phi_n}{\Phi_n}}  \Phi_n(z_i),
	\end{equation}
	for $i \in S$, is absolutely convergent.
\end{itemize}
\end{theorem}

\begin{proof}
	The Parseval formula \eqref{parseval-formula} is a consequence of the completeness of the eigenfunctions $\{\Phi_n\}_{n=0}^\infty$ in $L^2_\mu$, given by theorem \ref{real-complete}, and theorem IV.3.4 in \cite{bognar_indefinite_1974}. Similarly, the expansion \eqref{expansion-fulton-theorem1} is also due to completeness of the eigenfunctions. 
	
	We first prove that the series  \eqref{expansion-int-series} converges uniformly and absolutely for $z \in[z_1,z_2]$. We begin with the identity
	\begin{equation}\label{phi-green}
		\phi_n(z) = (\lambda - \lambda_n) \inner{G(z,\cdot,\lambda)}{\Phi_n}
	\end{equation} 
	where $\lambda \in \C$ is not an eigenvalue of $\Lc$, and $G$ is the $L^2_\mu$ Green's function [see equation (8) in \cite{hinton_expansion_1979}]. Then 
	\begin{equation}\label{green_convergence}
		\sum_{n=0}^\infty \lambda_n \frac{\abs{\phi_n}^2}{\abs{\lambda -\lambda_n}^2} = \sum_{n=0}^\infty \lambda_n \abs{\inner{G(z,\cdot,\lambda)}{\Phi_n}}^2 \leq \inner{G(z,\cdot,\lambda)}{\Lc G(z,\cdot,\lambda)}_+ \leq B_1(\lambda)
	\end{equation}
	where $\inner{\cdot}{\cdot}_+$ is the induced Hilbert space inner product given by equation \eqref{induced_hilbert_inner} and $B_1(\lambda)$ is a $z$ independent upper bound \cite[equation 9 in][]{hinton_expansion_1979}. In addition, since $\Psi \in D(\Lc)$, then $\inner{\Lc \Psi}{\Lc \Psi}_+ < \infty$. Thus, we obtain 
	\begin{equation}\label{Lc_psi_2_convergence}
		\sum_n \lambda_n^2 \abs{\inner{\Psi}{\Phi_n}}^2 = \inner{\Lc \Psi}{\Lc \Psi}_+ < \infty.
	\end{equation}
	
	The uniform and absolute convergence of \eqref{expansion-int-series} follows from
	\begin{align}
		\sum_{n=0}^\infty \abs{\frac{\inner{\Psi}{\Phi_n}}{\inner{\Phi_n}{\Phi_n}}  \phi_n} &= \sum_{n=0}^\infty \abs{ \left( \frac{\phi_n}{\lambda - \lambda_n}\right) \left(\lambda - \lambda_n \right) \frac{\inner{\Psi}{\Phi_n} }{\inner{\Phi_n}{\Phi_n}}} \\
			&\leq \sqrt{\left(\sum_{n=0}^\infty \abs{\frac{\phi_n}{\lambda-\lambda_n}}^2\right) \left( \sum_{n=0}^\infty \abs{\lambda-\lambda_n}^2 \abs{\inner{\Psi}{\Phi_n}}^2 \right)}
	\end{align}
	along with equations \eqref{green_convergence} and \eqref{Lc_psi_2_convergence}. The absolute convergence of the boundary series \eqref{boundary_series_appendix} follows as well.
	
	To show that the series \eqref{expansion-int-series} is term-by-term differentiable, it is sufficient to show that the differentiated series converges uniformly for $z \in [z_1,z_2]$ \citep[][section 6.14, theorem 33]{kaplan_advanced_1993}. The proof of the unform convergence of the differentiated series follows from the identity \citep{hinton_expansion_1979}
	\begin{equation}
		\frac{\phi_n'}{\lambda - \lambda_n} = \frac{\mathrm{d}}{\mathrm{d}z} \inner{G(z,\cdot,\lambda)}{\Phi_n} = \inner{\partial_z G(z,\cdot,\lambda)}{\Phi_n}.
	\end{equation}
	and a similar argument.
	
\end{proof}

\end{subappendices}

%% file: 5-QG_Modes/QG_Modes.tex
\chapter{On the Discrete Normal Modes of Quasigeostrophic Theory}\label{Ch-QG}

\begin{abstractchapter}
	The discrete baroclinic modes of quasigeostrophic theory are incomplete and the incompleteness manifests as a loss of information in the projection process. The incompleteness of the baroclinic modes is related to the presence of two previously unnoticed stationary step-wave solutions of the Rossby wave problem with flat boundaries. These step-waves are the limit of surface quasigeostrophic waves as boundary buoyancy gradients vanish. A complete normal mode basis for quasigeostrophic theory is obtained by considering the traditional Rossby wave problem with prescribed buoyancy gradients at the lower and upper boundaries. The presence of these boundary buoyancy gradients activates the previously inert boundary degrees of freedom. These Rossby waves have several novel properties such as the presence of multiple modes with no internal zeros, a finite number of modes with negative norms, and their vertical structures form a basis capable of representing any quasigeostrophic state with a differentiable series expansion. Using this complete basis, we are able to obtain a series expansion to the potential vorticity of Bretherton (with Dirac delta contributions). We also examine the quasigeostrophic vertical velocity modes and derive a complete basis for such modes as well. A natural application of these modes is the development of a weakly non-linear wave-interaction theory of geostrophic turbulence that takes topography into account.
\end{abstractchapter}

\section{Introduction}

\subsection{Background}

The vertical decomposition of quasigeostrophic motion into normal modes plays an important role in bounded stratified geophysical fluids \citep[e.g.,][]{charney_geostrophic_1971,flierl_models_1978,fu_nonlinear_1980,wunsch_vertical_1997, chelton_geographical_1998, smith_scales_2001, tulloch_quasigeostrophic_2009, lapeyre_what_2009, Ferrari2010a,ferrari_distribution_2010, de_la_lama_vertical_2016, lacasce_prevalence_2017,brink_structure_2019}. Most prevalent are the traditional baroclinic modes \cite[e.g., section 6.5.2 in][]{Vallis2017} that are the vertical structures of Rossby waves in a quiescent ocean with no topography or boundary buoyancy gradients. In a landmark contribution, \cite{wunsch_vertical_1997} partitions the ocean's kinetic energy into the baroclinic modes and finds that the zeroth and first baroclinic modes dominate over most of the extratropical ocean. Additionally, \cite{wunsch_vertical_1997} concludes that the surface signal primarily reflects the first baroclinic mode and, therefore, the motion of the thermocline.

However, the use of  baroclinic modes has come under increasing scrutiny in recent years \citep{lapeyre_what_2009,roullet_properties_2012,scott_assessment_2012,smith_surface-aware_2012}. \cite{lapeyre_what_2009} observes that the vertical shear of the baroclinic modes vanishes at the boundaries, thus leading to the concomitant vanishing of the boundary buoyancy. Consequently, \cite{lapeyre_what_2009} proposes that the baroclinic modes cannot be complete\footnote{A collection of functions is said to be complete in some function space, $\mathcal{F}$, if this collection forms a basis of $\mathcal{F}$. Specifying the underlying function space, $\mathcal{F}$, turns out to be crucial, as we see in section \ref{S-vertical-phase}.}  due to their inability to represent boundary buoyancy. To supplement the baroclinic modes, \cite{lapeyre_what_2009} includes a boundary-trapped exponential surface quasigeostrophic solution \cite[see][]{held_surface_1995} and suggests that the surface signal primarily reflects, not thermocline motion, but boundary-trapped surface quasigeostrophic dynamics \cite[see also][]{lapeyre_surface_2017}.

Appending additional functions to the collections of normal modes as in \cite{lapeyre_what_2009} or \cite{scott_assessment_2012} does not result in a set of normal modes since the appended functions are not  orthogonal to the original modes. It is only with \cite{smith_surface-aware_2012} that a set of normal modes capable of representing arbitrary surface buoyancy is derived.

Yet it is not clear how the normal modes of \cite{smith_surface-aware_2012} differ from the baroclinic modes or what these modes correspond to in linear theory. Indeed, \cite{rocha_galerkin_2015}, noting that the baroclinic series expansion of any sufficiently smooth function converges uniformly to the function itself,  argues that the incompleteness of the baroclinic modes has been ``overstated''. Moreover, \cite{de_la_lama_vertical_2016} and \cite{lacasce_prevalence_2017}, motivated by the observation that the leading empirical orthogonal function of \cite{wunsch_vertical_1997} vanishes near the ocean bottom, propose an alternate set of modes\textemdash the surface modes\textemdash that have a vanishing pressure at the bottom boundary.

We thus have a variety of proposed normal modes and it is not clear how their properties differ. Are the baroclinic modes actually incomplete? What about the surface modes? What does completeness mean in this context? 
The purpose of this paper is to  answer these questions. 

\subsection{Normal modes and eigenfunctions}

A normal mode is a linear motion in which all components of a system move coherently at a single frequency. Mathematically, a normal mode has the form
\begin{equation}\label{physical-mode}
	 \Phi_a( x,y,z) \, \mathrm{e}^{-\mathrm{i} \omega_a t},
\end{equation} 
where $\Phi_a$ describes the spatial structure of the mode and $\omega_a$ is its angular frequency. The function $\Phi_a$ is obtained by solving a differential eigenvalue problem and hence is an eigenfunction. The collection of all eigenfunctions forms a basis of some function space relevant to the problem. 

By an abuse of terminology, the spatial structure, $\Phi_a$, is often called a normal mode (e.g., the term ``Fourier mode'' is often used for $\mathrm{e}^{\mathrm{i} k \, x}$ where $k$ is a wavenumber). In linear theory, this misnomer is often benign as each $\Phi_a$ corresponds to a frequency $\omega_a$. For example, given some initial condition $\Psi(x,y,z)$, we decompose $\Psi$ as a sum of modes at $t=0$,
\begin{equation}\label{sum-of-eigen}
	\Psi(x,y,z) = \sum_{a}  c_a \, \Phi_a(x,y,z),
\end{equation}
where the $c_a$ are the Fourier coefficients, and the time evolution is then given by
\begin{equation}\label{sum-of-eigen-time}
	\sum_{a} c_a \, \Phi_a(x,y,z) \, \mathrm{e}^{-\mathrm{i} \omega_a t}.
\end{equation}

However, with non-linear dynamics, this abuse of terminology can be confusing. Given some spatial structure, $\Psi(x,y,z)$, in a fluid whose flow is non-linear, we can still exploit the basis properties of the eigenfunctions $\Phi_a$ to decompose $\Psi$ as in equation \eqref{sum-of-eigen}. Whereas in a linear fluid only wave motion of the form \eqref{physical-mode} is possible, a non-linear flow admits a larger collection of solutions (e.g., non-linear waves and coherent vortices) and so the linear wave solution \eqref{sum-of-eigen-time} no longer follows from the decomposition \eqref{sum-of-eigen}. 

For this reason, we call the linear solution \eqref{physical-mode} a \emph{physical} normal mode to distinguish it from the spatial structure $\Phi_a$, which is only an eigenfunction. Otherwise, we will use the terms ``normal mode'' and ``eigenfunction'' interchangeably to refer to the spatial structure $\Phi_a$, as is prevalent in the literature.

Our strategy here is then the following. We find the \emph{physical} normal modes [of the form \eqref{physical-mode}] to various Rossby wave problems  and examine the basis properties of their constituent eigenfunctions $\Phi_a$. Our goal is to find a collection of eigenfunctions (i.e., ``normal modes'' in the prevalent terminology) capable of representing every possible quasigeostrophic state.

\subsection{Contents of this chapter}

This chapter constitutes an examination of all collections of discrete (i.e., non-continuum\footnote{Continuum modes appear once a sheared mean-flow is present, e.g., \cite{drazin_rossby_1982}, \cite{balmforth_normal_1994,balmforth_singular_1995}, and \cite{brink_structure_2019}. }) quasigeostrophic normal modes. We include the baroclinic modes, the surface modes of \cite{de_la_lama_vertical_2016} and \cite{lacasce_prevalence_2017}, the surface-aware mode of \cite{smith_surface-aware_2012}, as well as various generalizations. To study the completeness of a set of normal modes, we must first define the underlying space in question. From general considerations, we introduce in section \ref{S-phase} the quasigeostrophic phase space, defined as the space of all possible quasigeostrophic states. Subsequently, in section \ref{S-linear} we use the general theory of differential eigenvalue problems with eigenvalue dependent boundary conditions, as developed in chapter \ref{Ch-modes}, to study Rossby waves in an ocean with prescribed boundary buoyancy gradients (e.g., topography, see section \ref{SS-pv}). Intriguingly, in an ocean with no topography, we find that, in addition to the usual baroclinic modes, there are two additional stationary step-mode solutions that have not been noted before. The stationary step-modes are the limits of boundary-trapped surface quasigeostrophic waves as the boundary buoyancy gradient vanishes. 

Our study of Rossby waves then leads us  examine  all possible discrete collections of normal modes in section \ref{S-expansions}.
As shown in this section,  the baroclinic modes are incomplete, as argued by \cite{lapeyre_what_2009}, and we point out that the incompleteness leads to a \emph{loss} of information after projecting a function onto the baroclinic modes.
In contrast, modes such as those suggested by \cite{smith_surface-aware_2012} are complete in the quasigeostrophic phase space so that projecting a function onto such modes provides an \emph{equivalent} representation of the function. 

We offer discussion of our analysis in Section \ref{S-discussion} and conclusions in Section \ref{S-conclusion}. Appendix A summarizes the key mathematical results pertaining to eigenvalue 
problems where the eigenvalue appears in the boundary conditions.  Appendix B then summarizes the polarization relations as well as the vertical velocity eigenvalue problem.

\section{Mathematics of the quasigeostrophic phase space}\label{S-phase}

\subsection{The potential vorticity}\label{SS-pv}

Consider a three-dimensional region $\Dc$ of the form 
\begin{equation}\label{region}
	\Dc = \Dc_0 \times \left[z_1,z_2\right].
\end{equation}
The area of the lower and upper boundaries is denoted by $\Dc_0$ and is a rectangle of area $A$ while $z_1$ (lower boundary) and $z_2$ (upper boundary) are constants. The horizontal boundaries are either rigid or periodic. 

The state of a quasigeostrophic fluid in $\Dc$ is determined by a charge-like quantity known as the quasigeostrophic potential vorticity \citep{hoskins_use_1985,schneider_boundary_2003}. If the potential vorticity is distributed throughout the three-dimensional region $\Dc$, we are concerned with the volume potential vorticity density, $Q$, with $Q$  related to the geostrophic streamfunction $\psi$ by [e.g., section 5.4 of \cite{Vallis2017}]
\begin{equation}\label{q-psi}
	Q = f + \lap \psi + \pd{}{z} \left(\frac{f_0^2}{N^2} \pd{\psi}{z} \right).
\end{equation}
Here, the latitude dependent Coriolis parameter is 
\begin{equation}
	f = f_0 + \beta\,y,
\end{equation}
$N(z)$ is the prescribed background buoyancy frequency,  $\nabla^{2}$ is the horizontal Laplacian operator, and 
\begin{equation}
\vec u = \unit z \times \grad \psi
\end{equation}
is the horizontal geostrophic velocity, $\vec u = (u,v)$. 

Additionally, the potential vorticity may be distributed over a two-dimensional region, say the lower and upper boundaries $\Dc_0$, to obtain surface potential vorticity densities $R_1$ and $R_2$.  The surface potential vorticity densities are related to the streamfunction by
\begin{equation}\label{r-psi}
	R_j = (-1)^{j+1} \left[ g_j + \left( \frac{f_0^2}{N^2} \, \pd{\psi}{z}\right)\Bigg |_{z=z_j}  \right]
\end{equation}
where $g_j$ is an imposed surface potential vorticity density at the lower or upper boundary and $j=1,2$. The density $g_j$ corresponds to a prescribed buoyancy 
\begin{equation}
 b_j =	\frac{N^2}{f_0} g_j
\end{equation}
at the $j$th boundary [see equation \eqref{buoyancy-boundary-evolution}]. Alternatively, $g_j$ may be thought of as an infinitesimal topography through
\begin{equation}\label{g_j_topography}
	g_j = f_0 h_j
\end{equation}
where $h_j$ represents infinitesimal topography at the $j$th boundary. Whereas $Q$ has dimensions of inverse time, $R_{j}$ has dimensions of length per time.


\subsection{Defining the quasigeostrophic phase space}

We define the quasigeostrophic phase space to be the space of all possible quasigeostrophic states, with a quasigeostrophic state  determined by the potential vorticity densities, $Q, R_1$, and $R_2$. Note that the volume potential vorticity density, $Q$, is defined throughout the whole fluid region $\Dc$, so that $Q = Q(x,y,z,t)$. In contrast, the surface potential vorticity densities, $R_1$ and $R_2$, are only defined on the two-dimensional lower and upper boundary surfaces, $\Dc_0$, so that $R_j = R_j(x,y,t)$.

It is useful to restate the previous paragraph with some added mathematical precision. For that purpose, let $L^2[\Dc]$ be the space of square-integrable functions\footnote{The definition of $L^2[\Dc]$ is   more subtle than presented here. Namely, elements of $L^2[\Dc]$ are not functions, but rather equivalence classes of functions leading to the unintuitive properties seen in this section. See chapter \ref{Ch-modes} and citations within for more details.} in the fluid volume $\Dc$, and let $L^2[\Dc_0]$ be the space of square-integrable functions on the boundary area $\Dc_0$. Elements of $L^2[\Dc]$ are functions of three spatial coordinates whereas elements of $L^2[\Dc_0]$ are functions of two spatial coordinates. Hence, $Q \in L^2[\Dc]$ and $R_1,R_2 \in L^2[\Dc_0]$.

Define the space $\Pb$ by
\begin{equation}\label{phase-space}
	\Pb = L^2[\Dc]\oplus L^2[\Dc_0] \oplus L^2[\Dc_0],
\end{equation} 
where $\oplus$ is the direct sum. Equation \eqref{phase-space} states that any element of $\Pb$ is a tuple  $(Q,R_1,R_2)$ of three functions, where $Q=Q(x,y,z,t)$ is a function on the volume $\Dc$ and hence element of $L^2[\Dc]$, while the functions $R_j=R_j(x,y,t)$, for $j=1,2$, are functions on the area $\Dc_0$ and hence are elements of $L^2[\Dc_0]$. We conclude that $(Q,R_1,R_2) \in \Pb$ and that $\Pb$ is the space of all possible quasigeostrophic states. We thus call $\Pb$ the quasigeostrophic phase space.

\subsection{The phase space in terms of the streamfunction}


Given an element $\Qf \in \Pb$, we can reconstruct a continuous function $\psi$ that contains the same dynamical information as $\Qf$. By inverting the problem
\begin{align}\label{int-pv-psi}
\begin{split}
	Q-f =  \lap \psi_{\textrm{int}}  + \pd{}{z} \left( \frac{f_0}{N^2} \pd{ \psi_{\textrm{int}}}{z} \right) \quad &\text{for } z \in (z_1,z_2) \\
	R_1 - g_1 = \frac{f_0^2}{N^2}\, \pd{ \psi_{\textrm{low}} }{z} \quad &\text{for } z=z_1\\
	R_2 + g_2  = -\frac{f_0^2}{N^2} \, \pd{ \psi_{\textrm{upp}}}{z} \quad &\text{for } z=z_2
\end{split}
\end{align}
we obtain a function $\psi(x,y,z)$ that is unique up to a gauge transformation \cite[see][]{schneider_boundary_2003}. Conversely, given a function $\psi(x,y,z)$, we can differentiate $\psi$ as in equations \eqref{int-pv-psi} to obtain $\Qf\in\Pb$. Thus, we can also consider the quasigeostrophic phase space $\Pb$ to be the space of all possible streamfunctions $\psi$.

Equations \eqref{int-pv-psi} motivate the definition of the relative potential vorticity densities, $q=Q-f$ and $r_j= R_j - (-1)^{j+1}\, g_j$, which are the portions of the potential vorticity providing a source for a streamfunction. Explicitly, the relative potential vorticity densities are 
\begin{subequations}
\begin{alignat}{2}
	q &= \lap \psi  + \pd{}{z} \left( \frac{f^2_0}{N^2} \pd{ \psi}{z}\right) 
	\quad &&\text{for } z \in (z_1,z_2)  
	\\
	r_1 &= \frac{f_0^2}{N^2} \pd{\psi}{z} 
	 \quad &&\text{for } z=z_1 
	\\
	 r_2 &= -\frac{f_0^2}{N^2} \pd{\psi }{z} 
	 \quad &&\text{for } z=z_2.
\end{alignat}
\end{subequations}

\subsection{The vertical structure phase space}\label{S-vertical-phase}

Since the fluid region, $\Dc$, is separable, we can expand the potential vorticity density distribution, $(q,r_1,r_2)$, and  the streamfunction $\psi$ in terms of the eigenfunctions, $e_{\vec k}$, of the horizontal Laplacian. For a horizontal domain $\Dc_0$, the eigenfunction $ e_{\vec k}(\vec x)$ satisfies
\begin{equation}
	- \lap e_{\vec k} = k^2\, e_{\vec k}.
\end{equation}
where $\vec x=(x,y)$ is the horizontal position vector, $\vec k=(k_x,k_y)$ is the horizontal wavevector, and $k=|\vec k|$ is the horizontal wavenumber. For example, in a horizontally periodic domain the eigenfunctions $e_{\vec k}(\vec x)$ are proportional to complex exponentials, $\mathrm{e}^{\mathrm{i}\vec k \cdot \vec x}$. 

Projecting the relative potential vorticity density distribution, $(q,r_1,r_2)$, onto the horizontal eigenfunctions, $e_{\vec k}$, yields 
\begin{subequations}
\begin{alignat}{2}
	q(\vec x,z,t) &= \sum_{\vec k} q_{\vec k}(z,t) \, e_{\vec k}(\vec x), \quad &&\text{for } z\in(z_1,z_2)\\
	r_j(\vec x,t) &= \sum_{\vec k} r_{j \vec k}(t) \, e_{\vec k}(\vec x) \quad &&\text{for } j=1,2.
\end{alignat}
\end{subequations}
Thus the Fourier coefficients of $(q,r_1,r_2)$ are $(q_{\vec k}, r_{1\vec k}, r_{2\vec k})$ where $q_{\vec k}$ is a function of $z$ and $r_{1 \vec k}$ and $r_{2\vec k}$ are independent of $z$. Hence, $q_{\vec k}$ is an element of $L^2[(z_1,z_2)]$ whereas  $r_{1 \vec k}$ and $ r_{2 \vec k}$ are elements of the space of complex numbers\footnote{Since all physical fields must be real, only a single degree of freedom is gained from $\C$. Furthermore, when complex notation is used (e.g., complex exponentials for the horizontal eigenfunctions $e_{\vec k}$) it is only the real part of the fields that is physical. } , $\C$.

We conclude that the vertical structure of the potential vorticity, given by $(q_{\vec k}, r_{1\vec k}, r_{2\vec k})$, is an element of
\begin{equation}\label{vertical-phase-space}
	\wh \Pb = L^2[(z_1,z_2)] \oplus \C \oplus \C,
\end{equation}
so that the vertical structures of the potential vorticity distribution are determined by a function, $q_{\vec k}$, in $L^2[(z_1,z_2)]$ and two $z$-independent elements, $r_{1 \vec k}$ and $r_{2  \vec k}$, of $\C$.  Similarly, the streamfunction can be represented as 
\begin{equation}\label{horizontal_fourier_amplitudes}
	\psi(\vec x,z,t) = \sum_{\vec k} \psi_{\vec k} (z,t) \, e_{\vec k} (\vec x),
\end{equation}
where $\psi_{\vec k}$ and $(q_{\vec k}, r_{1\vec k}, r_{2\vec k})$ are related by
\begin{subequations}\label{pv-hor-transform}
\begin{align}
	q_{\vec k} = -k^2 \, \psi_{\vec k} + \pd{}{z} \left(\frac{f_0^2}{N^2} \pd{\psi_{\vec k}}{z}\right) \\
	r_{j \vec k} = (-1)^{j+1} \left(\frac{f_0^2}{N^2} \pd{\psi_{\vec k}}{z} \right)\Bigg |_{z=z_j}.
\end{align}
\end{subequations}

As before, knowledge of the vertical structure of the streamfunction, $\psi_{\vec k}(z)$, is equivalent to knowing the vertical structure of the potential vorticity distribution, $(q_{\vec k}, r_{1\vec k}, r_{2\vec k})$. Thus $\wh \Pb$ is also the space of all possible streamfunction vertical structures.

That $\psi_{\vec k}$ belongs to $\wh \Pb$ and not $L^2[(z_1,z_2)]$ underlies much of the confusion over baroclinic modes. Assertions of completeness, based on Sturm-Liouville theory, assume that $\psi$ is an element of $L^2[(z_1,z_2)]$. However, as we have shown, that is an incorrect assumption. That $\psi$ belongs to $\wh \Pb$ will have consequences for the convergence and differentiability of normal mode expansions, as discussed in section \ref{S-expansions}. In the context of quasigeostrophic theory, the space $\wh \Pb$ first appeared in \cite{smith_surface-aware_2012}. More generally, $\wh \Pb$  appears in the presence of  non-trivial boundary dynamics (chapter \ref{Ch-modes}).

We call $\wh \Pb$ the vertical structure phase space, and for convenience we denote $L^2[(z_1,z_2)]$ by $L^2$ for the remainder of the chapter. The vertical structure phase space $\wh \Pb$ is then written as the direct sum 
\begin{equation}
	\wh \Pb = L^2 \oplus \C^2.
\end{equation}

\subsection{Representing the energy and potential enstrophy}
\label{subsection:decomposing-energy-enstrophy}

We find it convenient to represent several quadratic quantities in terms of the eigenfunctions of the horizontal Laplacian, $e_{\vec k}(\vec x)$. The energy per unit mass in the volume $\Dc$ is given by
\begin{align}
	E = \frac{1}{V}\int_{\Dc} \left[ \left |\grad \psi \right |^2 + \frac{f_0^2}{N^2} \left |\pd{\psi}{z} \right |^2 \right] \mathrm{d}A \, \mathrm{d}z = \sum_{\vec k} E_{\vec k},
\end{align}
where the horizontal energy mode is given by the vertical integral 
\begin{equation}
	E_{\vec k} =  \frac{1}{H}\intz \left[ k^2 \left| \psi_{\vec k} \right|^2 + \frac{f_0^2}{N^2}  \left |\pd{\psi_{\vec k}}{z} \right |^2 \right] \mathrm{d}z,
\end{equation}
with $V=A\,H$ the domain volume and $H =z_2-z_1$ the domain depth.

Similarly, for the relative volume potential enstrophy density, $Z$, we have
\begin{align}
	Z &= \frac{1}{V} \int_{\Dc} |q|^2  \mathrm{d}A \, \mathrm{d}z = \sum_{\vec k} Z_{\vec k},
\end{align}
where 
\begin{equation}
	Z_{\vec k} = \frac{1}{H} \intz  \left|q_{\vec k} \right|^2 \mathrm{d}z.
\end{equation}
Finally, analogous to $Z$, we have the relative surface potential enstrophy densities, $Y_j$, on the area $\Dc_0$
\begin{equation}\label{surface-enstrophy}
	Y_j = \frac{1}{A} \int_{\Dc_0} \left| r_j \right|^2 \mathrm{d}A = \sum_{\vec k} Y_{j\vec k},
\end{equation}
where 
\begin{equation}\label{surface-enstrophy-rjk}
	Y_{j\vec k} =  \left | r_{j\vec k} \right|^2.
\end{equation}

\section{Rossby waves in a quiescent ocean}\label{S-linear}

\begin{figure}
	\centerline{\includegraphics[width=\textwidth]{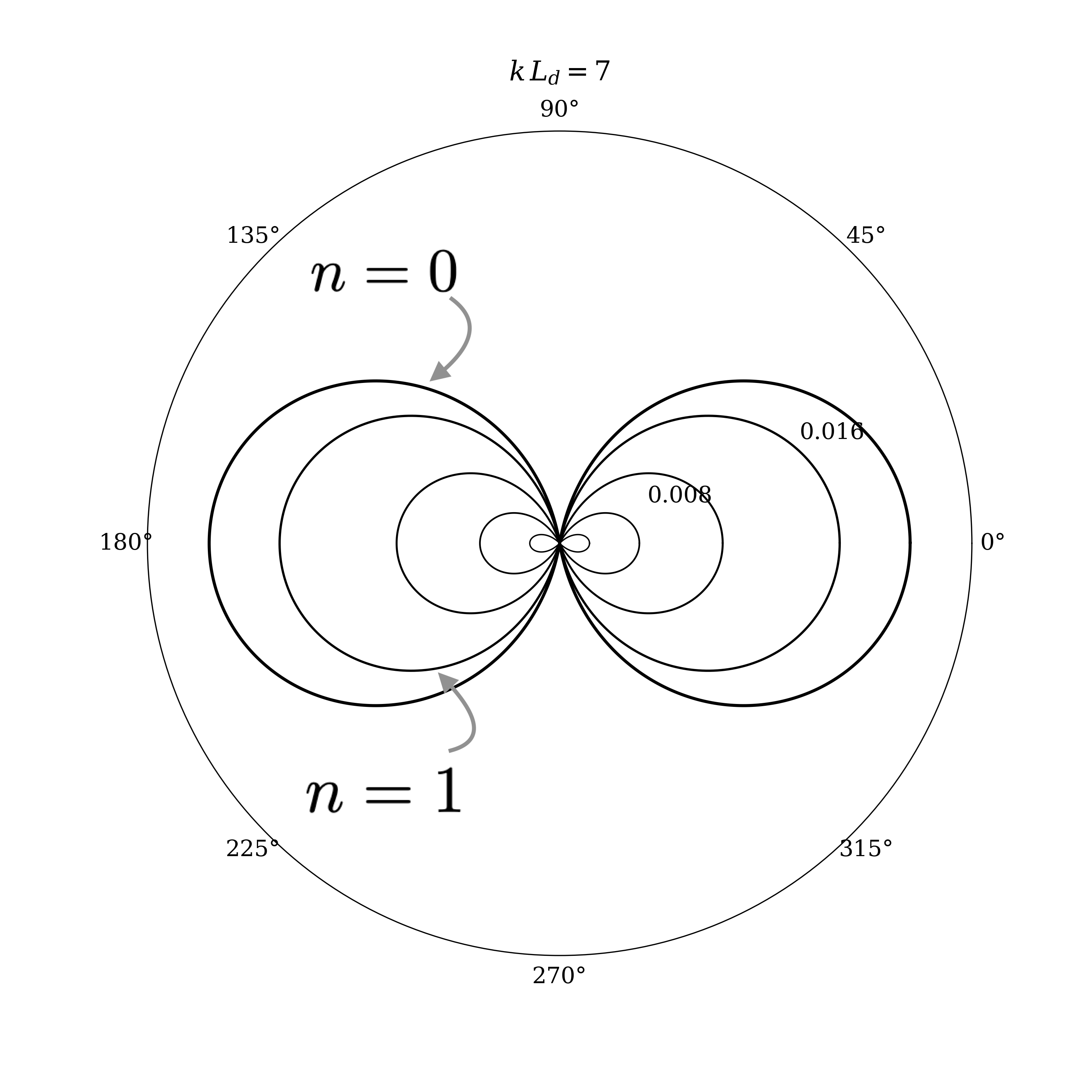}}
	\caption{Polar plots of the absolute value of the non-dimensional angular frequency $|\omega_n|/(\beta L_d)$ of the first five modes of the traditional eigenvalue problem (section \ref{SS-trad}) as a function of the wave propagation direction, $\vec k/|\vec k|$, for constant stratification. The outer most ellipse, with the largest absolute angular frequency, represents the angular frequency of the barotropic ($n=0$) mode. The higher modes have smaller absolute frequencies and are thus concentric and within the barotropic angular frequency curve. Since the absolute value of the angular frequency of the barotropic mode becomes infinitely large at small horizontal wavenumbers $k$, we have chosen a large wavenumber $k$, given by $k\,L_d = 7$, so that the angular frequency of the first five modes can be plotted in the same figure. We have chosen $f_0 = 10^{-4}~\textrm{s}^{-1},  \beta = 10^{-11}~\textrm{m}^{-1}~\textrm{s}^{-1}, N_0 = 10^{-2}~\textrm{s}^{-1}$ and $H = 1~\textrm{km}$ leading to a deformation radius $L_d = N_0 \, H/f_0= 100~\textrm{km}$. Numerical solutions to all eigenvalue problems in this paper are obtained using \texttt{Dedalus} \citep{burns_dedalus_2020}.}
	\label{F-baroclinic_angle}
\end{figure}

In this section, we study Rossby waves in an otherwise quiescent ocean; in other words, we examine the \emph{physical} normal modes of a quiescent ocean. The linear equations of motion are
\begin{subequations}\label{linear-time-evolution}
\begin{align}
	\label{linear-q-equation}
	\pd{q}{t} + \beta \, v = 0 \quad &\textrm{for } z\in(z_1,z_2)\\
	\label{linear-r-equation}
	\pd{r_{j}}{t} + \vec u  \cdot \grad\left[(-1)^{j+1} \, g_j\right]  = 0 \quad &\textrm{for } z=z_j.
\end{align}
\end{subequations}
We assume that the prescribed surface potential vorticity densities at the lower and upper boundaries, $g_1$ and $g_2$, are linear, which ensures the resulting eigenvalue problem is separable. Moreover, as the ocean is quiescent, $g_1$ and $g_2$ must refer to topographic slopes, as in equation \eqref{g_j_topography}.


The importance of the linear problem \eqref{linear-time-evolution} is that it provides all possible discrete Rossby wave normal modes in a quasigeostrophic flow. Substituting a wave ansatz of the form [compare with equation \eqref{physical-mode} for \emph{physical} normal modes]
\begin{equation}
	\psi(\vec x, z,t) = \hat \psi(z) \, e_{\vec k}(\vec x) \, \mathrm{e}^{-\mathrm{i}\omega t}
\end{equation}
into the  linear problem \eqref{linear-time-evolution} renders
\begin{equation}\label{q-eigen}
	\left(-\mathrm{i}\, \omega\right) \left[ -k^2 \, \hat \psi + \d{}{z}\left(\frac{f_0^2}{N^2} \d{\hat \psi}{z} \right) \right] + \mathrm{i} \, k_x \, \beta \, \hat \psi = 0,
\end{equation}
for $z\in(z_1,z_2)$, and 
\begin{equation}
	\label{r-eigen}
	\left(-\mathrm{i} \, \omega \right) \left( \frac{f_0^2}{N^2} \d{\hat \psi}{z} \right) + \mathrm{i} \, \unit z \cdot \left( \vec k \times \grad g_j \right) \hat \psi = 0,
\end{equation}
for $z=z_1,z_2$.

\subsection{Traditional Rossby wave problem}\label{SS-trad}

We first examine the traditional case of linear fluctuations to a quiescent ocean with isentropic lower and upper boundaries i.e., with no topography. Setting $\grad g_1=\grad g_2 = 0$ in the eigenvalue problem \eqref{q-eigen}\textendash\eqref{r-eigen} gives
\begin{subequations}\label{trad-rossby}
\begin{align}
	\label{trad-rossby-interior}
		\omega \left[ -k^2 \,  F + \d{}{z} \left(\frac{f_0^2}{N^2} \d{ F}{z}\right) \right] - \beta \, k_x \,  F = 0 \\
	\label{trad-rossby-boundary}
		\omega \left( \frac{f_0^2}{N^2} \d{F}{z}\right)\Bigg |_{z=z_j} = 0,
\end{align} 
\end{subequations}
where $\hat \psi(z) = \hat \psi_0 \, F(z)$ and $F$ is a non-dimensional function.
There are two cases to consider depending on whether $\omega$ vanishes.

\subsubsection{Traditional baroclinic modes} 

Assuming $\omega \neq 0$ in the eigenvalue problem \eqref{trad-rossby} renders a Sturm-Liouville eigenvalue problem in $L^2$
\begin{subequations}\label{trad-SL}
\begin{align}
	\label{trad-SL-interior}
	- \d{}{z}\left(\frac{f_0^2}{N^2}\d{ F}{z}\right) = \lambda \,  F \quad &\textrm{for } z\in (z_1,z_2) \\
	\label{trad-SL-boundary}
	\frac{f_0^2}{N^2} \d{ F}{z} = 0 \quad &\textrm{for } z=z_1,z_2,
\end{align}
\end{subequations}
where the eigenvalue, $\lambda$, is given by 
\begin{equation}\label{eigenvalue}
	\lambda = -k^2 - \frac{\beta \, k_x}{\omega}.
\end{equation}
See figure \ref{F-baroclinic_angle} for an illustration of the dependence of $|\omega|$ on the wavevector $\vec k$.

From Sturm-Liouville theory \cite[e.g.,][]{brown_fourier_1993}, the eigenvalue problem \eqref{trad-SL} has infinitely many eigenfunctions,  $F_0, \, F_1, \, F_2, \dots$ with distinct and ordered eigenvalues, $\lambda_n$, satisfying
\begin{equation}
	0 = \lambda_0 < \lambda_1 < \cdots \rightarrow \infty.
\end{equation}
The $n$th mode, $F_n$, has $n$ internal zeros in the interval $(z_1,z_2)$. 
The eigenfunctions are orthonormal with respect to
the inner product, $\left[\cdot ,\cdot\right]$, given by the vertical integral
\begin{equation}\label{trad-inner}
	\left[F , G\right] = \frac{1}{H} \intz F \, G \, \mathrm{d}z,
\end{equation}
with orthonormality meaning that  
\begin{equation}\label{trad-ortho}
	\delta_{mn} = \left[F_m,F_n\right]
\end{equation}
where $\delta_{mn}$ is the Kronecker delta. A powerful and commonly used result of Sturm-Liouville theory is that the set $\{F_n\}_{n=0}^\infty$ forms an orthonormal basis of $L^2$.

\subsubsection{Stationary step-modes}\label{SSS-step}

There are two additional solutions to the Rossby wave eigenvalue problem \eqref{trad-rossby} not previously noted in the literature. If $\omega = 0$ then the eigenvalue problem \eqref{trad-rossby} becomes
\begin{subequations}\label{trad-step}
\begin{align}
	\label{trad-step-interior}
		\beta \, k_x \, F = 0 \quad &\textrm{for } z\in(z_1,z_2) \\
	\label{trad-step-boundary}
		0 = 0 \quad &\textrm{for } z=z_1,z_2.
\end{align}
\end{subequations}
Consequently, if $k_x \neq 0$, then  $ F (z) =0$ for $z \in (z_1,z_2)$. That is, $F$ must vanish in the interior of the interval. However, since $\omega=0$ in \eqref{trad-rossby-boundary}, we obtain tautological boundary conditions \eqref{trad-step-boundary}. As a result, $F$ can take arbitrary values at the lower and upper boundaries. Thus two solutions are 
\begin{equation}\label{step-mode}
	F^\textrm{step}_{j}(z) = 
	\begin{cases}
		1  \quad \text{for } z=z_j\\
		0  \quad \text{otherwise.}
	\end{cases}
\end{equation}
The two step-mode solutions \eqref{step-mode} are independent of the traditional baroclinic modes, $F_{n}(z)$. An expansion of the step-mode $F^\textrm{step}_j$ in terms of the baroclinic modes will fail and produce a series that is identically zero.

The two stationary step-modes, $F^\textrm{step}_1$ and $F^\textrm{step}_2$, correspond to the two inert degrees of freedom in the eigenvalue problem \eqref{trad-rossby}. These two solutions are neglected in the traditional eigenvalue problem \eqref{trad-SL} through the assumption that $\omega \neq 0$. Although dynamically trivial, we will see that these two step-waves are obtained as limits of boundary-trapped modes as the boundary buoyancy gradients $N^2\,\grad g_j/f_0$ become small.

\subsubsection{The general solution}

For a wavevector $\vec k$ with $k_x \neq 0$, the vertical structure of the streamfunction must be of the form 
\begin{equation}\label{trad-vertical}
	\Psi(z) + \sum_{j=1}^2 \Psi_j \, F_j^\textrm{step}(z) = \psi_{\vec k} (z,t=0), 
\end{equation}
where $\Psi(z)$ is a twice differentiable function satisfying $\mathrm{d}\Psi(z_j)/\mathrm{d}z = 0$ for $j=1,2$ and $\Psi_1,\Psi_2$ are arbitrary constants. We can represent $\Psi$ according to the expansion,
\begin{equation}\label{trad-expansion}
	\Psi = \sum_{n=0}^\infty \left[\Psi,F_n\right] F_n,
\end{equation}
and so the time-evolution is 
\begin{equation}
	\psi_{\vec k} (z, t)= \sum_{n=0}^\infty  \left[\Psi,F_n\right] F_n \, \mathrm{e}^{-\mathrm{i}\omega_n t} + \sum_{j=1}^2 \Psi_j \, F^\textrm{step}_j.
\label{eq:linear-traditional-time-evolution}
\end{equation}
It is this time-evolution expression, which is valid only in linear theory for a quiescent ocean, that gives the baroclinic modes a clear physical meaning. More precisely, equation \eqref{eq:linear-traditional-time-evolution} states that the vertical structure $\Psi(z)$ disperses into its constituent Rossby waves with vertical structures $F_n$. Outside the linear theory of this section, baroclinic modes do not have a physical interpretation, although they remain a mathematical basis for $L^2$.

\subsection{The Rhines problem}\label{SS-rhines-bottom}

We now examine the case with a sloping lower boundary, $\grad g_1\neq 0$, and an isentropic upper boundary, $\grad g_2 =0$. The special case of a meridional bottom slope and constant stratification was first investigated by \cite{rhines_edge_1970}. Subsequently, \cite{charney_oceanic_1981} extended the analysis to realistic stratification and \cite{straub_dispersive_1994} examined the dependence of the waves on the propagation direction. Chapter \ref{Ch-modes} applies the mathematical theory of eigenvalue problems with $\lambda$-dependent boundary conditions and obtains various completeness and expansion results as well as a qualitative theory for the streamfunction modes. Below, we generalize these results, study the two limiting boundary conditions, and consider the corresponding vertical velocity modes.

\subsubsection{The eigenvalue problem}

Let $\hat \psi(z) = \hat \psi_0 \, G(z)$ where $G$ is a non-dimensional function. We then manipulate the eigenvalue problem \eqref{q-eigen}\textendash\eqref{r-eigen} to obtain (assuming $\omega \neq 0$)
\begin{subequations}\label{Rhines-eigen_bottom}
\begin{alignat}{2}	
		- \d{}{z} \left(\frac{f_0^2}{N^2} \d{G}{z}\right) &= \lambda \, G \quad &&\textrm{for } z\in(z_1,z_2) \\
			-k^2 G - \gamma_1^{-1} \left( \frac{f_0^2}{N^2}\d{G}{z} \right) &= \lambda \, G \quad &&\textrm{for } z=z_1,\\
			\d{G}{z} &= 0  \quad &&\textrm{for } z=z_2,
\end{alignat}
\end{subequations}
where the length-scale $\gamma_j$ is given by
\begin{align}\label{gamma}
\begin{split}
	\gamma_j &= (-1)^{j+1} \frac{\unit z\cdot\left(\vec k \times \grad g_j\right)}{\unit z\cdot\left(\vec k \times \grad f\right)}\\ &= (-1)^{j+1}  \left( \frac{\alpha_j \, k}{\beta \, k_x} \right) \, \sin\left(\Delta \theta_j\right)
	\end{split}
\end{align}
 where $\alpha_j = |\grad{g_j}|$ and $\Delta \theta_j$ is the angle between the wavevector $\vec k$ and $\grad g_j$ measured counterclockwise from $\vec k$. The parameter $\gamma_j$ depends only on the direction of the wavevector $\vec k$ and not its magnitude $k = |\vec k|$. If $\gamma_j = 0$, then the $j$th boundary condition can be written as a $\lambda$-independent boundary condition [as in the upper boundary condition at $z=z_{2}$ of the eigenvalue problem \eqref{Rhines-eigen_bottom}]. For now, we assume that $\gamma_1\neq 0$. 
 
 Since the eigenvalue, $\lambda$, appears in the differential equation and  one boundary condition in the eigenvalue problem \eqref{Rhines-eigen_bottom}, the eigenvalue problem takes place in $L^2\oplus \C$.
 
 \subsubsection{Characterizing the eigen-solutions}
The following is obtained by applying the theory summarized in appendix A to the eigenvalue problem \eqref{Rhines-eigen_bottom}.\footnote{To apply the theory of chapter \ref{Ch-modes}, summarized in Appendix A, let $\tilde \lambda = \lambda - k^2$ be the eigenvalue in place of $\lambda$; the resulting eigenvalue problem for $\tilde \lambda$ will then satisfy the positiveness conditions, equations \eqref{left-definite-1} and \eqref{left-definite-2}, of Appendix A.}
 
The eigenvalue problem \eqref{Rhines-eigen_bottom} has a countable infinity of eigenfunctions $G_0,\,G_1,\, G_2, \dots$ with ordered and distinct non-zero eigenvalues $\lambda_n$ satisfying
\begin{equation}\label{eigen-order}
	\lambda_0 < \lambda_1 < \lambda_2 < \cdots \rightarrow \infty.
\end{equation}
The inner product $\inner{\cdot}{\cdot}$ induced by the eigenvalue problem 
\eqref{Rhines-eigen_bottom}
is 
\begin{equation}\label{Rhines-inner}
	\inner{F}{G}= \frac{1}{H}\left( \intz F \, G \, \mathrm{d}z + \gamma_1 \, F(z_1)\, G(z_1) \right),
\end{equation}
which depends on the direction of the horizontal wavevector $\vec k$ through $\gamma_1$. Moreover, $\gamma_1$ is not necessarily positive\footnote{That $\gamma_1$ is not positive prevents us from applying the eigenvalue theory outlined in the appendix of \cite{smith_surface-aware_2012}.}, with one consequence being that some functions $G$ may have a negative square, $\inner{G}{G}<0$. Orthonormality  of the modes $G_n$ then takes the form
\begin{equation}
	\pm \delta_{mn} = \inner{G_m}{G_n},
\end{equation}
where at most one mode, $G_n$, satisfies $\inner{G_n}{G_n}=-1$. The eigenfunctions $\{G_n\}_{n=0}^\infty$ form an orthonormal basis of $L^2\oplus \C$ under the inner product $\eqref{Rhines-inner}$.

Appendix A provides the following inequality,
\begin{equation}
	\left(k^2 + \lambda_n \right) \inner{G_n}{G_n} > 0,
\end{equation}
which, using the dispersion relation \eqref{eigenvalue}, implies that modes $G_n$ with $\inner{G_n}{G_n} >0$ correspond to waves with a westward phase speed while modes $G_n$ with $\inner{G_n}{G_n}<0$ correspond to waves with an eastward phase speed (assuming $\beta>0$).

\begin{figure*}
	\centerline{\includegraphics[width=\textwidth]{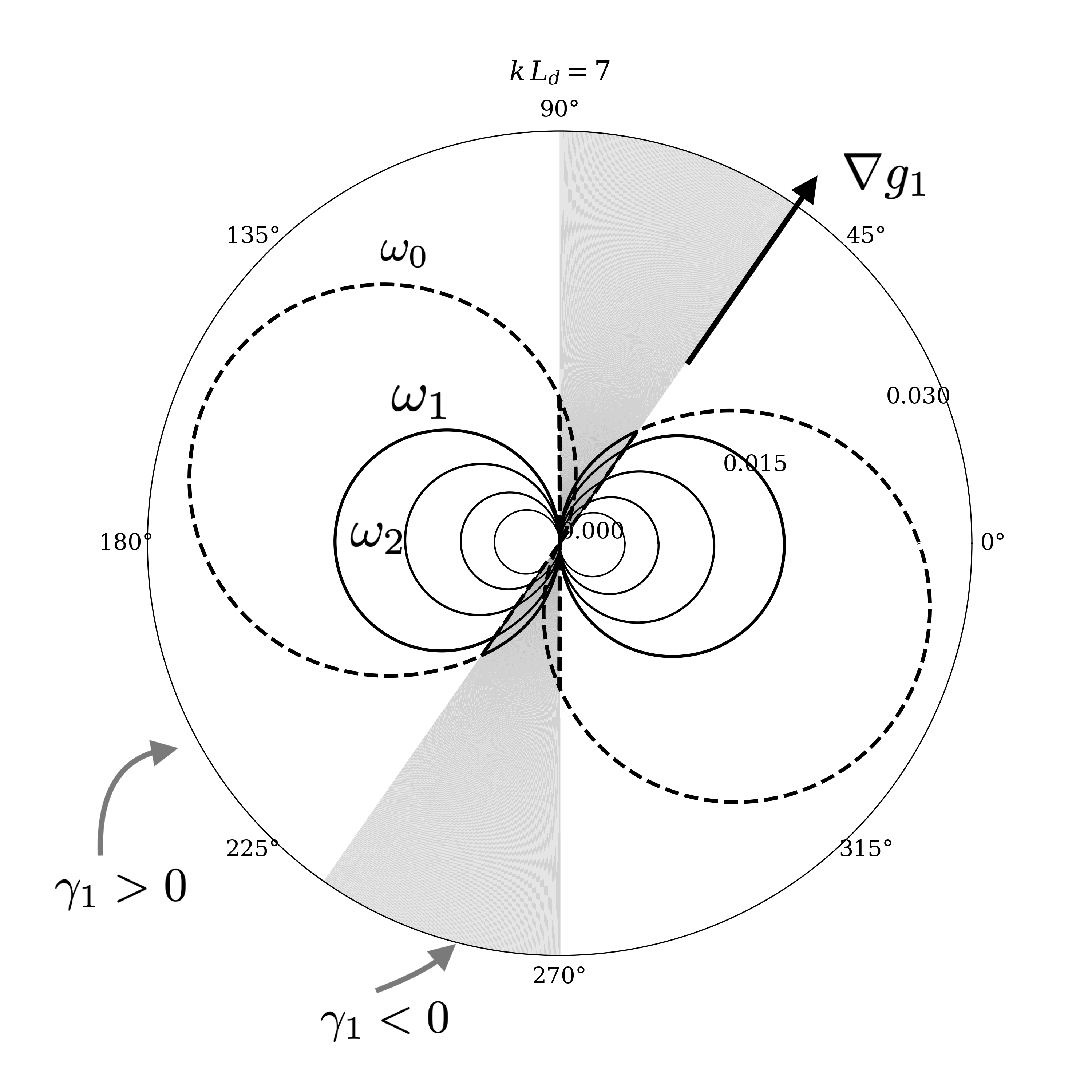}}
	\caption{Polar plots of the absolute value of the non-dimensional angular frequency $|\omega_n|/(\beta L_d)$ of the first five modes from section \ref{SS-rhines-bottom} as a function of the wave propagation direction $\vec k/|\vec k|$ for a horizontal wavenumber given by $k\, L_d = 7$ in constant stratification. The dashed line corresponds to $\omega_0$, this mode becomes boundary trapped at large wavenumbers $k = |\vec k|$. The remaining modes, $\omega_n$ for $n=1,2,3,4$, are shown with solid lines. White regions are angles where $\gamma_1>0$. All Rossby waves with a propagation direction lying in the white region have negative angular frequencies $\omega_n$ and so have a westward phase speed. Gray regions are angles where $\gamma_1<0$. Here, $\omega_0$ is positive while the remaining angular frequencies $\omega_n$ for $n>0$ are negative. Consequently, in the gray regions, $\omega_0$ corresponds to a Rossby wave with an eastward phase speed whereas the remaining Rossby waves have westward phase speeds.  The lower boundary buoyancy gradient, proportional to $\grad g_1$, points towards $55^\circ$ and corresoponds to a bottom slope of $|\grad h_1| = 1.5\times 10^{-5}$ leading to $\gamma_1/H = 0.15$. The remaining parameters are as in figure \ref{F-baroclinic_angle}.}
	\label{F-angle1}
\end{figure*}

We distinguish the following cases depending on the sign of $\gamma_1$. In the following, we assume $k\neq0$.
\begin{itemize}
	\item [i.] $\gamma_1 > 0$. All eigenvalues satisfy $\lambda_n > -k^2$, all modes satisfy $\inner{G_n}{G_n}>0$, and all waves propagate westward. The $n$th mode, $G_n$, has $n$ internal zeros \citep{binding_sturmliouville_1994}. See the regions in white in figure \ref{F-angle1}.
	
	\item [ii.]  $\gamma_1 < 0$. There is one mode, $G_0$, with a negative square, $\inner{G_0}{G_0} <0$, corresponding to an eastward propagating wave. The eastward propagating wave nevertheless travels pseudowestward (to the left of the upslope direction for $f_0>0$). The associated eigenvalue, $\lambda_0$, satisfies $\lambda_0 < -k^2$. The remaining modes, $G_n$ for $n>1$, have positive squares, $\inner{G_n}{G_n} >0$, corresponding to westward propagating waves and have eigenvalues, $\lambda_n$, satisfying $\lambda_n>-k^2$. Both $G_0$ and $G_1$ have no internal zeros whereas the remaining modes, $G_n$, have $n-1$ internal zeros for $n>1$ \citep{binding_sturmliouville_1994}. See the stippled regions in figures \ref{F-angle1}.
\end{itemize}


To elucidate the meaning of $\lambda_n < -k^2$, note that a pure surface quasigeostrophic mode\footnote{A pure surface quasigeostrophic mode is the mode found after setting $\beta =0$ with an upper boundary at $z_2=\infty$.}  has $\lambda = -k^2$. Thus $\lambda_0 < -k^2$ means that the bottom-trapped mode decays away from the boundary more rapidly than a pure surface quasigeostrophic wave. Indeed, the limit of $\lambda_0 \rightarrow -\infty$ yields the bottom step-mode \eqref{step-mode} of the previous subsection.

\begin{figure}
	\centerline{\includegraphics[width=27pc]{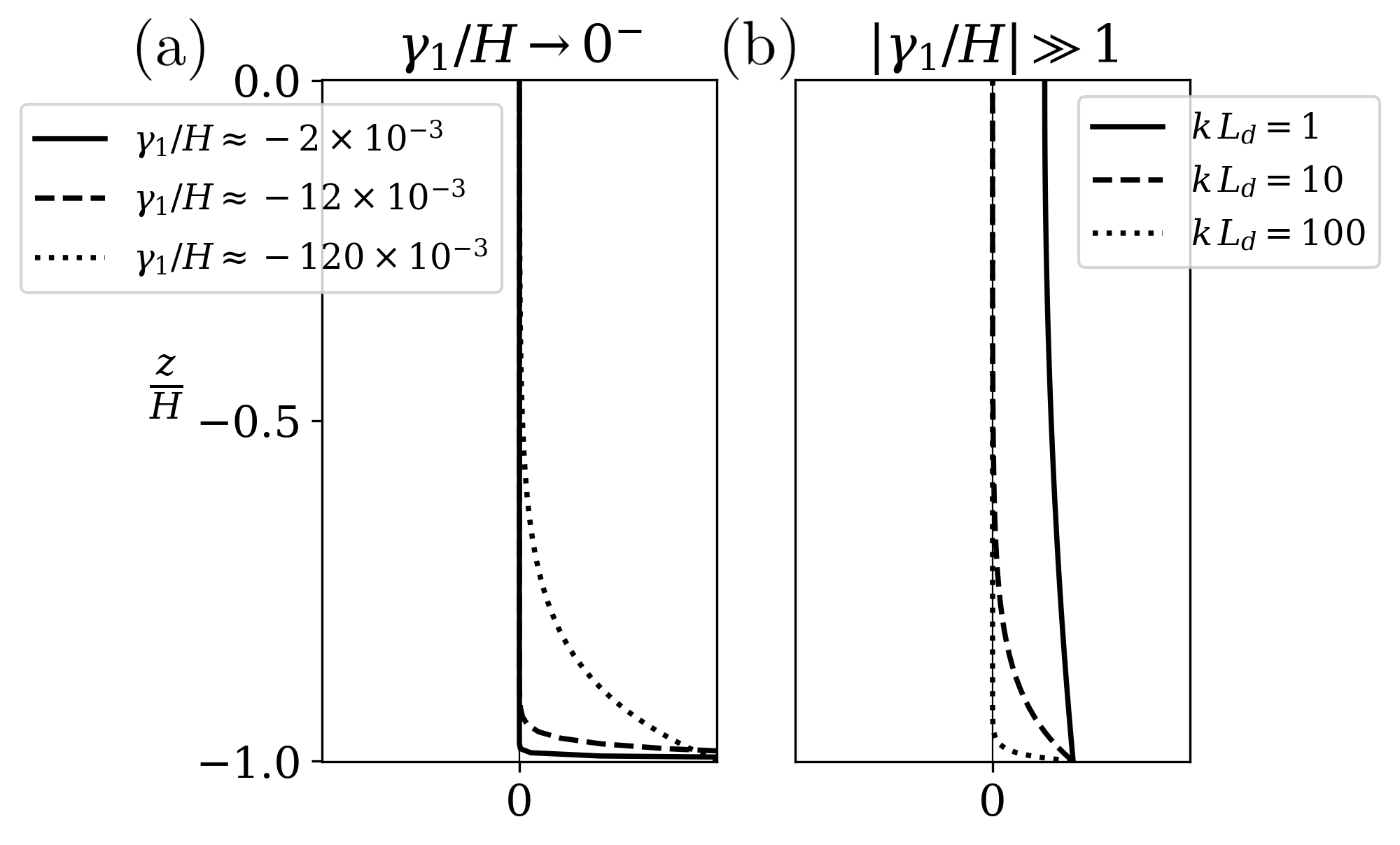}}
	\caption{The two limits of the boundary-trapped surface quasigeostrophic waves, as discussed in section \ref{SS-rhines}. (a) Convergence to the step mode given in equation \eqref{step-mode} with $j=1$ as $\gamma_1 \rightarrow 0^-$ for three values of $\gamma_1$ at a  wavenumber $k=|\vec k|$ given by $k \, L_d =1$. The phase speed approaches zero in the limit $\gamma_1 \rightarrow 0^{-}$. (b) Here, $\gamma_1/H \approx 10$ for the three vertical structures $G_n$ shown. Consequently, the bottom trapped wave has $\lambda \approx -k^2$ and the phase speeds are large. The vertical structure, $G$, for three values of $k \, L_d$ are shown, illustrating the dependence on $k$ of this mode, which behaves as a boundary-trapped exponential mode with an $\mathrm{e}$-folding scale of $|\lambda|^{-1/2} = k^{-1}$. In both (a) and (b), the wave propagation direction $\theta = 260^\circ$. All other parameters are identical to figure \ref{F-angle1}.}
	\label{F-step}
\end{figure}

The step-mode limit is obtained as $\gamma_1 \rightarrow 0^-$. This limit is found as either $|\grad g_1| \rightarrow 0$ for propagation directions in which $\gamma_1 < 0$ or as $\vec k$ becomes parallel or anti-parallel to $\grad g_1$ (whichever limit satisfies $\gamma_1 \rightarrow 0^-$). In this limit, we obtain a step-mode exactly confined at the boundary (that is, $|\lambda|^{-1/2}=0$) with zero phase speed [see figure \ref{F-step}(a)]. The remaining modes then satisfy the isentropic boundary condition
\begin{equation}\label{non-zero-boundary}
	\left(\frac{f_0^2}{N^2} \, \d{G_n}{z}\right)\Bigg |_{z=z_1}= 0.
\end{equation}

The other limit is that of $|\gamma_1| \rightarrow \infty$ which is obtained as the buoyancy gradient becomes large, $|\grad g_1| \rightarrow \infty$. In this limit, the eigenvalue $\lambda_0 \rightarrow -k^2$ [see figure \ref{F-step}(b)]. Moreover, the phase speed of the bottom-trapped wave becomes infinite, an indication that the quasigeostrophic approximation breaks down. Indeed, the large buoyancy gradient limit corresponds to steep topographic slopes and so we obtain the topographically-trapped internal gravity wave of \cite{rhines_edge_1970}, which has an infinite phase speed in quasigeostrophic theory. The remaining modes then satisfy the vanishing pressure boundary condition
\begin{align}\label{zero-boundary}
	G(z_1) = 0
\end{align}
as in the surface modes of \cite{de_la_lama_vertical_2016} and \cite{lacasce_prevalence_2017}.



\subsubsection{The general time-dependent solution}

At some wavevector $\vec k$, the observed vertical structure now has the form
\begin{equation}\label{rhines-vertical}
	\Psi(z)  = \psi_{\vec k} (z,t=0), 
\end{equation}
where $\Psi$ is a twice continuously differentiable function satisfying $d\Psi(z_2)/dz=0$. For such functions we can write (see appendix A)
\begin{equation}
	\Psi = \sum_{n=0}^\infty \frac{\inner{\Psi}{G_n}}{\inner{G_n}{G_n}} G_n,
\end{equation}
so that the time-evolution is
\begin{equation}\label{rhines-evolution}
	\psi_{\vec k} (z,t) = \sum_{n=0}^\infty \frac{\inner{\Psi}{G_n}}{\inner{G_n}{G_n}} G_n(z) \, \mathrm{e}^{-\mathrm{i}\omega_n t}.
\end{equation}
Again, it is the above expression, which is valid only in linear theory with a quiescent background state, that gives the generalized Rhines modes $G_n$ physical meaning. Outside the linear theory of this section, the generalized Rhines modes do not have any physical interpretation and instead merely serve as a mathematical basis for $L^2\oplus \C$.

Recall from section \ref{SS-trad} that an expansion of a step-mode \eqref{step-mode} in terms of the baroclinic modes $\{F_n\}_{n=0}^\infty$  produces a series that is identically zero. It follows that the step-modes are independent of the baroclinic modes\textemdash they constitute independent degrees of freedom. However, with the inclusion of bottom boundary dynamics, we may now expand the bottom step-mode, $F^\mathrm{step}_1(z)$, in terms of the $L^2\oplus\C^1$ modes, $\{G_n\}_{n=0}^\infty$, with the expansion given by
\begin{equation}
	F^\mathrm{step}_1(z) = \frac{\gamma_1}{H} \sum_{n=0}^\infty \frac{G_n(z_1)}{\inner{G_n}{G_n}} G_n(z).
\end{equation}

\begin{figure*}
	\centerline{\includegraphics[width=\textwidth]{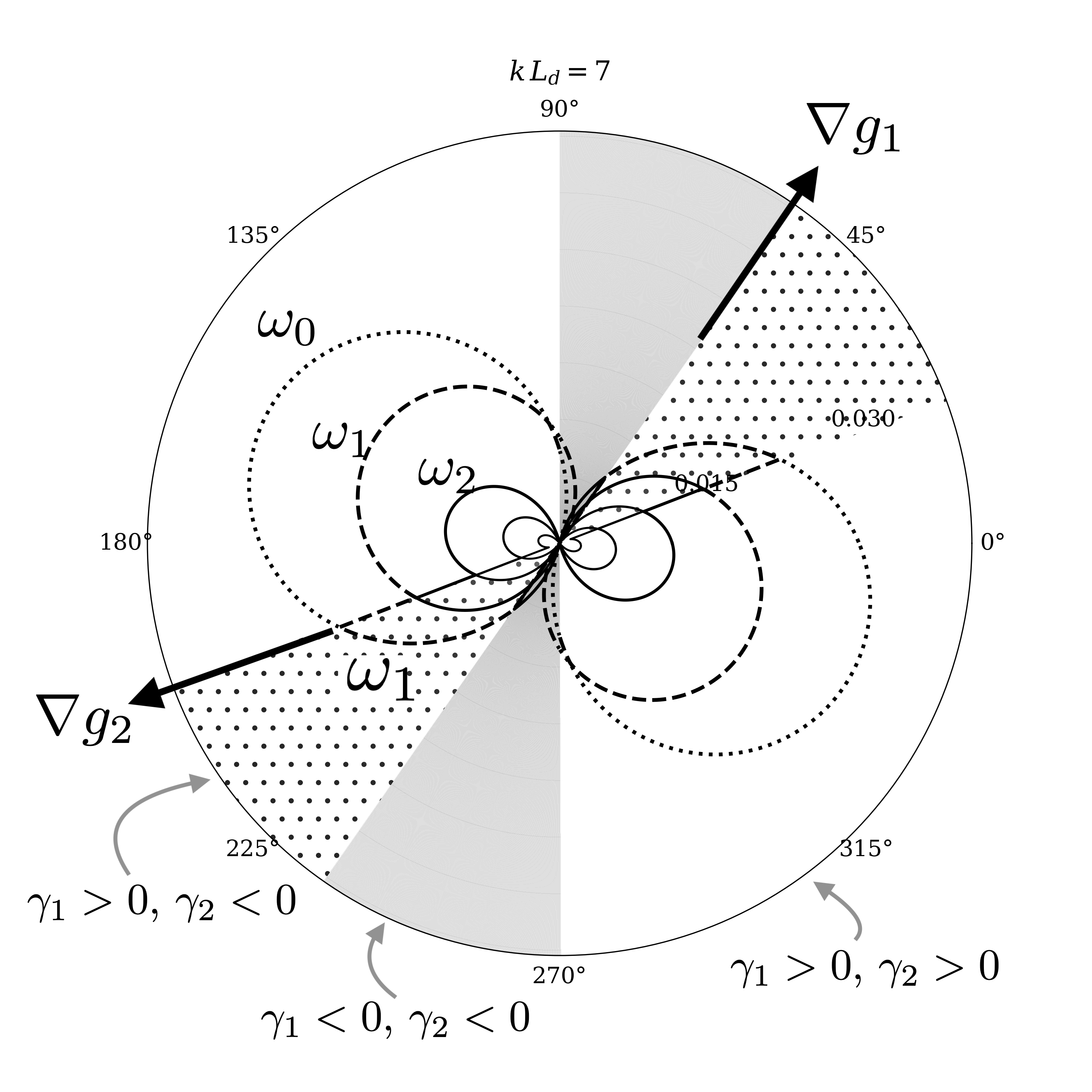}}
	\caption{As in figure \ref{F-angle1} but now with an upper slope $|\grad h_2| = 10^{-5}$ in the direction $200^\circ$ in addition to the bottom slope in figure \ref{F-angle1}. The upper slope corresponds to $\gamma_2/H =0.1$. The dotted line corresponds to $\omega_0$, the dashed line to $\omega_1$, with these two modes becoming boundary trapped at large wavenumbers $k$. The remaining modes, $\omega_n$ for $n=2,3,4$, are shown with solid lines. White regions are angles where $\gamma_1>0$ and $\gamma_2>0$. All Rossby waves with a propagation direction lying in the white region have negative angular frequencies $\omega_n$ and so have a westward phase speed. Gray regions are angles where $\gamma_1<0$ and $\gamma_2<0$. The two gravest angular frequencies $\omega_0$ and $\omega_1$ are both positive while the remaining angular frequencies $\omega_n$ for $n>1$ are negative. Consequently, in the gray regions, $\omega_0$ and $\omega_1$ each correspond to a Rossby waves with an eastward phase speed whereas the remaining Rossby waves have westward phase speeds . Stippled regions are angles where $\gamma_1>0$ and $\gamma_2<0$. In the stippled region, $\omega_0$ is positive and has an eastward phase speed. The remaining Rossby waves in the stippled region have negative angular frequencies and have westward phase speeds.}
	\label{F-angle2}
\end{figure*}

\subsection{The generalized Rhines problem}\label{SS-rhines}

The general problem with topography at both the upper and lower boundaries is
\begin{subequations}\label{Rhines-eigen}
\begin{align}	
		\label{Rhines-eigen-interior}		
		- \d{}{z} \left(\frac{f_0^2}{N^2} \d{G}{z}\right) = \lambda \, G \quad &\textrm{for } z\in(z_1,z_2)\\
		\label{Rhines-eigen-boundary}
			-k^2 G + (-1)^j \gamma_j^{-1} \left( \frac{f_0^2}{N^2}\d{G}{z} \right)= \lambda \, G \quad &\textrm{for } z=z_j,
\end{align}
\end{subequations}
for $j=1,2$, where the length-scale $\gamma_j$ is given by equation \eqref{gamma}. As the eigenvalue, $\lambda$, appears in both boundary conditions, the eigenvalue problem \eqref{Rhines-eigen} takes place in $L^2\oplus \C^2$. The inner product now has the form
\begin{equation}\label{Rhines-inner-general}
	\inner{F}{G}= \frac{1}{H}\left( \intz F \, G \, \mathrm{d}z + \sum_{j=1}^2 \gamma_j \, F(z_j)\, G(z_j) \right)
\end{equation}
which reduces to equation \eqref{Rhines-inner} when $\gamma_2=0$. Under this inner product, the eigenfunctions $\{G_n\}_{n=0}^\infty$ form a basis of $L^2\oplus \C^2$.

\begin{figure*}
	\noindent \includegraphics[width=\textwidth]{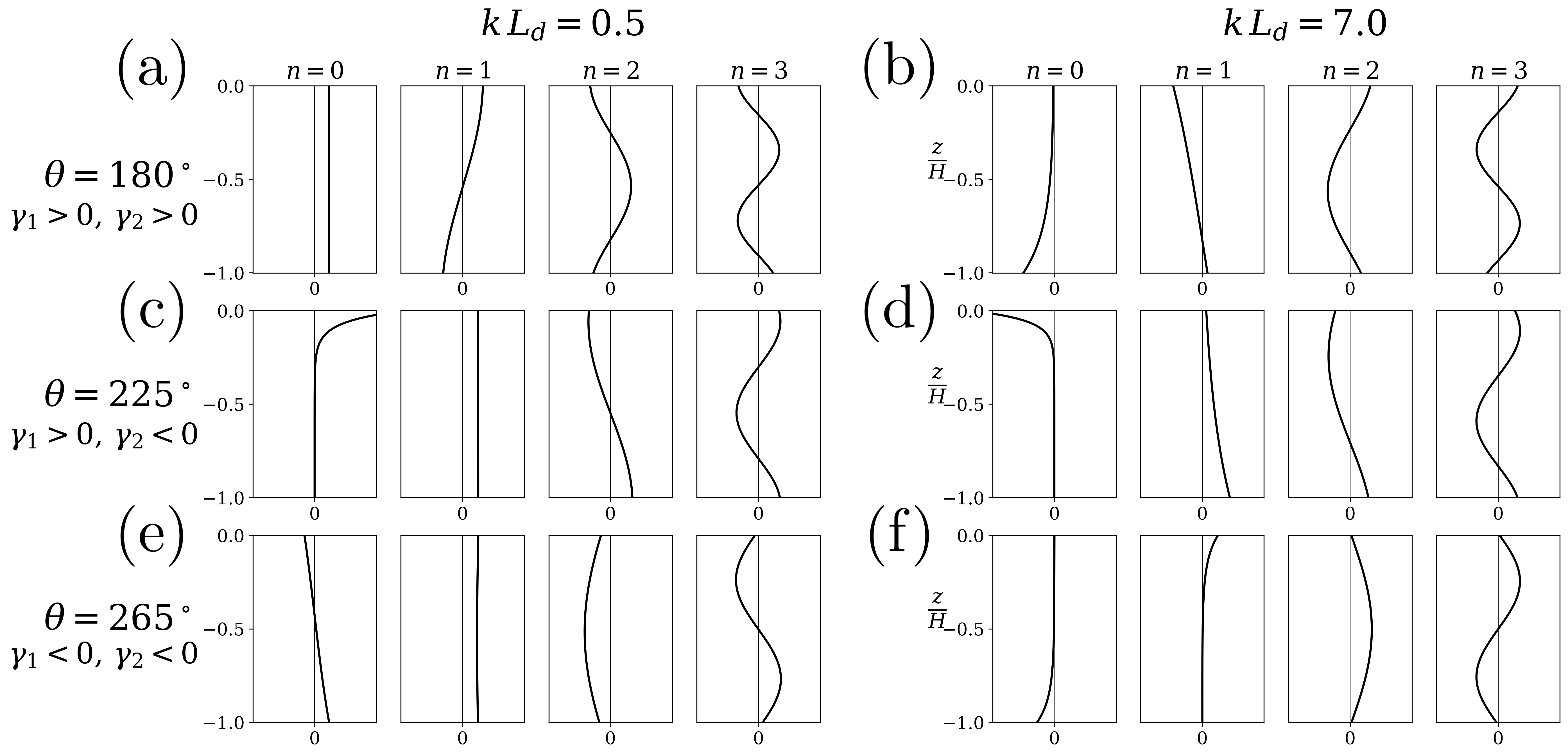} \\
	\caption{This figure illustrates the dependence of the vertical structure $G_n$ of the streamfunction to the horizontal wavevector $\vec k$ as discussed in section \ref{SS-rhines}. Three propagation directions are shown $\theta = 180^\circ, \, 225^\circ, \, 265^\circ$ and correspond to the rows in the figure [e.g., the row containing (a) and (b) are the vertical structures of waves at $\theta = 180^\circ$]; two wavenumbers $k\, L_d = 0.5, 7$ are shown (where $k=|\vec k|$) and they correspond to the columns in the above figure [e.g., (b), (d) and (f) are the vertical structure of waves with $k \, L_d = 7$]. The parameters for the above figure are identical to figure \ref{F-angle1}. We emphasize two features in this figure. First, note how the boundary modes ($n=0,1$) are typically only boundary-trapped at small horizontal scales (i.e., for $k \, L_{d} = 7$). At larger horizontal scales, we typically obtain a depth-independent mode along with another mode with large-scale features in the vertical. Second, note that for $\gamma_1,\gamma_2>0$, as in panels (a) and (b), the $n$th mode has $n$ internal zeros, as in Sturm-Liouville theory; for $\gamma_1>0,\gamma_2<0$, as in panels (c) and (d), the first two modes $(n=0,1)$ have no internal zeros; and for $\gamma_1,\gamma_2<0$, the zeroth mode $G_0$ has one internal zero, the first and second modes, $G_1$ and $G_2$ have no internal zeros, and the third mode $G_2$ has one internal zero. The zero-crossing for the $n=0$ mode in panel (f) is difficult to observe because the amplitude of $G_0$ is small near the zero-crossing.
	}
	\label{F-little}
\end{figure*}
 
There are now three cases depending on the signs of $\gamma_1$ and $\gamma_2$ and as depicted in figures \ref{F-angle2} and \ref{F-little}. In the following, we assume $k\neq0$. 
\begin{itemize}
	\item [i.] $\gamma_1 > 0$ and $\gamma_2 > 0$. Corresponds to case (i) in section \ref{SS-rhines-bottom}. See the regions in white in figure \ref{F-angle2} and plots (a) and (b) in figure \ref{F-little}.
	
	\item [ii.] $\gamma_1 \, \gamma_2 < 0$. This corresponds to case (ii) in section \ref{SS-rhines-bottom}. See the stippled regions in figure  and \ref{F-angle2} and plots (c) and (d) in figure \ref{F-little}.
	
	\item [iii.] $\gamma_1 < 0$ and $\gamma_2 < 0$. There are two modes $G_0$ and $G_1$ with negative squares, $\inner{G_n}{G_n}<0$, that propagate eastward and have eigenvalues, $G_n$, satisfying $G_n<-k^2$ for $n=1,2$. The remaining modes, $G_n$, for $n>1$ have positive squares, $\inner{G_n}{G_n}>0$, propagate westward, and have eigenvalues, $\lambda_n$, satisfying $\lambda_n>-k^2$. The zeroth mode, $G_0$, has one internal zero, the first and second modes, $G_1$ and $G_2$, have no internal zeros, and the remaining modes, $G_n$, have $n-2$ internal zeros for $n>2$ \citep{binding_left_1999}. See the shaded regions in figures \ref{F-angle1} and \ref{F-angle2} and panels (e) and (f) in figure \ref{F-little}.
\end{itemize}

\subsection{The vertical velocity eigenvalue problem}\label{SS-vertical-velocity}

Let $\hat w(z) = \hat w_0 \, \chi(z)$ where $\chi(z)$ is a non-dimensional function. For the Rossby waves with isentropic boundaries of section \ref{SS-trad} (the traditional baroclinic modes), the corresponding vertical velocity modes satisfy
\begin{equation}\label{vertical-baroclinic-interior}
	-\dd{\chi}{z} = \lambda  \left(\frac{N^2}{f_0^2}\right) \chi
\end{equation}
with vanishing vertical velocity boundary conditions
\begin{equation}\label{vertical-baroclinic-boundary}
	\chi(z_j) = 0
\end{equation}
(see appendix B for details). The resulting modes $\{\chi_n\}_{n=0}^\infty$ form an orthonormal basis of $L^2$ with orthonormality given by
\begin{equation}
	\delta_{mn} = \frac{1}{H} \intz \chi_m\, \chi_n \left(\frac{N^2}{f_0^2}\right)  \mathrm{d}z.
\end{equation}
One can obtain the eigenfunctions, $\chi_n$, by solving the eigenvalue problem \eqref{vertical-baroclinic-interior}\textendash\eqref{vertical-baroclinic-boundary} or by differentiating the streamfunction modes $F_n$ according to equation \eqref{psiz-w}.

\subsubsection*{Quasigeostrophic boundary dynamics}

As seen earlier, boundary buoyancy gradients activate  boundary dynamics in the quasigeostrophic problem. In this case, boundary conditions for the quasigeostrophic vertical velocity problem \eqref{vertical-baroclinic-interior} become
\begin{equation}\label{vertical-velocity-boundary}
	- (-1)^j \, \gamma_j \, k^2 \, \d{\chi}{z}\Big|_{z_j} = \lambda \left[\chi|_{z_j} + (-1)^j \, \gamma_j \, \d{\chi}{z}\Big|_{z_j} \right]
\end{equation}
(see the appendix B). The resulting modes $\{\chi_n\}_{n=0}^\infty$ satisfy a peculiar orthogonality relation given by equation \eqref{vertical-velocity-ortho}.

\begin{figure*}
	\noindent \includegraphics[width=\textwidth]{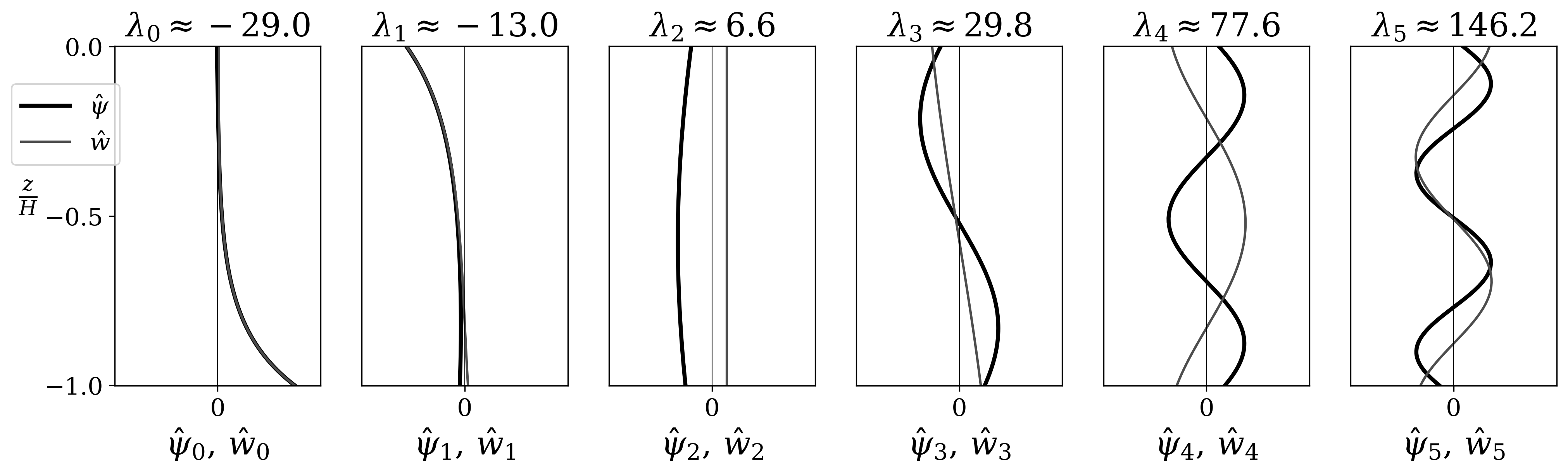} \\
	\caption{The first six vertical velocity normal modes $\chi_n$ (thin grey lines) and streamfunction normal modes $G_n$ (black lines) (see section \ref{SS-vertical-velocity}). The propagation direction is $\theta = 75^\circ$ with a wavenumber of $k\,L_d = 2$. The remaining parameters are as in figure \ref{F-angle1}. Note that $\chi_n$ and $G_n$ are nearly indistinguishable from the boundary-trapped modes $n=0,1$ while they are related by a vertical derivative for the internal modes $n>1$. The eigenvalue in the figure is non-dimensionalized by the deformation radius $L_d$.}
	\label{F-vertical-velocity}
\end{figure*}


\section{Eigenfunction expansions}\label{S-expansions}

Motivated by the Rossby waves of the previous section, we now investigate various sets of normal modes for quasigeostrophic theory. Let $\{F_n\}_{n=0}^\infty$ be a collection of $L^2$ normal modes, and assume $\psi_{\vec k}(z,t)$ is twice continuously differentiable in $z$. Define the eigenfunction expansion $\psi_{\vec k}^\textrm{exp}$ of $\psi$ by
\begin{equation}\label{F-expandion}
	\psi_{\vec k}^\textrm{exp}(z,t) = \sum_{n=0}^\infty \psi_{\vec k n}(t) \, F_n(z),
\end{equation}
where 
\begin{equation}
	\psi_{\vec k n} = \left[\psi_{\vec k},F_n\right].
\end{equation}
Because $\{F_n\}_{n=0}^\infty$ is a basis of $L^2$, the eigenfunction expansion $\psi_{\vec k}^\textrm{exp}$ satisfies \citep[e.g.,][] {brown_fourier_1993}
\begin{equation}\label{mean-square-equality}
	\intz |\psi_{\vec k}(z) - \psi_{\vec k}^\textrm{exp}(z)|^2 \mathrm{d}z = 0.
\end{equation}
Significantly, the vanishing of the integral \eqref{mean-square-equality} does not imply $\psi_{\vec k} = \psi_{\vec k}^\textrm{exp}$ because the two functions can still differ at some points $z \in [z_1,z_2]$.

In the following, we will only consider eigenfunctions expansions that diagonalize the energy and potential enstrophy integrals of section \ref{subsection:decomposing-energy-enstrophy}.

\subsection{The four possible $L^2$ modes}\label{SS-L2-modes}

There are only four $L^2$ bases in quasigeostrophic theory that diagonalize the energy and potential enstrophy integrals. All four sets of corresponding normal modes satisfy the differential equation
\begin{equation}
	- \d{}{z} \left(\frac{f_0^2}{N^2} \d{F}{z} \right) = \lambda \, F 
	 \quad z \in (z_1,z_2),
\end{equation}
but differ in boundary conditions according to the following (recall that $z_{1}$ is the bottom and $z_{2}$ the surface). 
\begin{itemize}
	\item \textit{Baroclinic modes}: Vanishing vertical velocity at both boundaries (Neumann),
		\begin{equation}
			\d{F(z_1)}{z} =0, \quad \d{F(z_2)}{z} = 0.
		\end{equation}
		\item \textit{Anti-baroclinic modes}: Vanishing pressure\footnote{Recall that the geostrophic streamfunction $\psi$ is proportional to pressure \cite[e.g.,][section 5.4]{Vallis2017}.} at both boundaries (Dirichlet),
		\begin{equation}
			F(z_1)=0, \quad F(z_2) = 0.
		\end{equation} 
	\item \textit{Surface modes}: (mixed Neumann/Dirichlet)
		\begin{equation}
			F(z_1) = 0, \quad \d{F(z_2)}{z} = 0.
		\end{equation}
	\item \textit{Anti-surface modes}: (mixed Neumann/Dirichlet)
		\begin{equation}
			\d{F(z_1)}{z} = 0, \quad F(z_2) = 0.
		\end{equation}
\end{itemize}
All four sets of modes are missing two modes. Each boundary condition of the form
\begin{equation}
	\d{F(z_j)}{z} = 0,
\end{equation}
implies a missing step-mode while a boundary condition of the form
\begin{equation}
	F(z_j) = 0,
\end{equation}
implies a missing boundary-trapped exponential mode [see the $\gamma_1\rightarrow \infty$ limit leading to equation \eqref{zero-boundary}].

\subsection{Expansions with $L^2$ modes}\label{SS-expansions-L2}

We here examine the pointwise convergence and the term-by-term differentiability of eigenfunction expansions in terms of $L^2$ modes. 
These properties of $L^2$ Sturm-Liouville expansions may be found in \cite{brown_fourier_1993} and \cite{levitan_introduction_1975}.\footnote{In particular, chapters 1 and 8 in \cite{levitan_introduction_1975} show that eigenfunction expansions have the same pointwise convergence and differentiability properties as the Fourier series with the analogous boundary conditions. The behaviour of Fourier series is discussed in \cite{brown_fourier_1993}.}

\subsubsection{Pointwise equality on $[z_1,z_2]$}

For all four sets of $L^2$ modes, if $\psi_{\vec k}$ is twice continuously differentiable in $z$, we obtain pointwise equality in the interior  
\begin{equation}
	\psi_{\vec k}(z) = \psi_{\vec k}^\textrm{exp}(z) \quad \textrm{for } z\in(z_1,z_z).
\end{equation}
The behaviour at the boundaries depends on the boundary conditions the modes $F_n$ satisfy. If the $F_n$ satisfy the vanishing pressure boundary condition at the $j$th boundary
\begin{equation}
	F_n(z_j) = 0 
\end{equation}
then
\begin{equation}
	\psi_{\vec k}^\textrm{exp}(z_j) = 0
\end{equation}
regardless of the values of $\psi_{\vec k}(z_j)$. It follows that $\psi_{\vec k}^\textrm{exp}$ will be continuous over $(z_1,z_2)$ and will generally have a jump discontinuity at the boundaries [unless $\psi_{\vec k}(z_j)=0$ for $j=1,2$]. In contrast, if the $F_n$ satisfy a zero vertical velocity boundary condition at the $j$th boundary
\begin{equation}
	\d{F_n(z_j)}{z} = 0
\end{equation} 
then 
\begin{equation}
	 \psi_{\vec k}(z_j) = \psi_{\vec k}^\textrm{exp}(z_j).
\end{equation}
Consequently, of the four sets of $L^2$ modes, only with the baroclinic modes do we obtain the pointwise equality $\psi_{\vec k}(z) = \psi_{\vec k}^\textrm{exp}(z)$ on the \emph{closed} interval $[z_1,z_2]$. 

However, even though $\psi_{\vec k}^\textrm{exp}$ converges pointwise to $\psi_{\vec k}$ when the baroclinic modes are used, we are unable to represent the corresponding velocity $w_{\vec k}$ in terms of the vertical velocity baroclinic modes since the modes vanish at both boundaries. Analogous considerations show that only the anti-baroclinic vertical velocity modes (see appendix B) can represent arbitrary vertical velocities.

\subsubsection{Differentiability of the series expansion}

Although  we obtain pointwise equality on the whole interval $[z_1,z_2]$ with the streamfunction baroclinic modes, we have lost two degrees of freedom in the expansion process. Recall that the degrees of freedom in the quasigeostrophic phase space are determined by the potential vorticity. The volume potential vorticity, $q_{\vec k}$, is associated with the $L^2$ degrees of freedom while the surface potential vorticities, $r_{1\vec k}$ and $r_{2 \vec k}$, are associated with the $\C^2$ degrees of freedom.

The series expansion $\psi_{\vec k}^\textrm{exp}$ of $\psi_{\vec k}$ in terms of the baroclinic modes is differentiable in the interior $(z_1,z_2)$. Consequently, we can differentiate the $\psi_{\vec k}^\textrm{exp}$ series for $z\in(z_1,z_2)$ to recover $q_{\vec k}$, that is,
\begin{equation}\label{q-trad-expansion}
	q_{\vec k}  = \sum_{n=0}^\infty  q_{\vec k n}\, F_n,
\end{equation}
where 
\begin{equation}\label{qn-trad-expansion}
	q_{\vec k n} = - (k^2 + \lambda_n) \, \psi_{\vec k n}.
\end{equation}
However, $\psi_{\vec k}^\textrm{exp}$ is not differentiable at the boundaries, $z=z_1,z_2$, so we are unable to recover the surface potential vorticities, $r_{1 \vec k}$ and $r_{2 \vec k}$. Two degrees of freedom are lost by projecting onto the baroclinic modes.\footnote{To see that $\psi_{\vec k}^\textrm{exp}$ is non-differentiable at $z=z_1,z_2$, suppose that the series $\psi_{\vec k}^\textrm{exp}$ is differentiable and that $\mathrm{d}\psi_{\vec k}(z_j)/\mathrm{d}z \neq 0$ for $j=1,2$. But then $$ 0 \neq \d{\psi_{\vec k}(z_j)}{z} = \sum_{n=0}^\infty \psi_{\vec k n} \d{F_n(z_j)}{z} = 0,$$ which is a contradiction.}

The energy at wavevector $\vec k$ is indeed partitioned between the modes,
\begin{align}\label{trad-energy}
	E_{\vec k} = \sum_{n=0}^\infty (k^2 + \lambda_n) \, \psi_{\vec k n},
\end{align}
and similarly for the potential enstrophy,
\begin{align}\label{trad-enstrophy-expansion}
		Z_{\vec k} = \sum_{n=0}^\infty (k^2 + \lambda_n)^2 \psi_{\vec k n}.
\end{align}
However, as we have lost $r_{1 \vec k}$ and $r_{2 \vec k}$ in the projection process, the surface potential enstrophies $Y_{1 \vec k}$ and $Y_{2 \vec k}$, defined in equation \eqref{surface-enstrophy-rjk}, are not partitioned.

\subsection{Quasigeostrophic $L^2 \oplus \C^2$ modes}

Consider the eigenvalue problem
\begin{subequations}\label{SV-eigen}
\begin{align}
			- \d{}{z} \left(\frac{f_0^2}{N^2} \d{G}{z}\right) = \lambda \, G \quad \textrm{for } z\in(z_1,z_2)&\\
			-k^2 G + (-1)^j D_j^{-1} \, \left(\frac{f_0^2}{N^2}\d{G}{z}\right) = \lambda \, G \quad \textrm{for } z=z_j&
\end{align}
\end{subequations}
where $D_1$ and $D_2$ are non-zero real constants. This eigenvalue problem differs from the generalized Rhines eigenvalue problem \eqref{Rhines-eigen} in that $D_j$ are generally not equal to the $\gamma_j$ defined in equation \eqref{gamma}. The inner product $\inner{\cdot}{\cdot}$ induced by the eigenvalue problem \eqref{SV-eigen} is given by equation \eqref{Rhines-inner-general} with the $\gamma_j$ replaced by the $D_j$.

\cite{smith_surface-aware_2012} investigate an equivalent eigenvalue problem to \eqref{SV-eigen} and conclude that, when $D_1$ and $D_2$ are positive, the resulting eigenfunctions form a basis of $L^2\oplus \C^2 $. However, such a completeness result is insufficient for the Rossby wave problem of  section \ref{SS-rhines}, in which case $D_j=\gamma_j$ and $\gamma_j$ can be negative. 


\subsection{Expansion with $L^2\oplus \C^2$ modes}

When $D_1,D_2$ in the eigenvalue problem \eqref{SV-eigen} are finite and non-zero, the resulting eigenmodes $\{G_n\}_{n=0}^\infty$ form a basis for the vertical structure phase space $L^2\oplus \C^2$. Thus, the projection 
\begin{equation}
	\psi_{\vec k}^\textrm{exp}(z) = \sum_{n=0}^\infty \psi_{\vec k n} \, G_n(z)
\end{equation}
where
\begin{equation}
	\psi_{\vec k n} =  \frac{\inner{\psi_{\vec k}}{G_n}}{\inner{G_n}{G_n}}
\end{equation}
is an \emph{equivalent} representation of $\psi_{\vec k}$. Not only do we have pointwise equality
\begin{equation}
	\psi_{\vec k}(z) = \psi_{\vec k}^\textrm{exp}(z) \quad \textrm{for } z\in[z_1,z_2],
\end{equation}
but the series $\psi_{\vec k}^\textrm{exp}$ is also differentiable on the \emph{closed} interval $[z_1,z_2]$ [the case of $D_j>0$ is due to \cite{fulton_two-point_1977} whereas the case of $D_j<0$ is from chapter  \ref{Ch-modes}.]. Thus given $\psi_{\vec k}^\textrm{exp}$, we can differentiate to obtain both $q_{\vec k}$ and $r_{j \vec k}$ and thereby recover all quasigeostrophic degrees of freedom. Indeed, we have 
\begin{align}
	q_{\vec k}(z,t) &= \sum_{n=0}^\infty q_{\vec k n}(t) \, G_n(z),\\
	r_{j \vec k}(t) &= \sum_{n=0}^\infty r_{j \vec k n}(t) \, G_n(z_j),
\end{align}
where
\begin{align}
	q_{\vec k n} &= - (k^2+\lambda_n) \frac{\inner{\Psi}{G_n}}{\inner{G_n}{G_n}}, \\
	\label{rj_q_relation}
	r_{j \vec k n} &= D_j \, q_{\vec k n},
\end{align}
for $j=1,2$. 

In addition, the energy, $E_{\vec k}$, volume potential enstrophy, $Z_{\vec k}$, and surface potential enstrophies, $Y_{1\vec k}$ and $Y_{2\vec k}$, are partitioned (diagonalized) between the modes
\begin{align}\label{Rhines-energy}
	E_{\vec k} &=  \sum_{n=0}^\infty (k^2 + \lambda_n) \psi_{\vec k n}, \\
	Z_{\vec k} + \frac{1}{H} \sum_{j=1}^2 \frac{1}{D_j} Y_{j\vec k} &= \sum_{n=0}^\infty (k^2 + \lambda_n)^2 \psi_{\vec k n}.
\end{align}

\section{Discussion}\label{S-discussion}

The traditional baroclinic modes are useful since they are the vertical structures of linear Rossby waves in a resting ocean and they can be used for wave-turbulence studies such as in \cite[e.g.,][]{hua_numerical_1986,smith_scales_2001}. Therefore, any basis we choose should not only be complete in $L^2\oplus \C^2$, but should also represent the vertical structure of Rossby waves in the linear (quiescent ocean) limit. Such a basis would then amenable to wave-turbulence arguments and can permit a dynamical interpretation of field observations. The basis suggested by \cite{smith_surface-aware_2012} does not correspond to Rossby waves in the linear limit. It is a mathematical basis with two-independent parameters $D_1,D_2>0$ that diagonalizes the energy and potential enstrophy integrals. 

The Rhines modes of section \ref{SS-rhines-bottom} offer a basis of $L^2\oplus \C$ that corresponds to Rossby wave over topography in the linear limit. These Rhines modes do not contain any free parameters. Indeed, if we set $D_2=0$ in the eigenvalue problem \eqref{SV-eigen} and let $D_1=\gamma_1$, we then obtain the Rhines modes. Note that since  $D_1=\gamma_1=\gamma_1(\vec k)$ may be negative, the \cite{smith_surface-aware_2012} modes do not apply. Instead, the case of negative $D_j$ is examined in this chapter and in chapter \ref{Ch-modes}.

However, the Rhines modes, as a basis of $L^2\oplus \C$ are not a basis of the whole vertical structure phase space $L^2\oplus \C^2$ since they exclude surface buoyancy anomalies at the upper boundary. To solve this problem, we can use the modes of the eigenvalue problem \eqref{SV-eigen} with $D_1=\gamma_1$ but leaving $D_2$ arbitrary as in \cite{smith_surface-aware_2012}. Although this basis now only has one free parameter, $D_2$, it still does not correspond to Rossby waves in the linear limit. We can even eliminate this free parameter by interpreting surface buoyancy gradients as topography e.g., by defining 
\begin{equation}
	g_\textrm{buoy} = \left[\frac{f_0^2}{N^2} \d{\psi_B}{z}\right]_{z=z_2}
\end{equation}
where $\psi_B$ corresponds to the background flow, and using $g_\textrm{buoy}$ in place of $g_2$ in the generalized Rhines modes of section \ref{SS-rhines}. However the waves resulting from topographic gradients generally differ from those resulting from vertically-sheared mean-flows (in particular, one must take into account advective continuum modes) and so this resolution is artificial.

\subsection*{Galerkin approximations with $L^2$ modes}
Both the $L^2$ baroclinic modes and the $L^2\oplus \C^2$ modes have infinitely many degrees of freedom. In contrast, numerical simulations only contain a finite number of degrees of freedom. Consequently, it should be possible to use baroclinic modes to produce a Galerkin approximation to quasigeostrophic theory with non-trivial boundary dynamics. Such an approach has been proposed by \cite{rocha_galerkin_2015}.

Projecting $\psi_{\vec k}$ onto the baroclinic modes produces a series expansion, $\psi_{\vec k}^\textrm{exp}$, that is differentiable in the interior but not at the boundaries. By differentiating the series \emph{in the interior} we obtain equation \eqref{qn-trad-expansion} for $q_{\vec k n}$. If instead we integrate by parts twice and avoid differentiating $\psi^\textrm{exp}_{\vec k}$, we obtain
\begin{equation}\label{rocha-qn2}
	q_{\vec k n} = -(k^2+\lambda_n) \psi_{\vec k n} - \frac{1}{H} \sum_{j=1}^2 r_{j\vec k} \, F_n(z_j).
\end{equation}
The two expressions \eqref{qn-trad-expansion} and \eqref{rocha-qn2} are only equivalent when $r_{1\vec k}=r_{2\vec k}=0$. For non-zero $r_{1\vec k}$ and $r_{2\vec k}$, the singular nature of the expansion means we have a choice between equations \eqref{qn-trad-expansion} and \eqref{rocha-qn2}.

By choosing equation \eqref{rocha-qn2} and avoiding the differentiation of $\psi^\textrm{exp}_{\vec k}$, \cite{rocha_galerkin_2015} produced a least-squares approximation to quasigeostrophic dynamics that conserves the surface potential enstrophy integrals \eqref{surface-enstrophy}. This is a conservation property underlying their approximation's success.

\section{Conclusion}
\label{S-conclusion}

In this chapter, we have studied all possible non-continuum collections of streamfunction normal modes that diagonalize the energy and potential enstrophy. There are four possible $L^2$ modes: the baroclinic modes, the anti-baroclinic modes, the surface modes, and the anti-surface modes. Additionally, we explored the properties of the family of $L^2\oplus\C^2$ bases introduced by \cite{smith_surface-aware_2012} which contain two free parameters $D_1,D_2$ and generalized the family to allow for $D_1,D_2<0$. This generalization is necessary for Rossby waves in the presence of bottom topography. If $D_j=\gamma_j$, where $\gamma_j$ is given by equation \eqref{gamma} for $j=1,2$, the resulting modes are the vertical structure of Rossby waves in a quiescent ocean with prescribed boundary buoyancy gradients (i.e., topography). We have also examined the associated $L^2$ and $L^2\oplus \C^2$ vertical velocity modes.

For the streamfunction $L^2$ modes, only the baroclinic modes are capable of converging pointwise to any quasigeostrophic state on the interval $[z_1,z_2]$, whereas for the vertical velocity $L^2$ modes, only the anti-baroclinic modes are capable. However, in both cases, the resulting eigenfunction expansion is not differentiable at the boundaries, $z=z_1,z_2$. Consequently, while we can recover the volume potential vorticity density, $q_{\vec k}$, we cannot recover the surface potential vorticity densities, $r_{1\vec k}$ and $r_{2 \vec k}$. Thus, we lose two degrees of freedom when projecting onto the baroclinic modes. In contrast, $L^2\oplus \C^2$ modes provide an equivalent representation of the function in question. Namely, the eigenfunction expansion is differentiable on the closed interval $[z_1,z_2]$ so that we can recover $q_{\vec k}$, $r_{1 \vec k}$, $r_{2 \vec k}$ from the series expansion.

We have also introduced a new set of modes, the Rhines modes, that form a basis of $L^2\oplus \C$ and correspond to the vertical structures of Rossby waves over topography. A natural application of these normal modes is to the study of weakly non-linear wave-interaction theories of geostrophic turbulence found in \cite{fu_nonlinear_1980} and \cite{smith_scales_2001}, extending their work to include bottom topography.
 
\begin{subappendices}

\section{Sturm-Liouville eigenvalue problems with $\lambda$-dependent boundary conditions}\label{A-math}

Consider the differential eigenvalue problem
\begin{equation} \label{EigenDiff}
	-\d{}{z}\left(p \, \d{F}{z}\right) + q \, F = \lambda \, r \, F, 
\end{equation}
in the interval $(z_1,z_2)$ with boundary conditions
\begin{equation}\label{EigenB}
	- \left[a_j F - b_j  \left(p  \d{F}{z}\right)(z_j)\right] =
		 \lambda \left[c_j  F(z_j) - d_j \left(p \d{F}{z}\right)(z_j)\right]
\end{equation}
for $j=1,2$, where $1/p(z), q(z), r(z)$ are real-valued integrable functions and $a_j,b_j,c_j,d_j$ are real numbers. Moreover, we assume $p>0, r>0$, that $p$ and $r$ are twice continuously differentiable, that $q$ is continuous, and that $(a_j,b_j)\neq(0,0)$. 

Define the two boundary parameters $D_j$ for $j=1,2$ by
\begin{equation}
	D_j = (-1)^{j+1} \left(a_j\,d_j - b_j\,c_j\right).
\end{equation}
Then the natural inner product for the eigenvalue problem is given by
\begin{equation}\label{appendix-inner}
	\inner{F}{G} = \intz F\,G\,\mathrm{d}z + \sum_{j=1}^{2} D_j^{-1}\left(\mathcal{C}_jF\right) \,\left( \mathcal{C}_j G\right)
\end{equation}
where the boundary operator $\mathcal{C}_j$ is defined by
\begin{equation}
	\mathcal{C}_j F = c_j \, F(z_j) - d_j \,\left( p\,\d{F}{z}\right)(z_j).
\end{equation}

The eigenvalue problem takes place in the space $L^2\oplus \C^N$ where $N$ is the number of non-zero $D_j$. Assume for the following that $N=2$; the case when $N=1$ is similar. If 
\begin{equation}\label{D_j_positive}
	D_j>0
\end{equation}
for $j=1,2$ then the inner product \eqref{appendix-inner} is positive definite\textemdash that is, all non-zero $F$ satisfy $\inner{F}{F}>0$. Therefore $L^2\oplus \C^2$, equipped with the inner product \eqref{appendix-inner}, is a Hilbert space. In this Hilbert space settings, the eigenfunctions $\{F_n\}_{n=0}^\infty$ form and orthonormal basis of $L^2\oplus \C^2$ and that the eigenvalues distinct and bounded below as in equation \eqref{eigen-order} \citep{evans_non-self-adjoint_1970,walter_regular_1973,fulton_two-point_1977}. The appendix of \cite{smith_surface-aware_2012} also proves this result in the case when $d_1=d_2=0$. The convergence properties of normal mode expansions in this case are due to \cite{fulton_two-point_1977}.

However, as we observe in section \ref{S-linear}, the $D_j>0$ case is not sufficient for the Rossby wave problem with topography. In general, the space $L^2\oplus \C^2$ with the indefinite inner product \eqref{appendix-inner} is a Pontryagin space \citep[see][]{iohvidov_spectral_1960,bognar_indefinite_1974}. Pontryagin spaces are analogous to Hilbert spaces except that they have a finite-dimensional subspace of elements satisfying $\inner{F}{F}<0$. If $\Pi$ is a Pontryagin space with inner product $\inner{\cdot}{\cdot}$, then $\Pi$ admits a decomposition 
\begin{equation}
	\Pi = \Pi^{+} \oplus \Pi^{-},
\end{equation}
where $\Pi^{+}$ is a Hilbert space under the inner product $\inner{\cdot}{\cdot}$ and $\Pi^{-}$ is a finite-dimensional Hilbert space under the inner product $-\inner{\cdot}{\cdot}$. If $\{G_n\}_{n=0}$ is an orthonormal basis for the Pontryagin space $\Pi$, then an element $\Psi\in\Pi$ can be expressed
\begin{equation}\label{pont-exp}
	\Psi = \sum_{n=0}^{\infty} \frac{\inner{\Psi}{G_n}}{\inner{G_n}{G_n}}.
\end{equation}
Even though $\{G_n\}_{n=0}^\infty$ is normalized, the presence of $\inner{G_n}{G_n}=\pm 1$ in the denominator of equation \eqref{pont-exp} is essential since this term may be negative.

One can rewrite the eigenvalue problem \eqref{EigenDiff}\textendash\eqref{EigenB} in the form $\mathcal{L}\, F = \lambda \, F$ for some operator $\mathcal{L}$ \cite[e.g.,][]{langer_spectral_1991}. The operator $\mathcal{L}$ is a positive operator if 
\begin{itemize}
	\item  for the $\lambda$-dependent boundary conditions, we have 
			\begin{equation}\label{left-definite-1}
				\frac{a_i \, c_i}{D_i} \leq 0, \quad \frac{b_i \, d_i}{D_i} \leq 0, \quad (-1)^i \frac{a_i \, d_i}{D_i} \geq 0
			\end{equation} 
	\item for the $\lambda$-independent boundary conditions, we have 
			\begin{equation}\label{left-definite-2}
				b_i = 0  \quad \text{or} \quad  (-1)^{i+1}\frac{a_i}{b_i} \geq 0 \quad \text{ if } b_i \neq 0.
			\end{equation}
\end{itemize}
Chapter \ref{Ch-modes} shows that, if $\mathcal{L}$ is positive, the eigenfunctions $\{F_n\}_{n=0}^\infty$ of the eigenvalue problem \eqref{EigenDiff}\textendash\eqref{EigenB} form an orthonormal basis of $L^2\oplus{\C}^2$, that the eigenvalues are real, and that the eigenvalues are ordered as in equation \eqref{eigen-order}. Moreover, since $\mathcal{L}$ is positive, we have the relationship
\begin{equation}
	\lambda \inner{F}{F} = \inner{\mathcal{L}F}{F} \geq 0.
\end{equation}
Finally, chapter \ref{Ch-modes} shows that the normal mode expansion results of \cite{fulton_two-point_1977} extend to this case as well.

\section{Polarization relations and the vertical velocity eigenvalue problem}

\subsection{Polarization relations}
 

The linear quasigeostrophic vorticity and buoyancy equations, computed about a resting background state, are
\begin{align}
	\label{zeta-equation-linear}
	\pd{\zeta}{t} + \beta \, \pd{\psi}{x} &= f_0 \pd{w}{z}, \\
	\label{b-equation-linear}
	\pd{b}{t} &= -N^2 \, w,
\end{align}
in the interior $z\in(z_1,z_2)$. The vorticity, $\zeta$, and buoyancy, $b$, are given in terms of the geostrophic streamfunction via 
\begin{align}
	\label{zeta}
	\zeta = \lap \psi \\
	\label{b}
	b = f_0 \pd{\psi}{z}.
\end{align} 
The no-normal flow at the lower and upper boundaries implies
\begin{equation}\label{vv-boundary}
	f_0 \, w = \vec u \cdot \grad g_j,
\end{equation}
for $j=1,2$. Substituting equation \eqref{vv-boundary} into the linear buoyancy equation \eqref{b-equation-linear}, yields the boundary conditions
\begin{equation}\label{buoyancy-boundary-evolution}
	\partial_t b + \vec u \cdot \grad \left(\frac{N^2}{f_0} g_j\right) = 0 \quad \textrm{for } z=z_j.
\end{equation}

We now assume solutions of the form 
\begin{equation}
	\psi= \hat \psi(z) \, e_{\vec k}(\vec x) \, \mathrm{e}^{-\mathrm{i}\omega t},
\end{equation}
and similarly for $w$. Substituting such solutions into equations \eqref{zeta-equation-linear}\textendash\eqref{b-equation-linear} and using $\vec u = \unit z \times \grad \psi$ gives
\begin{align}
	\label{psiz-w}
	\d{\hat \psi}{z} &= -\mathrm{i} \, \frac{N^2}{f_0\, \omega} \hat w\\
	\label{wz-psi}
	\d{\hat w}{z} &= \mathrm{i} \, \frac{\omega}{f_0} \left[k^2 + \frac{\beta \, k_x}{\omega} \right] \hat \psi,
\end{align}
for $z \in (z_1,z_2)$. 
At the boundaries $z=z_1,z_2$, we use equations \eqref{vv-boundary} and \eqref{buoyancy-boundary-evolution} to obtain
\begin{align}
	\label{bhat-boundary}
	\hat b &= - \frac{N^2}{f_0 \, \omega} \, \hat{\vec u} \cdot \grad g_j \\
	\label{what-boundary}
	\hat w &= \mathrm{i} \frac{1}{f_0} \, \hat{\vec u} \cdot \grad g_j.
\end{align}

\subsection{The vertical velocity eigenvalue problem}

Taking the vertical derivative of \eqref{wz-psi} and using \eqref{psiz-w} yields
\begin{equation}\label{vertical-velocity-interior}
	- \dd{\chi}{z} = \lambda \left(\frac{N^2}{f_0^2}\right)  \chi,
\end{equation}
where $\hat w = w_0 \, \chi(z)$ and $\chi$ is non-dimensional. The boundary conditions at $z = z_{j}$ are
\begin{equation}\label{vertical-velocity-boundary-again}
	- (-1)^j \, \gamma_j \, k^2 \, \d{\chi}{z} = \lambda \left[\chi + (-1)^j \, \gamma_j \, \d{\chi}{z} \right],
\end{equation}
as obtained by using equations \eqref{wz-psi} and \eqref{psiz-w} in boundary conditions \eqref{Rhines-eigen-boundary}. The orthonormality condition is
\begin{align}\label{vertical-velocity-ortho}
	\pm \delta_{mn} = \frac{1}{H} \left[ \intz \chi_m \, \chi_n \left( \frac{N^2}{f_0^2}\right) \mathrm{dz}   -\frac{1}{k^2} \sum_{j=1}^2 \frac{1}{\gamma_j} \left(\mathcal{C}_j \chi_m \right) \left(\mathcal{C}_j \chi_n\right)  \right],
\end{align}
where 
\begin{equation}
	\mathcal{C}_j \chi = \chi(z_j) + (-1)^j \, \gamma_j \d{\chi(z_j)}{z}.
\end{equation}

When only one boundary condition is $\lambda$-dependent (e.g., $\gamma_2=0$) the eigenvalue problem \eqref{vertical-velocity-interior}\textendash\eqref{vertical-velocity-boundary-again} satisfies equation \eqref{D_j_positive} when $\gamma_1>0$ and equations \eqref{left-definite-1} and \eqref{left-definite-2} when $\gamma_1<0$; thus the reality of the eigenvalues and the completeness results follow. However, when both boundary conditions are $\lambda$-dependent the problem no longer satisfies these conditions for all $\vec k$. Instead, in this case, one exploits the relationship between the vertical velocity eigenvalue problem \eqref{vertical-velocity-interior}\textendash\eqref{vertical-velocity-boundary-again} and the streamfunction problem \eqref{Rhines-eigen-interior}\textendash\eqref{Rhines-eigen-boundary} given by equations \eqref{psiz-w} and \eqref{wz-psi} to conclude that the two problem have the identical eigenvalues (for $\omega \neq 0$) and then use the simplicity of the eigenvalues to conclude that no generalized eigenfunctions can arise.

\subsection{The vertical velocity $L^2$  modes} 
 
Analogously with the streamfunction $L^2$ modes, we have the following sets of vertical velocity $L^2$ modes.

\begin{itemize}
	\item \textit{Baroclinic modes}: Vanishing vertical velocity at both boundaries,
	\begin{equation}
		\chi(z_1) = 0, \quad \chi(z_2) = 0.
	\end{equation}
	\item \textit{Anti-baroclinic modes}: Vanishing pressure at both boundaries,
	\begin{equation}
		\d{\chi(z_1)}{z} = 0, \quad \d{\chi(z_2)}{z}=0.
	\end{equation}
	\item \textit{Surface modes}:
	\begin{equation}
		\d{\chi(z_1)}{z} = 0, \quad \chi(z_2) = 0.
	\end{equation}
	\item \textit{Anti-surface modes}:
	\begin{equation}
		\chi(z_1) =0, \quad \d{\chi(z_2)}{z} = 0. 
	\end{equation} 
\end{itemize} 

\end{subappendices}

%% file: 6-Conclusion/Conclusion.tex
\chapter{Conclusion}\label{Ch-conc}

\section{Modal truncations with non-isentropic boundaries}\label{SS-nogo_thm}

We now show that no energy conserving modal truncation of the quasigeostrophic equations is possible in the presence of non-isentropic boundaries. Consider a fluid with some linear bottom topography, $h_1$, but with an isentropic upper boundary. Then the appropriate vertical modes are given by the Rhines eigenvalue problem \eqref{Rhines-eigen_bottom}. We obtain modes $\varphi_{\vec k 0}, \, \varphi_{
\vec k 1}, \, \varphi_{\vec k 2},\, \dots$ with corresponding eigenvalues
\begin{equation}
	\lambda_{\vec k 0}  < \lambda_{\vec k 1} < \lambda_{\vec k 2} < \cdots \rightarrow \infty.
\end{equation}
The eigenfunctions are orthonormal with respect to the inner product 
\begin{equation}\label{eq:inner_conc}
	\inner{F}{G}_{\vec k} = \frac{1}{H}\left( \intz F \, G \, \mathrm{d}z + \gamma_1(\vec k) \, F(z_1)\, G(z_1) \right),
\end{equation}
where the lower boundary parameter is
\begin{align}
	\gamma_1(\vec k) &= \frac{\unit z\cdot\left(\vec k \times f_0\grad  \,h_1\right)}{\unit z\cdot\left(\vec k \times \grad f\right)}.
\end{align}
Given a streamfunction satisfying $\partial_z \psi = 0$ at the upper boundary, we have the expansion
\begin{equation}\label{rhines-evolution}
	\psi_{\vec k} (z,t) = \sum_{n=0}^\infty \psi_{\vec k n}(t) \, \varphi_{\vec k n}(z),
\end{equation}
where $\psi_{\vec k}(z,t)$ is the amplitude of the horizontal Fourier expansion \eqref{horizontal_fourier_amplitudes}, and 
\begin{equation}
	\psi_{\vec k n} = \frac{\inner{\psi_{\vec k}}{\varphi_{\vec k n}}_{\vec k}}{\inner{\varphi_{\vec k n}}{\varphi_{\vec k n}}_{\vec k}}
\end{equation}
is the amplitude of the vertical mode $n$ in the expansion of the vertical structure $\psi_{\vec k}(z)$.

Substituting the horizontal Fourier expansion \eqref{horizontal_fourier_amplitudes} into the time-evolution equations
	\begin{align}
	\label{eq:time_q}
		\pd{q}{t} + \beta \pd{\psi}{x} +  \J{\psi}{q}  &= 0 \quad \text{for } z\in(z_1,z_2),\\
	\label{eq:time_r}
		\pd{r_1}{t} + \unit z\cdot\left(f_0 \grad  \, h_1 \times \grad \psi\right) +  \J{\psi}{r_1}  &= 0 \quad \text{at } \text{for } z=z_1,
\end{align}
we obtain
	\begin{align}
		\label{eq:time_q_f}
		\pd{q_{\vec k} }{t} + \mathrm{i}\, \unit z \cdot \left(\vec k  \times \grad f  \right) \psi_{\vec k} +  \sum_{\vec a \vec b} A_{\vec a \vec b \vec k} \, \psi_{\vec a} \,  q_{\vec b} = 0  &\quad \text{for } z\in(z_1,z_2),\\
		\label{eq:time_r_f}
		\pd{r_{1 \vec k}}{t} + \mathrm{i} \, \unit z \cdot \left(\vec k \times  f_0 \grad  h_1  \right) \psi_{\vec k} +  \sum_{\vec a \vec b} A_{\vec a \vec b \vec k} \, \psi_{\vec a} \,  r_{1 \vec b} = 0 &\quad \text{at } z=z_1,
	\end{align}
where the horizontal coupling coefficient is given by
\begin{equation}
	A_{\vec a \vec b \vec k} = - \unit z \cdot \left(\vec a \times \vec b\right) \, \delta_{\vec a + \vec b, \vec k}.
\end{equation}
To combine the two Fourier space time-evolution equations \eqref{eq:time_q_f}  and  \eqref{eq:time_r_f} into a single equation for the modal amplitudes, we expand the interior potential vorticity as 
	\begin{equation}
		q_{\vec k n}(z) = \sum_{n=0}^\infty q_{\vec k n} \, \varphi_{\vec k n}(z) \quad \text{for } z\in(z_1,z_2),
	\end{equation}
	and the surface potential vorticity as 
	\begin{equation}
		r_{1 \vec k n} = \sum_{n=0}^\infty r_{1 \vec k n} \, \varphi_{\vec k n}(z_1),
	\end{equation}
	where $q_{\vec k}$ and $r_{1 \vec k}$ are related to $\psi_{\vec k}$ through the Fourier transforms of their physical space diagnostic relation [equations \eqref{pv-hor-transform}], and where
	\begin{equation}\label{eq:qkn_conc}
		 q_{\vec k n} =- \lambda_{\vec k n} \, \psi_{\vec k n}
	\end{equation}
	is the modal amplitude of the interior potential vorticity
	 and  
	 \begin{equation}\label{eq:rkn_conc}
	 	r_{1 \vec k n} = - \gamma_1(\vec k) \, \lambda_{\vec k n} \, \psi_{\vec k n}
	 \end{equation}
	is the modal amplitude of the surface potential vorticity. 
Then substituting these two series expansions into the Fourier space time-evolution equations \eqref{eq:time_q_f} and \eqref{eq:time_r_f} and using the identities \eqref{eq:qkn_conc} and \eqref{eq:rkn_conc}, we obtain
	\begin{equation}
		\sum_{n=0}^\infty \left[ \d{q_{\vec k n}}{t} + \mathrm{i} \, \beta \, k_x \psi_{\vec k n} \right] \varphi_{\vec k n} + \sum_{\vec a \vec b} \sum_{l m } A_{\vec a \vec b \vec k}\, \psi_{\vec a l}\, q_{\vec b m}\, \varphi_{\vec a l}\,  \varphi_{\vec b m} = 0,
	\end{equation}
	if $\gamma_1(\vec k ) \neq 0$. Applying the inner product $\inner{\varphi_{\vec k n}}{\cdot}_{\vec k}$ [equation \eqref{eq:inner_conc}] to this equation then gives the time-evolution equation for modal amplitudes
	\begin{equation}\label{eq:time_modal_boundary}
		\inner{ \varphi_{\vec k n}}{\varphi_{\vec k n}}_{\vec k} \left( \d{q_{\vec k n}}{t} + \mathrm{i} \, \beta \, k_x \, \psi_{\vec k n} \right)+  \sum_{\vec a , \vec b} \sum_{ l m} A_{\vec a \vec b \vec k} \, \varepsilon_{l m n}^{\vec a \vec b \vec k} \, \psi_{\vec a l} \, q_{\vec b m} = 0,
	\end{equation}
	where the vertical coupling coefficient is
	\begin{equation}\label{eq:vertical_coupling_boundary}
		\varepsilon_{l m n}^{\vec a \vec b \vec k} = \inner{\varphi_{\vec a l} \, \varphi_{\vec b m}}{\varphi_{\vec k n}}_{\vec k} = \frac{1}{H} \left( \int_{z_1}^{z_2} \varphi_{\vec a l}\,  \varphi_{\vec b m} \, \varphi_{\vec k n} \mathrm{d}z  + \gamma_1(\vec k) \, \left[\varphi_{\vec a l}\,  \varphi_{\vec b m} \, \varphi_{\vec k n} \right]|_{z=z_1}  \right).
	\end{equation}

With isentropic boundaries, the vertical coupling coefficient \eqref{eq:vertical_coupling_boundary} is independent of the wavectors of the interacting modes. However, with non-isentropic boundaries, the vertical coupling coefficient depends on both the propagation direction as well as the horizontal length scale of the interacting modes. Multiplying the modal time-evolution equation \eqref{eq:time_modal_boundary} by the complex conjugate, $\psi_{\vec k n}^*$, taking the real part, and then summing over $\vec k$ and $n$ gives the energy equation
	\begin{equation}
		\d{}{t} \left( \sum_{\vec k n} \frac{1}{2} \lambda_{\vec k n} \inner{\varphi_{\vec k n}}{\varphi_{\vec kn}}_{\vec k}  \abs{\psi_{\vec k n}}^2  \right)
		+ \sum_{\vec a \vec b \vec k }\sum_{l m n} A_{\vec a \vec b \vec k} \, \varepsilon_{l m n}^{\vec a \vec b \vec k}\,  \Re\left\{\psi_{\vec a l} \, q_{\vec k - \vec a \, m} \psi_{\vec k n}^* \right\} = 0.
	\end{equation}
	If we truncate at $n=N$, the nonlinear sum does not vanish because the modal interaction
	\begin{equation}
		(\vec a, l) + (\vec b, m) \rightarrow (\vec k, n),
	\end{equation} 
	no longer provides the opposite contribution to the energy as the modal interaction 
	\begin{equation}
		(\vec k, n) + (\vec b, m) \rightarrow (\vec a, l)
	\end{equation} 
	because $\varepsilon_{lmn}^{\vec a \vec b \vec k} \neq  \varepsilon_{nml}^{\vec k \vec b \vec a}$.
	Therefore, modal truncations do not conserve a truncated form of the energy. 
	
Physically, the inability of modally truncated models to conserve a truncated energy means the following. Suppose we initialize a quasigeostrophic state so that there is energy only in the lowest $N$ vertical modes. For a quasigeostrophic system with isentropic boundaries, the energy will remain in the lowest $N$ modes for all time; we can view this trapping of the energy in the lowest modes as a consequence of the vertical inverse cascade \citep{charney_geostrophic_1971}. As a result, if we truncate the model at some $n=N$, the truncated model conserves a truncated energy. In contrast, for a quasigeostrophic system with non-isentropic boundaries, the energy does not necessarily remain in the lowest $N$ vertical modes and energy exchanges with the higher modes are possible. Because of these energy exchanges with the higher modes, any truncation at $n=N$ does not conserve energy.

We can further examine the nature of these energy exchanges by considering the form of the vertical coupling coefficient, $\varepsilon_{l m n}^{\vec a \vec b \vec k}$ in equation \eqref{eq:vertical_coupling_boundary}. The energy exchanges between the lowest $N$ modes and the higher modes is a consequence of the term multiplying $\gamma_1(\vec k)$, which couples the vertical modes at the lower boundary. As $n$ becomes large, then $\lambda_{\vec k n} \rightarrow \infty$ and so we obtain an approximate bottom boundary condition of $\varphi_{\vec k n}\approx 0$ in the Rhines eigenvalue problem \eqref{Rhines-eigen_bottom}. Therefore, for high vertical modes (those with large $n$), the term multiplying  $\gamma_1(\vec k)$ in the vertical coupling coefficient \eqref{eq:vertical_coupling_boundary} is negligible. It is for the lowest modes that the energy exchange is greatest; these low modes describe the interactions of the gravest potential vorticity induced modes with the boundary buoyancy induced mode. Thus, the possibility of these energy exchanges indicates that there are non-trivial energetic interactions between boundary buoyancy induced dynamics and interior potential vorticity induced dynamics.

\section{Summary}

This dissertation consisted of two parts. The first part, consisting of chapters \ref{Ch-SQG} and \ref{Ch-jets}, explored how variable stratification modifies the geostrophic turbulence of boundary buoyancy anomalies. The second part, consisting of chapters \ref{Ch-modes} and \ref{Ch-QG}, explored the properties of normal modes in the presence of boundary confined restoring forces (e.g., boundary buoyancy anomalies in quasigeostrophy), with the ultimate aim of creating a modal truncation of the quasigeostrophic equations that account for boundary buoyancy anomalies. However, we showed in section \ref{SS-nogo_thm} that such a generalization is not possible as quasigeostrophic modal truncations in the presence of non-isentropic boundaries do not conserve energy.

Chapter \ref{Ch-SQG} showed that the vertical stratification controlled the interaction range of surface buoyancy anomalies. Over vertically decreasing stratification, $N'(z) \leq 0$, surface buoyancy anomalies generate long range velocity fields whereas over vertically increasing stratification, $N'(z) \geq 0$, they generate short range velocity fields. Consequently, the vertical stratification controls the shape of the surface kinetic energy spectrum in surface quasigeostrophic turbulence. 

We therefore suggested that variable stratification may be what accounts for the discrepancy between the expected surface kinetic energy spectrum from surface quasigeostrophic theory and the observed surface kinetic energy spectrum.
Observations and numerical modelling suggest that the surface geostrophic velocity over wintertime extratropical currents are largely induced by surface buoyancy anomalies \citep{isernfontanet_three-dimensional_2008,lapeyre_what_2009,gonzalez-haro_global_2014,qiu_reconstructability_2016,qiu_reconstructing_2020,miracca-lage_can_2022}. For horizontal scales between 1-100 km, we expect a dual cascade: the energetically dominant pycnocline baroclinic instability forces the surface flow at larger scales whereas the faster mixed-layer baroclinic instability forces the surface flow at smaller scales. Uniformly stratified surface quasigeostrophic theory then predicts a surface kinetic energy spectrum between $k^{-1}$ and $k^{-5/3}$ \citep{blumen_uniform_1978}, which is too shallow to be consistent with the observed $k^{-2}$ spectrum \citep{mensa_seasonality_2013,sasaki_impact_2014,callies_seasonality_2015}. However, in chapter \ref{Ch-SQG}, we found that over mixed-layer like stratification, we expect a surface kinetic energy spectrum between $k^{-4/3}$ and $k^{-7/3}$, which is consistent with the $k^{-2}$ spectrum.

In chapter \ref{Ch-jets}, we investigated surface quasigeostrophic dynamics in the presence of a latitudinal buoyancy gradient, which allows for the propagation of westward propagating, surface-trapped Rossby waves. We found a close connection between the spatial locality of the flow and the dispersion of Rossby waves. Over decreasing stratification [$N'(z) \leq 0$] , the flow is spatially non-local, with long range vortices, and highly dispersive Rossby waves. In contrast, over increasing stratification [$N'(z) \geq 0$], the flow is spatially local, with short range vortices, and weakly dispersive Rossby waves. The interaction of Rossby waves with turbulence results in latitudinally inhomogeneous mixing that, in the presence of a sufficiently strong latitudinal buoyancy gradient, results in a staircase structure consisting of homogenized zones of surface buoyancy punctuated by sharp surface buoyancy gradients. Eastward jets are centred at the sharp buoyancy gradients with weaker westward flows in between. The dynamics of these jets depends on the vertical stratification. Over decreasing stratification we obtain straight jets perturbed by highly dispersive, eastward propagating, along jet waves, similar to $\beta$-plane barotropic turbulence. In contrast, over increasing stratification, we obtain meandering jets whose shape evolves in time due to the westward propagation of weakly dispersive along jet waves. In addition, the energy spectrum in the staircase limit depends on the vertical stratification, with a steeper energy spectrum over decreasing stratification [$N'(z)\leq 0$] than over increasing stratification [$N'(z)\geq 0$].

In the next two chapters, chapters \ref{Ch-modes} and \ref{Ch-QG},  we investigated normal modes in the presence of both volume-permeating and boundary-confined restoring forces with the ultimate aim of creating a modal truncation of the quasigeostrophic equations that takes non-isentropic boundaries into account. This aim was motivated by the four mode model of \cite{tulloch_quasigeostrophic_2009}; their model consists of two interior modes (a barotropic and a first baroclinic mode) coupled to a surface quasigeostrophic mode at the upper boundary and a surface quasigeostrophic mode at the lower boundary. However, because these modes are not orthogonal, the model does not conserve energy. To obtain an orthogonal set of modes, we consider linear wave problems with dynamically active boundaries.  In chapter \ref{Ch-modes}, we investigated geophysical waves in the presence of both volume-permeating and boundary-confined restoring forces, with a special emphasis on the mathematical properties of the resulting vertical modes. Then in chapter \ref{Ch-QG}, we applied this formalism to obtain all possible discrete normal modes in quasigeostrophy that diagonalize the energy and the potential enstrophy. However, although we obtained normal modes that account for boundary buoyancy anomalies and form an orthogonal set, the vertical coupling between the modes became dependent on the wavevector. As a consequence, energy is not conserved after any finite modal truncation, and so there are no energy conserving modal truncations of the quasigeostrophic equations that diagonalize the energy and surface potential enstrophy in the presence of non-isentropic boundaries.

\section{Future work} 

\subsection{Geostrophic turbulence with non-isentropic boundaries}
Geostrophic turbulence with isentropic boundaries is characterized by two properties. The first is its energy cycle in which baroclinic instability  extracts energy from a background vertical shear and cascades it downscale towards the deformation radius where it is then transferred into the barotropic mode; the barotropic mode then cascades the energy back to larger horizontal scales where it is then dissipated through bottom drag. The second property is that the barotropic mode dominates the large-scale dynamics, with the time-evolution of the baroclinic mode reduced to the advection of a nearly passive scalar by the barotropic mode. This property is a combined consequence of the long interaction range of the barotropic mode along with the short interaction range of the baroclinic modes. 

The main open question here is how these properties are modified in the presence of boundary buoyancy anomalies, which generate their own velocity fields. First, over sufficiently steep topography, both the upper surface quasigeostrophic flow and the interior quasigeostrophic flow will nearly vanish at the bottom boundary \cite[chapter \ref{Ch-SQG},][]{lacasce_prevalence_2017}. The bottom boundary has its own surface quasigeostrophic flow; for weak bottom friction, the inverse cascade in the bottom surface quasigeostrophic mode can lead to a nearly depth-independent bottom buoyancy induced flow at sufficiently large horizontal scales, and so we recover a barotropic-like mode.
  However, if the inverse cascade in the bottom surface quasigeostrophic mode is arrested by bottom friction before the bottom mode extends significantly upwards into the water column, then we expect the bottom surface quasigeostrophic mode to be nearly decoupled from the flow at the surface and in the interior. 
  The surface and interior flow then are insulated from the direct effects of bottom friction; instead, energy leaks from the surface and interior through interactions with the bottom mode. 
  In this case, we expect the effective damping rate on the surface and interior flows to be determined by nonlinear interactions with the bottom mode instead of by bottom friction.
   Moreover, the absence of a depth-independent flow in this regime then implies that the mode with the longest interaction range at the surface is generally the upper surface quasigeostrophic mode, and it may dominate the large scale dynamics in a similar manner to the barotropic mode. These considerations indicate that bottom topography may alter both the details of the energy cycle in quasigeostrophic turbulence as well as the large-scale dynamics.

 \subsection{The geostrophic turbulence of surface modes}
 If we neglect upper surface buoyancy anomalies, then we can derive a two mode model for quasigeostrophic turbulence in the steep topography limit. 
 As shown in chapter \ref{Ch-intro}, we can think of the two-layer model \eqref{eq:time-barotropic}\textendash\eqref{eq:time-baroclinic} as a two mode truncation of the potential vorticity time-evolution equation \eqref{eq:time_q} over isentropic boundaries. The model consists of two time-evolution equations: one for the barotropic mode
\begin{equation}
	q_0 = \lap \psi_0,
\end{equation} 
and another for the first baroclinic mode
\begin{equation}
	q_1 = (\lap - \lambda_1) \psi_1,
\end{equation}
where $L_1 = 1/\sqrt{\lambda_1}$ is the first mode deformation radius.
However, \cite{lacasce_prevalence_2017} argues that steep bottom topography prevents a barotropic mode from forming at horizontal scales relevant for quasigeostrophic dynamics. In the strong slope limit, we obtain the surface modes instead, which vanish at the bottom. As a consequence, the potential vorticity in the surface modes is
\begin{equation}
	q_n = (\lap -\lambda_n) \psi_n,
\end{equation}
where $\lambda_n>0$ for all $n$ (because there is no barotropic mode). Therefore, the gravest surface mode has a finite interaction range determined by the deformation radius $L_0 = 1/\sqrt{\lambda_0} < \infty$. Truncating the potential vorticity time-evolution equation \eqref{eq:time_q} at $n=1$ gives
	\begin{align}
		\label{eq:time-barotropic_surface}
		\pd{q_0}{t} + \beta \, \pd{\psi_0}{x} + \varepsilon_{000}\, \J{\psi_0}{q_0} +\varepsilon_{010}\, \J{\psi_0}{q_1} +  \varepsilon_{100}\, \J{\psi_1}{q_0}+ \varepsilon_{110}\, \J{\psi_1}{q_1} =0, \\
		\label{eq:time-baroclinic_surface}
		\pd{q_1}{t} + \beta \, \pd{\psi_1}{x} + \varepsilon_{001}\, \J{\psi_0}{q_0} +\varepsilon_{011}\, \J{\psi_0}{q_1} +  \varepsilon_{101}\, \J{\psi_1}{q_0}+ \varepsilon_{111}\, \J{\psi_1}{q_1} =0.
	\end{align}
Previously, the barotropic mode imposed the selection rule \eqref{eq:coulping_barotropic} for modal interactions ($\varepsilon_{mn0}=\delta_{mn}$),  which prevented off-diagonal interactions with the gravest mode (i.e., $\varepsilon_{01n} = \varepsilon_{10n}=0$ for $n=0,1$). With the surface modes, off-diagonal interactions are now possible. Although there are dynamics at the bottom boundary in the steep slope limit, this model filters out these dynamics, and so energy loss to the bottom mode must be parametrized.

\subsection{Energy transfers from weakly nonlinear wave theory}

One way to examine the energy transfers between the potential vorticity induced dynamics and the boundary buoyancy induced dynamics is through weakly nonlinear wave interaction theory \citep{nazarenko_wave_2011}. In this theory, the strength of the interactions between different modes is determined by the vertical coupling coefficient \eqref{eq:vertical_coupling_boundary}. The vertical coupling between different modes was considered in \cite{fu_nonlinear_1980} in the case of isentropic boundaries and  surface-intensified stratification; weakly nonlinear theory predicts the concentration of energy in the first mode, and this prediction was later verified by \cite{smith_scales_2001,smith_scales_2002} using nonlinear simulations. The presence of a bottom slope complicates the problem, with the vertical coupling of wave triads depending on both their propagation directions as well as their wavelength.  However, such an approach may provide an estimate of the energy loss of the interior modes to the bottom-trapped dynamics.

\subsection{Jets and non-isentropic boundaries}

There is also the question of jet formation in the presence of bottom topographic gradients, upper surface buoyancy gradients, and the planetary $\beta$ effect. With isentropic boundaries, the dynamics depend on the value of the bottom friction. For weak bottom friction,  the inverse cascade reaches the barotropic mode and so jet dynamics are similar to $\beta$-plane barotropic turbulence; otherwise, if the inverse cascade is arrested by bottom friction before significant energy reaches the barotropic mode, then jet dynamics are similar to  an equivalent barotropic model with a finite deformation radius. For bottom topographic slopes, the numerical simulations reported in \cite{lacasce_geostrophic_2000} indicate that a bottom slope may result in bottom-trapped along slope structures.  In contrast, the characteristics of surface jets in the presence of both upper surface buoyancy gradients and the planetary $\beta$ effect will depend on the properties of Rossby waves in vertical shear; we anticipate that their propagation direction, their vertical structure, as well as their dispersion will control the dynamics of the resulting jets.


\subsection{Coherent structures in the ocean}

Another question concerns the nature of quasigeostrophic turbulence in the ocean. 
Vertical decompositions of oceanic motion into vertical modes can be misleading. For example, both \cite{wunsch_vertical_1997} and \cite{de_la_lama_vertical_2016} found that the leading empirical orthogonal structure of ocean currents typically is a monotonic function that decays away from the ocean surface and nearly vanishes at the bottom. \cite{wunsch_vertical_1997} interpreted this vertical structure as the sum of a barotropic and baroclinic mode whereas \cite{de_la_lama_vertical_2016} and \cite{lacasce_prevalence_2017} interpreted this vertical structure a surface mode over steep topography. These two interpretations imply distinct dynamics. Wunsch's interpretation implies the existence of coherent barotropic motion whereas the surface mode interpretation does not.

To distinguish between these two interpretations, we can use the spectral proper orthogonal decomposition method to identify coherent structures in the turbulence \citep{taira_modal_2017,towne_spectral_2018}. This method identifies an empirical orthogonal basis for the flow that, for a given number of modes, captures the largest fraction of the flow variance. These modes depend on both space \emph{and} time and are orthogonal with respect to a \emph{spacetime} dependent inner product; consequently, they optimally express the spatiotemporal coherence in the flow \citep{schmidt_guide_2020}. 
One can apply this method to a high resolution numerical ocean model to form a census of three-dimensional oceanic coherent structures. With this approach, we can empirically determine the nature of oceanic geostrophic turbulence.